%% file: main.tex
\pdfoutput=1

\documentclass[acmsmall,screen,authorversion]{acmart}

\acmPrice{}
\acmDOI{10.1145/3290371}
\acmYear{2019}
\copyrightyear{2019}
\acmJournal{PACMPL}
\acmVolume{3}
\acmNumber{POPL}
\acmArticle{58}
\acmMonth{1}
\setcopyright{rightsretained}

\usepackage{booktabs}
\usepackage{subcaption}
\input{preamble_ext}

\def\githuburlDS{\url{https://github.com/Wolff09/TMRexp/releases/tag/POPL19-chkds}}
\def\githuburlSMR{\url{https://github.com/Wolff09/TMRexp/releases/tag/POPL19-chksmr}}

\begin{document}

\title[Decoupling Lock-Free Data Structures from Memory Reclamation for Static Analysis]{Decoupling Lock-Free Data Structures from\linebreak[1] Memory Reclamation for Static Analysis}
\subtitle{Extended Version}

\author{Roland Meyer}
\orcid{0000-0001-8495-671X} 
\affiliation{%
  \institution{TU Braunschweig}
  \country{Germany}
}
\email{roland.meyer@tu-bs.de}
\author{Sebastian Wolff}
\orcid{0000-0002-3974-7713} 
\affiliation{%
  \institution{TU Braunschweig}
  \country{Germany}
}
\email{sebastian.wolff@tu-bs.de}

\input{content/abstract}

%
%
%
%


\maketitle
\input{content/base}

\end{document}

%% file: preamble_ext.tex
\usepackage[utf8]{inputenc}
\usepackage[american]{babel}
\usepackage[T1]{fontenc}
\usepackage{microtype}

\usepackage{paralist}
\usepackage{enumitem}
\usepackage{xcolor}
\usepackage{wrapfig}
\usepackage{tabularx}

\usepackage{todonotes}
\usepackage{wasysym}

\usepackage{tikz}
\usetikzlibrary{automata,arrows,patterns,shapes,decorations,calc}
\tikzset{
    cell/.style={
        draw,
        rectangle split,
        rectangle split horizontal,
        rectangle split parts=2,
        rectangle split part align=base,
        align=center,
        rounded corners,
        text width=4mm,
        node distance=2cm,
    },
    next/.style={
      circle,
      minimum size=1.5mm,
      inner sep=0pt,
      fill=black,
      draw=black,
      node distance=3mm,
    },
    ptrbox/.style={
      draw,
      rectangle,
        text width=1.5mm,
        text height=1.5mm,
        node distance=1cm
    },
    ptr/.style={
      circle,
      minimum size=1.5mm,
      inner sep=0pt,
      fill=black,
      draw=black,
      node distance=0mm,
    }
}

\newcounter{ObserverStateCounter}
\makeatletter
\newcommand{\observerstatelabel}[2]{%
   \protected@write \@auxout {}{\string \newlabel {#1}{{#2}{\thepage}{#2}{#1}{}} }%
   \hypertarget{#1}{}
}
\makeatother
\newcommand{\mkstatename}[1]{\stepcounter{ObserverStateCounter}$s_{\scriptscriptstyle\arabic{ObserverStateCounter}}$\observerstatelabel{#1}{\ensuremath{l_{\arabic{ObserverStateCounter}}}}}

\newcommand{\tikzimagespace}{\\~\phantom{.}\\[-1mm]}

\usepackage{tcolorbox}
\tcbset{
  colback=black!3,
  boxrule=.5pt,
  boxsep=5pt,
  top=0pt,
  bottom=0pt,
  right=0pt,
  left=0pt,
  sharp corners,
}

\usepackage{listings}
\lstdefinestyle{condensed}{
  language=Java,
  basicstyle=\footnotesize\ttfamily,
  keywordstyle=\color{teal}\ttfamily,
  commentstyle=\color{gray}\ttfamily,
    tabsize=2,
  numberbychapter=false,
  morekeywords={NULL,CAS,shared,threadlocal,struct,atomic},
    moredelim=**[is][\bfseries]{@}{@},
    moredelim=**[is][\bfseries\tiny]{§}{§},
    keepspaces=true,
  numbers=left,
  numberstyle=\scriptsize,
  numbersep=6pt,
  firstnumber=last,
  mathescape=true,
  xleftmargin=13pt,
}
\lstdefinestyle{inline}{
    basicstyle=\small\ttfamily,
  keywords={},
    mathescape=true,
}
\newcommand{\code}[1]{\lstinline[style=inline]!#1!}

\newtheorem{assumption}{Assumption}

\usepackage{amsmath,stackengine}
\stackMath
\newcommand\xxrightarrow[1]{\raisebox{-.75pt}{\ensuremath{\smash{\mathrel{%
  \setbox2=\hbox{\stackon{\scriptstyle#1}{\scriptstyle#1}}%
  \stackon[-3.5pt]{%
    \xrightarrow{\makebox[\dimexpr\wd2\relax]{}}%
  }{%
   \scriptstyle#1\,%
  }%
}}}}}

\usepackage{cleveref} 
\crefname{line}{line}{lines}
\Crefname{line}{Line}{Lines}
\crefname{assumption}{assumption}{assumptions}
\Crefname{assumption}{Assumption}{Assumptions}
\crefname{property}{Property}{Properties}
\Crefname{property}{Property}{Properties}

\crefname{aux}{auxiliary}{auxiliaries}
\creflabelformat{aux}{#2(#1)#3}

\newcommand{\CREF}[1]{\protect\texorpdfstring{\Cref{#1}}{}}

\newenvironment{casedistinction}{%
   \let\descriptionlabel\mycasedistinctionlabel
   \description[style=nextline,leftmargin=1pt,itemsep=2.5ex,topsep=1.5ex]
}{%
   \enddescription
}
\setdescription{listparindent=\parindent}
\newcommand{\ad}[1]{\vspace{2mm}\underline{\textit{Ad #1.}}}

\newcounter{auxpropcounter}[section]%
\newcounter{auxpropcounterSave}[section]%
\makeatletter%
\newenvironment{auxiliary}{%
  \setcounter{auxpropcounterSave}{\value{equation}}%
  \setcounter{equation}{\value{auxpropcounter}}%
  \align%
}{%
  \endalign%
  \setcounter{auxpropcounter}{\value{equation}}%
  \setcounter{equation}{\value{auxpropcounterSave}}%
}%
\makeatother

\usepackage{turnstile}
\usepackage{shuffle}

\input{math}

%% file: math.tex
\renewcommand{\emptyset}{\varnothing}
\newcommand{\set}[1]{\{#1\}}
\newcommand{\setcond}[2]{\{#1\mid\,#2\}}

\newcommand{\bnf}{\:\mid\:}

\newcommand{\rangeof}[1]{\mathit{range}(#1)}
\newcommand{\domof}[1]{\mathit{dom}(#1)}

\newcommand{\restrict}[2]{#1|_{#2}}

\newcommand{\aheap}{\mathit{m}}
\newcommand{\aheapp}{\mathit{m}'}

\newcommand{\anexp}{\mathit{e}}
\newcommand{\anexpp}{\mathit{e'}}
\newcommand{\pvars}{\mathit{PVar}}
\newcommand{\apvar}{\mathit{p}}
\newcommand{\apvarp}{\mathit{q}}
\newcommand{\dvars}{\mathit{DVar}}
\newcommand{\advar}{\mathit{x}}
\newcommand{\advarp}{\mathit{y}}
\newcommand{\lhs}{\mathit{lhs}}
\newcommand{\rhs}{\mathit{rhs}}
\newcommand{\psels}{\mathit{PSel}}
\newcommand{\dsels}{\mathit{DSel}}

\newcommand{\retire}{\mathtt{retire}}
\newcommand{\retireof}[1]{\retire(#1)}
\newcommand{\guard}{\mathtt{protect}}
\newcommand{\guardof}[2]{\guard(#1,#2)}
\newcommand{\unguard}{\mathtt{unprotect}}

\newcommand{\leaveQ}{\mathtt{leaveQ}}
\newcommand{\enterQ}{\mathtt{enterQ}}
\newcommand{\free}{\mathtt{free}}
\newcommand{\freeof}[1]{\free(#1)}

\newcommand{\anact}{\mathit{act}}
\newcommand{\anactp}{\mathit{act}'}
\newcommand{\act}[3]{(#1,#2,#3)}
\newcommand{\athread}{t}

\newcommand{\acom}{\mathit{com}}
\newcommand{\anup}{\mathit{up}}
\newcommand{\anupp}{\mathit{up}'}

\newcommand{\vecof}[1]{\bar{#1}}
\newcommand{\trans}[1]{\xxrightarrow{#1}}
\newcommand{\translab}[2]{\ensuremath{#1,\,#2}}
\newcommand{\translabbr}[2]{\ensuremath{#1},\\\ensuremath{#2}}

\newcommand{\evt}[2]{\mathtt{#1}(#2)}

\newcommand{\sem}[1]{[\![#1]\!]}

\newcommand{\gcplussem}[2][1]{\sem{#2}_{\mathit{gc{\scriptscriptstyle+#1}}}}

\newcommand{\asem}[2][\ads]{\sem{#1}_{#2}}
\newcommand{\allsem}[1][\ads]{\sem{#1}_{\adr}}

\newcommand{\onesem}[1][\ads]{\sem{#1}_{\mathit{one}}}
\newcommand{\adrsem}[2][\ads]{\sem{#1}_{\set{#2}}}

\newcommand{\validof}[1]{\mathit{valid}_{#1}}
\newcommand{\invalidof}[1]{\mathit{invalid}_{#1}}

\newcommand{\segval}{\mathtt{seg}}
\newcommand{\anint}{\mathit{i}}
\newcommand{\sel}[2]{#1.\mathtt{#2}}
\newcommand{\psel}[2][\mkern-1mu]{#2.\mathtt{next_{#1}}}

\newcommand{\dsel}[2][\mkern-1mu]{#2.\mathtt{data_{#1}}}
\newcommand{\malloc}{\mathtt{malloc}}

\newcommand{\freesof}[1]{\mathit{frees}({#1})}
\newcommand{\freedof}[1]{\mathit{freed}({#1})}
\newcommand{\freshof}[1]{\mathit{fresh}({#1})}
\newcommand{\pexp}{\mathit{PExp}}
\newcommand{\adexp}{\mathit{dexp}}
\newcommand{\apexp}{\mathit{pexp}}
\newcommand{\apexpp}{\mathit{qexp}}
\newcommand{\dexp}{\mathit{DExp}}
\newcommand{\op}{\mathtt{op}}
\newcommand{\opof}[1]{\op(#1)}
\newcommand{\acond}{\mathit{cond}}

\newcommand{\assert}{\mathtt{assert}}
\newcommand{\assertof}[1]{\assert\ #1}

\newcommand{\anadr}{a}
\newcommand{\anadrp}{b}
\newcommand{\anadrpp}{c}
\newcommand{\advalue}{d}
\newcommand{\avalue}{v}
\newcommand{\avaluep}{w}

\newcommand{\aprog}{\mathit{P}}

\newcommand{\dom}{\ensuremath{\mathit{Dom}}}

\newcommand{\adr}{\mathit{Adr}}
\newcommand{\adrof}[1]{\mathit{adr}(#1)}

\newcommand{\heapcomput}[1]{\aheap_{#1}}
\newcommand{\heapcomputof}[2]{\heapcomput{#1}(#2)}

\newcommand{\afunc}{\mathit{func}}
\newcommand{\afuncof}[1]{\afunc(#1)}

\newcommand{\aval}[1][]{\mathit{v_{#1}}}
\newcommand{\historyof}[2][]{\mathcal{H}_{#1}({#2})}
\newcommand{\specof}[1]{\mathcal{S}({#1}\mkern+.5mu)}
\newcommand{\anevent}{\mathit{evt}}

\newcommand{\aguard}{\mathit{g}}
\newcommand{\alocation}{\mathit{l}}
\newcommand{\alocationp}{\alocation'}
\newcommand{\astate}{\mathit{s}}
\newcommand{\astatep}{\astate'}
\newcommand{\anobs}[1][]{\mathcal{O}_{#1}}
\newcommand{\smrobs}[1][\mkern-1mu\mathit{SMR}]{\anobs[#1]}

\newcommand{\baseobs}{\anobs[\mathit{Base}]}
\newcommand{\hpobs}{\anobs[\mathit{HP}]}
\newcommand{\ebrobs}{\anobs[\mathit{EBR}]}
\newcommand{\gcobs}{\anobs[\mathit{GC}]}

\newcommand{\thehpobs}{\baseobs\times\hpobs}
\newcommand{\theebrobs}{\baseobs\times\ebrobs}
\newcommand{\ahist}{h}
\newcommand{\ahistp}{h'}
\newcommand{\controlof}[1]{\mathit{ctrl}(#1)}

\newcommand{\absevt}[2]{#1(#2)}
\newcommand{\hconcat}{\mkern+1mu.\mkern+2.5mu}
\newcommand{\ansmr}{\mathit{R}}
\newcommand{\adsraw}{\mathit{D}}
\newcommand{\ads}[1][\smrobs]{\adsraw(#1\mkern+.5mu)}

\newcommand{\anovar}{\mathit{u}}
\newcommand{\anovarp}{\mathit{v}}
\newcommand{\anovarpp}{\mathit{w}}

\newcommand{\swap}[1]{\mathit{swap}_{#1}}
\newcommand{\swapadrraw}{\swap{\mathit{adr}}}
\newcommand{\swapadrinvraw}{\swap{\mathit{adr}}^{-1}}
\newcommand{\swapexpraw}{\swap{\mathit{exp}}}
\newcommand{\swapexpinvraw}{\swap{\mathit{exp}}^{-1}}
\newcommand{\swaphistraw}{\swap{\mathit{hist}}}
\newcommand{\swaphistinvraw}{\swap{\mathit{hist}}^{-1}}
\newcommand{\swapadr}[1]{\swapadrraw(#1)}
\newcommand{\swapadrinv}[1]{\swapadrinvraw(#1)}
\newcommand{\swapexp}[1]{\swapexpraw(#1)}
\newcommand{\swapexpinv}[1]{\swapexpinvraw(#1)}
\newcommand{\swaphist}[1]{\swaphistraw(#1)}
\newcommand{\swaphistinv}[1]{\swaphistinvraw(#1)}

\newcommand{\computequiv}{\sim}

\newcommand{\heapequiv}[1][\emptyset]{\simeq_{#1}}
\newcommand{\obsrel}[1][\!]{\lessdot_{#1}}

\newcommand{\freeableof}[2]{\mathcal{F}(#1,\mkern+2mu#2)}

\newcommand{\sigmagoal}{\sigma_{\!\textsc{is}}}

\newcommand{\logicequiv}{\sdststile{}{}}
\newcommand{\avar}{\mathit{var}}
\newcommand{\avarp}{\mathit{var}'}

\newcommand{\anhpobs}[3]{\anobs[\mathit{HP}]({{#1},{#2},{#3}})}

\newcommand{\vadrof}[1]{\adrof{\restrict{\heapcomput{#1}}{\validof{#1}}}}

\newcommand{\allocatedof}[1]{\mathit{allocated}_{#1}}

\newcommand{\renamingof}[3]{#1[#2 / #3]}

\newcommand{\prall}[1]{\mkern+2mu{#1}\mkern+2mu}
\newcommand{\pralll}[1]{\mkern+4mu{#1}\mkern+4mu}

\newcommand{\views}{V}
\newcommand{\sequential}{\mathit{seq}}
\newcommand{\sequentialof}[1]{\sequential(#1)}
\newcommand{\seqof}[1]{\sequentialof{#1}}
\newcommand{\interference}{\mathit{int}}
\newcommand{\interferenceof}[1]{\interference(#1)}
\newcommand{\intof}[1]{\interferenceof{#1}}
\newcommand{\aview}{\nu}
\newcommand{\aviewp}{\omega}

\newcommand{\enter}{\mathtt{enter}}
\newcommand{\enterof}[1]{\enter\:#1}
\newcommand{\exit}{\mathtt{exit}}
\newcommand{\exitof}[1]{\exit(#1)}
\newcommand{\smrcom}{\mathit{smr}}
\newcommand{\satisfies}[2]{#1\models#2}
\newcommand{\inv}{\mathit{inv}}
\newcommand{\invof}[1]{\inv{\mkern+1.5mu:\mkern+1mu}#1}
\newcommand{\ret}{\mathit{ret}}
\newcommand{\retof}[1]{\ret{\mkern+1mu:\mkern+1mu}#1}

\newcommand{\pvarssome}[1]{\pvars^{#1}}
\newcommand{\dvarssome}[1]{\dvars^{#1}}
\newcommand{\pvarsds}{\pvarssome{\adsraw}}
\newcommand{\dvarsds}{\dvarssome{\adsraw}}
\newcommand{\pvarssmr}{\pvarssome{\ansmr}}
\newcommand{\dvarssmr}{\dvarssome{\ansmr}}
\newcommand{\heapcomputsome}[2]{\aheap^{#1}_{#2}}
\newcommand{\heapcomputsomeof}[3]{\heapcomputsome{#1}{#2}(#3)}
\newcommand{\heapcomputall}[1]{\heapcomputsome{*}{#1}}
\newcommand{\heapcomputds}[1]{\heapcomputsome{\adsraw}{#1}}
\newcommand{\heapcomputdsof}[2]{\heapcomputds{#1}(#2)}
\newcommand{\heapcomputsmr}[1]{\heapcomputsome{\ansmr}{#1}}

\newcommand{\controlsome}[2]{\mathit{ctrl}_{#1}^{\mkern+2mu#2}}
\newcommand{\controlallof}[2][]{\controlsome{#1}{*}(#2)}
\newcommand{\controldsof}[2][]{\controlsome{#1}{\adsraw}(#2)}
\newcommand{\controlsmrof}[2][]{\controlsome{#1}{\ansmr}(#2)}
\newcommand{\apc}{\mathit{pc_\adsraw}}
\newcommand{\apcp}{\mathit{pc_\ansmr}}
\newcommand{\bad}{\mathit{Bad}}

%% file: content/abstract.tex

\begin{abstract}
	Verification of concurrent data structures is one of the most challenging tasks in software verification.
	The topic has received considerable attention over the course of the last decade.
	Nevertheless, human-driven techniques remain cumbersome and notoriously difficult while automated approaches suffer from limited applicability.
	The main obstacle for automation is the complexity of concurrent data structures.
	This is particularly true in the absence of garbage collection.
	The intricacy of lock-free memory management paired with the complexity of concurrent data structures makes automated verification prohibitive.

	In this work we present a method for verifying concurrent data structures and their memory management separately.
	We suggest two simpler verification tasks that 
	imply the correctness of the data structure.
	The first task establishes an over-approximation of the reclamation behavior of the memory management.
	The second task exploits this over-approximation to verify the data structure without the need to consider the implementation of the memory management itself.
	To make the resulting verification tasks tractable for automated techniques, we establish a second result.
	We show that a verification tool needs to consider only executions where a single memory location is reused.
	We implemented our approach and were able to verify linearizability of Michael\&Scott's queue and the DGLM queue for both hazard pointers and epoch-based reclamation.
	To the best of our knowledge, we are the first to verify such implementations fully automatically.
\end{abstract}

%% file: content/base.tex

\input{content/introduction}
\input{content/illustration}
\input{content/programs}
\input{content/observer}
\input{content/reduction}
\input{content/evaluation}
\input{content/relatedwork}
\input{content/conclusion}

\begin{acks}
	We thank the POPL'19 reviewers for their valuable feedback and suggestions for improvements.
\end{acks}

\bibliography{bibliography}

\appendix
\input{content/appendix/base}

%% file: content/introduction.tex

\section{Introduction}
\label{sec:introduction}

Data structures are a basic building block of virtually any program.
Efficient implementations are typically a part of a programming language's standard library.
With the advent of highly concurrent computing being available even on commodity hardware, concurrent data structure implementations are needed.
The class of lock-free data structures has been shown to be particularly efficient.
Using fine-grained synchronization and avoiding such synchronization whenever possible results in unrivaled performance and scalability.

Unfortunately, this use of fine-grained synchronization is what makes lock-free data structures also unrivaled in terms of complexity.
Indeed, bugs have been discovered in published lock-free data structures \cite{Michael:1995:CMM:898203,DBLP:conf/spaa/DohertyDGFLMMSS04}.
This confirms the need for formal proofs of correctness.
The de facto standard correctness notion for concurrent data structures is linearizability \cite{DBLP:journals/toplas/HerlihyW90}.
Intuitively, linearizability provides the illusion that the operations of a data structure appear atomically.
Clients of linearizable data structures can thus rely on a much simpler sequential specification.

Establishing linearizability for lock-free data structures is challenging.
The topic has received considerable attention over the past decade (cf. \Cref{sec:related_work}).
For instance, \citet{DBLP:conf/forte/DohertyGLM04} give a mechanized proof of a lock-free queue.
Such proofs require a tremendous effort and a deep understanding of the data structure and the verification technique.
Automated approaches remove this burden.
\Citet{DBLP:conf/vmcai/Vafeiadis10,DBLP:conf/cav/Vafeiadis10}, for instance, verifies singly-linked structures fully automatically.

However, many linearizability proofs rely on a garbage collector.
What is still missing are automated techniques that can handle lock-free data structures with manual memory management.
The major obstacle in automating proofs for such implementations is that lock-free memory management is rather complicated---in some cases even as complicated as the data structure using it.
The reason for this is that memory deletions need to be deferred until all unsynchronized, concurrent readers are done accessing the memory.
Coping with lock-free memory management is an active field of research.
It is oftentimes referred to as \emph{Safe Memory Reclamation (SMR)}.
The wording underlines its focus on safely reclaiming memory for lock-free programs.
This results in the system design depicted in \Cref{fig:system-view}.
The clients of a lock-free data structure are unaware of how it manages its memory.
The data structure uses an allocator to acquire memory, for example, using \code{malloc}.
However, it does not \code{free} the memory itself.
Instead, it delegates this task to an SMR algorithm which defers the \code{free} until it is \emph{safe}.
The deferral can be controlled by the data structure through an API the functions of which depend on the actual SMR algorithm.

\input{content/figures/systemview}
In this paper we tackle the challenge of verifying lock-free data structures which use SMR.
To make the verification tractable, we suggest a compositional approach which is inspired by the system design from \Cref{fig:system-view}.
We observe that the only influence the SMR implementation has on the data structure are the \code{free} operations it performs.
So we introduce SMR specifications that capture when a \code{free} can be executed depending on the history of invoked SMR API functions.
With such a specification at hand, we can verify that a given SMR implementation adheres to the specification.
More importantly, it allows for a compositional verification of the data structure.
Intuitively, we replace the SMR implementation with the SMR specification.
If the SMR implementation adheres to the specification, then the specification over-approximates the frees of the implementation.
Using this over-approximation for verifying the data structure is sound because frees are the only influence the SMR implementation has on the data structure.

Although our compositional approach localizes the verification effort, it leaves the verification tool with a hard task: verifying shared-memory programs with memory reuse.
Our second finding eases this task by taming the complexity of reasoning about memory reuse.
We prove sound that it suffices to consider reusing a single memory location only.
This result relies on data structures being invariant to whether or not memory is actually reclaimed and reused.
Intuitively, this requirement boils down to ABA freedom and is satisfied by data structures from the literature.

To substantiate the usefulness of our approach, we implemented a linearizability checker which realizes the approaches presented in this paper, compositional verification and reduction of reuse to a single address.
Our tool is able to establish linearizability of well-known lock-free data structures, such as \emph{Treiber's stack} \cite{opac-b1015261}, \emph{Michael\&Scott's queue} \cite{DBLP:conf/podc/MichaelS96}, and the \emph{DGLM queue} \cite{DBLP:conf/forte/DohertyGLM04}, when using SMR, like \emph{Hazard Pointers} \cite{DBLP:conf/podc/Michael02} and \emph{Epoch-Based Reclamation} \cite{DBLP:phd/ethos/Fraser04}.
We remark that we needed both results for the benchmarks to go through.
To the best of our knowledge, we are the first to verify lock-free data structures with SMR fully automatically.
We are also the first to automatically verify the DGLM queue with any manual memory management.

Our contributions and the outline of our paper are summarized as follows:
\begin{compactitem}
	\item[§\ref{sec:observers}]
		introduces a means for specifying SMR algorithms and establishes how to perform compositional verification of lock-free data structures and SMR implementations,
	\item[§\ref{sec:reduction}]
		presents a sound verification approach which considers only those executions of a program where at most a single memory location is reused,
	\item[§\ref{sec:evaluation}]
		evaluates our approach on well-known lock-free data structures and SMR algorithms, and demonstrates its practicality.
\end{compactitem}
We illustrate our contributions informally in §\ref{sec:illustration}, introduce the programing model in §\ref{sec:programs}, discuss related work in §\ref{sec:related_work}, and conclude the paper in §\ref{sec:conclusion}.

This technical report extends the conference version \cite{popl} with missing details.
For companion material refer to \url{https://wolff09.github.io/TMRexp/}.

%% file: content/figures/systemview.tex

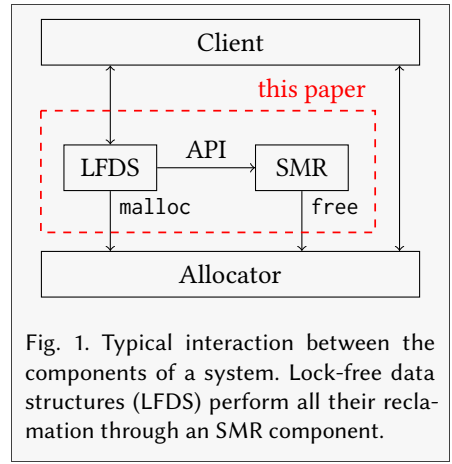
\begin{wrapfigure}{r}{0.42\textwidth} 
	\vspace{-12pt}
	\begin{tcolorbox}
	\center%
	\begin{tikzpicture}
		\node[draw,rectangle,anchor=west,text width=4.8cm,minimum height=6mm,align=center] at (0,.15) (C) { Client };

		\node[draw,rectangle,anchor=west,text width=1cm,minimum height=6mm,align=center] at (.3,-1.5) (D) { LFDS };
		\node[draw,rectangle,anchor=west,text width=1cm,minimum height=6mm,align=center] at (2.85,-1.5) (S) { SMR };

		\node[draw,rectangle,anchor=west,text width=4.8cm,minimum height=6mm,align=center] at (0,-2.9) (A) { Allocator };

		\draw[->] (D) -- node [above]{API} (S);
		\draw[->] (D) -- node [right,pos=.25]{\code{malloc}} (D |- A.north);
		\draw[->] (S) -- node [right,pos=.25]{\code{free}} (S |- A.north);
		\draw[<->] ($(C.south)+(2.25,0)$) -- ($(A.north)+(2.25,0)$);
		\draw[<->] (D) -- (D |- C.south);

		\draw[dashed,color=red,line width=0.65pt] (0,-.75) rectangle (4.45,-2.35);

		\node [anchor=south east,color=red] at (4.45,-.78) (L) {this paper}; 
	\end{tikzpicture}
	\caption{%
		Typical interaction between the components of a system.
		Lock-free data structures (LFDS) perform all their reclamation through an SMR component.
	}
	\label{fig:system-view}
	\end{tcolorbox}
	\vspace{-11pt}
\end{wrapfigure}

%% file: content/illustration.tex

\section{The Verification Approach on an Example}
\label{sec:illustration}

The verification of lock-free data structures is challenging due to their complexity.
One source of this complexity is the avoidance of traditional synchronization.
This leads to subtle thread interactions and imposes a severe state space explosion.
The problem becomes worse in the absence of garbage collection.
This is due to the fact that lock-free memory reclamation is far from trivial.
Due to the lack of synchronization it is typically impossible for a thread to judge whether or not certain memory will be accessed by other threads.
Hence, naively deleting memory is not feasible.
To overcome this problem, programmers employ SMR algorithms.
While this solves the memory reclamation problem, it imposes a major obstacle for verification.
For one, SMR implementations are oftentimes as complicated as the data structure using it.
This makes the already hard verification of lock-free data structures prohibitive.

\input{content/figures/michaelscott_hp}

We illustrate the above problems on the lock-free queue from \citet{DBLP:conf/podc/MichaelS96}.
It is a practical example in that it is used for Java's \href{https://docs.oracle.com/javase/7/docs/api/java/util/concurrent/ConcurrentLinkedQueue.html}{\code{ConcurrentLinkedQueue}} and C++ Boost's \href{https://www.boost.org/doc/libs/1_68_0/boost/lockfree/queue.hpp}{\code{lockfree:$\!$:queue}}, for instance.
The implementation is given in \Cref{fig:michaelscott_hp} (ignore the lines marked by \code{H} for a moment).
The queue maintains a \code{NULL}-terminated singly-linked list of nodes.
New nodes are enqueued at the end of that list.
If the \code{Tail} pointer points to the end of the list, a new node is appended by linking \code{Tail->$\,$next} to the new node.
Then, \code{Tail} is updated to point to the new node.
If \code{Tail} is not pointing to the end of the list, the enqueue operation first moves \code{Tail} to the last node and then appends a new node as before.
The dequeue operation on the other hand removes nodes from the front of the queue.
To do so, it first reads the data of the second node in the queue (the first one is a dummy node) and then swings \code{Head} to the subsequent node.
Additionally, dequeues ensure that \code{Head} does not overtake \code{Tail}.
Hence, a dequeue operation may have to move \code{Tail} towards the end of the list before it moves \code{Head}.
It is worth pointing out that threads read from the queue without synchronization.
Updates synchronize on single memory words using~atomic~Compare-And-Swap~(\code{CAS}).

In terms of memory management, the queue is flawed.
It leaks memory because dequeued nodes are not reclaimed.
A naive fix for this leak would be to uncomment the \code{delete head} statement in \Cref{code:michaelscott_hp:free}.
However, other threads may still hold and dereference pointers to the then deleted node.
Such \emph{use-after-free} dereference are \emph{unsafe}.
In \texttt{C/C++}, for example, the behavior is \emph{undefined} and can result in a system crash due to a \code{segfault}.

\input{content/figures/hpimpl}

To avoid both memory leaks and unsafe accesses, programmers employ SMR algorithms like \emph{Hazard Pointers (HP)} \cite{DBLP:conf/podc/Michael02}.
An example HP implementation is given in \Cref{fig:hpimpldyn}.
Each thread holds a \code{HPRec} record containing two single-writer multiple-reader pointers \code{hp0} and \code{hp1}.
A thread can use these pointers to \code{protect} nodes it will subsequently access without synchronization.
Put differently, a thread requests other threads to defer the deletion of a node by protecting it.
The deferred deletion mechanism is implemented as follows.
Instead of explicitly deleting nodes, threads \code{retire} them.
Retired nodes are stored in a thread-local \code{retiredList} and await reclamation.
Eventually, a thread tries to \code{reclaim} the nodes collected in its list.
Therefore, it reads the hazard pointers of all threads and copies them into a local \code{protectedList}.
The nodes in the intersection of the two lists, \code{retiredList$\,\cap\,$protectedList}, cannot be reclaimed because they are still protected by some thread.
The remaining nodes, \code{retiredList$\,\setminus\,$protectedList}, are reclaimed.

Note that the HP implementation from \Cref{fig:hpimpldyn} allows for threads to join and part dynamically.
In order to join, a thread allocates an \code{HPRec} and appends it to a shared list of such records.
Afterwards, the thread uses the \code{hp0} and \code{hp1} fields of that record to issue protections.
Subsequent \code{reclaim} invocations of any thread are aware of the newly added hazard pointers since \code{reclaim} traverses the shared list of \code{HPRec} records.
To part, threads simply \code{unprotect} their hazard pointers.
They do \emph{not} reclaim their \code{HPRec} record \cite{DBLP:conf/podc/Michael02}.
The reason for this is that reclaiming would yield the same difficulties that we face when reclaiming in lock-free data structures, as discussed before.

To use hazard pointers with Michael\&Scott's queue we have to modify the implementation to \code{retire} dequeued nodes and to \code{protect} nodes that will be accessed without synchronization.
The required modifications are marked by \code{H} in \Cref{fig:michaelscott_hp}.
Retiring dequeued nodes is straight forward, as seen in \Cref{code:michaelscott_hp:retire}.
Successfully protecting a node is more involved.
A typical pattern to do this is implemented by \Cref{code:michaelscott_hp:read-head,code:michaelscott_hp:protect-head,code:michaelscott_hp:check-head}, for instance.
First, a local copy \code{head} of the shared pointer \code{Head} is created, \Cref{code:michaelscott_hp:read-head}.
The node referenced by \code{head} is subsequently protected, \Cref{code:michaelscott_hp:protect-head}.
Simply issuing this protection, however, does not have the intended effect.
Another thread could concurrently execute \code{reclaim} from \Cref{fig:hpimpldyn}.
If the reclaiming thread already computed its \code{protectedList}, i.e., executed \code{reclaim} up to \Cref{code:hpimpldyn:protectedListComputed}, then it does not \emph{see} the later protection and thus may reclaim the node referenced by \code{head}.
The check from \Cref{code:michaelscott_hp:check-head} safeguards the queue from such situations.
It ensures that \code{head} has not been retired since the protection was issued.%
 \footnote{The reasoning is a bit more complicated. We discuss this in more detail in \Cref{sec:reduction:abas}.}
Hence, no concurrent \code{reclaim} considers it for deletion.
This guarantees that subsequent dereferences of \code{head} are safe.
This pattern exploits a simple temporal property of hazard pointers, namely that \emph{a retired node is not reclaimed if it has been protected continuously since before the retire} \cite{DBLP:conf/esop/GotsmanRY13}.

As we have seen, the verification of lock-free data structures becomes much more complex when considering SMR code.
On the one hand, the data structure needs additional logic to properly use the SMR implementation.
On the other hand, the SMR implementation is complex in itself.
It is lock-free (as otherwise the data structure would not be lock-free) and uses multiple lists.

Our contributions make the verification tractable.
First, we suggest a compositional verification technique which allows us to verify the data structure and the SMR implementation separately.
Second, we reduce the impact of memory management for the two new verification tasks.
We found that both contributions are required to automate the verification of data structures like Michael\&Scott's queue with hazard pointers.

\subsection{Compositional Verification}

We propose a compositional verification technique.
We split up the single, monolithic task of verifying a lock-free data structure together with its SMR implementation into two separate tasks: verifying the SMR implementation and verifying the data structure implementation without the SMR implementation.
At the heart of our approach is a specification of the SMR behavior.
Crucially, this specification has to capture the influence of the SMR implementation on the data structure.
Our main observation is that it has \emph{none}, as we have seen conceptually in \Cref{fig:system-view} and practically in \Cref{fig:michaelscott_hp,fig:hpimpldyn}.
More precisely, there is no \emph{direct} influence.
The SMR algorithm influences the data structure only \emph{indirectly} through the underlying allocator: the data structure passes to-be-reclaimed nodes to the SMR algorithm, the SMR algorithm eventually reclaims those nodes using \code{free} of the allocator, and then the data structure can reuse the reclaimed memory with \code{malloc} of the allocator.

In order to come up with an SMR specification, we exploit the above observation as follows.
We let the specification define when reclaiming retired nodes is allowed.
Then, the SMR implementation is correct if the reclamations it performs are a subset of the reclamations allowed by the specification.
For verifying the data structure, we use the SMR specification to over-approximate the reclamation of the SMR implementation.
This way we over-approximate the influence the SMR implementation has on the data structure, provided that the SMR implementation is correct.
Hence, our approach is sound for solving the original verification task.

Towards lightweight SMR specifications, we rely on the insight that SMR implementations, despite their complexity, implement rather simple temporal properties \cite{DBLP:conf/esop/GotsmanRY13}.
We have already seen that hazard pointers implement that \emph{a retired node is not reclaimed if it has been protected continuously since before the retire}.
These temporal properties are incognizant of the actual SMR implementation.
Instead, they reason about those points in time when a call of an SMR function is invoked or returns.
We exploit this by having SMR specifications judge when reclamation is allowed based on the \emph{history} of SMR invocations and returns.

\input{content/figures/observer_hp_illu}

For the actual specification we use observer automata.
A simplified specification for hazard pointers is given in \Cref{fig:observer_hpprotect}.
The automaton $\anhpobs{\athread}{\anadr}{\anint}$ is parametrized by a thread $\athread$, an address $\anadr$, and an integer $\anint$.
Intuitively, $\anhpobs{\athread}{\anadr}{\anint}$ specifies when the $\anint$-th hazard pointer of $\athread$ forces a \code{free} of $\anadr$ to be deferred.
Technically, the automaton reaches an accepting state if a free is performed that should have been deferred.
That is, we let observers specify \emph{bad} behavior.
We found this easier than to formalize the \emph{good} behavior.
For an example, consider the following histories:
\newcommand{\illuevt}[3]{#2(\athread_#1,#3)}%
\begin{align*}
	\ahist_1 &=
	\illuevt{1}{\invof{\guard}}{\anadr,0}
	\hconcat\illuevt{1}{\retof{\guard}}{\anadr,0}
	\hconcat\illuevt{2}{\invof{\retire}}{\anadr}
	\hconcat\illuevt{2}{\retof{\retire}}{\anadr}
	\hconcat\freeof{\anadr}
	~\;\text{and}\\
	\ahist_2 &=
	\illuevt{1}{\invof{\guard}}{\anadr,0}
	\hconcat\illuevt{2}{\invof{\retire}}{\anadr}
	\hconcat\illuevt{1}{\retof{\guard}}{\anadr,0}
	\hconcat\illuevt{2}{\retof{\retire}}{\anadr}
	\hconcat\freeof{\anadr}
	\ .
\end{align*}
History $\ahist_1$ leads $\anhpobs{\athread}{\anadr}{\anint}$ to an accepting state.
Indeed, that $\anadr$ is protected before being retired forbids a free of $\anadr$.
History $\ahist_2$ does not lead to an accepting state because the retire is issued before the protection has returned.
The free of $\anadr$ can be observed if the threads are scheduled in such a way that the protection of $\athread_1$ is not visible while $\athread_2$ computes its \code{retiredList}, as in the scenario described above for motivating why the check at \Cref{code:michaelscott_hp:check-head} is required.

Now, we are ready for compositional verification.
Given an observer, we first check that the SMR implementation is correct wrt. to that observer.
Second, we verify the data structure.
To that end, we strip away the SMR implementation and let the observer execute the frees.
More precisely, we non-deterministically free those addresses which are allowed to be freed according to the observer.

\begin{theorem}[Proven by \Cref{thm:compositionality}]
	Let $\ads[\ansmr]$ be a data structure $\adsraw$ using an SMR implementation $\ansmr$.
	Let $\anobs$ be an observer.
	If $\ansmr$ is correct wrt. $\anobs$ and if $\ads[\anobs]$ is correct, then $\ads[\ansmr]$ is correct.
\end{theorem}

A thorough discussion of the illustrated concepts is given in \Cref{sec:observers}.

\subsection{Taming Memory Management for Verification}

Factoring out the implementation of the SMR algorithm and replacing it with its specification reduces the complexity of the data structure code under scrutiny.
What remains is the challenge of verifying a data structure with manual memory management.
As suggested by \citet{DBLP:conf/tacas/AbdullaHHJR13,DBLP:conf/vmcai/HazizaHMW16} this makes the analysis scale poorly or even intractable.
To overcome this problem, we suggest to perform verification in a simpler semantics.
Inspired by the findings of the aforementioned works we suggest to avoid reallocations as much as possible.
As a second contribution we prove the following \namecref{thm:reuse-conjecture}.

\begin{theorem}[Proven by \Cref{thm:reduction}]
	\label{thm:reuse-conjecture}
	For a sound verification of safety properties it suffices to consider executions where at most a single address is reused.
\end{theorem}

The rational behind this \namecref{thm:reuse-conjecture} is the following.
From the literature we know that avoiding memory reuse altogether is not sound for verification \cite{DBLP:conf/podc/MichaelS96}.
Put differently, correctness under garbage collection (GC) does not imply correctness under manual memory management (MM).
The difference of the two program semantics becomes evident in the ABA problem.
An ABA is a scenario where a pointer referencing address $\anadr$ is changed to point to address $\anadrp$ and changed back to point to $\anadr$ again.
Under MM a thread might erroneously conclude that the pointer has never changed if the intermediate value was not seen due to a certain interleaving.
Typically, the root of the problem is that address $\anadr$ is removed from the data structure, deleted, reallocated, and reenters the data structure.
Under GC, the exact same code does not suffer from this problem.
A pointer referencing $\anadr$ would prevent it from being reused.

From this we learn that avoiding memory reuse does not allow for a sound analysis due to the ABA problem.
So we want to check with little overhead to a GC analysis whether or not the program under scrutiny suffers from the ABA problem.
If not, correctness under GC implies correctness under MM.
Otherwise, we reject the program as buggy.

To check whether a program suffers from ABAs it suffices to check for \emph{first} ABAs.
Fixing the address $\anadr$ of such a first ABA allows us to avoid reuse of any address except $\anadr$ while retaining the ability to detect the ABA.
Intuitively, this is the case because the first ABA is the first time the program reacts differently on a reused address than on a fresh address.
Hence, replacing reallocations with allocations of fresh addresses before the first ABA retains the behavior of the program.

A formal discussion of the presented result is given in \Cref{sec:reduction}.

%% file: content/figures/michaelscott_hp.tex

\begin{figure}
	\center%
\begin{tcolorbox}
\begin{lstlisting}[style=condensed]
struct Node { data_t data; Node* next; };
shared Node* Head, Tail;
atomic init() { Head = new Node(); Head->next = NULL; Tail = Head; }
\end{lstlisting}%
	\begin{minipage}[t]{.46\textwidth}%
\begin{lstlisting}[style=condensed]
void enqueue(data_t input) {
	Node* node = new Node();
	node->data = input;
	node->next = NULL;
	while (true) {
		Node* tail = Tail;
§H§		@protect(tail, 0);@
§H§		@if (tail != Tail) continue;@
		Node* next = tail->next;
		if (tail != Tail) continue;
		if (next != NULL) {
			CAS(&Tail, tail, next);
			continue;
		}
		if (CAS(&tail->next, next, node))
			break
	}
	CAS(&Tail, tail, node);
§H§	@unprotect(0);@
}
\end{lstlisting}%
	\end{minipage}%
	\hfill%
	\begin{minipage}[t]{.475\textwidth}%
\begin{lstlisting}[style=condensed]
data_t dequeue() { $\label[line]{code:michaelscott_hp:dequeue}$
	while (true) {
		Node* head = Head; $\label[line]{code:michaelscott_hp:read-head}$
§H§		@protect(head, 0);@ $\label[line]{code:michaelscott_hp:protect-head}$
§H§		@if (head != Head) continue;@ $\label[line]{code:michaelscott_hp:check-head}$
		Node* tail = Tail; $\label[line]{code:michaelscott_hp:read-tail}$
		Node* next = head->next; $\label[line]{code:michaelscott_hp:read-next}$
§H§		@protect(next, 1);@
		if (head != Head) continue;
		if (next == NULL) return EMPTY;
		if (head == tail) {
			CAS(&Tail, tail, next);
			continue;
		} else {
			data_t output = next->data;
			if (CAS(&Head, head, next)) { $\label[line]{code:michaelscott_hp:cas}$
				// delete head; $\label[line]{code:michaelscott_hp:free}$
§H§				@retire(head);@ $\label[line]{code:michaelscott_hp:retire}$
§H§				@unprotect(0); unprotect(1);@
				return output; $\label[line]{code:michaelscott_hp:dequeue-return}$
}	}	}	}
\end{lstlisting}%
	\end{minipage}%
	\caption{%
		Michael\&Scott's non-blocking queue \cite{DBLP:conf/podc/MichaelS96} extended with hazard pointers \cite{DBLP:conf/podc/Michael02} for safe memory reclamation.
		The modifications needed for using hazard pointers are marked with \code{H}.
		The implementation requires two hazard pointers per thread.
	}
	\label{fig:michaelscott_hp}
\end{tcolorbox}
\end{figure}

%% file: content/figures/hpimpl.tex

\begin{figure}
	\center%
\begin{tcolorbox}
\begin{lstlisting}[style=condensed]
struct HPRec { HPRec* next; Node* hp0; Node* hp1; }
shared HPRec* Records;
threadlocal HPRec* myRec;
threadlocal List<Node*> retiredList;
atomic init() { Records = NULL; }
\end{lstlisting}%
	\begin{minipage}[t]{.46\textwidth}%
\begin{lstlisting}[style=condensed]
void join() {
	myRec = new HPRec();
	while (true) {
		HPRec* rec = Records;
		myRec->next = rec;
		if (CAS(Records, rec, myRec))
			break;
	}
}

void part() {
	unprotect(0); unprotect(1);
}

void protect(Node* ptr, int i) {
	if (i == 0) myRec->hp0 = ptr;
	if (i == 1) myRec->hp1 = ptr;
	assert(false);
}

void unprotect(int i) {
	protect(NULL, i);
}
\end{lstlisting}%
	\end{minipage}%
	\hfill%
	\begin{minipage}[t]{.475\textwidth}%
\begin{lstlisting}[style=condensed]
void retire(Node* ptr) {
	if (ptr != NULL)
		retiredList.add(ptr);
	if (*) reclaim();
}

void reclaim() {
	List<Node*> protectedList;
	HPRec* cur = Records;
	while (cur != NULL) {
		Node* hp0 = cur->hp0;
		Node* hp1 = cur->hp1;
		protectedList.add(hp0);
		protectedList.add(hp1);
		cur = cur->next;
	}
	for (Node* ptr : retiredList) { $\label[line]{code:hpimpldyn:protectedListComputed}$
		if (protectedList.contains(ptr))
			continue;
		retiredList.remove(ptr);
		delete ptr;
	}
}
\end{lstlisting}%
	\end{minipage}%
	\caption{%
		Simplified hazard pointer implementation \cite{DBLP:conf/podc/Michael02}.
		Each thread is equipped with two hazard pointers.
		Threads can dynamically join and part.
		Note that the record used to store a thread's hazard pointers is not reclaimed upon parting.
	}
	\label{fig:hpimpldyn}
\end{tcolorbox}
\end{figure}

%% file: content/figures/observer_hp_illu.tex

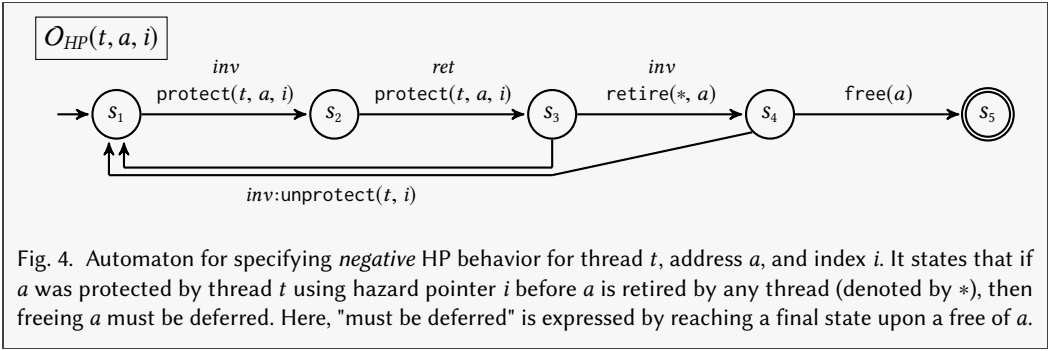
\begin{figure}
\begin{tcolorbox}
	\center%
	\begin{tikzpicture}[->,>=stealth',shorten >=1pt,auto,node distance=2.9cm,thick,initial text={}]
		\node [xshift=-0.2cm,yshift=1cm,draw,thin] {$\anhpobs{\athread}{\anadr}{\anint}$};
		\tikzstyle{every state}=[minimum size=1.5em]
		\tikzset{every edge/.append style={font=\footnotesize}}
		\node[initial,state]    (A)              {\mkstatename{obs:hpillu:init}};
		\node[state]            (E) [right of=A] {\mkstatename{obs:hpillu:invprotected}};
		\node[state]            (B) [right of=E] {\mkstatename{obs:hpillu:retprotected}};
		\node[state]            (C) [right of=B] {\mkstatename{obs:hpillu:invretired}};
		\node[accepting,state]  (D) [right of=C] {\mkstatename{obs:hpillu:final}};
		\coordinate             [below of=A, yshift=+2.1cm]  (X)  {};
		\coordinate             [below of=B, yshift=+2.1cm]  (Y)  {};
		\path
			(A) edge node[align=center] {$\inv$\\$\evt{\guard}{\athread,\anadr,\anint}$} (E)
			(E) edge node[align=center] {$\ret$\\$\evt{\guard}{\athread,\anadr,\anint}$} (B)
			(B) edge node[align=center] {$\inv$\\$\evt{\retire}{*,\anadr}$} (C)
			(C) edge node {$\freeof{\anadr}$} (D)
			(C.south west) edge[-,shorten >=0pt] (Y)
			([xshift=-1mm]X) edge ([xshift=-1mm]A.south)
			(Y) edge[-,shorten >=0pt] node {$\evt{\invof{\unguard}}{\athread,\anint}$} ([xshift=-1mm]X)
			(B) edge[-,shorten >=0pt] ([yshift=1mm]Y.north)
			([yshift=1mm]Y) edge[-,shorten >=0pt] ([xshift=1mm,yshift=1mm]X)
			([xshift=1mm,yshift=1mm]X) edge ([xshift=1mm]A.south)
			;
	\end{tikzpicture}
	\caption{%
		Automaton for specifying \emph{negative} HP behavior for thread $\athread$, address $\anadr$, and index $\anint$.
		It states that if $\anadr$ was protected by thread $\athread$ using hazard pointer $\anint$ before $\anadr$ is retired by any thread (denoted by $*$), then freeing $\anadr$ must be deferred.
		Here, "must be deferred" is expressed by reaching a final state upon a free of $\anadr$.
	}
	\label{fig:observer_hpprotect}
\end{tcolorbox}
\end{figure}

%% file: content/programs.tex

\section{Programs with Safe Memory Reclamation}
\label{sec:programs}

We define shared-memory programs that use an SMR library to reclaim memory.
Here, we focus on the program.
We leave unspecified the internal structure of SMR libraries (typically, they use the same constructs as programs), our development does not depend on it.
We show in \Cref{sec:observers} how to strip away the SMR implementation for verification.

\paragraph{Memory}

A memory $\aheap$ is a partial function
$\aheap:\pexp\uplus\dexp\nrightarrow\adr\uplus\set{\segval}\uplus\dom$ which maps pointer expressions $\pexp$ to addresses from $\adr\uplus\set{\segval}$ and data expressions $\dexp$ to values from $\dom$, respectively.
A pointer expression is either a pointer variable or a pointer selector: $\pexp=\pvars\uplus\psels$.
Similarly, we have $\dexp=\dvars\uplus\dsels$.
The selectors of an address $\anadr$ are $\psel{\anadr}\in\psels$ and $\dsel{\anadr}\in\dsels$.
A generalization to arbitrary selectors is straight forward.
We use $\segval\notin\adr$ to denote undefined/uninitialized pointers. 
We write $\aheap(e)=\bot$ if $e\notin\domof{\aheap}$.
An address $\anadr$ is in-use if it is referenced by some pointer variable or if one of its selectors is defined.
The set of such in-use addresses in $\aheap$ is $\adrof{\aheap}$.
For a formal definition refer to \Cref{appendix:definitions:programs}.

\paragraph{Programs}

\begin{figure}
\begin{tcolorbox}
	\begin{align*}
		\acond\;&::=\phantom{\bnf}
			\apvar = \apvarp
			\bnf
			\apvar \neq \apvarp
			\bnf
			\advar = \advarp
			\bnf
			\advar \neq \advarp
			\bnf
			\advar < \advarp
		\\
		\acom\;&::=\phantom{\bnf}
			\apvar:=\apvarp
			\bnf
			\apvar:=\psel{\apvarp}
			\bnf
			\psel{\apvar}:=\apvarp
			\bnf
			\advar:=\opof{\advar_1,\ldots, \advar_n}
			\bnf
			\advar:=\dsel{\apvarp}
			\\&\phantom{::=}\bnf
			\dsel{\apvar}:=\advar
			\bnf
			\assert\ \acond
			\bnf
			\apvar:=\malloc
			\bnf
			\enterof{\afunc(\vecof{\apvar},\vecof{\advar})}
			\bnf
			\exit
			\bnf
			\smrcom
		\\
		\smrcom\;&::=\phantom{\bnf}
			\freeof{\apvar}
			\bnf
			\dots
	\end{align*}
	\caption{%
		The syntax of atomic commands.
		Here, $\advar,\advarp\in\dvars$ are data variables, and $\apvar,\apvarp\in\pvars$ are pointer variables.
		We write $\vecof{\apvar}$ instead of $\apvar_1,\dots,\apvar_n$ and similarly for $\vecof{\advar}$.
		Besides $\free$, SMR implementations may use not further specified commands $\smrcom$.
	}
	\label{fig:commands}
\end{tcolorbox}
\end{figure}

We consider computations of data structures $\adsraw$ using an SMR library $\ansmr$, written $\ads[\ansmr]$.
A computation $\tau$ is a sequence of actions.
An action $\anact$ is of the form $\anact=\act{\athread}{\acom}{\anup}$.
Intuitively, $\anact$ states that thread $\athread$ executes command $\acom$ resulting in the memory update $\anup$.
An action stems either from executing $\adsraw$ or from executing functions offered by $\ansmr$.

The commands are given in \Cref{fig:commands}.
The commands of $\adsraw$ include assignments, memory accesses, assertions, and allocations with the usual meaning.
We make explicit when a thread enters and exits a library function with $\enter$ and $\exit$, respectively.
That is, we assume that computations are well-formed in the sense that no commands from $\adsraw$ and all commands from $\ansmr$ of a thread occur between $\enter$ and $\exit$.
Besides deallocations we leave the commands of $\ansmr$ unspecified.

The memory resulting from a computation $\tau$, denoted by $\heapcomput{\tau}$, is defined inductively by its updates.
Initially, pointer variables $\apvar$ are uninitialized, $\heapcomputof{\epsilon}{\apvar}=\segval$, and data variables $\advar$ are default initialized, $\heapcomputof{\epsilon}{\advar}=0$.
For a computation $\tau.\anact$ with $\anact=(\athread,\acom,\anup)$ we have $\heapcomput{\tau.\anact}=\heapcomput{\tau}[\anup]$.
With the memory update $\aheap'=\aheap[\anexp\mapsto\aval]$ we mean $\aheap'(\anexp)=\aval$ and $\aheap'(\anexpp)=\aheap(\anexpp)$ for all $\anexpp\neq \anexp$.

\paragraph{Semantics}

The semantics of $\ads[\ansmr]$ is defined relative to a set $X\subseteq\adr$ of addresses that may be reused.
It is the set of allowed executions, denoted by $\asem[{\ads[\ansmr]}]{X}$.
To make the semantics precise, let $\freshof{\tau}$ and $\freedof{\tau}$ be those sets of addresses which have never been allocated and have been freed since their last allocation, respectively.
(See \Cref{appendix:definitions:programs} for formal definitions.)
Then, the semantics is defined inductively.
In the base case we have the empty execution $\epsilon\in\asem[{\ads[\ansmr]}]{X}$.
In the induction step we have $\tau.\anact\in\asem[{\ads[\ansmr]}]{X}$ if one of the following rules applies.
\begin{description}
	\item[(Malloc)]
		If $\anact=(\athread,\apvar:=\malloc,[\apvar\mapsto\anadr,\psel{\anadr}\mapsto\segval,\dsel{\anadr}\mapsto\advalue])$ where $\advalue\in\dom$ is arbitrary and $\anadr\in\adr$ is fresh or available for reuse, that is, $\anadr\in\freshof{\tau}$ or $\anadr\in\freedof{\tau}\cap X$.
	\item[(FreePtr)]
		If $\anact=(\athread,\freeof{\apvar},[\psel{\anadr}\mapsto\bot,\dsel{\anadr}\mapsto\bot])$ with $\heapcomputof{\tau}{\apvar}=\anadr\in\adr$.
	\item[(Enter)]
		If $\anact=(\athread,\enterof{\afuncof{\vecof{\apvar},\vecof{\advar}}},\emptyset)$ with $\vecof{\apvar}=\apvar_1,\dots,\apvar_k$ and $\heapcomputof{\tau}{\apvar_i}\in\adr$ for all $1\leq i\leq k$.
	\item[(Exit)]
		If $\anact=(\athread,\exit,\emptyset)$.
	\item[(Assign1)]
		If $\anact=(\athread,\apvar:=\apvarp,[\apvar\mapsto\heapcomputof{\tau}{\apvarp}])$.
	\item[(Assign2)]
		If $\anact=(\athread,\psel{\apvar}:=\apvarp,[\psel{\anadr}\mapsto\heapcomputof{\tau}{\apvarp}])$ with $\heapcomputof{\tau}{\apvar}=\anadr\in\adr$.
	\item[(Assign3)]
		If $\anact=(\athread,\apvar:=\psel{\apvarp},[\apvar\mapsto\heapcomputof{\tau}{\psel{\anadr}}])$ with $\heapcomputof{\tau}{\apvarp}=\anadr\in\adr$.
	\item[(Assign4)]
		If $\anact=(\athread,\advar:=\opof{\advarp_1,\dots, \advarp_n},[\advar\mapsto\advalue])$ with $\advalue=\opof{\heapcomputof{\tau}{\advarp_1},\dots, \heapcomputof{\tau}{\advarp_n}}$.
	\item[(Assign5)]
		If $\anact=(\athread,\dsel{\apvar}:=\advarp,[\dsel{\anadr}\mapsto\heapcomputof{\tau}{\advarp}])$ with $\heapcomputof{\tau}{\apvar}=\anadr\in\adr$.
	\item[(Assign6)]
		If $\anact=(\athread,\advar:=\dsel{\apvarp},[\advar\mapsto\heapcomputof{\tau}{\dsel{\anadr}}])$ with $\heapcomputof{\tau}{\apvarp}=\anadr\in\adr$.
	\item[(Assert)]
		If $\anact=(\athread,\assertof{\lhs\triangleq\rhs},\emptyset)$ if $\heapcomputof{\tau}{\lhs}\triangleq\heapcomputof{\tau}{\rhs}$.
\end{description}

We assume that computations respect the control flow (program order) of threads.
The control location after $\tau$ is denoted by $\controlof{\tau}$.
We deliberately leave this unspecified as we will express only properties of the form $\controlof{\tau}=\controlof{\sigma}$ to state that after $\tau$ and $\sigma$ the threads can execute the same commands.

%% file: content/observer.tex

\section{Compositional Verification}
\label{sec:observers}

Our first contribution is a compositional verification approach for data structures $\ads[\ansmr]$ which use an SMR library $\ansmr$.
The complexity of SMR implementations makes the verification of data structure and SMR implementation in a single analysis prohibitive.
To overcome this problem, we suggest to verify both implementations independently of each other.
More specifically, we
\begin{inparaenum}[(i)]
	\item introduce a means for specifying SMR implementations, then
	\item verify the SMR implementation $\ansmr$ against its specification, and
	\item verify the data structure $\adsraw$ relative to the SMR specification rather than the SMR implementation.
\end{inparaenum}
If both verification tasks succeed, then the data structure using the SMR implementation, $\ads[\ansmr]$, is correct.

Our approach compares favorably to existing techniques.
Manual techniques from the literature consider a monolithic verification task where both the data structure and the SMR implementaetion are verfied together \cite{DBLP:conf/popl/ParkinsonBO07,DBLP:conf/concur/FuLFSZ10,DBLP:conf/ictac/TofanSR11,DBLP:conf/esop/GotsmanRY13,DBLP:journals/pacmpl/KrishnaSW18}.
Consequently, only simple implementations using SMR have been verified.
Existing automated techniques rely on non-standard program semantics and support only simplistic SMR techniques \cite{DBLP:conf/tacas/AbdullaHHJR13,DBLP:conf/vmcai/HazizaHMW16}.
Refer to \Cref{sec:related_work} for a more~detailed~discussion.

Towards our result, we first introduce observer automata for specifying SMR algorithms.
Then we discuss the two new verification tasks and show that they imply the desired correctness result.

\paragraph{Observer Automata}
An observer automaton $\anobs$ consists of observer locations, observer variables, and transitions.
There is a dedicated initial location and some accepting locations.
Transitions are of the form $\alocation\trans{\translab{f(\vecof{r})}{\aguard}}\alocationp$ with observer locations $\alocation,\alocationp$, event $f(\vecof{r})$, and guard $\aguard$.
Events $\absevt{f}{\vecof{r}}$ consist of a type $f$ and parameters $\vecof{r}=r_1,\dots,r_n$.
The guard is a Boolean formula over equalities of observer variables and the parameters $\vecof{r}$.
An observer state $\astate$ is a tuple $(\alocation,\varphi)$ where $\alocation$ is a location and $\varphi$ maps observer variables to values.
Such a state is initial if $\alocation$ is initial, and similarly accepting if $\alocation$ is accepting.
Then, $(\alocation,\varphi)\trans{f(\vecof{v})}(\alocationp,\varphi)$ is an observer step, if $\alocation\trans{\translab{f(\vecof{r})}{\aguard}}\alocationp$ is a transition and $\varphi(\aguard[\vecof{r}\mapsto\vecof{v}])$ evaluates to $\mathit{true}$.
With $\varphi(\aguard[\vecof{r}\mapsto\vecof{v}])$ we mean $\aguard$ where the formal parameters $\vecof{r}$ are replaced with the actual values $\vecof{v}$ and where the observer variables are replaced by their $\varphi$-mapped values.
Initially, the valuation $\varphi$ is chosen non-deterministically; it is not changed by observer steps.

A \emph{history} $\ahist=f_1(\vecof{v}_1) \ldots f_n(\vecof{v}_n)$ is a sequence of events.
We write $\astate\trans{\ahist}\astatep$ if there are steps $\astate\trans{f_1(\vecof{v}_1)}\cdots\trans{f_n(\vecof{v}_n)}\astatep$.
If $\astatep$ is accepting, then we say that $\ahist$ is accepted by $\astate$. 
We use observers to characterize \emph{bad behavior}.
So we say $\ahist$ is in the specification of $\astate$, denoted by $\ahist\in\specof{\astate}$, if it is not accepted by $\astate$.
Formally, the specification of $\astate$ is the set $\specof{\astate}:=\setcond{\ahist}{\forall\astatep.~\astate\trans{\ahist}\astatep\implies\astatep\text{ not final}}$ of histories that are not accepted by $\astate$.
The specification of $\anobs$ is the set  of histories that are not accepted by any initial state of $\anobs$,
$\specof{\anobs}:=\bigcap\setcond{\specof{\astate}}{\!\astate\text{ initial}}$.
The cross-product $\anobs[1]\times\anobs[2]$ denotes an observer with $\specof{\anobs[1]\times\anobs[2]}=\specof{\anobs[1]}\cap\specof{\anobs[2]}$.

\paragraph{SMR Specifications}

To use observers for specifying SMR algorithms, we have to instantiate appropriately the histories they observe.
Our instantiation crucially relies on the fact that programmers of lock-free data structures rely solely on simple temporal properties that SMR algorithms implement \cite{DBLP:conf/esop/GotsmanRY13}.
These properties are typically incognizant of the actual SMR implementation.
Instead, they allow reasoning about the implementation's behavior based on the temporal order of function invocations and responses.
With respect to our programming model,
$\enter$ and $\exit$ actions provided the necessary means to deduce from the data structure computation how the SMR implementation behaves.

We instantiate observers for specifying SMR as follows.
As event types we use
\begin{inparaenum}[(i)]
	\item $\afunc_1,\dots,\afunc_n$, the functions offered by the SMR algorithm,
	\item $\exit$, and
	\item $\free$.
\end{inparaenum}
The parameters to the events are
\begin{inparaenum}[(i)]
	\item the executing thread and the parameters to the call in case of type $\afunc_i$,
	\item the executing thread for type $\exit$, and
	\item the parameters to the call for type $\free$.
\end{inparaenum}
Here, type $\afunc_i$ represents the corresponding $\enter$ command of a call to $\afunc_i$.
The corresponding $\exit$ event is uniquely defined because both $\afunc_i$ and $\exit$ events contain the executing thread.

\input{content/figures/observer_hp_real}

For an example, consider the hazard pointer specification $\thehpobs$ from \Cref{fig:observer}.
It consists of two observers.
First, $\baseobs$ specifies that no address must be freed that has not been retired.
Second, $\hpobs$ implements the temporal property that no address must be freed if it has been protected continuously since before the retire.
For observer $\thehpobs$ we assume that no address is retired multiple times before being reclaimed (freed) by the SMR implementation.
This avoids handling such \emph{double-retire} scenarios in the observer, keeping it smaller and simpler.
The assumption is reasonable because in a setting where SMR is used a double-retire is the analogue of a double-free and thus avoided.
Our experiments confirm this intuition.

With an SMR specification in form of an observer $\smrobs\mkern+1mu$ at hand, our task is to check whether or not a given SMR implementation $\ansmr$ satisfies this specification.
We do this by converting a computation $\tau$ of $\ansmr$ into its induced history $\historyof{\tau}$ and check if $\historyof{\tau}\in\specof{\smrobs}$.
The induced history $\historyof{\tau}$ is a projection of $\tau$ to $\enter$, $\exit$, and $\free$ actions.
This projection replaces the formal parameters in $\tau$ with their actual values.
For example, $\historyof{\tau.(\athread,\guardof{\apvar}{\advar},\anup)}=\historyof{\tau}.\evt{\guard}{\athread,\heapcomputof{\tau}{\apvar},\heapcomputof{\tau}{\advar}}$.
For a formal definition, consider \Cref{appendix:definitions:programs}.
Then, $\tau$ satisfies $\smrobs$ if $\historyof{\tau}\in\specof{\smrobs}$.
The SMR implementation $\ansmr$ satisfies $\smrobs$ if every possible usage of $\ansmr$ produces a computation that satisfies $\smrobs$.
To generate all such computations, we use a most general client (MGC) for $\ansmr$ which concurrently executes arbitrary~sequences~of~SMR~functions.

\begin{definition}[SMR Correctness]
	An SMR implementation $\ansmr$ \emph{is correct wrt. a specification} $\smrobs$, denoted by $\satisfies{\ansmr}{\smrobs}$, if for every $\tau\in\allsem[MGC(\ansmr)]$ we have $\historyof{\tau}\in\specof{\smrobs}$.
\end{definition}

From the above definition follows the first new verification task: prove that the SMR implementation $\ansmr$ cannot possibly violate the specification $\smrobs$.
Intuitively, this boils down to a reachability analysis of accepting states in the cross-product of $MGC(\ansmr)$ and $\smrobs$.
Since we can understand $\ansmr$ as a lock-free data structure itself, this task is similar to our next one, 
namely verifying the data structure relative to $\smrobs$.
In the remainder of the paper we focus on this second task because it is harder than the first one.
The reason for this lies in that SMR implementations typically do not reclaim the memory they use.
This holds true even if the SMR implementation supports dynamic thread joining and parting \cite{DBLP:conf/podc/Michael02} (cf. \code{part()} from \Cref{fig:hpimpldyn}).
The absence of reclamation allows for a simpler\footnote{%
	In terms of \Cref{sec:reduction}, the absence of reclamation results in SMR implementations being free from pointer races and harmful ABAs since pointers do not become invalid.
	Intuitively, this allows us to combine our results with ones from \citet{DBLP:conf/vmcai/HazizaHMW16} and verify the SMR implementation in a garbage-collected semantics.
} and more efficient analysis.
We confirm this in our experiments where we automatically verify the Hazard Pointer implementation from \Cref{fig:hpimpldyn} wrt. $\thehpobs$.

\paragraph{Compositionality}

The next task is to verify the data structure $\ads[\ansmr]$ avoiding the complexity of $\ansmr$.
We have already established correctness of $\ansmr$ wrt. a specification $\smrobs$.
Intuitively, we now replace $\ansmr$ by $\smrobs$.
Because $\smrobs$ is an observer, and not program code like $\ansmr$, we cannot just execute $\smrobs$ in place of $\ansmr$.
Instead, we remove the SMR implementation from $\ads[\ansmr]$.
The result is $\ads[\epsilon]$ the computations of which correspond to the ones of $\ads[\ansmr]$ with SMR-specific actions between $\enter$ and $\exit$ being removed.
To account for the frees that $\ansmr$ executes, we introduce \emph{environment steps}.
We non-deterministically check for every address $\anadr$ whether or not $\smrobs$ allows freeing it.
If so, we free the address.
Formally, the semantics $\asem{X}$ corresponds to $\asem[{\ads[\epsilon]}]{X}$ as defined in \Cref{sec:programs} plus a new rule for frees from the environment.
\begin{description}
	\item[(Free)]
		If $\tau\in\asem[{\ads[\epsilon]}]{X}$ and $\anadr\in\adr$ can be freed, i.e., $\historyof{\tau}.\freeof{\anadr}\in\specof{\smrobs}$, then we have\linebreak\phantom{.~}$\tau.\anact\in\asem[{\ads[\epsilon]}]{X}$ with $\anact=(\athread,\freeof{\anadr},[\psel{\anadr}\mapsto\bot,\dsel{\anadr}\mapsto\bot])$.
\end{description}
Note that $\freeof{\anadr}$ from the environment has the same update as $\freeof{\apvar}$ from $\ansmr$ if $\apvar$ points to $\anadr$.
With this definition, $\ads$ performs more frees than $\ads[\ansmr]$ provided $\satisfies{\ansmr}{\smrobs}$.

With the semantics of data structures with respect to an SMR specification rather than an SMR implementation set up, we can turn to the main result of this section.
It states that the correctness of $\ansmr$ wrt. $\smrobs$ and the correctness of $\ads$ entail the correctness of the original program $\ads[\ansmr]$.
Here, we focus on the verification of safety properties.
It is known that this reduces to control location reachability \cite{DBLP:conf/lics/Vardi87}.
So we can assume that there is a dedicated \emph{bad} control location in $\adsraw$ the unreachability of which is equivalent to the correctness of $\ads[\ansmr]$.
To establish the result, we require that the interaction between $\adsraw$ and $\ansmr$ follows the one depicted in \Cref{fig:system-view} and discussed on an example in \Cref{sec:illustration}.
That is, the only influence $\ansmr$ has on $\adsraw$ are frees.
In particular, this means that $\ansmr$ does not modify the memory accessed by $\adsraw$.
We found this restriction to be satisfied by many SMR algorithms from the literature.
We believe that our development can be generalized to reflect memory modifications performed by the SMR algorithm.
A proper investigation of the matter, however, is beyond the scope of this paper.

\begin{theorem}[Compositionality]
	\label{thm:compositionality}
	Let $\satisfies{\ansmr}{\smrobs}$.
	If $\allsem$ is correct, so is $\allsem[{\ads[\ansmr]}]$.
\end{theorem}

For the technical details of the above result, see \Cref{appendix:theory-compositionality}.

Compositionality is a powerful tool for verification.
It allows us to verify the data structure and the SMR implementation independently of each other.
Although this simplifies the verification, reasoning about lock-free programs operating on a shared memory remains hard.
In \Cref{sec:reduction} we build upon the above result and propose a sound verification of $\allsem$ which considers only executions reusing a single addresses.

%% file: content/figures/observer_hp_real.tex

\begin{figure}
\begin{tcolorbox}
	\begin{subfigure}[t]{\textwidth}
			\center
			\begin{tikzpicture}[->,>=stealth',shorten >=1pt,auto,node distance=4cm,thick,initial text={}]
				\node [xshift=-0.4cm,yshift=.8cm,draw,thin] {$\baseobs$};
				\tikzstyle{every state}=[minimum size=1.5em]
				\tikzset{every edge/.append style={font=\footnotesize}}
				\node[accepting,state]  (C)                               {\mkstatename{obs:base:final}};
				\node[initial above, state]   (A) [right of=C]                  {\mkstatename{obs:base:init}};
				\node[state]            (B) [right of=A]                  {\mkstatename{obs:base:retired}};
				\path
					([yshift=1mm]B.west) edge node [above]{\translab{\freeof{\anadr}}{\anadr=\anovarp}} ([yshift=1mm]A.east)
					([yshift=-1mm]A.east) edge node [below]{\translab{\evt{\retire}{\athread,\anadr}}{\anadr=\anovarp}} ([yshift=-1mm]B.west)
					(A) edge node [above]{\translab{\freeof{\anadr}}{\anadr=\anovarp}} (C)
					;
			\end{tikzpicture}
			\subcaption{%
				Observer specifying that address $\anovarp$ may be freed only if it has been retired and not been freed since.
			}
			\label{fig:observer:base}
	\end{subfigure}
	\\[4mm]
	\begin{subfigure}[t]{\textwidth}
			\center
			\begin{tikzpicture}[->,>=stealth',shorten >=1pt,auto,node distance=2.6cm,thick,initial text={}]
				\node [xshift=-0.4cm,yshift=.8cm,draw,thin] {$\hpobs$};
				\tikzstyle{every state}=[minimum size=1.5em]
				\tikzset{every edge/.append style={font=\footnotesize}}
				\node[initial,state]    (A)              {\mkstatename{obs:hp:init}};
				\node[state]            (B) [right of=A,xshift=1cm] {\mkstatename{obs:hp:invoked}};
				\node[state]            (E) [right of=B] {\mkstatename{obs:hp:protected}};
				\node[state]            (C) [right of=E] {\mkstatename{obs:hp:retired}};
				\node[accepting,state]  (D) [right of=C] {\mkstatename{obs:hp:final}};
				\coordinate             [below of=A, yshift=+1.75cm]  (X)  {};
				\coordinate             [below of=E, yshift=+1.75cm]  (Y)  {};
				\coordinate             [below of=B, yshift=+1.75cm]  (Z)  {};
				\path
					(A) edge node[align=center] {\translabbr{\evt{\guard}{\athread,\anadr,\anint}}{\athread=\anovar\wedge\anadr=\anovarp\wedge\anint=\anovarpp}} (B)
					(B) edge node[align=center] {\translabbr{\evt{\exit}{\athread}}{\athread=\anovar}} (E)
					(E) edge node[align=center] {\translabbr{\evt{\retire}{\athread,\anadr}}{\anadr=\anovarp}} (C)
					(C) edge node[align=center] {\translabbr{\freeof{\anadr}}{\anadr=\anovarp}} (D)
					(Y) edge[-,shorten >=0pt] node[text width=4.45cm] {\translab{\evt{\guard}{\athread,\anadr,\anint}}{\athread=\anovar\wedge\anadr\neq\anovarp\wedge\anint=\anovarpp}\newline\translab{\evt{\unguard}{\athread,\anint}}{\athread=\anovar\wedge\anint=\anovarpp}} ([xshift=-1.5mm]X)
					(E) edge[-,shorten >=0pt] ([yshift=1mm]Y.north)
					(C.south west) edge[-,shorten >=0pt] (Y)
					([xshift=-1.5mm]X) edge ([xshift=-1.5mm]A.south)
					([yshift=1mm]Y) edge[-,shorten >=0pt] ([xshift=0mm,yshift=1mm]X)
					([xshift=0mm,yshift=1mm]X) edge ([xshift=0mm]A.south)
					(B) edge[-,shorten >=0pt] ([yshift=2mm]Z.north)
					([yshift=2mm]Z) edge[-,shorten >=0pt] ([xshift=1.5mm,yshift=2mm]X)
					([xshift=1.5mm,yshift=2mm]X) edge ([xshift=1.5mm]A.south)
					;
			\end{tikzpicture}
			\subcaption{%
				Observer specifying when HP defers frees.
				A retired cell $\anovarp$ may not be freed if it has been protected continuously by the $\anovarpp$-th hazard pointer of thread $\anovar$ since before being retired.
			}
			\label{fig:observer:hp}
	\end{subfigure}
	\caption{%
		Observer $\thehpobs$ characterizes the histories that violate the Hazard Pointer specification.
		Three observer variables, $\anovar$, $\anovarp$, and $\anovarpp$, are used to observe a thread, an address, and an integer, respectively.
		For better legibility we omit self-loops for every location and every event that is missing an outgoing transition from that location.
	}
	\label{fig:observer}
\end{tcolorbox}
\end{figure}

%% file: content/reduction.tex

\section{Taming Memory Reuse}
\label{sec:reduction}

As a second contribution we demonstrate that \emph{one can soundly verify a data structure $\ads$ by considering only those computations where at most a single cell is reused}.
This avoids the need for a state space exploration of full $\allsem$.
Such explorations suffer from a severe state space explosion.
In fact, we were not able to make our analysis from \Cref{sec:evaluation} go through without this second contribution.
Previous works \cite{DBLP:conf/tacas/AbdullaHHJR13,DBLP:conf/vmcai/HazizaHMW16,DBLP:conf/sas/HolikMVW17} have not required such a result since they did not consider fully fledged SMR implementations like we do.
For a thorough discussion of related work refer to \Cref{sec:related_work}.

Our results are independent of the actual safety property and the actual observer $\smrobs$ specifying the SMR algorithm.
To achieve this, we establish that for every computation from $\allsem$ there is a \emph{similar} computation which reuses only a single address.
We construct such a similar computation by eliding reuse in the original computation.
With elision we mean that we replace in a computation a freed address with a fresh one.
This allows a subsequent allocation to \code{malloc} the elided address fresh instead of reusing it.
Our notion of similarity makes sure that in both computations the threads reach the same control locations.
This allows for verifying safety properties.

The remainder of the section is structured as follows.
\Cref{sec:reduction:similarity} introduces our notion of similarity.
\Cref{sec:reduction:pointer-races} formalizes requirements on $\ads$ such that the notion of similarity suffices to prove the desired result.
\Cref{sec:reduction:abas} discusses how the ABA problem can affect soundness of our approach and shows how to detect those cases.
\Cref{sec:reduction:result} presents the reduction result.

\subsection{Similarity of Computations}
\label{sec:reduction:similarity}

Our goal is to mimic a computation $\tau$ where memory is reused arbitrarily with a computation $\sigma$ where memory reuse is restricted.
As noted before, we want that the threads in $\tau$ and $\sigma$ reach the same control locations in order to verify safety properties of $\tau$ in $\sigma$.
We introduce a \emph{similarity relation} among computations such that $\tau$ and $\sigma$ are similar if they can execute the same actions.
This results in both reaching the same control locations as desired.
However, control location equality alone is insufficient for $\sigma$ to mimic subsequent actions of $\tau$, that is, to preserve similarity for subsequent actions.
This is because most actions involve memory interaction.
Since $\sigma$ reuses memory differently than $\tau$, the memory of the two computations is not equal.
Similarity requires a non-trivial correspondence wrt. the memory.
Towards a formal definition let~us~consider~an~example.

\begin{example}
	\label{ex:reuse-vs-elision}
	Let $\tau_1$ be a computation of a data structure $\ads[\thehpobs]$ using hazard pointers:
	\begin{align*}
		\tau_1 =~&
		(\athread,\apvar:=\malloc,[\apvar\mapsto\anadr,\dots])
		\hconcat
		(\athread,\retireof{\apvar},\emptyset)
		\hconcat
		(\athread,\freeof{\anadr},[\dots])
		\hconcat
		(\athread,\exit,\emptyset)
		\hconcat
		\\&
		(\athread,\apvarp:=\malloc,[\apvarp\mapsto\anadr,\dots])
		\ .
	\end{align*}
	In this computation, thread $\athread$ uses pointer $\apvar$ to allocate address $\anadr$.
	The address is then retired and freed.
	In the subsequent allocation, $\athread$ acquires another pointer $\apvarp$ to $\anadr$; $\anadr$ is reused.
	
	If $\sigma_1$ is a computation where $\anadr$ shall not be reused, then $\sigma_1$ is not able to execute the exact same sequence of actions as $\tau_1$.
	However, it can mimic $\tau_1$ as follows:
	\begin{align*}
		\sigma_1 =~&
		(\athread,\apvar:=\malloc,[\apvar\mapsto\anadrp,\dots])
		\hconcat
		(\athread,\retireof{\apvar},\emptyset)
		\hconcat
		(\athread,\freeof{\anadrp},[\dots])
		\hconcat
		(\athread,\exit,\emptyset)
		\hconcat
		\\&
		(\athread,\apvarp:=\malloc,[\apvarp\mapsto\anadr,\dots])
		\ ,
	\end{align*}
	where $\sigma_1$ coincides with $\tau_1$ up to replacing the first allocation of $\anadr$ with another address $\anadrp$.
	We say that $\sigma_1$ elides the reuse of $\anadr$.
	Note that the memories of $\tau_1$ and $\sigma_1$ differ on $\apvar$ and agree on $\apvarp$.
\end{example}

In the above example, $\apvar$ is a dangling pointer.
Programmers typically avoid using such pointers because it is unsafe.
For a definition of similarity, this practice suggests that similar computations must coincide only on the non-dangling pointers and may differ on the dangling ones.
To make precise which pointers in a computation are dangling, we use the notion of \emph{validity}.
That is, we define a set of valid pointers.
The dangling pointers are then the complement of the valid pointers.
We take this detour because we found it easier to formalize the valid pointers.

Initially, no pointer is valid.
A pointer becomes valid if it receives its value from an allocation or another valid pointer.
A pointer becomes invalid if its referenced memory location is deleted or it receives its value from an invalid pointer.
A deletion of a memory cell makes invalid its pointer selectors and all pointers to that cell.
A subsequent reallocation of that cell makes valid only the receiving pointer; all other pointers to that cell remain invalid.
Assertions of the form $\apvar=\apvarp$ validate $\apvar$ if $\apvarp$ is valid, and vice versa.
A formal definition of the valid pointers in a computation $\tau$, denoted by $\validof{\tau}$, can be found in \Cref{appendix:definitions:reduction}.

\begin{example}[Continued]
	In both $\tau_1$ and $\sigma_1$ from the previous example, the last allocation renders valid pointer $\apvarp$.
	On the other hand, the free to $\anadr$ in $\tau_1$ renders $\apvar$ invalid.
	The reallocation of $\anadr$ does not change the validity of $\apvar$, it remains invalid.
	In $\sigma_1$, address $\anadrp$ is allocated and freed rendering $\apvar$ invalid.
	It remains invalid after the subsequent allocation of $\anadr$.
	That is, both $\tau_1$ and $\sigma_1$ agree on the validity of $\apvarp$ and the invalidity of $\apvar$.
	Moreover, $\tau_1$ and $\sigma_1$ agree on the valuation of the valid $\apvarp$ and disagree (here by chance) on the valuation of the invalid $\apvar$.
\end{example}

The above example illustrates that eliding reuse of memory leads to a different memory valuation.
However, the elision can be performed in such a way that the valid memory is not affected.
So we say that two computations are similar if they agree on the resulting control locations of threads and the valid memory.
The valid memory includes the valid pointer variables, the valid pointer selectors, the data variables, and the data selectors of addresses that are referenced by a valid pointer variable/selector.
Formally, this is a restriction of the entire \mbox{memory to the valid pointers, written $\restrict{\heapcomput{\tau}}{\validof{\tau}}$}.

\begin{definition}[Restrictions]
	A restriction of $\aheap$ to a set $P\subseteq\pexp$, denoted by $\restrict{\aheap}{P}$, is a new $\aheapp$ with
	\(\domof{\aheapp} := P \cup \dvars \cup \setcond{\dsel{\anadr}\in\dexp}{\anadr\in\aheap(P)}\)
	and $\aheap(e) = \aheapp(e)$ for all $e\in\domof{\aheapp}$.
\end{definition}

We are now ready to formalize the notion of similarity among computations.
Two computations are similar if they agree on the control location of threads and the valid memory.

\begin{definition}[Computation Similarity]
	Two computations $\tau$ and $\sigma$ are \emph{similar}, denoted by $\tau\computequiv\sigma$, if $\controlof{\tau}=\controlof{\sigma}$ and $\restrict{\heapcomput{\tau}}{\validof{\tau}}=\restrict{\heapcomput{\sigma}}{\validof{\sigma}}$.
\end{definition}

If two computations $\tau$ and $\sigma$ are similar, then each action enabled after $\tau$ can be mimicked in $\sigma$.
An action $\anact=(\athread,\acom,\anup)$ can be mimicked by another action $\anactp=(\athread,\acom,\anupp)$.
Both actions agree on the executing thread and the executed command but may differ in the memory update.
The reason for this is that similarity does not relate the invalid parts of the memory.
This may give another update in $\sigma$ if $\acom$ involves invalid pointers.

\begin{example}[Continued]
	Consider the following continuation of $\tau_1$ and $\sigma_1$:
	\begin{align*}
		\tau_2 = \tau_1 \hconcat (\athread,\apvar:=\apvar,\anup)
		\qquad\text{and}\qquad
		\sigma_2 = \sigma_1 \hconcat (\athread,\apvar:=\apvar,\anupp)
		\ ,
	\end{align*}
	where we append an assignment of $\apvar$ to itself.
	The prefixes $\tau_1$ and $\sigma_1$ are similar, $\tau_1\computequiv\sigma_1$.
	Nevertheless, the updates $\anup$ and $\anupp$ differ because they involve the valuation of the invalid pointer $\apvar$ which differs in $\tau_1$ and $\sigma_1$.
	The updates are $\anup=[\apvar\mapsto\anadr]$ and $\anupp=[\apvar\mapsto\anadrp]$.
	Since the assignment leaves $\apvar$ invalid, similarity is preserved by the appended actions, $\tau_2\computequiv\sigma_2$.
	We say that $\anactp$ mimics $\anact$.

	Altogether, similarity does not guarantee that the exact same actions are executable.
	It guarantees that every action can be mimicked such that similarity is preserved.
\end{example}

In the above we omitted an integral part of the program semantics.
Memory reclamation is not based on the control location of threads but on an observer examining the history induced by a computation.
The enabledness of a $\free$ is not preserved by similarity.
On the one hand, this is due to the fact that invalid pointers can be (and in practice are) used in SMR calls which lead to different histories.
On the other hand, similar computations end up in the same control location but may perform different sequences of actions to arrive there, for instance, execute different branches of conditionals.
That is, to mimic $\free$ actions we need to correlate the behavior of the observer rather than the behavior of the program.
We motivate the definition of an appropriate relation.

\begin{example}[Continued]
	\label{ex:observer-relation}
	Consider the following continuation of $\tau_2$ and $\sigma_2$:
	\begin{align*}
		\begin{array}{r @{{}\quad{}} c @{{}={}} c @{{}\hconcat{}} c @{{}\hconcat{}} c @{{}\hconcat{}} c @{{}\hconcat{}} c @{{}{}} l}
			& \tau_3 & \tau_2 & (\athread,\guardof{\apvar}{\anint},\emptyset) & (\athread,\exit,\emptyset) & (\athread,\retireof{\apvarp},\emptyset) & (\athread,\exit,\emptyset)
			\\ \text{and} &
			\sigma_3 & \sigma_2 & (\athread,\guardof{\apvar}{\anint},\emptyset) & (\athread,\exit,\emptyset) & (\athread,\retireof{\apvarp},\emptyset) & (\athread,\exit,\emptyset) & \ ,
		\end{array}
	\end{align*}
	where $\athread$ issues a protection and a retirement using $\apvar$ and $\apvarp$, respectively.
	The histories induced by those computations are:
	\begin{align*}
		\begin{array}{r @{{}\quad{}} c @{{}={}} c @{{}\hconcat{}} c @{{}\hconcat{}} c @{{}\hconcat{}} c @{{}\hconcat{}} c @{{}{}} l}
			& \historyof{\tau_3} & \historyof{\tau_1} & \guardof{\athread,\anadr}{\anint} & \exitof{\athread} & \retireof{\athread,\anadr} & \exitof{\athread} &
			\\ \text{and}\quad &
			\historyof{\sigma_3} & \historyof{\sigma_1} & \guardof{\athread,\anadrp}{\anint} & \exitof{\athread} & \retireof{\athread,\anadr} & \exitof{\athread}
			& \ .
		\end{array}
	\end{align*}
	Recall that $\tau_2$ and $\sigma_2$ are similar.
	Similarity guarantees that the events of the $\retire$ call coincide because $\apvarp$ is valid.
	The events of the $\guard$ call differ because the valuations of the invalid $\apvar$ differ.
	That is, SMR calls do not necessarily emit the same event in similar computations.
	Consequently, the observer states after $\tau_3$ and $\sigma_3$ differ.
	More precisely, $\hpobs$ from \Cref{fig:observer:hp} has the following runs from the initial observer state $(\ref{obs:hp:init},\varphi)$ with $\varphi=\set{\anovar\mapsto\athread,\anovarp\mapsto\anadr,\anovarpp\mapsto\anint}$:
	\begin{align*}
		\begin{array}{r @{{}\quad{}} c @{{}{}} c @{{}{}} c @{{}{}} c @{{}{}} c @{{}{}} c @{{}{}} c @{{}{}} c @{{}{}} c @{{}{}} c @{{}{}} c @{{}{}} c}
			&(\ref{obs:hp:init},\varphi)
			&\trans{\historyof{\tau_1}}
			&(\ref{obs:hp:init},\varphi)
			&\trans{\guardof{\athread,\anadr}{\anint}}
			&(\ref{obs:hp:invoked},\varphi)
			&\trans{\exitof{\athread}}
			&(\ref{obs:hp:protected},\varphi)
			&\trans{\retireof{\athread,\anadr}}
			&(\ref{obs:hp:retired},\varphi)
			&\trans{\exitof{\athread}}
			&(\ref{obs:hp:retired},\varphi)
			&\\ \text{and}
			&(\ref{obs:hp:init},\varphi)
			&\trans{\historyof{\sigma_1}}
			&(\ref{obs:hp:init},\varphi)
			&\trans{\guardof{\athread,\anadrp}{\anint}}
			&(\ref{obs:hp:init},\varphi)
			&\trans{\exitof{\athread}}
			&(\ref{obs:hp:init},\varphi)
			&\trans{\retireof{\athread,\anadr}}
			&(\ref{obs:hp:init},\varphi)
			&\trans{\exitof{\athread}}
			&(\ref{obs:hp:init},\varphi)
			&\ .
		\end{array}
	\end{align*}
	This prevents $\anadr$ from being freed after $\tau_3$ (a $\freeof{\anadr}$ would lead to the final state $(\ref{obs:hp:final},\varphi)$ and is thus not enabled) but allows for freeing it after $\sigma_3$.
\end{example}

The above example shows that eliding memory addresses to avoid reuse changes observer runs.
The affected runs involve freed addresses.
Like for computation similarity, we define a relation among computations which captures the observer behavior on the \emph{valid addresses}, i.e., those addresses that are referenced by valid pointers, and ignores all other addresses.
Here, we do not use an equivalence relation.
That is, we do not require observers to reach the exact same state for valid addresses.
Instead, we express that the mimicking $\sigma$ allows for more observer behavior on the valid addresses than the mimicked $\tau$ does.
We define an \emph{observer behavior inclusion} among computations.
This is motivated by the above example.
There, address $\anadr$ is valid because it is referenced by a valid pointer, namely $\apvarp$.
Yet the observer runs for $\anadr$ differ in $\tau_3$ and $\sigma_3$.
After $\sigma_3$ more behavior is possible; $\sigma_3$ can free $\anadr$ while $\tau_3$ cannot.

To make this intuition precise, we need a notion of \emph{behavior on an address}.
Recall that the goal of the desired behavior inclusion is to enable us to mimic $\free$s.
Intuitively, the behavior allowed by $\smrobs$ on address $\anadr$ is the set of those histories that \emph{lead} to a free of $\anadr$.

\begin{definition}[Observer Behavior]
	The behavior allowed by $\smrobs$ on address $\anadr$ after history $\ahist$ is the set
	$\freeableof{\ahist}{\anadr} := \setcond{\ahistp}{\ahist.\ahistp\in\specof{\smrobs}\wedge\freesof{\ahistp}\subseteq{\anadr}}$.
\end{definition}

Note that $\ahistp\in\freeableof{\ahist}{\anadr}$ contains $\free$ events for address $\anadr$ only.
This is necessary because an address may become invalid before being freed if, for instance, the address becomes unreachable from valid pointers.
The mimicking computation $\sigma$ may have already freed such an address while $\tau$ has not, despite similarity.
Hence, the free is no longer allowed after $\sigma$ but still possible after $\tau$.
To prevent such invalid addresses from breaking the desired inclusion on valid addresses, we strip from $\freeableof{\ahist}{\anadr}$ all $\free$s that do not target $\anadr$.
Note that we do not even retain frees of valid addresses here.
This way, only actions which emit an event influence $\freeableof{\ahist}{\anadr}$.

The observer behavior inclusion among computations is defined such that $\sigma$ includes at least the behavior of $\tau$ on the valid addresses.
Formally, the valid addresses in $\tau$ are $\vadrof{\tau}$.

\begin{definition}[Observer Behavior Inclusion]
	Computation $\sigma$ \emph{includes the (observer) behavior} of $\tau$, denoted by $\tau\obsrel\sigma$,
	if $\freeableof{\tau}{\anadr}\subseteq\freeableof{\sigma}{\anadr}$ holds for all $\anadr\in\vadrof{\tau}$.
\end{definition}

\subsection{Preserving Similarity}
\label{sec:reduction:pointer-races}

The development in \Cref{sec:reduction:similarity} is idealized.
There are cases where the introduced relations do not guarantee that an action can be mimicked.
All such cases have in common that they involve the usage of invalid pointers.
More precisely,
\begin{inparaenum}[(i)]
 	\item the computation similarity may not be strong enough to mimic actions that dereference invalid pointers, and
 	\item the observer behavior inclusion may not be strong enough to mimic calls involving invalid pointers.
\end{inparaenum}
For each of those cases we give an example and restrict our development.
We argue throughout this section that our restrictions are reasonable.
Our experiments confirm this.
We begin with the computation similarity.

\begin{example}[Continued]
	Consider the following continuation of $\tau_3$ and $\sigma_3$:
	\begin{align*}
		\begin{array}{r @{{}\quad{}} c @{{}={}} c @{{}\hconcat{}} c @{{}\hconcat{}} c @{{}{}} l}
			& \tau_4 & \tau_3 & (\athread,\psel{\apvarp}:=\apvarp,[\psel{\anadr}\mapsto\anadr]) & (\athread,\psel{\apvar}:=\apvar,[\psel{\anadr}\mapsto\anadr])
			\\ \text{and}\quad &
			\sigma_4 & \sigma_3 & (\athread,\psel{\apvarp}:=\apvarp,[\psel{\anadr}\mapsto\anadr]) & (\athread,\psel{\apvar}:=\apvar,[\psel{\anadrp}\mapsto\anadrp]) & \ .
		\end{array}
	\end{align*}
	The first appended action updates $\psel{\anadr}$ in both computations to $\anadr$.
	Since $\apvarp$ is valid after both $\tau_3$ and $\sigma_3$ this assignment renders valid $\psel{\anadr}$.
	The second action assigns to $\psel{\anadr}$ in $\tau_4$.
	This results in $\psel{\anadr}$ being invalid after $\tau_4$ because the right-hand side of the assignment is the invalid $\apvar$.
	In $\sigma_4$ the second action updates $\psel{\anadrp}$ which is why $\psel{\anadr}$ remains valid.
	That is, the valid memories of $\tau_4$ and $\sigma_4$ differ.
	We have executed an action that cannot be mimicked on the valid memory despite the computations being similar.
\end{example}

The problem in the above example is the dereference of an invalid pointer.
The computation similarity does not give any guarantees about the valuation of such pointers.
Consequently, it cannot guarantee that an action using invalid pointers can be mimicked.
To avoid such problems, we forbid programs to dereference invalid pointers.

The rational behind this is as follows.
Recall that an invalid pointer is dangling.
That is, the memory it references has been freed.
If the memory has been returned to the underlying operating system, then a subsequent dereference is unsafe, that is, prone to a system crash due to a \emph{segfault}.
Hence, such dereferences should be avoided.
The dereference is only safe if the memory is guaranteed to be accessible.
To decide this, the invalid pointer needs to be compared with a definitely valid pointer.
As we mentioned in \Cref{sec:reduction:similarity}, such a comparison renders valid the invalid pointer.
This means that dereferences of invalid pointers are always unsafe.
We let verification fail if unsafe accesses are performed.
That performance-critical and lock-free code is free from unsafe accesses was validated experimentally by \citet{DBLP:conf/vmcai/HazizaHMW16} and is confirmed by our experiments.

\begin{definition}[Unsafe Access]
	A computation $\tau.(\athread,\acom,\anup)$ performs an \emph{unsafe access} if $\acom$ contains $\dsel{\apvar}$ or $\psel{\apvar}$ with $\apvar\notin\validof{\tau}$.
\end{definition}

Forbidding unsafe accesses makes the computation similarity strong enough to mimic all desired actions.
A discussion of cases where the observer behavior inclusion cannot be preserved is in order.
We start with an example.

\begin{example}[Continued]
	\label{ex:breaking-obsrel}
	Consider the following continuations of $\tau_1$ and $\sigma_1$ from \Cref{ex:reuse-vs-elision}:
	\begin{align*}
		\begin{array}{r @{{}\quad{}} c @{{}={}} c @{{}\hconcat{}} c @{{}\quad\text{with}\quad{}} c @{{}\,=\,{}} c @{{}\hconcat{}} c @{{}{}} l}
			& \tau_5 & \tau_1 & (\athread,\retireof{\apvar},\emptyset) & \historyof{\tau_5} & \historyof{\tau_1} & \retireof{\athread,\anadr}
			\\ \text{and}\quad &
			\sigma_5 & \sigma_1 & (\athread,\retireof{\apvar},\emptyset) & \historyof{\sigma_5} & \historyof{\sigma_1} & \retireof{\athread,\anadrp}
			& \ .
		\end{array}
	\end{align*}
	The observer behavior of $\tau_1$ is included in $\sigma_1$, $\tau_1\obsrel\sigma_1$.
	After $\tau_5$ a deletion of $\anadr$ is possible because it was retired.
	After $\sigma_5$ a deletion of $\anadr$ is prevented by $\baseobs$ because $\anadr$ was not retired.
	Technically, we have $\freeof{\anadr}\in\freeableof{\tau_5}{\anadr}$ and $\freeof{\anadr}\notin\freeableof{\sigma_5}{\anadr}$.
	However, $\anadr$ is a valid address because it is referenced by the valid pointer $\apvarp$.
	That is, the behavior inclusion among $\tau_1$ and $\sigma_1$ is not preserved by the subsequent action.
\end{example}

The above example showcases that calls to the SMR algorithm can break the observer behavior inclusion.
This is the case because an action can emit different events in similar computations.
The event emitted by an SMR call differs only if it involves invalid pointers.

The naive solution would prevent using invalid pointers in calls altogether.
In practice, this is too strong a requirement.
As discussed in \Cref{sec:illustration}, a common pattern for protecting an address (cf.~\Cref{fig:michaelscott_hp}, \Cref{code:michaelscott_hp:read-head,code:michaelscott_hp:protect-head,code:michaelscott_hp:check-head}) is to
\begin{inparaenum}[(i)]
	\item read a pointer $\apvar$ into a local variable $\apvarp$,
	\item issue a protection using $\apvarp$, and
	\item repeat the process if $\apvar$ and $\apvarp$ do not coincide.
\end{inparaenum}
After reading into $\apvarp$ and before protecting $\apvarp$ the referenced memory may be freed.
Hence, the protection is prone to use invalid pointers.
Forbidding such protections would render our theory inapplicable to lock-free data structures using hazard pointers.

To fight this problem, we forbid only those calls involving invalid pointers which are prone to \emph{break} the observer behavior inclusion.
Intuitively, this is the case if a call with the values of the invalid pointers replaced arbitrarily allows for more behavior on the valid addresses than the original call.
Actually, we keep precise the address the behavior of which is under consideration.
This allows us to support more scenarios where invalid pointers are used.

\begin{definition}[Racy SMR Calls]
	\label{def:racy-call}
	A computation $\tau.\anact$ with $\anact=(\athread,\enterof{\afuncof{\vecof{\apvar},\vecof{\advar}}},\emptyset)$, ${\historyof{\tau}=\ahist}$, ${\heapcomputof{\tau}{\vecof{\apvar}}=\vecof{\anadr}}$, and $\heapcomputof{\tau}{\vecof{\advar}}=\vecof{\advalue}$ performs a \emph{racy call} if:
	\begin{align*}
		\exists\,\anadrpp~
		\exists\,\vecof{\anadrp}.~~
		\big(
			\forall i.~ 
			(
			\anadr_i=\anadrpp\vee\apvar_i\in\validof{\tau}
			)
			\implies
			\anadr_i=\anadrp_i
		\big)
		\wedge
		\freeableof{\ahist.\afunc(\athread,\vecof{\anadrp},\vecof{\advalue})}{\anadrpp}
		\not\subseteq
		\freeableof{\ahist.\afunc(\athread,\vecof{\anadr},\vecof{\advalue})}{\anadrpp}
		\ .
	\end{align*}
\end{definition}

It follows immediately that calls containing valid pointers only are not racy.
In practice, $\retire$ is always called using valid pointers, thus avoiding the problematic scenario from \Cref{ex:breaking-obsrel}.
The rational behind this is that freeing invalid pointers may lead to system crashes, just like dereferences.
For $\guard$ calls of hazard pointers one can show that they never race.
We have already seen this in \Cref{ex:observer-relation}.
There, a call to $\guard$ with invalid pointers has not caused the observer relation to break.
Instead, the mimicking computation ($\sigma_3$) could perform (strictly) more frees than the computation it mimicked ($\tau_3$).

We uniformly refer to the above situations where the usage of an invalid pointer can break the ability to mimic an actions as a \emph{pointer race}.
It is a race indeed because the usage and the free of a pointer are not properly synchronized.

\begin{definition}[Pointer Race]
	\label{definition:RPR}
	A computation $\tau.\anact$ is a \emph{pointer race} if $\anact$ performs
	\begin{inparaenum}[(i)]
		\item an unsafe access, or
		\item a racy SMR call.
	\end{inparaenum}
\end{definition}

With pointer races we restrict the class of supported programs.
The restriction to pointer race free programs is reasonable in that we can handle common non-blocking data structures from the literature as shown in our experiments.
Since we want to give the main result of this section in a general fashion that does not rely on the actual observer used to specify the SMR implementation, we have to restrict the class of supported observers as well.

We require that the observer supports the elision of reused addresses, as done in \Cref{ex:reuse-vs-elision}.
Intuitively, elision is a two-step process the observer must be insensitive to.
First, an address $\anadr$ is replaced with a fresh address $\anadrp$ upon an allocation where $\anadr$ should be reused but cannot.
In the resulting computation, $\anadr$ is fresh and thus the allocation can be performed without reusing $\anadr$.
The process of replacing $\anadr$ with $\anadrp$ must not affect the behavior of the observer on addresses other than $\anadr$ and $\anadrp$.
Second, the observer must allow for more behavior on the fresh address than on the reused address.
This is required to preserve the observer behavior inclusion because the allocation of $\anadr$ renders it a valid address.

Additionally, we require a third property:
the observer behavior on an address must not be influenced by frees to another address.
This is needed because computation similarity and behavior inclusion do not guarantee that frees of invalid addresses can be mimicked, as discussed before.
Since such frees do not affect the valid memory they need not be mimicked.
The observer has to allow us to do so, that is, simply \emph{skip} such frees when mimicking a computation.

For a formal definition of our intuition we write $\renamingof{\ahist}{\anadr}{\anadrp}$ to denote the history that is constructed from $\ahist$ by replacing every occurrence of $\anadr$ with $\anadrp$ (cf. \Cref{appendix:definitions:reduction}).

\begin{definition}[Elision Support]
	\label{def:elision-support}
	The observer $\smrobs$ \emph{supports elision of memory reuse} if
	\begin{compactenum}[(i)]
		\item \label[property]{def:elision-support:replace} for all $\ahist_1,\ahist_2,\anadr,\anadrp,\anadrpp$ with $\anadr\neq\anadrpp\neq\anadrp$ and $\ahist_2=\renamingof{\ahist_1}{\anadr}{\anadrp}$ we have $\freeableof{\ahist_1}{\anadrpp}=\freeableof{\ahist_2}{\anadrpp}$,
		\item \label[property]{def:elision-support:fresh} for all $\ahist_1,\ahist_2,\anadr,\anadrp$ with $\freeableof{\ahist_1}{\anadr}\pralll{\subseteq}\freeableof{\ahist_2}{\anadr}$ and $\anadrp\pralll{\in}\freshof{\ahist_2}$ we have $\freeableof{\ahist_1}{\anadrp}\pralll{\subseteq}\freeableof{\ahist_2}{\anadrp}$, and
		\item \label[property]{def:elision-support:free} for all $\ahist,\anadr,\anadrp$ with $\anadr\neq\anadrp$ we have $\freeableof{\ahist.\freeof{\anadr}}{\anadrp}=\freeableof{\ahist}{\anadrp}$.
	\end{compactenum}
\end{definition}

We found this definition practical in that the observers we use for our experiments support elision (cf. \Cref{sec:smr-in-eval}).
The hazard pointer observer $\thehpobs$, for instance, supports elision.

\subsection{Detecting ABAs}
\label{sec:reduction:abas}

So far we have introduced restrictions, namely pointer race freedom and elision support, to rule out cases where our idea of eliding memory reuse would not work, that is, break the similarity or behavior inclusion.
If those restrictions were strong enough to carry out our development, then we could remove any reuse from a computation and get a similar one where no memory is reused.
That the resulting computation does not reuse memory means, intuitively, that it is executed under garbage collection.
As shown in the literature \cite{DBLP:conf/podc/MichaelS96}, the ABA problem is a subtle bug caused by manual memory management which is prevented by garbage collection.
So eliding all reuses jeopardizes soundness of the analysis---it could miss ABAs which result in a safety violation.
With this observation,
we elide all reuses except for one address per computation.
This way we analyse a semantics that is close to garbage collection, can detect ABA problems, and is much simpler than full $\allsem$.

The semantics that we suggest to analyse is $\onesem:=\bigcup_{\anadr\in\adr}\adrsem{\anadr}$.
It is the set of all computations that reuse at most a single address.
A single address suffices to detect the ABA problem.
The ABA problem manifests as an assertion of the form $\assertof{\apvar=\apvarp}$ where the addresses held by $\apvar$ and $\apvarp$ coincide but stem from different allocations.
That is, one of the pointers has received its address, the address was freed and then reallocated, before the pointer is used in the assertion.
Note that this implies that for an assertion to be ABA one of the involved pointers must be invalid.
Pointer race freedom does not forbid this.
Nor do we want to forbid such assertions.
In fact, most programs using hazard pointers contain ABAs.
They are written in a way that ensures that the ABA is \emph{harmless}.
Consider an example.

\begin{example}[ABAs in Michael\&Scott's queue using hazard pointers]
	Consider the code of Michael\&Scott's queue from \Cref{fig:michaelscott_hp}.
	More specifically, consider \Cref{code:michaelscott_hp:read-head,code:michaelscott_hp:protect-head,code:michaelscott_hp:check-head}.
	In \Cref{code:michaelscott_hp:read-head} the value of the shared pointer \code{Head} is read into the local pointer \code{head}.
	Then, a hazard pointer is used in \Cref{code:michaelscott_hp:protect-head} to protect \code{head} from being freed.
	In between reading and protecting \code{head}, its address could have been deleted, reused, and reentered the queue.
	That is, when executing \Cref{code:michaelscott_hp:check-head} the pointers \code{Head} and \code{head} can coincide although the \code{head} pointer stems from an earlier allocation.
	This scenario is an ABA.
	Nevertheless, the queue's correctness is not affected by this ABA.
	The ABA prone assertion is only used to guarantee that the address protected in \Cref{code:michaelscott_hp:protect-head} is indeed protected after \Cref{code:michaelscott_hp:check-head}.
	Wrt. to the observer $\hpobs$ from \Cref{fig:observer}, the assertion guarantees that the protection was issued before a retirement (after the latest reallocation) so that $\hpobs$ is guaranteed to be in $\ref{obs:hp:protected}$ and thus prevent future retirements from freeing the protected memory.
	The ABA does not void this guarantee, it is harmless.
\end{example}

The above example shows that lock-free data structures may perform ABAs which do not affect their correctness.
To soundly verify such algorithms, our approach is to detect every ABA and decide whether it is harmless indeed.
If so, our verification is sound.
Otherwise, we report to the programmer that the implementation suffers from a harmful ABA problem.

A discussion of how to detect ABAs is in order.
Let $\tau\in\allsem$ and $\sigma\in\adrsem{\anadr}$ be two similar computations.
Intuitively, $\sigma$ is a computation which elides the reuses from $\tau$ except for some address $\anadr$.
The address $\anadr$ can be used in $\sigma$ in exactly the same way as it is used in $\tau$.
Let $\anact=(\athread,\assertof{\apvar=\apvarp},\emptyset)$ be an ABA assertion which is enabled after $\tau$.
To detect this ABA under $\adrsem{\anadr}$ we need $\anact$ to be enabled after $\sigma$.
We seek to have $\sigma.\anact\in\adrsem{\anadr}$.
This is not guaranteed.
Since $\anact$ is an ABA it involves at least one invalid pointer, say~$\apvar$.
Computation similarity does not guarantee that $\apvar$ has the same valuation in both $\tau$ and $\sigma$.
However, if $\apvar$ points to $\anadr$ in $\tau$, then it does so in $\sigma$ because $\anadr$ is (re)used in $\sigma$ in the same way as in $\tau$.
Thus, we end up with $\heapcomputof{\tau}{\apvar}=\heapcomputof{\sigma}{\apvar}$ although $\apvar$ is invalid.
In order to guarantee this, we introduce a \emph{memory equivalence} relation.
We use this relation to precisely track how the reusable address $\anadr$ is used.

\begin{definition}[Memory Equivalence]
	\label{def:mem-equiv}
	Two computations $\tau$ and $\sigma$ are \emph{memory equivalent wrt. address}~$\anadr$, denoted by $\tau\heapequiv[\anadr]\sigma$, if
	\begin{align*}
		&
		\forall\, \apvar\in\pvars.~ \heapcomputof{\tau}{\apvar}=\anadr \iff \heapcomputof{\sigma}{\apvar}=\anadr
		\\\text{ and }\quad&
		\forall\, \anadrp\in\heapcomputof{\tau}{\validof{\tau}}.~ \heapcomputof{\tau}{\psel{\anadrp}}=\anadr \iff \heapcomputof{\sigma}{\psel{\anadrp}}=\anadr
		\\\text{ and }\quad&
		\anadr\in\freshof{\tau}\cup\freedof{\tau}\iff\anadr\in\freshof{\sigma}\cup\freedof{\sigma}
		\\\text{ and }\quad&
		\freeableof{\tau}{\anadr}\subseteq\freeableof{\sigma}{\anadr}
		\ .
	\end{align*}
\end{definition}

The first line in this definition states that the same pointer variables in $\tau$ and $\sigma$ are pointing to $\anadr$.
Similarly, the second line states this for the pointer selectors of valid addresses.
We have to exclude the invalid addresses here because $\tau$ and $\sigma$ may differ on the in-use addresses due to eliding reuse.
The third line states that $\anadr$ can be allocated in $\tau$ iff it can be allocated in $\sigma$.
The last line states that the observer allows for more behavior on $\anadr$ in $\sigma$ than in $\tau$.
These properties combined guarantee that $\sigma$ can mimic actions of $\tau$ involving $\anadr$ no matter if invalid pointers are used.

The memory equivalence lets us detect ABAs in $\onesem$.
Intuitively, we can only detect \emph{first} ABAs because we allow for only a single address to be reused.
Subsequent ABAs on different addresses cannot be detected.
To detect ABA sequences of arbitrary length, an arbitrary number of reusable addresses is required.
To avoid this, i.e., to avoid an analysis of full $\allsem$, we formalize the idea of \emph{harmless ABA} from before.
We say that an ABA is harmless if executing it leads to a system state which can be explored without performing the ABA.
That the system state can be explored without performing the ABA means that every ABA is also a first ABA.
Thus, any sequence of ABAs is explored by considering only first ABAs.
Note that this definition is independent of the actual correctness notion.

\begin{definition}[Harmful ABA]
	\label{def:harmful-ABA}
	$\onesem$ is \emph{free from harmful ABAs} if the following holds:
	\begin{align*}
		&\forall\,\sigma_\anadr.\anact\in\adrsem{\anadr}~
		\forall\,\sigma_\anadrp\in\adrsem{\anadrp}~
		\exists\,\sigma_\anadrp'\in\adrsem{\anadrp}.~~
		\\
		&\qquad\sigma_\anadr\computequiv\sigma_\anadrp
		\:\wedge\:
		\anact=(\_,\assertof{\_},\_)
		\:\implies\:
		\sigma_\anadr.\anact\computequiv\sigma_\anadrp'
		\:\wedge\:
		\sigma_\anadrp\heapequiv[\anadrp]\sigma_\anadrp'
		\:\wedge\:
		\sigma_\anadr.\anact\obsrel\sigma_\anadrp'
		\ .
	\end{align*}
\end{definition}

To understand how the definition implements our intuition, consider some $\tau.\anact\in\allsem$ where $\anact$ performs an ABA on address $\anadr$.
Our goal is to mimic $\tau.\anact$ in $\adrsem{\anadrp}$, that is, we want to mimic the ABA without reusing address $\anadr$ (for instance, to detect subsequent ABAs on address $\anadrp$).
Assume we are given some $\sigma_\anadrp\in\adrsem{\anadrp}$ which is similar and memory equivalent wrt. $\anadrp$ to $\tau$.
This does not guarantee that $\anact$ can be mimicked after $\sigma_\anadrp$; the ABA may not be enabled because it involves invalid pointers the valuation of which may differ in $\tau$ and $\sigma_\anadrp$.
However, we can construct a computation $\sigma_\anadr$ which is similar and memory equivalent wrt. $\anadr$ to $\tau$.
After $\sigma_\anadr$ the ABA is enabled, i.e., we have $\sigma_\anadr.\anact\in\adrsem{\anadr}$.
For those two computations $\sigma_\anadr.\anact$ and $\sigma_\anadrp$ we invoke the above definition.
It yields another computation $\sigma_\anadrp'\in\adrsem{\anadrp}$ which, intuitively, coincides with $\sigma_\anadrp$ but where the ABA has already been executed.
Put differently, $\sigma_\anadrp'$ is a computation which mimics the execution of $\anact$ after $\sigma_\anadrp$ (although $\anact$ is not enabled).

\begin{example}[Continued]
	Consider the following computation of Michael\&Scott's queue:
	\begin{align*}
		\tau =~ & \tau_6 \hconcat (\athread,\mathtt{head}:=\mathtt{Head},[\mathtt{head}\mapsto\anadr]) \hconcat \tau_7 \hconcat \freeof{\anadr} \hconcat \tau_8 \hconcat \\& (\athread,\enterof{\guardof{\mathtt{head}}{0}},\emptyset) \hconcat (\athread,\exit,\emptyset) \hconcat (\athread,\assertof{\mathtt{head}=\mathtt{Head}},\emptyset) \ .
	\end{align*}
	This computation resembles a thread $\athread$ executing \Cref{code:michaelscott_hp:read-head,code:michaelscott_hp:protect-head,code:michaelscott_hp:check-head} while an interferer frees address $\anadr$ referenced by $\mathtt{head}$.
	Note that the $\assert$ resembles the conditional from \Cref{code:michaelscott_hp:check-head} and states that the condition evaluates to \emph{true}.
	That is, the last action in $\tau$ is an ABA.

	Reusing address $\anadr$ allows us to mimic $\tau$ with a computation $\sigma_\anadr\in\adrsem{\anadr}$ which detects the ABA.
	For simplicity, assume $\tau=\sigma_\anadr$.
	Mimicking $\tau$ with another computation $\sigma_\anadrp\in\adrsem{\anadrp}$ is not possible.
	In $\sigma_\anadrp\in\adrsem{\anadrp}$ the first allocation of $\anadr$ will be elided such that the assertion is not enabled.
	However, rescheduling the actions gives rise to $\sigma_\anadrp'\in\adrsem{\anadrp}$ which coincides with $\sigma_\anadrp$ but where the assertion has also been executed:
	\begin{align*}
		\sigma_\anadrp' =~ & \tau_6 \hconcat \tau_7 \hconcat \freeof{\anadr} \hconcat \tau_8 \hconcat (\athread,\mathtt{head}:=\mathtt{Head},[\mathtt{head}\mapsto\anadr]) \hconcat \\& (\athread,\enterof{\guardof{\mathtt{head}}{0}},\emptyset) \hconcat (\athread,\exit,\emptyset) \hconcat (\athread,\assertof{\mathtt{head}=\mathtt{Head}},\emptyset) \ .
	\end{align*}
	Requiring the existence of such a $\sigma_\anadrp'$ guarantees that an analysis can \emph{see past} ABAs on address $\anadr$, although $\anadr$ is not reused.
\end{example}

A key aspect of the above definition is that checking for harmful ABAs can be done in the simpler semantics $\onesem$.
Altogether, this means that we can rely on $\onesem$ for both the actual analysis and a soundness (absence of harmful ABAs) check.
Our experiments show that the above definition is practical.
There were no harmful ABAs in the benchmarks we considered.

\subsection{Reduction Result}
\label{sec:reduction:result}

We show how to exploit the concepts introduced so far to soundly verify safety properties in the simpler semantics $\onesem$ instead of full $\allsem$.

\begin{lemma}
	Let $\onesem$ be free from pointer races and harmful ABAs, and let $\smrobs$ support elision.
	For all $\tau\prall{\in}\allsem$ and $\anadr\prall{\in}\adr$ there is $\sigma\prall{\in}\adrsem{\anadr}$ with $\tau\prall{\computequiv}\sigma$, $\tau\prall{\obsrel}\sigma$, and $\tau\prall{\heapequiv[\anadr]}\sigma$.
\end{lemma}

\begin{proof}[Proof Sketch]
	We construct $\sigma$ inductively by mimicking every action from $\tau$ and eliding reuses as needed.
	For the construction, consider $\tau.\anact\in\allsem$ and assume we have already constructed, for every $\anadr\in\adr$, an appropriate $\sigma_\anadr\in\adrsem{\anadr}$.
	Consider some address $\anadr\in\adr$.
	The task is to mimic $\anact$ in $\sigma_\anadr$.
	If $\anact$ is an assignment or an SMR call, then pointer race freedom guarantees that we can mimic $\anact$ by executing the same command with a possibly different update.
	We discussed this in \Cref{sec:reduction:pointer-races}.
	The interesting cases are ABAs, frees, and allocations.

	First, consider the case where $\anact$ executes an ABA assertion $\assertof{\apvar=\apvarp}$.
	That the assertion is an ABA means that at least one of the pointers is invalid, say $\apvar$.
	That is, $\anact$ may not be enabled after $\sigma_\anadr$.
	Let $\apvar$ point to $\anadrp$ in $\tau$.
	By induction, we have already constructed $\sigma_\anadrp$ for $\tau$.
	The ABA is enabled after $\sigma_\anadrp$.
	This is due to $\tau\heapequiv[\anadrp]\sigma_\anadrp$.
	It implies that $\apvar$ points to $\anadrp$ in $\tau$ iff $\apvar$ points to $\anadrp$ in $\sigma_\anadrp$ (independent of the validity), and likewise for $\apvarp$.
	That is, the comparison has the same outcome in both computations.
	Now, we can exploit the absence of harmful ABAs to find a computation mimicking $\tau.\anact$ for $\anadr$.
	Applying \Cref{def:harmful-ABA} to $\sigma_\anadrp.\anact$ and $\sigma_\anadr$ yields some $\sigma_\anadr'$ that satisfies the required properties.

	Second, consider the case of $\anact$ performing a $\freeof{\anadrp}$.
	If $\anact$ is enabled after $\sigma_\anadr$ nothing needs to be shown.
	In particular, this is the case if $\anadrp$ is a valid address or $\anadr=\anadrp$.
	Otherwise, $\anadrp$ must be an invalid address.
	Freeing an invalid address does not change the valid memory.
	It also does not change the control location of threads as frees are performed by the environment.
	Hence, we have $\tau.\anact\computequiv\sigma_\anadr$.
	By the definition of elision support, \Cref{def:elision-support}\ref{def:elision-support:free}, the $\free$ does not affect the behavior of the observer on other addresses.
	So we get $\tau.\anact\obsrel\sigma_\anadr$.
	With the same arguments we conclude $\tau.\anact\heapequiv[\anadr]\sigma_\anadr$.
	That is, we do not need to mimic frees of invalid addresses.

	Last, consider $\anact$ executing an allocation $\apvar:=\malloc$ of address $\anadrp$.
	If $\anadrp$ is fresh in $\sigma_\anadr$ or $\anadr=\anadrp$, then $\anact$ is enabled.
	The allocation makes $\anadrp$ a valid address.
	That $\obsrel$ holds for this address follows from elision support, \Cref{def:elision-support}\ref{def:elision-support:fresh}.
	Otherwise, $\anact$ is not enabled because $\anadrp$ cannot be reused.
	We replace in $\sigma_\anadr$ every occurrence of $\anadrp$ with a fresh address $\anadrpp$.
	Let us denote the result with $\renamingof{\sigma_\anadr}{\anadrp}{\anadrpp}$.
	Relying on elision support, \Cref{def:elision-support}\ref{def:elision-support:replace}, one can show $\sigma_\anadr\prec\renamingof{\sigma_\anadr}{\anadrp}{\anadrpp}$ and thus $\tau\prec\renamingof{\sigma_\anadr}{\anadrp}{\anadrpp}$ for all $\prec{\in}\,\set{\computequiv,\obsrel\mkern+2mu,\heapequiv[\anadr]}$.
	Since $\anadrp$ is fresh in $\renamingof{\sigma_\anadr}{\anadrp}{\anadrpp}$, we conclude by enabledness of $\anact$.
\end{proof}

From the above follows the second result of the paper.
It states that for every computation there is a similar one which reuses at most a single address.
Since similarity implies control location equality, our result allows for a much simpler verification of safety properties.
We stress that the result is independent of the actual observer used.

\begin{theorem}
	\label{thm:reduction}
	If $\onesem$ is pointer race free, supports elision, and is free from harmful ABAs, then $\allsem\computequiv\onesem$.
\end{theorem}

One can generalize the above results to a strictly weaker premise, see \Cref{appendix:theory-generalization}.
For the proofs of the above results, refer to \Cref{appendix:proofs-reduction}.

%% file: content/evaluation.tex

\section{Evaluation}
\label{sec:evaluation}

We implemented our approach in a tool.
It is a thread-modular analysis for
\begin{inparaenum}[(i)]
	\item verifying linearizability of singly-linked lock-free data structures using SMR specifications, and
	\item verifying SMR implementations against their specification.
\end{inparaenum}
In the following, we elaborate on the SMR implementations used for our benchmarks and the implemented analysis, and evaluate our tool.

\subsection{SMR Algorithms}
\label{sec:smr-in-eval}

For our experiments, we consider two well-known SMR algorithms: \emph{Hazard Pointers (HP)} and \emph{Epoch-Based Reclamation (EBR)}.
Additionally, we include as a baseline for our experiments a simplistic \emph{GC} SMR algorithm which does not allow for memory to be reclaimed.
We already introduced HP and gave a specification (cf. \Cref{fig:observer}) in form of the observer $\thehpobs$.
We briefly introduce EBR.

Epoch-based reclamation \cite{DBLP:phd/ethos/Fraser04} relies on two assumption:
\begin{inparaenum}[(i)]
	\item threads cannot have pointers to any node of the data structure in-between operation invocations, and
	\item nodes are retired only after being removed from the data structure, i.e., after being made unreachable from the shared variables.
\end{inparaenum}
Those assumptions imply that no thread can acquire a pointer to a removed node if every thread has been in-between an invocation since the removal.
So it is safe to delete a retired node if every thread has been in-between an invocation since the retire.
Technically, EBR introduces \emph{epoch counters}, a global one and one for each thread.
Similar to hazard pointers, thread epochs are single-writer multiple-reader counters.
Whenever a thread invokes an operation, it reads the global epoch $e$ and announces this value by setting its thread epoch to $e$.
Then, it scans the epochs announced by the other threads.
If they all agree on $e$, the global epoch is set to $e+1$.
The fact that all threads must have announced the current epoch $e$ for it to be updated to $e+1$ means that all threads have invoked an operation after the epoch was changed from $e-1$ to $e$.
That is, all threads have been in-between invocations.
Thus, deleting nodes retired in the global epoch $e-1$ becomes safe from the moment when the global epoch is updated from $e$ to $e+1$.
To perform those deletions, every thread keeps a list of retired nodes for every epoch and stores nodes passed to \code{retire} in the list for the current thread epoch.
For the actual deletion it is important to note that the thread-local epoch may lack behind the global epoch by up to $1$.
As a consequence, a thread may put a node retired during the global epoch $e$ into its retire-list for epoch $e-1$.
So for a thread during its local epoch $e$. 
it is not safe to delete the nodes in the retired-list for $e-1$ because it may have been retired during the global epoch $e$.
It is only safe to delete the nodes contained in the retired-list for epochs $e-2$ and smaller.
Hence, it suffices to maintaining three retire-lists.
Progressing to epoch $e+1$ allows for deleting the nodes from the local epoch $e-2$ and to reuse that retire-list for epoch $e+1$.

\input{content/figures/observer_ebr}

\emph{Quiescent-State-Based Reclamation (QSBR)} \cite{McKenney1998ReadcopyUU} generalizes EBR by allowing the programmer to manually identify when threads are \emph{quiescent}.
A thread is quiescent if it does not hold pointers to any node from the data structure.
Threads signal this by calling \code{enterQ} and \code{leaveQ} upon entering and leaving a quiescent phase, respectively.

We use the observer $\theebrobs$ from \Cref{fig:observer-ebr} for specifying both EBR and QSBR (they have the same specification).
As for HP, we use $\baseobs$ to ensure that only those addresses are freed that have been retired.
Observer $\ebrobs$ implements the~actual~EBR/QSBR~behavior~described~above.

To use HP and EBR with our approach for restricting reuse during an analysis, we have to show that their observers support elision.
This is established by the following \namecref{thm:our-observers-support-elision}.
We consider it future work to extend our tool such that it performs the appropriate checks automatically.

\begin{lemma}
	\label{thm:our-observers-support-elision}
	The observers $\thehpobs$ and $\theebrobs$ support elision.
\end{lemma}

\subsection{Thread-Modular Linearizability Analysis}
\label{sec:thread-modular-analysis}

Proving a data structure correct for an arbitrary number of client threads requires a thread-modular analysis \cite{DBLP:conf/cav/BerdineLMRS08,DBLP:journals/toplas/Jones83}.
Such an analysis abstracts a system state into so-called \emph{views}, partial configurations reflecting a single thread's perception of the system state.
A view includes a thread's program counter and, in the case of shared-memory programs, the memory reachable from the shared and thread-local variables.
An analysis then saturates a set $\views$ of reachable views.
This is done by computing the least solution to the recursive equation $\views=\views\mkern+1mu\cup\mkern+1mu\seqof{{\mkern+1mu\views}}\mkern+1mu\cup\mkern+1mu\intof{{\mkern+1mu\views}}$.
Function $\sequential$ computes a \emph{sequential step}, the views obtained from letting each thread execute an action on its own views.
Function $\interference$ accounts for \emph{interference} among threads.
It updates the shared memory of views by actions from other threads.
We follow the analysis from \citet{DBLP:conf/tacas/AbdullaHHJR13,DBLP:journals/sttt/AbdullaHHJR17}.
There, $\interference$ is computed by combining two views, letting one thread perform an action, and projecting the result to the other thread.
More precisely, computing $\intof{{\mkern+1mu\views}}$ requires for every pair of views $\aview_1,\aview_2\in\views$ to
\begin{inparaenum}[(i)]
	\item compute a combined view $\aviewp$ of $\aview_1$ and $\aview_2$,
	\item perform for $\aviewp$ a sequential step for the thread of $\aview_2$, and
	\item project the result of the sequential step to the perception of the thread from $\aview_1$.
\end{inparaenum}
This process is required only for views $\aview_1$ and $\aview_2$ that \emph{match}, i.e., agree on the shared memory both views have in common.
Otherwise, the views are guaranteed to reflect different system states.
Thus, interference is not needed for an exhaustive state space exploration.

To check for linearizability, we assume that the program under scrutiny is annotated with linearization points, points at which the effect of operations take place logically and become visible to other threads.
Whether or not the sequence of emitted linearization points is indeed linearizable can be checked using an observer automaton implementing the desired specification \cite{DBLP:conf/tacas/AbdullaHHJR13,DBLP:journals/sttt/AbdullaHHJR17}, in our case the one for stacks and queues.
The state of this automaton is stored in the views.
If a final state is reached, verification fails.

For brevity, we omit a discussion of the memory abstraction we use.
It is orthogonal to the analysis.
For more details, we refer the reader to \cite{DBLP:conf/tacas/AbdullaHHJR13,DBLP:journals/sttt/AbdullaHHJR17}.
We are not aware of another memory abstraction which can handle reuse and admits automation.

We extend the above analysis by our approach to integrate SMR and restrict reuse to a single address.
To integrate SMR, we add the necessary observers to views.
Note that observers have a pleasant interplay with thread-modularity.
In a view for thread $\athread$ only those observer states are required where $\athread$ is observed.
For the hazard pointer observer $\thehpobs$ this means that only observer states with $\anovar$ capturing $\athread$ need to be stored in the view for $\athread$.
Similarly, the memory abstraction induces a set of addresses that need to be observed (by $\anovarp$).
However, we do not keep observer states for \emph{shared} addresses, i.e., addresses that are reachable from the shared variables.
Instead, we maintain the invariant that they are never retired nor freed.
For the analysis, we then assume that the ignored observer states are arbitrary (but not in observer locations implying retiredness or freedness of shared addresses).
We found that all benchmark programs satisfied this invariant and that the resulting precision allowed for successful verification.
Altogether, this keeps the number of observer states per view small in practice.

As discussed in \Cref{sec:observers}, observers $\thehpobs$ and $\theebrobs$ assume that a client does not perform double-retires.
We integrate a check for this invariant, relying on observer $\baseobs$: if $\baseobs$ is in a state $(\ref{obs:base:retired},\varphi)$, then a double-retire occurs if an event of the form $\retireof{\_,\varphi(\anovarp)}$ is emitted.

To guarantee that the restriction of reuse to a single cell is sound, we have to check for pointer races and harmful ABAs.
To check for pointer races we annotate views with validity information.
This information is updated accordingly during sequential and interference steps.
If a pointer race is detected, verification fails.
For this check, we rely on \Cref{thm:races-are-in-retire} below and deem racy any invocation of \code{retire} with invalid pointers.
That is, the pointer race check boils down to scanning dereferences and \code{retire} invocations for invalid pointers.

\begin{lemma}
	\label{thm:races-are-in-retire}
	If a call is racy wrt. $\theebrobs$ or $\thehpobs$, then it is a call of function \textnormal{\code{retire}} using an invalid pointer.
\end{lemma}

Last, we add a check for harmful ABAs on top of the state space exploration.
This check has to implement \Cref{def:harmful-ABA}.
That a computation $\sigma_\anadr.\anact$ contains a harmful ABA can be detected in the view $\aview_\anadr$ for thread $\athread$ which performs $\anact$.
Like for computations, the view abstraction $\aview_\anadrp$ of $\sigma_\anadrp$ for $\athread$ cannot perform the ABA.
To establish that the ABA is harmless, we seek a $\aview_\anadrp'$ which is similar to $\aview_\anadr$, memory equivalent to $\aview_\anadrp$, and includes the observer behavior of $\aview_\anadrp$.
(The relations introduced in \Cref{sec:reduction} naturally extend from computations to views.)
If no such $\aview_\anadrp'$ exists, verification fails.

In the thread-modular setting one has to be careful with the choice of $\aview_\anadrp'$.
It is not sufficient to find \emph{just some} $\aview_\anadrp'$ satisfying the desired relations.
The reason lies in that we perform the ABA check on a thread-modular abstraction of computations.
To see this, assume the view abstraction of $\sigma_\anadrp$ is $\alpha(\sigma_\anadrp)=\set{\aview_\anadrp,\aview}$ where $\aview_\anadrp$ is the view for thread $\athread$ which performs the ABA in $\sigma_\anadr.\anact$.
For \emph{just some} $\aview_\anadrp'$ it is not guaranteed that there is a computation $\sigma_\anadrp'$ such that $\alpha(\sigma_\anadrp')=\set{\aview_\anadrp',\aview}$.
The sheer existence of $\aview_\anadrp'$ and $\aview$ in $\views$ does not guarantee that there is a computation the abstraction of which yields those two views.
Put differently, we cannot construct computations from views.
Hence, a simple search for $\aview_\anadrp'$ cannot prove the existence of the required $\sigma_\anadrp'$.

To overcome this problem, we use a method to search for a $\aview_\anadrp'$ that guarantees the existence of $\sigma_\anadrp'$; in terms of the above example, guarantees that there is $\sigma_\anadrp'$ with $\alpha(\sigma_\anadrp')=\set{\aview_\anadrp',\aview}$.
We take the view $\aview_\anadrp$ that cannot perform the ABA.
We apply sequential steps to $\aview_\anadrp$ until it \emph{comes back} to the same program counter.
The rational behind is that ABAs are typically conditionals that restart the operation if the ABA is not executable.
Restarting the operation results in reading out pointers anew (this time without interference from other threads).
Consequently, the ABA is now executable.
The resulting view is a candidate for $\aview_\anadrp'$.
If it does not satisfy \Cref{def:harmful-ABA}, verification fails.
Although simple, this approach succeeded in all benchmarks.

\subsection{Linearizability Experiments}
\label{sec:evaluation:verify-ds}

We implemented the above analysis in a \texttt{\small C++} tool.\footnote{Available at: \githuburlDS}
We evaluated the tool on singly-linked lock-free data structures from the literature, like Treiber's stack \cite{opac-b1015261,DBLP:conf/podc/Michael02}, Michael\&Scott's lock-free queue \cite{DBLP:conf/podc/MichaelS96,DBLP:conf/podc/Michael02}, and the DGLM lock-free queue \cite{DBLP:conf/forte/DohertyGLM04}.
The findings are listed in \Cref{table:experiments}.
They include
\begin{inparaenum}[(i)]
	\item the time taken for verification, i.e., to explore exhaustively the state space and check linearizability,
	\item the size of the explored state space, i.e., the number of reachable views,
	\item the number of ABA prone views, i.e., views where a thread is about to perform an $\assert$ containing an invalid pointer,
	\item the time taken to establish that no ABA is harmful, and
	\item the verdict of the linearizability check.
\end{inparaenum}
The experiments were conducted on an Intel Xeon X5650@2.67GHz running Ubuntu 16.04 and using Clang version 6.0.

\input{content/figures/eval}

Our approach is capable of verifying lock-free data structures using HP and EBR.
We were able to automatically verify Treiber's stack, Michael\&Scott's queue, and the DGLM queue.
To the best of our knowledge, we are the first to verify data structures using the aforementioned SMR algorithms fully automatically.
Moreover, we are also the first to verify automatically the DGLM queue under any manual memory management technique.

An interesting observation throughout the entire test suite is that the number of ABA prone views is rather small compared to the total number of reachable views.
Consequently, the time needed to check for harmful ABAs is insignificant compared to the verification time.
This substantiates the usefulness of \emph{ignoring} ABAs during the actual analysis and checking afterwards that no harmful ABA exists.

\input{content/figures/treibers_push}%
Our tool could not establish linearizability for the optimized version of Treiber's stack with hazard pointers by \citet{DBLP:conf/podc/Michael02}.
The reason for this is that the \code{push} operation does not use any hazard pointers.
This leads to pointer races and thus verification failure although the implementation is correct.
To see why, consider the code of \code{push} from \Cref{fig:treibers-push}.
The operation allocates a new node, reads the top-of-stack pointer into a local variable \code{top} in \Cref{code:treibers-push:read-tos}, links the new node to the top-of-stack in \Cref{code:treibers-push:set-next}, and swings the top-of-stack pointer to the new node in \Cref{code:treibers-push:cas}.
Between \Cref{code:treibers-push:read-tos} and \Cref{code:treibers-push:set-next} the node referenced by \code{top} can be popped, reclaimed, reused, and reinserted by an interferer.
That is, the \code{CAS} in \Cref{code:treibers-push:cas} is ABA prone.
The reclamation of the node referenced by \code{top} renders both the \code{top} pointer and the \code{next} field of the new node invalid.
As discussed in \Cref{sec:reduction}, the comparison of the valid \code{ToS} with the invalid \code{top} makes \code{top} valid again.
However, the \code{next} field of the new node remains invalid.
That is, the \code{push} succeeds and leaves the stack in a state with \code{ToS->$\,$next} being invalid.
This leads to pointer races because no thread can acquire valid pointers to the nodes following \code{ToS}.
Hence, reading out data of such subsequent nodes in the \code{pop} procedure, for example, raises a pointer race.

To solve this issue, the \code{CAS} in \Cref{code:treibers-push:cas} has to validate the pointer \code{node->next}.
One could annotate the \code{CAS} with an invariant \code{Tos$\:$==$\:$node->$\,$next}.
Treating invariants and assertions alike would then result in the \code{CAS} validating \code{node->next} (cf. \Cref{sec:reduction:similarity}) as desired.
That the annotation is an invariant indeed, could be checked during the analysis.
We consider a proper investigation as future work.

For the DGLM queue, our tool required hints.
The DGLM queue is similar to Michael\&Scott's queue but allows the \code{Head} pointer to overtake the \code{Tail} pointer by at most one node.
Due to imprecision in the memory abstraction, our tool explored states with malformed lists where \code{Head} overtook \code{Tail} by more than one node.
We implemented a switch to increase the precision of the abstraction and ignore cases where \code{Head} overtakes \code{Tail} by more than one node.
This allowed us to verifying the DGLM queue.
While this change is ad hoc, it does not jeopardize the principledness of our approach because it affects only the memory abstraction which we took from the literature.

\subsection{Verifying SMR Implementations}
\label{sec:evaluation:verify-smr}

It remains to verify that a given SMR implementation is correct wrt. an observer $\smrobs$.
As noted in \Cref{sec:observers}, an SMR implementation can be viewed as a lock-free data structure where the stored \emph{data} are pointers.
Consequently, we can reuse the above analysis.
We extended our implementation with an abstraction for (sets of) data values.\footnote{Available at: \githuburlSMR}
The main insight for a concise abstraction is that it suffices to track a single observer state per view.
If the SMR implementation is not correct wrt. $\smrobs$, then by definition there is $\tau\in\allsem[MGC(\ansmr)]$ with $\historyof{\tau}\notin\specof{\smrobs}$.
Hence, there must be some observer state $\astate$ with $\historyof{\tau}\notin\specof{\astate}$.
Consider the observers $\thehpobs$ and $\theebrobs$ where $\astate$ is of the form $\astate=(\alocation,\set{\anovar\mapsto\athread,\anovarp\mapsto\anadr,\anovarpp\mapsto\anint})$.
This single state $\astate$ induces a simple abstraction of data values $\advalue$: either $\advalue=\anadr$ or $\advalue\neq\anadr$.
Similarly, an abstraction of sets of data values simply tracks whether or not the set contains $\anadr$.

To gain adequate precision, we retain in every view the thread-local pointers of $\athread$.
Wrt. \Cref{fig:hpimpldyn}, this keeps the thread-$\athread$-local \code{HPRec} in every view.
It makes the analysis recognize that $\athread$ has indeed protected $\anadr$.
Moreover, we store in every view whether or not the last $\retire$ invocation stems from the thread of that view.
With this information, we avoid unnecessary matches during interference of views $\aview_1$ and $\aview_2$: if both threads $\athread_1$ of $\aview_1$ and $\athread_2$ of $\aview_2$ have performed the last $\retire$ invocation, then $\athread_1$ and $\athread_2$ are the exact same thread.
Hence, interference is not needed as threads have unique identities.
We found this extension necessary to gain the precision required to verify our benchmarks.

\input{content/figures/smreval}

\Cref{table:smrexp} shows the experimental results for the HP implementation from \Cref{fig:hpimpldyn} and an EBR implementation.
Both SMR implementations allow threads to dynamically join and part.
We conducted the experiments in the same setup as before.
As noted in \Cref{sec:observers}, the verification is simpler and thus more efficient than the previous one.
The reason for this is the absence of memory reclamation.

%% file: content/figures/observer_ebr.tex

\begin{figure}
\begin{tcolorbox}
	\center
	\begin{tikzpicture}[->,>=stealth',shorten >=1pt,auto,node distance=2.8cm,thick,initial text={}]
		\node [xshift=-0.4cm,yshift=.8cm,draw,thin] {$\ebrobs$};
		\tikzstyle{every state}=[minimum size=1.5em]
		\tikzset{every edge/.append style={font=\footnotesize}}
		\node[initial,state]    (A)              {\mkstatename{obs:ebr:init}};
		\node[state]            (B) [right of=A] {\mkstatename{obs:ebr:invoked}};
		\node[state]            (E) [right of=B] {\mkstatename{obs:ebr:active}};
		\node[state]            (C) [right of=E] {\mkstatename{obs:ebr:retired}};
		\node[accepting,state]  (D) [right of=C] {\mkstatename{obs:ebr:final}};
		\coordinate             [below of=A, yshift=+1.75cm]  (X)  {};
		\coordinate             [below of=E, yshift=+1.75cm]  (Y)  {};
		\coordinate             [below of=B, yshift=+1.75cm]  (Z)  {};
		\path
			(A) edge node[align=center] {\translabbr{\evt{\leaveQ}{\athread}}{\athread=\anovar}} (B)
			(B) edge node[align=center] {\translabbr{\evt{\exit}{\athread}}{\athread=\anovar}} (E)
			(E) edge node[align=center] {\translabbr{\evt{\retire}{\athread,\anadr}}{\anadr=\anovarp}} (C)
			(C) edge node[align=center] {\translabbr{\freeof{\anadr}}{\anadr=\anovarp}} (D)
			(C.south west) edge[-,shorten >=0pt] (Y)
			([xshift=-1.5mm]X) edge ([xshift=-1.5mm]A.south)
			(Y) edge[-,shorten >=0pt] node {\translab{\evt{\enterQ}{\athread}}{\athread=\anovar}} ([xshift=-1.5mm]X)
			(E) edge[-,shorten >=0pt] ([yshift=1mm]Y.north)
			([yshift=1mm]Y) edge[-,shorten >=0pt] ([xshift=0mm,yshift=1mm]X)
			([xshift=0mm,yshift=1mm]X) edge ([xshift=0mm]A.south)
			(B) edge[-,shorten >=0pt] ([yshift=2mm]Z.north)
			([yshift=2mm]Z) edge[-,shorten >=0pt] ([xshift=1.5mm,yshift=2mm]X)
			([xshift=1.5mm,yshift=2mm]X) edge ([xshift=1.5mm]A.south)
			;
	\end{tikzpicture}
	\caption{%
		Observer specifying when EBR/QSBR defers deletion using two variables, $\anovar$ and $\anovarp$, to observe a thread and an address, respectively.
		The observer implements the property that a cell $\anovarp$ retired during the non-quiescent phase of a thread $\anovar$ may not be freed until the thread becomes quiescent.
		The full specification of EBR/QSBR is the observer $\theebrobs$.
	}
	\label{fig:observer-ebr}
\end{tcolorbox}
\end{figure}

%% file: content/figures/eval.tex

\begin{table}%
\begin{tcolorbox}
	\caption{Experimental results for verifying singly-linked data structures using SMR. The experiments were conducted on an Intel Xeon~X5650@2.67GHz running Ubuntu~16.04 and using Clang~6.0.}%
	\label{table:experiments}%
	\begin{minipage}{\columnwidth}
		\center%
		\newcommand{\seprule}{\midrule[.2pt]}
		\begin{tabularx}{\textwidth}{Xlrrrrc}%
			\toprule
			Program\footnote{The code for the benchmark programs can be found in \Cref{appendix:benchmark-programs}.}
				& SMR
				& Time Verif.
				& States
				& ABAs
				& Time ABA
				& Linearizable
				\\
			\midrule
			Coarse stack
				& GC
				& $0.44s$
				& $300$
				& $0$
				& $0s$
				& yes
				\\
				& None
				& $0.5s$
				& $300$
				& $0$
				& $0s$
				& yes
				\\
			\seprule
			Coarse queue
				& GC
				& $1.7s$
				& $300$
				& $0$
				& $0s$
				& yes
				\\
				& None
				& $1.7s$
				& $300$
				& $0$
				& $0s$
				& yes
				\\
			\seprule
			Treiber's stack
				& GC
				& $3.2s$
				& $806$
				& $0$
				& $0s$
				& yes
				\\
				& EBR
				& $16s$
				& $1822$
				& $0$
				& $0s$
				& yes
				\\
				& HP
				& $19s$
				& $2606$
				& $186$
				& $0.06s$
				& yes
				\\
			Opt. Treiber's stack
				& HP
				& $0.8s$
				& ---
				& ---
				& ---
				& no\footnote{Pointer race due to an ABA in \texttt{push}: the \texttt{next} pointer of the new node becomes invalid and the CAS succeeds.}
				\\
			\seprule
			Michael\&Scott's queue
				& GC
				& $414s$
				& $2202$
				& $0$
				& $0s$
				& yes
				\\
				& EBR
				& $2630s$
				& $7613$
				& $0$
				& $0s$
				& yes
				\\
				& HP
				& $7075s$
				& $19028$
				& $536$
				& $0.9s$
				& yes
				\\
			\seprule
			DGLM queue
				& GC
				& $714s$
				& $9934$
				& $0$
				& $0s$
				& yes\textsuperscript{\ref{footnote:DGLMprecision}}
				\\
				& EBR
				& $3754s$
				& $27132$
				& $0$
				& $0s$
				& yes\textsuperscript{\ref{footnote:DGLMprecision}}
				\\
				& HP
				& $7010s$
				& $41753$
				& $2824$
				& $26s$
				& yes\footnote{Imprecision in the memory abstraction required hinting that \texttt{Head} cannot overtake \texttt{Tail} by more than one node.\label{footnote:DGLMprecision}}
				\\
			\bottomrule
		\end{tabularx}
	\end{minipage}
\end{tcolorbox}
\end{table}

%% file: content/figures/treibers_push.tex

\begin{wrapfigure}{r}{0.42\textwidth} 
\begin{tcolorbox}
	\center%
\begin{lstlisting}[style=condensed,xleftmargin=20pt]
struct Node {/* ... */}
shared Node* ToS;
void push(data_t input) {
	Node* node = new Node();
	node->data = input;
	while (true) {
		Node* top = ToS; $\label[line]{code:treibers-push:read-tos}$
		node->next = top; $\label[line]{code:treibers-push:set-next}$
		if (CAS&ToS, top, next) $\label[line]{code:treibers-push:cas}$
			break;
	}
\end{lstlisting}
	\caption{%
		The \code{push} operation of Treiber's lock-free stack \cite{opac-b1015261}.
	}
	\label{fig:treibers-push}
\end{tcolorbox}
\end{wrapfigure}

%% file: content/figures/smreval.tex

\begin{table}%
\begin{tcolorbox}
	\caption{Experimental results for verifying SMR implementations against their observer specifications. The experiments were conducted on an Intel Xeon~X5650@2.67GHz running Ubuntu~16.04 and using Clang~6.0.}%
	\label{table:smrexp}%
	\begin{minipage}{\columnwidth}
		\center%
		\newcommand{\seprule}{\midrule[.2pt]}
		\begin{tabularx}{\textwidth}{Xlrrc}%
			\toprule
			SMR Implementation\footnote{The code for the benchmark programs can be found in \Cref{appendix:benchmark-programs}.}
				& Specification
				& Verification Time
				& States
				& Correct
				\\
			\midrule
			Hazard Pointers
				& $\thehpobs$
				& $1.5s$
				& $5437$
				& yes
				\\
			\seprule
			Epoch-Based Reclamation
				& $\theebrobs$
				& $11.2s$
				& $11528$
				& yes
				\\
			\bottomrule
		\end{tabularx}
	\end{minipage}
\end{tcolorbox}
\end{table}

%% file: content/relatedwork.tex

\section{Related Work}
\label{sec:related_work}

We discuss the related work on SMR implementations and on the verification of linearizability.

\paragraph{Safe Memory Reclamation}

Besides HP and EBR further SMR algorithms have been proposed in the literature.
\emph{Free-lists} is the simplest such mechanism.
Retired nodes are stored in a thread-local free-list.
The nodes in this list are never reclaimed.
Instead, a thread can reuse nodes instead of allocating new ones.
\emph{Reference Counting (RC)} adds to nodes a counter representing the number of pointers referencing that node.
Updating such counters safely in a lock-free fashion, however, requires the use of hazard pointers \cite{DBLP:journals/tocs/HerlihyLMM05} or double-word CAS \cite{DBLP:conf/podc/DetlefsMMS01}, which is not available on most hardware.
Besides free-lists and RC, most SMR implementations from the literature combine techniques.
For example, \emph{DEBRA} \cite{DBLP:conf/podc/Brown15} is an optimized QSBR implementation.
\Citet{DBLP:conf/wdag/Harris01} extends EBR by adding epochs to nodes to detect when reclamation is safe.
\emph{Cadence} \cite{DBLP:conf/spaa/BalmauGHZ16}, the work by \citet{DBLP:conf/podc/AghazadehGW13a}, and the work by \citet{DBLP:conf/iwmm/DiceHK16} are HP implementations improving on the original implementation due to \citet{DBLP:conf/podc/Michael02}.
\emph{ThreadScan} \cite{DBLP:conf/spaa/AlistarhLMS15}, \emph{StackTrack} \cite{DBLP:conf/eurosys/AlistarhEHMS14}, and \emph{Dynamic Collect} \cite{DBLP:conf/podc/DragojevicHLM11} borrow the mechanics of hazard pointers to protect single cells at a time.
\emph{Drop the Anchor} \cite{DBLP:conf/spaa/BraginskyKP13}, \emph{Optimistic Access} \cite{DBLP:conf/spaa/CohenP15}, \emph{Automatic Optimistic Access} \cite{DBLP:conf/oopsla/CohenP15}, \emph{QSense} \cite{DBLP:conf/spaa/BalmauGHZ16}, \emph{Hazard Eras} \cite{DBLP:conf/spaa/RamalheteC17}, and \emph{Interval-Based Reclamation} \cite{DBLP:conf/ppopp/WenICBS18} are combinations of EBR and HP.
\emph{Beware\&Cleanup} \cite{DBLP:conf/ispan/GidenstamPST05} is a combination of HP and RC.
\emph{Isolde} \cite{DBLP:conf/iwmm/YangW17} is an implementation of EBR and RC.
We omit \emph{Read-Copy-Update (RCU)} here because it does not allow for non-blocking deletion \cite{phd/McKenney04}.
While we have implemented an analysis for EBR and HP only, we believe that our approach can handle most of the above works with little to no modifications.
We refer the reader to \Cref{sec:more_smr} for a more detailed discussion.

\paragraph{Linearizability}

Linearizability of lock-free data structures has received considerable attention over the last decade.
The proposed techniques for verifying linearizability can be classified roughly into
\begin{inparaenum}[(i)]
	\item testing,
	\item non-automated proofs, and
	\item automated proofs.
\end{inparaenum}
Linearizability testing \cite{DBLP:conf/pldi/VechevY08,DBLP:conf/fm/LiuCLS09,DBLP:conf/pldi/BurckhardtDMT10,DBLP:conf/cav/CernyRZCA10,DBLP:conf/icse/Zhang11a,DBLP:journals/tse/Liu0L0ZD13,DBLP:conf/hvc/TravkinMW13,DBLP:conf/pldi/EmmiEH15,DBLP:conf/forte/HornK15a,DBLP:journals/concurrency/Lowe17,DBLP:journals/corr/YangKLW17,DBLP:journals/pacmpl/EmmiE18} enumerates an incomplete portion of the state space and checks whether or not the discovered computations are linearizable.
This approach is useful for bug-hunting and for analyzing huge code bases.
However, it cannot prove an implementation linearizable as it might miss non-linearizable computations.

In order to prove a given implementation linearizable, one can conduct a manual (pen\&paper) or a mechanized (tool-supported, not automated) proof.
Such proofs are cumbersome and require a human to have a deep understanding of both the implementation under scrutiny and the verification technique used.
Common verification techniques are program logics and simulation relations.
Since our focus lies on automated proofs, we do not discuss non-automated proof techniques in detail.
For a survey refer to \cite{DBLP:journals/corr/DongolD14}.

Interestingly, most non-automated proofs of lock-free code rely on a garbage collector \cite{DBLP:journals/entcs/ColvinDG05,DBLP:conf/cav/ColvinGLM06,DBLP:conf/iceccs/Groves07,DBLP:conf/cats/Groves08,DBLP:conf/wdag/DohertyM09,DBLP:conf/tacas/ElmasQSST10,DBLP:conf/podc/OHearnRVYY10,DBLP:journals/fac/BaumlerSTR11,DBLP:journals/toplas/DerrickSW11,DBLP:journals/fac/Jonsson12,DBLP:conf/popl/LiangFF12,DBLP:conf/pldi/LiangF13,DBLP:journals/toplas/LiangFF14,DBLP:conf/wdag/HemedRV15,DBLP:conf/pldi/SergeyNB15,DBLP:conf/esop/SergeyNB15,DBLP:conf/cav/BouajjaniEEM17,DBLP:conf/ecoop/DelbiancoSNB17,DBLP:conf/esop/KhyzhaDGP17} to avoid the complexity of memory reclamation.
\citet{DBLP:conf/cav/SchellhornWD12,DBLP:conf/concur/HenzingerSV13} verify a lock-free queue by \citet{DBLP:journals/toplas/HerlihyW90} which does not reclaim memory.

There is less work on manual verification in the presence of reclamation.
\Citet{DBLP:conf/forte/DohertyGLM04} verify a lock-free queue implementation using tagged pointers and free-lists.
\Citet{DBLP:conf/popl/DoddsHK15} verify a time-stamped stack using tagged pointers.
\Citet{DBLP:journals/pacmpl/KrishnaSW18} verify Harris' list \cite{DBLP:conf/wdag/Harris01} with reclamation.
Finally, there are works \cite{DBLP:conf/popl/ParkinsonBO07,DBLP:conf/concur/FuLFSZ10,DBLP:conf/ictac/TofanSR11,DBLP:conf/esop/GotsmanRY13} which verify implementations using safe memory reclamation.
With the exception of \cite{DBLP:conf/esop/GotsmanRY13}, they only consider implementations using HP.
\Citet{DBLP:conf/esop/GotsmanRY13} in addition verify implementations using EBR.
With respect to lock-free data structures, these works prove linearizability of stacks.
Unlike in our approach, data structure and SMR code are verified together.
In theory, those works are not limited to such simple data structures.
In practice, however, we are not aware of any work that proves linearizable more complicated implementations using HP.
The (temporal) specifications for HP and EBR from \citet{DBLP:conf/esop/GotsmanRY13} are reflected in our observer automata.

\Citet{DBLP:conf/cav/AlglaveKT13,DBLP:journals/sigops/DesnoyersMD13,DBLP:conf/spin/Kokologiannakis17,DBLP:conf/date/LiangMKM18} test RCU implementations.
\Citet{DBLP:conf/esop/GotsmanRY13,DBLP:conf/pldi/TassarottiDV15} specify RCU and verify data structures using it non-automatically.
We do not discuss such approaches because memory reclamation in RCU is blocking.

Automated approaches relieve a human checker from the complexity of manual/mechanized proofs.
In turn, they have to automatically synthesize invariants and finitely encode all possible computations.
The state-of-the-art methodology to do so is thread-modularity \cite{DBLP:conf/cav/BerdineLMRS08,DBLP:journals/toplas/Jones83}.
We have already discussed this technique in \Cref{sec:thread-modular-analysis}.
Its main advantage is the ability to verify library code under an unbounded number of concurrent clients.

Similar to non-automated verification, most automated approaches assume a garbage collector \cite{DBLP:conf/cav/AmitRRSY07,DBLP:conf/cav/BerdineLMRS08,DBLP:conf/aplas/SegalovLMGS09,DBLP:conf/spin/VechevYY09,DBLP:conf/vmcai/Vafeiadis10,DBLP:conf/cav/Vafeiadis10,DBLP:conf/spin/SethiTM13,DBLP:conf/cav/ZhuPJ15,DBLP:conf/sas/AbdullaJT16}.
Garbage collection has the advantage of ownership: an allocation is always owned by the allocating thread, it cannot be accessed by any other thread.
This allows for guiding thread-modular techniques, resulting in faster convergence times and more precise analyses.
Despite this, there are many techniques that suffer from poor scalability.
To counteract the state space explosion, they neglect SMR code making the programs under scrutiny simpler (and shorter).

Few works \cite{DBLP:conf/tacas/AbdullaHHJR13,DBLP:conf/vmcai/HazizaHMW16,DBLP:conf/sas/HolikMVW17} address the challenge of verifying lock-free data structures under manual memory management.
These works assume that accessing deleted memory is safe and that tag fields are never overwritten with non-tag values.
Basically, they integrate free-lists into the semantics.
Their tools are able to verify Treiber's stack and Micheal\&Scott's queue using tagged pointers.
For our experiments we implemented an analysis based on \cite{DBLP:conf/tacas/AbdullaHHJR13,DBLP:conf/vmcai/HazizaHMW16}.
For our theoretical development, we borrowed the observer automata from \citet{DBLP:conf/tacas/AbdullaHHJR13} as a specification means for SMR algorithms.
From \citet{DBLP:conf/vmcai/HazizaHMW16} we borrowed the notion of validity and pointer races.
We modified the definition of validity to allow for ABAs and lifted the pointer race definition to SMR calls.
This allowed us to verify lock-free data structures without the complexity of SMR implementations and without considering all possible reallocations.
With our approach we could verify the DGLM queue which has not been verified automatically for manual memory management before.

To the best of our knowledge, there are no works which propose an automated approach for verifying linearizability of lock-free data structures similar to ours.
We are the first to
\begin{inparaenum}[(i)]
	\item propose a method for automatically verifying linearizability which decouples the verification of memory reclamation from the verification of the data structure, and
	\item propose a method for easing the verification task by avoiding reuse of memory except for a single memory location.
\end{inparaenum}

%% file: content/conclusion.tex

\section{Conclusion and Future Work}
\label{sec:conclusion}

In this paper, we have shown that the way lock-free data structures and their memory reclamation are implemented allows for compositional verification.
The memory reclamation can be verified against a simple specification.
In turn, this specification can be used to verify the data structure without considering the implementation of the memory reclamation.
This breaks verification into two tasks each of which has to consider only a part of the original code under scrutiny.

However, the resulting tasks remain hard for verification tools.
To reduce their complexity and make automation tractable, we showed that one can rely on a simpler semantics without sacrificing soundness.
The semantics we proposed is simpler in that, instead of arbitrary reuse, only a single memory location needs to be considered for reuse.
To ensure soundness, we showed how to check in such a semantics for ABAs.
We showed how to tolerate certain, harmless ABAs to handle SMR implementations like hazard pointers.
That is, we need to give up verification only if there are harmful ABAs, that is, true bugs.

We evaluated our approach in an automated linearizability checker.
Our experiments confirmed that our approach can handle complex data structures with SMR the verification of which is beyond the capabilities of existing automatic approaches.

As future work we would like to extend our implementation to support more data structures and more SMR implementations from the literature.
As stated before, this may require some generalizations of our theory.

%% file: content/appendix/base.tex

\clearpage
\newpage
\appendix

\include{content/appendix/code/base}

\include{content/appendix/smr/base}
\include{content/appendix/theory/base}

%% file: content/appendix/code/base.tex

\section{Missing Benchmark Programs}
\label{appendix:benchmark-programs}

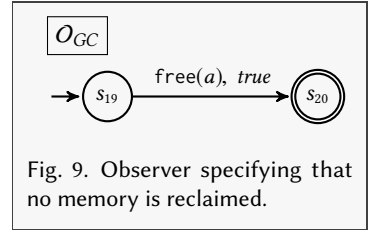
\begin{wrapfigure}{r}{0.34\textwidth}
	\vspace{-10pt}
	\small
	\begin{tcolorbox}
	\center%
	\begin{tikzpicture}[->,>=stealth',shorten >=1pt,auto,node distance=2.8cm,thick,initial text={}]
		\node [xshift=-0.4cm,yshift=.8cm,draw,thin] {$\gcobs$};
		\tikzstyle{every state}=[minimum size=1.5em]
		\tikzset{every edge/.append style={font=\footnotesize}}
		\node[initial,state]    (A)              {\mkstatename{obs:gc:init}};
		\node[accepting,state]  (B) [right of=A] {\mkstatename{obs:gc:final}};
		\path
			(A) edge node {\translab{\evt{\free}{\anadr}}{\mathit{true}}} (B)
			;
	\end{tikzpicture}
	\caption{%
		Observer specifying that no memory is reclaimed.
	}
	\label{fig:observer-gc}
	\end{tcolorbox}
\end{wrapfigure}
In \Cref{sec:evaluation} we evaluated our tool on various data structures and SMR algorithms.
With reference to \Cref{table:experiments}, the observers for \emph{HP}, \emph{EBR}, \emph{GC}, and \emph{None} are $\thehpobs$, $\theebrobs$, $\gcobs$, and $\baseobs$, respectively.
They can be found in \Cref{fig:observer-ebr,fig:observer,fig:observer-gc}.

The code for the \emph{Coarse stack} can be found in \Cref{fig:cstack}, the \emph{Coarse queue} in \Cref{fig:cqueue}, \emph{Treiber's stack} in \Cref{fig:treibers}, the \emph{Optimized Treiber's stack} in \Cref{fig:treibers-opt}, \emph{Michael\&Scott's queue} in \Cref{fig:michaelscott_ebrhp}, and the \emph{DGLM queue} in \Cref{fig:dglm}.
The logic for \emph{HP}, \emph{EBR}, \emph{GC}, and \emph{None} is highlighted with \code{H}, \code{E}, \code{G}, and \code{N}, respectively.

The code for the \emph{Hazard Pointers} implementation can be found in \Cref{fig:hpimpldyn}, the code for the \emph{Epoch-Based Reclamation} implementation in \Cref{fig:ebrimpl}.

\input{content/appendix/code/coarse_stack}
\input{content/appendix/code/coarse_queue}
\input{content/appendix/code/treibers}
\input{content/appendix/code/treibers_opt}
\input{content/appendix/code/ms}
\input{content/appendix/code/dglm}
\input{content/appendix/code/ebr}

%% file: content/appendix/code/coarse_stack.tex

\begin{figure}[b]
	\begin{tcolorbox}%
	\center%
\begin{lstlisting}[style=condensed]
/* Coarse Stack */
struct Node { data_t data; Node* next; };
shared Node* ToS;
atomic init() { ToS = NULL; }
\end{lstlisting}%
	\begin{minipage}[t]{.46\textwidth}%
\begin{lstlisting}[style=condensed]
void push(data_t input) {
	Node* node = new Node();
	node->data = input;
	atomic {
		node->next = ToS;
		ToS = node;
	}
}
\end{lstlisting}%
	\end{minipage}%
	\hfill%
	\begin{minipage}[t]{.475\textwidth}%
\begin{lstlisting}[style=condensed]
data_t pop() {
	atomic {
		Node* top = ToS;
		if (top == NULL) return EMPTY;
		data_t output = top->data;
		ToS = Tos->next;
§GN§	@retire(top);@
		return output;
}	}
\end{lstlisting}%
	\end{minipage}%
	\caption{%
		Coarse Stack implementation.
		No SMR is needed due to the use of \code{atomic} blocks.
	}
	\label{fig:cstack}
	\end{tcolorbox}
\end{figure}

%% file: content/appendix/code/coarse_queue.tex

\begin{figure}[b]
	\begin{tcolorbox}%
	\center%
\begin{lstlisting}[style=condensed]
/* Coarse Queue */
struct Node { data_t data; Node* next; };
shared Node* Head, Tail;
atomic init() { Head = new Node(); Head->next = NULL; Tail = Head; }
\end{lstlisting}%
	\begin{minipage}[t]{.46\textwidth}%
\begin{lstlisting}[style=condensed]
void enqueue(data_t input) {
	Node* node = new Node();
	node->data = input;
	node->next = NULL;
	atomic {
		Tail->next = node;
		Tail = node;
	}
}
\end{lstlisting}%
	\end{minipage}%
	\hfill%
	\begin{minipage}[t]{.475\textwidth}%
\begin{lstlisting}[style=condensed]
data_t dequeue() {
	atomic {
		Node* head = Head;
		Node* next = head->next;
		if (next == NULL) return EMPTY;
		data_t output = next->data;
		Head = next;
§GN§	@retire(head);@
		return output;
}	}
\end{lstlisting}%
	\end{minipage}%
	\caption{%
		Coarse Queue implementation.
		No SMR is needed due to the use of \code{atomic} blocks.
	}
	\label{fig:cqueue}
	\end{tcolorbox}
\end{figure}

%% file: content/appendix/code/treibers.tex

\begin{figure}[b]
	\begin{tcolorbox}%
	\center%
\begin{lstlisting}[style=condensed]
/* Treiber's stack */
struct Node { data_t data; Node* next; };
shared Node* ToS;
atomic init() { ToS = NULL; }
\end{lstlisting}%
	\begin{minipage}[t]{.46\textwidth}%
\begin{lstlisting}[style=condensed]
void push(data_t input) {
§E§	@leaveQ();@
	Node* node = new Node();
	node->data = input;
	while (true) {
		Node* top = ToS;
§H§		@protect(top, 0);@
§H§		@if (top != ToS) continue;@
		node->next = top;
		if (CAS(&ToS, top, node))
			break
	}
§H§	@unprotect(0);@
§E§	@enterQ();@
}
\end{lstlisting}%
	\end{minipage}%
	\hfill%
	\begin{minipage}[t]{.475\textwidth}%
\begin{lstlisting}[style=condensed]
data_t pop() {
§E§	@leaveQ();@
	while (true) {
		Node* top = ToS;
		if (top == NULL) return EMPTY;
§H§		@protect(top, 0);@
§H§		@if (top != ToS) continue;@
		Node* next = top->next;
		if (CAS(&ToS, top, next)) {
			data_t output = top->data;
§HEG§		@retire(top);@
§H§			@unprotect(0);@
§E§			@enterQ();@
			return output;
}	}	}
\end{lstlisting}%
	\end{minipage}%
	\caption{%
		Treiber's non-blocking stack \cite{opac-b1015261} with hazard pointer code.
		The implementation requires a single hazard pointer per thread.
	}
	\label{fig:treibers}
	\end{tcolorbox}
\end{figure}

%% file: content/appendix/code/treibers_opt.tex

\begin{figure}[b]
	\begin{tcolorbox}%
	\center%
\begin{lstlisting}[style=condensed]
/* Optimized Treiber's stack */
struct Node { data_t data; Node* next; };
shared Node* ToS;
atomic init() { ToS = NULL; }
\end{lstlisting}%
	\begin{minipage}[t]{.46\textwidth}%
\begin{lstlisting}[style=condensed]
void push(data_t input) {
	Node* node = new Node();
	node->data = input;
	while (true) {
		Node* top = ToS;
		node->next = top;
		if (CAS(&ToS, top, node))
			break
	}
}
\end{lstlisting}%
	\end{minipage}%
	\hfill%
	\begin{minipage}[t]{.475\textwidth}%
\begin{lstlisting}[style=condensed]
data_t pop() {
	while (true) {
		Node* top = ToS;
		if (top == NULL) return EMPTY;
§H§		@protect(top, 0);@
§H§		@if (top != ToS) continue;@
		Node* next = top->next;
		if (CAS(&ToS, top, next)) {
			data_t output = top->data;
§H§			@retire(top);@
§H§			@unprotect(0);@
			return output;
}	}	}
\end{lstlisting}%
	\end{minipage}%
	\caption{%
		Optimized version of Treiber's non-blocking stack \cite{opac-b1015261} with hazard pointer code due to \cite{DBLP:conf/podc/Michael02}.
		The \code{push} operation does not require hazard pointers.
		The implementation of \code{push} is the same as in \Cref{fig:treibers-push}.
	}
	\label{fig:treibers-opt}
	\end{tcolorbox}
\end{figure}

%% file: content/appendix/code/ms.tex

\begin{figure}
	\begin{tcolorbox}%
	\center%
\begin{lstlisting}[style=condensed]
/* Michael&Scott's queue */
struct Node { data_t data; Node* next; };
shared Node* Head, Tail;
atomic init() { Head = new Node(); Head->next = NULL; Tail = Head; }
\end{lstlisting}%
	\begin{minipage}[t]{.46\textwidth}%
\begin{lstlisting}[style=condensed]
void enqueue(data_t input) {
§E§	@leaveQ;@
	Node* node = new Node();
	node->data = input;
	node->next = NULL;
	while (true) {
		Node* tail = Tail;
§H§		@protect(tail, 0);@
§H§		@if (tail != Tail) continue;@
		Node* next = tail->next;
		if (tail != Tail) continue;
		if (next != NULL) {
			CAS(&Tail, tail, next);
			continue;
		}
		if (CAS(&tail->next, next, node))
			break
	}
	CAS(&Tail, tail, node);
§E§	@enterQ;@
§H§	@unprotect(0);@
}
\end{lstlisting}%
	\end{minipage}%
	\hfill%
	\begin{minipage}[t]{.475\textwidth}%
\begin{lstlisting}[style=condensed]
data_t dequeue() { $\label[line]{code:michaelscott_ebrhp:dequeue}$
§E§	@leaveQ;@
	while (true) {
		Node* head = Head; $\label[line]{code:michaelscott_ebrhp:read-head}$
§H§		@protect(head, 0);@ $\label[line]{code:michaelscott_ebrhp:protect-head}$
§H§		@if (head != Head) continue;@ $\label[line]{code:michaelscott_ebrhp:check-head}$
		Node* tail = Tail; $\label[line]{code:michaelscott_ebrhp:read-tail}$
		Node* next = head->next; $\label[line]{code:michaelscott_ebrhp:read-next}$
§H§		@protect(next, 1);@
		if (head != Head) continue;
		if (next == NULL) return EMPTY;
		if (head == tail) {
			CAS(&Tail, tail, next);
			continue;
		} else {
			data_t output = next->data;
			if (CAS(&Head, head, next)) { $\label[line]{code:michaelscott_ebrhp:cas}$
§HEG§			@retire(head);@
§E§				@enterQ;@
§H§				@unprotect(0); unprotect(1);@
				return output; $\label[line]{code:michaelscott_ebrhp:dequeue-return}$
}	}	}	}
\end{lstlisting}%
	\end{minipage}%
	\caption{%
		Michael\&Scott's non-blocking queue with SMR code.
		The added lines are marked with \code{E} and \code{H} for the modifications needed for using EBR/QSBR and HP \cite{DBLP:journals/tpds/Michael04}, respectively.
		The HP version requires two hazard pointers per thread which are set using calls of the form \code{protect(ptr, k)}.
	}
	\label{fig:michaelscott_ebrhp}
	\end{tcolorbox}
\end{figure}

%% file: content/appendix/code/dglm.tex

\begin{figure}[b]
	\begin{tcolorbox}%
	\center%
\begin{lstlisting}[style=condensed]
/* DGLM queue */
struct Node { data_t data; Node* next; };
shared Node* Head, Tail;
atomic init() { Head = new Node(); Head->next = NULL; Tail = Head; }
\end{lstlisting}%
	\begin{minipage}[t]{.46\textwidth}%
\begin{lstlisting}[style=condensed]
void enqueue(data_t input) {
§E§	@leaveQ;@
	Node* node = new Node();
	node->data = input;
	node->next = NULL;
	while (true) {
		Node* tail = Tail;
§H§		@protect(tail, 0);@
§H§		@if (tail != Tail) continue;@
		Node* next = tail->next;
		if (tail != Tail) continue;
		if (next != NULL) {
			CAS(&Tail, tail, next);
			continue;
		}
		if (CAS(&tail->next, next, node))
			break
	}
	CAS(&Tail, tail,node);
§H§	@unprotect(0);@
§E§	@enterQ;@
}
\end{lstlisting}%
	\end{minipage}%
	\hfill%
	\begin{minipage}[t]{.475\textwidth}%
\begin{lstlisting}[style=condensed]
data_t dequeue() {
§E§	@leaveQ;@
	while (true) {
		Node* head = Head;
§H§		@protect(head, 0);@
§H§		@if (head != Head) continue;@
		Node* next = head->next;
§H§		@protect(next, 1);@
		if (head != Head) continue;
		if (next == NULL) return EMPTY;
		data_t output = next->data;
		if (CAS(&Head, head, next)) {
			Node* tail = Tail;
			if (head == tail) {
				CAS(Tail, tail, next);
			}
§HEG§		@retire(head);@
§H§			@unprotect(0); unprotect(1);@
§E§			@enterQ;@
			return output;
}	}	}
\end{lstlisting}%
	\end{minipage}%
	\caption{%
		DGLM non-blocking queue \cite{DBLP:conf/forte/DohertyGLM04} with hazard pointer code.
		The implementation is similar to Michael\&Scott's non-blocking queue but allows the \code{Head} to overtake the \code{Tail}.
	}
	\label{fig:dglm}
	\end{tcolorbox}
\end{figure}

%% file: content/appendix/code/ebr.tex

\begin{figure}
	\center%
\begin{tcolorbox}
\begin{lstlisting}[style=condensed]
enum epoch_t { 0, 1, 2 };
struct EBRRec { EBRRec* next; epoch_t epoch; }
shared EBRRec* Records;
shared epoch_t Epoch;
threadlocal EBRRec* myRec;
threadlocal List<void*> bag0;
threadlocal List<void*> bag1;
threadlocal List<void*> bag2;
\end{lstlisting}%
	\begin{minipage}[t]{.46\textwidth}%
\begin{lstlisting}[style=condensed]
void join() {
	myRec = new EBRRec();
	while (true) {
		EBRRec* rec = Records;
		myRec->next = rec;
		if (CAS(Records, rec, myRec))
			break;
	}
	myEpoch = Epoch;
	myRec->epoch = myEpoch;
}

void part() {
}

void enterQ() {
}

void retire(Node* ptr) {
	bag0.add(ptr);
}
\end{lstlisting}%
	\end{minipage}%
	\hfill%
	\begin{minipage}[t]{.475\textwidth}%
\begin{lstlisting}[style=condensed]
void reclaim() {
	if (*) {
		if (myEpoch != Epoch) { 
			for (Node* ptr : bag2) {
				delete ptr;
			}
			bag2 = bag1;
			bag1 = bag0;
			bag0 = List<void*>::empty();
			myEpoch = Epoch;
			myRec->epoch = myEpoch;
		}

		EBRrec* cur = Records;
		while (cur != NULL) {
			if (epoch != cur->epoch) {
				break;
			}
			cur = cur->next;
		}
		CAS(Epoch, epoch, (epoch+1)%3);
	}
}
\end{lstlisting}%
	\end{minipage}%
	\caption{%
		Implementation for epoch-based reclamation.
		Note that parting of a thread prevents the other threads from any subsequent reclamation.
	}
	\label{fig:ebrimpl}
\end{tcolorbox}
\end{figure}

%% file: content/appendix/smr/base.tex

\section{Extended Discussion of SMR Algorithms}
\label{sec:more_smr}
\label{sec:more_observers}

We extend the discussion of SMR algorithms from \Cref{sec:related_work}.

\input{content/appendix/smr/hp}
\input{content/appendix/smr/freelist}
\input{content/appendix/smr/rc}
\input{content/appendix/smr/dta}
\input{content/appendix/smr/oa_aoa}
\input{content/appendix/smr/st}
\input{content/appendix/smr/ts}
\input{content/appendix/smr/dc}
\input{content/appendix/smr/bc}
\input{content/appendix/smr/isolde}
\input{content/appendix/smr/cadence}
\input{content/appendix/smr/qsense}
\input{content/appendix/smr/he}
\input{content/appendix/smr/ibr}
\input{content/appendix/smr/debra}
\input{content/appendix/smr/dice}

%% file: content/appendix/smr/hp.tex

\paragraph{Hazard Pointers (HP), continued}
We already discussed HP throughout the paper.
We gave a formal specification in \Cref{sec:observers}, \Cref{fig:observer}.
The presented observer handles each hazard pointer individually.
In practice, a common pattern is to protect the nodes of a list in a \emph{hand-over-hand} fashion:
\begin{inparaenum}[(1)]
	\item hazard pointer $h_0$ is set to protect node $a$,
	\item $h_1$ is set to protect the successor node $a'$, \label{item:hand-over-hand:set1}
	\item $h_1$ is transferred into $h_0$ meaning that now also $h_0$ protects $a'$,
	\item $h_1$ is set to protect the successor $a''$ of $a'$, and so on.
\end{inparaenum}
Using this pattern, node $a'$ is continuously protected; the protections of $h_0$ and $h_1$ overlap in time.
If $a'$ is retired after step \ref{item:hand-over-hand:set1} it must not be deleted as long as $h_0$ protects it.
However, $\hpobs$ does not account for this.
Because it considers every hazard pointer individually, the protection by $h_0$ appears after $a'$ is retired.
So the $\hpobs$ does allow for $a'$ being deleted.

Such \emph{spurious} deletions are likely to cause false alarms during an analysis.
The reason for this is that one typically performs consistency checks ensuring that $a''$ is still a successor of $a'$ before moving on.
Due to the spurious deletion this check results in an unsafe memory access reported as bug by most analyses.
To overcome this problem, one can modify $\hpobs$ such that it supports protections by multiple hazard pointers.
To that end, it needs to track multiple hazard pointers simultaneously.
Defining such an observer for a fixed number of hazard pointers per thread is straight forward and omitted here.

%% file: content/appendix/smr/freelist.tex

\paragraph{Free-Lists}
The simplest mechanism to avoid use-after-free bugs is to not free any memory.
Instead, retired memory is stored in a \emph{free-list}.
Memory in that list can be reused immediately.
Instead of an allocation it is checked whether there is memory available in the free-list.
If so, it is removed from the list and reused.
Otherwise, fresh memory is allocated.
There are two major drawbacks with this approach.
First, the memory consumption of an application can never shrink.
Once allocated, memory always remains allocated for the process.
Second, the possibility for memory being immediately reused allows for harmful ABAs \cite{DBLP:conf/podc/MichaelS96}.
To overcome this, pointers are instrumented to carry an integer \emph{tag}, or modification counter.
Updating a pointer then not only updates the pointer but also increases the tag.
To solve avoid harmful ABAs, comparisons of pointers consider their corresponding tags too.
This requires to use double-word \code{CAS} operations or to \emph{steal} unused bits of pointers to use as storage for the tag.

That retired nodes are never freed can be formalized with observers easily.
Moving retired nodes from the free-list to the data structure, however, does not fit into our development---we assume that reuse is implemented by the SMR algorithm freeing nodes and the data structure reallocating the memory.
One could follow the approach by \citet{DBLP:conf/tacas/AbdullaHHJR13,DBLP:conf/vmcai/HazizaHMW16}.
Instead of retiring nodes they use explicit frees and allow the data structure to access freed memory.
As discussed in \Cref{sec:reduction}, such accesses result in pointer races and lead to verification failure.
To counteract this, \citet{DBLP:conf/vmcai/HazizaHMW16} introduce a relaxed version of pointer races to tolerate certain unsafe accesses.
We refrain from doing so since we believe that this would jeopardize our reduction result.
We rather focus on \emph{proper} SMR algorithms which actually reclaim memory.

%% file: content/appendix/smr/rc.tex

\paragraph{Reference Counting (RC)}
Reference counting adds to every node a counter representing the number of pointers currently pointing to that node.
This count is increased whenever a new pointer is created and decreased when an existing pointer is destroyed.
When the count drops to zero upon decrease, then the node is deleted.
However, in order to update the reference count of a node, a pointer to that node is required.
A RC implementation must ensure that a node cannot be deleted between reading a pointer to it and increasing its reference count.
This becomes problematic if the implementation is supposed to be lock-free.
Then, one has to either
\begin{inparaenum}[(i)]
 	\item perform the read and update atomically with a double-word CAS \cite{DBLP:conf/podc/DetlefsMMS01}, which is not available on most hardware, or
 	\item use hazard pointers to protect nodes form being freed in the aforementioned time frame \cite{DBLP:journals/tocs/HerlihyLMM05}.
\end{inparaenum}
Moreover, reference counting has shown to be less efficient than EBR and HPs \cite{DBLP:journals/jpdc/HartMBW07}

Since RC is implemented on top of HP, one can consider the RC implementation as part of the data structure on top of an HP implementation.
Alternatively, to directly fit RC into our theory, one could to instrument the program under scrutiny such that creating new and destroying existing pointers accounts for an increase and decrease in the reference count, respectively.
For pointers $\psel{\anadr}$ one can easily come up with an observer that prevents deleting the referenced address.
To account for pointer variables, one has to either assume a fixed number of shared and thread-local pointer variables or generalize observers to be able to observe such variables.
We believe that such a generalization is possible.

%% file: content/appendix/smr/dta.tex

\paragraph{Drop the Anchor (DTA) \cite{DBLP:conf/spaa/BraginskyKP13}}
A version of EBR which uses HP to fight thread failures.
DTA has no notion of global epoch; instead each thread has a local epoch which is always updated on operation invocation to a value higher than all other thread epochs.
Additionally, the thread epochs are written to the nodes of the data structure to denote when the node was added/removed.
Then, a node can be deallocated if its removal time is smaller than all thread epochs.
To allow for deallocations after thread failure, a version of hazard pointers is used.
Here, a hazard pointer does not protect a single node, but a certain (configurable) number $k$ of consecutive nodes starting from the protected node.
When a threads local epoch lacks behind too much, it is marked as potentially crashed, or \emph{stuck}.
Then, the sublist of nodes protected by the crashed thread is replaced with a fresh copy.
The removed sublist is said to be \emph{frozen}.
On the one hand, the sublist can be reclaimed only if the crashed thread resurrects.
On the other hand, the sublist is marked such that a resurrected thread can detect what happened.
For future deallocations the crashed thread can be ignored because it does not acquire new pointers into the data structure.

\newcommand{\obsDTA}{\anobs[\mathit{DTA}]}
\newcommand{\obsFROZEN}{\anobs[\mathit{Frozen}]}

To fit DTA into our framework we consider the recovery logic to be part of the data structure rather than the SMR algorithm.
Using a feedback mechanism for the recovery, similar to QSBR quiescent phases, allows for DTA to be specified with the observer observers $\baseobs\times\obsDTA\times\obsFROZEN$.
For observers $\obsDTA$ and $\obsFROZEN$ refer to \Cref{fig:observer:dta}.
The observers assume API functions \code{recovered($\athread$)} and \code{freeze(ptr)} to signal that a failed thread has been recovered and that an address is frozen, respectively.

\begin{figure}
	\begin{tcolorbox}
		\begin{tikzpicture}[->,>=stealth',shorten >=1pt,auto,node distance=2.9cm,thick,initial text={}]
			\node [xshift=-0.4cm,yshift=.9cm,draw,thin] {$\obsDTA$};
			\tikzstyle{every state}=[minimum size=1.5em]
			\tikzset{every edge/.append style={font=\footnotesize}}
			\node[initial,state]    (A)              {\mkstatename{obs:dta:init}};
			\node[state]            (B) [right of=A] {\mkstatename{obs:dta:active-inv}};
			\node[state]            (E) [right of=B] {\mkstatename{obs:dta:active-res}};
			\node[state]            (C) [right of=E] {\mkstatename{obs:dta:retired}};
			\node[accepting,state]  (D) [right of=C] {\mkstatename{obs:dta:final}};
			\coordinate             [below of=A, yshift=+1.9cm]  (X)  {};
			\coordinate             [below of=E, yshift=+1.9cm]  (Y)  {};
			\coordinate             [below of=B, yshift=+1.9cm]  (Z)  {};
			\path
				(A) edge node {\translab{\evt{\leaveQ}{\athread}}{\athread=\anovar}} (B)
				(B) edge node {\translab{\evt{\exit}{\athread}}{\athread=\anovar}} (E)
				(E) edge node[align=center] {\translabbr{\evt{\retire}{\athread,\anadr}}{\anadr=\anovarp}} (C)
				(C) edge node {\translab{\freeof{\anadr}}{\anadr=\anovarp}} (D)
				(C.south west) edge[-,shorten >=0pt] (Y)
				([xshift=-1.5mm]X) edge ([xshift=-1.5mm]A.south)
				(Y) edge[-,shorten >=0pt] node[text width=2.85cm] {\translab{\evt{\enterQ}{\athread}}{\athread=\anovar}\newline\translab{\evt{\mathtt{recovered}}{\athread,\athread'}}{\athread'=\anovar}} ([xshift=-1.5mm]X)
				(E) edge[-,shorten >=0pt] ([yshift=1mm]Y.north)
				([yshift=1mm]Y) edge[-,shorten >=0pt] ([xshift=0mm,yshift=1mm]X)
				([xshift=0mm,yshift=1mm]X) edge ([xshift=0mm]A.south)
				(B) edge[-,shorten >=0pt] ([yshift=2mm]Z.north)
				([yshift=2mm]Z) edge[-,shorten >=0pt] ([xshift=1.5mm,yshift=2mm]X)
				([xshift=1.5mm,yshift=2mm]X) edge ([xshift=1.5mm]A.south)
				;
		\end{tikzpicture}
		\tikzimagespace
		\begin{tikzpicture}[->,>=stealth',shorten >=1pt,auto,node distance=3.87cm,thick,initial text={}]
			\node [xshift=-0.3cm,yshift=.9cm,draw,thin] {$\obsFROZEN$};
			\tikzstyle{every state}=[minimum size=1.5em]
			\tikzset{every edge/.append style={font=\footnotesize}}
			\node[initial,state]   (C)              {\mkstatename{obs:frozen:final}};
			\node[state]           (E) [right of=C] {\mkstatename{obs:frozen:init}};
			\node[state]           (A) [right of=E] {\mkstatename{obs:frozen:exit}};
			\node[accepting,state] (B) [right of=A] {\mkstatename{obs:frozen:retired}};
			\path
				(C) edge node [above,align=center]{\translabbr{\evt{\mathtt{freeze}}{\athread,\anadr}}{\athread=\anovar\wedge\anadr=\anovarp}} (E)
				(E) edge node [above]{\translab{\evt{\exit}{\athread}}{\athread=\anovar}} (A)
				(A) edge node [above]{\translab{\freeof{\anadr}}{\anadr=\anovarp}} (B)
				;
		\end{tikzpicture}
		\caption{%
			Observers for specifying DTA.
			Observer $\obsDTA$ extends $\ebrobs$ from \Cref{fig:observer-ebr} in that a recovered thread stops protecting addresses.
			Observer $\obsFROZEN$ prevents frozen addresses from being deleted.
		}
		\label{fig:observer:dta}
	\end{tcolorbox}
\end{figure}

%% file: content/appendix/smr/oa_aoa.tex

\paragraph{Optimistic Access (OA) \cite{DBLP:conf/spaa/CohenP15}}
OA relies on the assumption that accessing freed memory never leads to system crashes.
Hence, memory reads are performed optimistically in that they may accesses reclaimed memory.
To detect such situations, each thread has a \emph{warning bit}, i.e., a simplified thread local epoch.
The warning bit is set if any thread performs reclamation.
Hence, if the warning bit is set after a memory accesses, then the result may be undefined as it stems from reclaimed memory.
If so the thread needs to restart its operation---a common pattern in lock-free code.
To protect writes from corrupting the data structure, they must not write to reclaimed memory.
This is done by using hazard pointers.
The node to be written is protected with a hazard pointer and the integrity of the protection is guaranteed if the warning bit is not set before and after the protection.

The specification of OA is identical to the one of hazard pointers as discussed in \Cref{sec:observers}.
Since data structure implementations using OA access the warning bit, this bit has to be retained in the data structure.
To update the bit, the semantics of \code{free} can be adapted.
This adaptation can be problematic in our theory: a \code{free} may not be mimicked and thus the update of the warning bit may break our development.
To fight this, one could introduce \emph{spurious} updates of the warning bit.

Another problem with OA is that it deliberately accesses potentially reclaimed memory.
That is, it raises pointer races.
To support this one would need a relaxation of pointer races like the one by \citet{DBLP:conf/vmcai/HazizaHMW16}.
This might lead to a severe state space explosion.

\paragraph{Automatic Optimistic Access (AOA) \cite{DBLP:conf/oopsla/CohenP15}}
Extends OA by using a normalized form applicable for lock-free code and by using the properties of OA (namely, a node must not be collected if it is reachable from the shared variables or from a hazard pointer) to release the programmer from the burden of explicitly adding \code{retire} calls to the code.
Moreover, the normalization of the code allows for automatically adding the necessary code to perform OA.
Hence, the authors consider AOA to be a restricted form of garbage collection.

Since AOA is an extension of OA, it has the same requirements as OA.
We believe that a generalization to OA allows for AOA support.

%% file: content/appendix/smr/st.tex

\paragraph{StackTrack (ST) \cite{DBLP:conf/eurosys/AlistarhEHMS14}}
The basic idea is to use transactions to perform optimistic memory accesses.
The transactional memory guarantees that a transaction aborts if it reads from memory that has been freed concurrently.
Since not every operation can be performed in a single transaction, it is split up into multiple transactions.
To guarantee that references are not freed between two transactions without being noticed by the executing thread, the relevant pointers are protected using hazard pointers.

%% file: content/appendix/smr/ts.tex

\paragraph{ThreadScan (TS) \cite{DBLP:conf/spaa/AlistarhLMS15}}
Modified version of hazard pointers which uses signaling to protect nodes on demand.
More specifically, whenever a thread starts reclaiming memory, it uses the operating system's signaling infrastructure to inform all other threads about the reclamation intent.
Then, the other threads are interrupted and execute the signal handler.
This handler scans the list of to-be-deleted nodes of the reclaiming thread and marks those entries the executing thread holds thread-local references to.
The reclaimer waits until all other threads have signaled completion of the marking process.
Then, the reclaimer deletes all unmarked nodes.

\newcommand{\obsTSmark}{\anobs[\mathit{TS-mark}]}
\newcommand{\obsTSdone}{\anobs[\mathit{TS-done}]}

To verify data structures using TS with our theory, one has to integrate signaling into the programing language from \Cref{sec:programs}.
For a specification with observers we assume the following API: \code{reclaim()}, \code{retire(ptr)}, \code{mark(ptr)}, and \code{markdone()}.
Function \code{reclaim()} is called by a thread when it starts reclaiming memory.
Function \code{retire(ptr)} behaves as for ordinary HP.
Function \code{makr(ptr)} is used by the signal handler of threads to mark addresses that shall not be reclaimed.
After a thread has completed the marking, function \code{markdone()} is called.
We use \code{mark(ptr)} instead of \code{protect(ptr)} because it does guarantee protection only during the current reclamation phase.
Then, the observer $\baseobs\times\obsTSmark\times\obsTSdone$ from \Cref{fig:observer:ts} specifies TS.

\begin{figure}
	\begin{tcolorbox}
		\center
		\begin{tikzpicture}[->,>=stealth',shorten >=1pt,auto,node distance=3.8cm,thick,initial text={}]
			\node [xshift=-0.4cm,yshift=.8cm,draw,thin] {$\obsTSmark$};
			\tikzstyle{every state}=[minimum size=1.5em]
			\tikzset{every edge/.append style={font=\footnotesize}}
			\node[initial, state]         (A)                         {\mkstatename{obs:tsmark:init}};
			\node[state]                  (B) [right of=A]            {\mkstatename{obs:tsmark:retired}};
			\coordinate [below=1.0cm of A](X) {};
			\node[accepting,state]        (C) [right of=X]            {\mkstatename{obs:tsmark:final}};
			\path
				([yshift=-1mm]B.west) edge node [below,text width=2.3cm]{\translab{\evt{mark}{\athread,\anadr}}{\anadr=\anovarp}\\\translab{\evt{\retire}{\athread,\anadr}}{\anadr=\anovarp}} ([yshift=-1mm]A.east)
				([yshift=1mm]A.east) edge node [above]{\translab{\evt{reclaim}{\athread}}{\mathit{true}}} ([yshift=1mm]B.west)
				(A) edge[-,shorten >=0pt] (X)
				(X) edge node [below]{\translab{\freeof{\anadr}}{\anadr=\anovarp}} (C)
				;
		\end{tikzpicture}
		\hspace{1cm}
		\begin{tikzpicture}[->,>=stealth',shorten >=1pt,auto,node distance=3.8cm,thick,initial text={}]
			\node [xshift=-0.4cm,yshift=.8cm,draw,thin] {$\obsTSdone$};
			\tikzstyle{every state}=[minimum size=1.5em]
			\tikzset{every edge/.append style={font=\footnotesize}}
			\node[initial, state]         (A)                         {\mkstatename{obs:tsdone:init}};
			\node[state]                  (B) [right of=A]            {\mkstatename{obs:tsdone:retired}};
			\coordinate [below=1.0cm of A](X) {};
			\node[accepting,state]        (C) [right of=X]            {\mkstatename{obs:tsdone:final}};
			\path
				([yshift=-1mm]B.west) edge node [below]{\translab{\evt{reclaim}{\athread}}{\mathit{true}}} ([yshift=-1mm]A.east)
				([yshift=1mm]A.east) edge node [above]{\translab{\evt{markdone}{\athread}}{\athread=\anovar}} ([yshift=1mm]B.west)
				(A) edge[-,shorten >=0pt] (X)
				(X) edge node [below]{\translab{\freeof{\anadr}}{\anadr=\anovarp}} (C)
				;
		\end{tikzpicture}
		\caption{%
			Observer $\obsTSmark$ states that marked addresses must not be deleted.
			Additionally, it guarantees that addresses retired after reclamation starts are not deleted either---this is required because after a thread executed its signal handler to mark addresses, it continues its execution and ca retire addresses; those address must not be deleted by a reclamation phase that started earlier.
			Observer $\obsTSdone$ states that deletions can be performed only if all threads have finished marking nodes.
		}
		\label{fig:observer:ts}
	\end{tcolorbox}
\end{figure}

%% file: content/appendix/smr/dc.tex

\paragraph{Dynamic Collect (DC) \cite{DBLP:conf/podc/DragojevicHLM11}}
DC is an abstract framework to describe SMR schemes like hazard pointers, that is, schemes where single addresses can be protected.
A DC object supplies the user with the ability to
\begin{inparaenum}[(i)]
	\item register,
	\item deregister,
	\item update, and
	\item collect entries.
\end{inparaenum}
Intuitively, this means to dynamically
\begin{inparaenum}[(i)]
	\item create,
	\item purge, and
	\item update hazard pointers, and to
	\item query all hazard pointers together with their values.
\end{inparaenum}
Moreover, the authors explore how such a dynamic collect object could be implemented efficiently using transactions.

The specification of DC objects follows the one of hazard pointers.
That the implementation of such objects may use transactions is irrelevant for the use with our results since we use the specification rather than the implementation of the SMR algorithm for an analysis.

%% file: content/appendix/smr/bc.tex

\paragraph{Beware\&Cleanup (BC) \cite{DBLP:conf/ispan/GidenstamPST05}}
This technique combines hazard pointers and reference counting.
HPs are used to protect thread-local references (pointers).
RC is used for nodes.
The reference count of a node, however, reflects only the number of existing references from other nodes.
That is, acquiring/purging thread-local pointers does not affect the reference count of nodes.

Additionally, BC features a reclamation procedure which guarantees that the number of unreclaimed nodes is finite at all times.
To that end, the references contained in logically deleted nodes (those that await reclamation) are altered to point to nodes still contained in the data structure.
However, updating logically deleted nodes is a modification of the data structure.
As stated by the authors, BC requires this modification to \emph{retain semantics}.

Similar to ordinary RC, BC can be handled with our approach when considering the reference counting as part of the data structure that is done on top of HP.
With respect to the aforementioned instrumentation for RC, BC has the advantage that no generalization of our observers is required since pointer variables are covered by HP rather than RC.

%% file: content/appendix/smr/isolde.tex

\paragraph{Isolde \cite{DBLP:conf/iwmm/YangW17}}

Isolde combines epoch-based reclamation and reference counting.
Similar to Beware\&Cleanup, reference counting is used for incoming pointers from nodes.
References due to thread-local pointers are ignored for the count.
Instead, EBR is used for pointer variables.
The main contribution of this work is a type system for the proposed reclamation approach.

%% file: content/appendix/smr/cadence.tex

\paragraph{Cadence \cite{DBLP:conf/spaa/BalmauGHZ16}}
Cadence is an optimized implementation of hazard pointers.
A major drawback of HP under weak memory is that each protection must be followed by a memory fence in order be visible to other threads.
This may degrade performance because it requires a fence for every node when traversing a data structure.
To avoid those fences, the authors observe that a context switch in modern operating systems can act as a fence.
Hence, Cadence adds sufficiently many dummy processes which simply sleep for a predefined time interval $T$ in an infinite loop.
This way, every process is forced into a context switch after $T$ time has elapsed.
For hazard pointers this means that a protection is visible to other threads at most $T$ time after the protection took place.
So for memory reclamation, processes wait an additional time $T$ between retiring and trying to delete a node.

As a hazard pointer implementation, Cadence uses the exact same specification as ordinary HP.

%% file: content/appendix/smr/qsense.tex

\paragraph{QSense \cite{DBLP:conf/spaa/BalmauGHZ16}}
Combination of QSBR and Cadence.
By default, QSBR is used.
If a thread failure/delay is detected, then the system switches to Cadence, i.e., HP.
Only if all threads make progress, the system switchs from Cadence back to QSBR.

\newcommand{\obsQSebr}{\anobs[\mathit{QS-fast}]}
\newcommand{\obsQShp}{\anobs[\mathit{QS-slow}]}

For a specification of QSense with observers, we assume an EBR and HP API which is extended by two functions \code{goSlow()} and \code{goFast()} that switch from QSBR to Cadence, and vice versa.
Then, QSense can be specified by $\baseobs\times\obsQSebr\times\obsQShp$.
Observer $\obsQSebr$ can be found in \Cref{fig:qsense}.
It is a variant of $\ebrobs$ which can be switched on and off on demand.
Observer $\obsQShp$ follows from $\hpobs$ with a similar construction.

\begin{figure}
	\begin{tcolorbox}
		\center
		\begin{tikzpicture}[->,>=stealth',shorten >=1pt,auto,node distance=3.8cm,thick,initial text={}]
			\node [xshift=-0.4cm,yshift=2.3cm,draw,thin] {$\obsQSebr$};
			\tikzstyle{every state}=[minimum size=1.5em]
			\tikzset{every edge/.append style={font=\footnotesize}}
			\node[initial,state]    (E)              {\mkstatename{obs:qs-ebr:init}};
			\node[state]            (F) [right of=E] {\mkstatename{obs:qs-ebr:active}};
			\node[state]          (foo) [right of=F] {\mkstatename{obs:qs-ebr:foo}};
			\node[state]            (G) [right of=foo] {\mkstatename{obs:qs-ebr:retired}};
			\node[accepting,state]  (H) [above of=G,yshift=-1.8cm] {\mkstatename{obs:qs-ebr:final}};
			\coordinate             [above of=E, yshift=-2.75cm]  (U)  {};
			\coordinate             [above of=foo, yshift=-2.75cm]  (V)  {};
			\coordinate             [above of=F, yshift=-2.75cm]  (W)  {};
			\node[state]            (A) [below of=E,yshift=2cm] {\mkstatename{obs:qs-ebr:init-copy}};
			\node[state]            (B) [right of=A] {\mkstatename{obs:qs-ebr:active-copy}};
			\node[state]          (bar) [right of=B] {\mkstatename{obs:qs-ebr:active-bar}};
			\node[state]            (C) [right of=bar] {\mkstatename{obs:qs-ebr:retired-copy}};
			\coordinate             [below of=A, yshift=+2.75cm]  (X)  {};
			\coordinate             [below of=bar, yshift=+2.75cm]  (Y)  {};
			\coordinate             [below of=B, yshift=+2.75cm]  (Z)  {};
			\path
				(E) edge node[below] {\translab{\evt{\leaveQ}{\athread}}{\athread=\anovar}} (F)
				(F) edge node[below] {\translab{\evt{\exit}{\athread}}{\athread=\anovar}} (foo)
				(foo) edge node[below] {\translab{\evt{\retire}{\athread,\anadr}}{\anadr=\anovarp}} (G)
				(G) edge node[left,pos=.7] {\translab{\freeof{\anadr}}{\anadr=\anovarp}} (H)
				(G.north west) edge[-,shorten >=0pt] (V)
				([xshift=-1.5mm]U) edge ([xshift=-1.5mm]E.north)
				([xshift=-1.5mm]U) edge[-,shorten >=0pt] node {\translab{\evt{\enterQ}{\athread}}{\athread=\anovar}} (V)
				(foo) edge[-,shorten >=0pt] ([yshift=-1mm]V.north)
				([yshift=-1mm]V) edge[-,shorten >=0pt] ([xshift=0mm,yshift=-1mm]U)
				([xshift=0mm,yshift=-1mm]U) edge ([xshift=0mm]E.north)
				(F) edge[-,shorten >=0pt] ([yshift=-2mm]W.south)
				([yshift=-2mm]W) edge[-,shorten >=0pt] ([xshift=1.5mm,yshift=-2mm]U)
				([xshift=1.5mm,yshift=-2mm]U) edge ([xshift=1.5mm]E.north)
				;
			\path
				(A) edge node {\translab{\evt{\leaveQ}{\athread}}{\athread=\anovar}} (B)
				(B) edge node {\translab{\evt{\exit}{\athread}}{\athread=\anovar}} (bar)
				(bar) edge node {\translab{\evt{\retire}{\athread,\anadr}}{\anadr=\anovarp}} (C)
				(C.south west) edge[-,shorten >=0pt] (Y)
				([xshift=-1.5mm]X) edge ([xshift=-1.5mm]A.south)
				(Y) edge[-,shorten >=0pt] node {\translab{\evt{\enterQ}{\athread}}{\athread=\anovar}} ([xshift=-1.5mm]X)
				(bar) edge[-,shorten >=0pt] ([yshift=1mm]Y.north)
				([yshift=1mm]Y) edge[-,shorten >=0pt] ([xshift=0mm,yshift=1mm]X)
				([xshift=0mm,yshift=1mm]X) edge ([xshift=0mm]A.south)
				(B) edge[-,shorten >=0pt] ([yshift=2mm]Z.north)
				([yshift=2mm]Z) edge[-,shorten >=0pt] ([xshift=1.5mm,yshift=2mm]X)
				([xshift=1.5mm,yshift=2mm]X) edge ([xshift=1.5mm]A.south)
				;
			\path
				([xshift=1mm]E.south) edge[] node {\$} ([xshift=1mm]A.north)
				([xshift=-1mm]A.north) edge node {\#} ([xshift=-1mm]E.south)
				([xshift=1mm]F.south) edge[] node {\$} ([xshift=1mm]B.north)
				([xshift=-1mm]B.north) edge node {\#} ([xshift=-1mm]F.south)
				([xshift=1mm]G.south) edge[] node {\$} ([xshift=1mm]C.north)
				([xshift=-1mm]C.north) edge node {\#} ([xshift=-1mm]G.south)
				([xshift=1mm]foo.south) edge[] node {\$} ([xshift=1mm]bar.north)
				([xshift=-1mm]bar.north) edge node {\#} ([xshift=-1mm]foo.south)
				;
			\node (text) [text width=2.75cm,align=left,above of=F,yshift=-1.6cm,xshift=3.5cm]  {%
				\footnotesize
				\[
				\begin{aligned}
					\$ &:= \translab{\evt{goSlow}{\athread}}{\mathit{true}}
					\\
					\# &:= \translab{\evt{goFast}{\athread}}{\mathit{true}}
				\end{aligned}
				\]
			};
		\end{tikzpicture}
		\caption{%
			Observer specifying a version of QSBR for QSense which can be switched on and off on demand.
		}
		\label{fig:qsense}
	\end{tcolorbox}
\end{figure}

%% file: content/appendix/smr/he.tex

\paragraph{Hazard Eras (HE) \cite{DBLP:conf/spaa/RamalheteC17}}
A hazard pointer implementation that borrows a global clock from QSBR.
HE augments nodes with two additional time stamp fields: a creation time and a retire time.
Whenever a node is created/retired, its corresponding field is set to the current value of the global clock.
Additionally, the global clock is advanced if a node is retired.
Unlike with hazard pointers, a protection does not announce the pointer of the to-be-protected node but the current global clock.
Then, the deletion of a retired node is deferred if some thread has announced a time that lies in between that nodes creation and retirement time.

This adaptation of hazard pointers was specifically designed to feature the same API as HP.
For a specification of HE, however, we need to make explicit the time stamps that are announced and contained in nodes.
To support this SMR algorithm in our theory, one would need to adapt the definition of histories such that an event containing an address $\anadr$ is extended to also contain the values of $\anadr$ time stamp selector.
We believe that such a generalization can be implemented with minor changes to our development.

%% file: content/appendix/smr/ibr.tex

\paragraph{Interval-Based Reclamation (IBR) \cite{DBLP:conf/ppopp/WenICBS18}}
Reclamation scheme similar to HE.
Like in HE, there is a global clock and nodes carry a creation and retirement time stamp.
Unlike in HE, the global clock is advanced after a configurable number of allocations.
For protection, a thread announces a time interval.
Then, every node the creation-retirement interval of which has a non-empty intersection with the announced interval is protected.
That is, upon operation start a thread announces an interval containing only the current global epoch.
Whenever a pointer is read, the announced interval is extended to include the current epoch.

%% file: content/appendix/smr/debra.tex

\paragraph{DEBRA \cite{DBLP:conf/podc/Brown15}}
DEBRA is an QSBR implementation.
There is also a version that adds fault tolerance using signals and non-local gotos.
The main idea is to \emph{neutralize} threads by sending them a signal.
The signal handler then forces the thread to enter a quiescent state and restart its operation.
Intuitively, the restart is done via non-local goto: it resets the thread's state back to when it started executing the operation.

Since this work gives an QSBR implementation, it uses the same specification as ordinary QSBR.

%% file: content/appendix/smr/dice.tex

\paragraph{\cite{DBLP:conf/iwmm/DiceHK16}}
Optimized HP implementation for weak memory which exploits memory page management of modern operating systems.
More specifically, the authors observe that issuing a \emph{write-protect} to a memory page forces each processor to first finish pending writes to the to-be-protected memory page.
One can exploit this behavior write-protecting all pages that contain hazard pointers before starting to reclaim memory.

Since this work gives a HP implementation, it uses the exact same specification as ordinary HP.

%% file: content/appendix/theory/base.tex

\newcommand{\addtodoc}[1]{\clearpage\newpage\input{#1}}

\section{Missing Details}
\label{appendix:definitions}

We provide details missing in the main paper.

\input{content/appendix/theory/definitions}
\input{content/appendix/theory/composition}

\input{content/appendix/theory/generalization}
\clearpage\newpage\input{content/appendix/theory/theory}

%% file: content/appendix/theory/definitions.tex

\subsection{Programs}
\label{appendix:definitions:programs}

\begin{definition}[Lifted Memory]
	We lift a memory $\aheap$ to sets by $\aheap(E):=\setcond{\aheap(e)}{e\in E}\setminus\set{\segval}$.
\end{definition}

\begin{definition}[In-use Addresses]
	The \emph{in-use addresses} in a memory $\aheap$ are \[\adrof{\aheap}:=(\rangeof{\aheap}\cup\domof{\aheap})\cap\adr\] where we use $\set{\psel{\anadr}}\cap\adr=\anadr$ and similarly for data selectors.
\end{definition}

\begin{definition}[Fresh Addresses]
	\label{def:fresh-addresses}
	The \emph{fresh} addresses in $\tau$, denoted by $\freshof{\tau}\subseteq\adr$, are:
	\begin{align*}
		\freshof{\epsilon} = &~ \adr \\
		\freshof{\tau.\anact} = &~ \freshof{\tau}\setminus\set{\anadr} &&\text{if } \anact=(\athread,\freeof{\anadr},\anup) \\
		\freshof{\tau.\anact} = &~ \freshof{\tau}\setminus\set{\anadr} &&\text{if } \anact=(\athread,\freeof{\apvar},\anup) \wedge \heapcomputof{\tau}{\apvar}=\anadr \\
		\freshof{\tau.\anact} = &~ \freshof{\tau}\setminus\set{\anadr} &&\text{if } \anact=(\athread,\apvar:=\malloc,\anup) \wedge \heapcomputof{\tau.\anact}{\apvar}=\anadr \\
		\freshof{\tau.\anact} = &~ \freshof{\tau} &&\text{otherwise.}
	\end{align*}
\end{definition}

\begin{definition}[Freed Addresses]
	\label{def:freed-addresses}
	The \emph{freed} addresses in $\tau$, denoted by $\freedof{\tau}\subseteq\adr$, are:
	\begin{align*}
		\freedof{\epsilon} = &~ \emptyset \\
		\freedof{\tau.\anact} = &~ \freedof{\tau} \cup \set{\anadr} &&\text{if } \anact=(\athread,\freeof{\anadr},\anup) \\
		\freedof{\tau.\anact} = &~ \freedof{\tau} \cup \set{\anadr} &&\text{if } \anact=(\athread,\freeof{\apvar},\anup) \wedge \heapcomputof{\tau}{\apvar}=\anadr \\
		\freedof{\tau.\anact} = &~ \freedof{\tau} \setminus \set{\anadr} &&\text{if } \anact=(\athread,\apvar:=\malloc,\anup) \wedge \heapcomputof{\tau.\anact}{\apvar}=\anadr \\
		\freedof{\tau.\anact} = &~ \freedof{\tau} &&\text{otherwise.}
	\end{align*}
\end{definition}

\begin{definition}[Allocated Addresses]
	The \emph{allocated} addresses in $\tau$ are those addresses that are neither fresh nor freed: $\allocatedof{\tau}:=\adr\setminus(\freshof{\tau}\cup\freedof{\tau})$.
\end{definition}

\begin{definition}[Induced Histories]
	\label{def:induced_histories}
	The history \emph{induced by} a computation $\tau$, denoted by $\historyof{\tau}$, is defined by:
	\begin{align*}
		\historyof{\epsilon} = &~ \epsilon \\
		\historyof{\tau.(\athread,\freeof{\anadr},\anup)} = &~ \historyof{\tau} \hconcat \freeof{\anadr} \\
		\historyof{\tau.(\athread,\freeof{\apvar},\anup)} = &~ \historyof{\tau} \hconcat \freeof{\heapcomputof{\tau}{\apvar}} \\
		\historyof{\tau.(\athread,\enterof{\afuncof{\vecof{\apvar},\vecof{\advar}}},\anup)} = &~ \historyof{\tau} \hconcat \evt{\afunc}{\athread,\heapcomputof{\tau}{\vecof{\apvar}},\heapcomputof{\tau}{\vecof{\advar}}} \\
		\historyof{\tau.(\athread,\exit,\anup)} = &~ \historyof{\tau} \hconcat \exitof{\athread} \\
		\historyof{\tau.\anact} = &~  \historyof{\tau} &&\text{otherwise.}
	\end{align*}
\end{definition}

\begin{assumption}[Environment Frees]
	\label{assumption:frees-do-not-affect-control}
	For any $\tau.\anact\in\allsem$ with $\anact=(\athread,\freeof{\anadr},\anup)$ we have $\controlof{\tau}=\controlof{\tau.\anact}$.
\end{assumption}

%% file: content/appendix/theory/composition.tex

\subsection{Compositionality}
\label{appendix:theory-compositionality}

Consider $\ads[\ansmr]$, a data structure $\adsraw$ using an SMR implementation $\ansmr$.
We split the variables of $\ads[\ansmr]$ into the variables of $\adsraw$ and $\ansmr$.
That is, we have $\pvars=\pvarsds\uplus\pvarssmr$ where $\pvarsds$ and $\pvarssmr$ are the pointer variables of $\adsraw$ and $\ansmr$, respectively.
Similarly, we have $\dvars=\dvarsds\uplus\dvarssmr$.

We use $\controlallof[\athread]{\tau}$ to make precise the return label for calls of $\ansmr$ performed by thread $\athread$ of $\adsraw$ in computation $\tau$.
That is, $\controlallof[\athread]{\tau}=(\apc,\apcp)$ where $\apc$ and $\apcp$ are labels for the next command to be executed in $\adsraw$ and $\ansmr$, respectively.
If $\apcp=\bot$, this meas that $\athread$ is currently executing $\adsraw$ code and is not inside an SMR call.
That $\apcp\neq\bot$ means $\athread$ is inside an SMR function and executing $\ansmr$ code.
In such cases, $\apc$ holds the return location.
With respect to \Cref{sec:programs}, this means that $\apc$ holds a label to the corresponding $\exit$ command.
We assume that program steps change $\apc$ only if $\apcp=\bot$.
For simplicity, we write $\controlallof{\tau}$ and mean a map from threads $\athread$ to $\controlallof[\athread]{\tau}$.

Now, we introduce a separation of memory into the parts of $\adsraw$ and $\ansmr$.
The separation is denoted by $\heapcomputall{\tau}=(\heapcomputds{\tau},\heapcomputsmr{\tau})$ and is inductively defined similarly to $\heapcomput{\tau}$.
In the base case, we have $\heapcomputsomeof{\aprog}{\epsilon}{\apvar}=\segval$, $\heapcomputsomeof{\aprog}{\epsilon}{\advar}=0$, and $\heapcomputsomeof{\aprog}{\epsilon}{e}=\bot$ for all $\apvar\in\pvarssome{\aprog}$, $\advar\in\dvarssome{\aprog}$, and $e\notin\pvarssome{\aprog}\cup\dvarssome{\aprog}$ for $\aprog\in\set{\adsraw,\ansmr}$.
In the induction step we have $\heapcomputall{\tau.\anact}=(\heapcomputds{\tau.\anact},\heapcomputsmr{\tau.\anact})$ for $\anact=(\athread,\acom,\anup)$ as follows.
\begin{compactitem}
	\item If $\acom\equiv\freeof{\apvar}$, then $\heapcomputds{\tau.\anact}=\heapcomputds{\tau}[\anup]$ and $\heapcomputsmr{\tau.\anact}=\heapcomputsmr{\tau}[\anup]$.
	\item If $\controlsmrof{\tau}=\bot$, then then $\heapcomputds{\tau.\anact}=\heapcomputds{\tau}[\anup]$.
	\item If $\controlsmrof{\tau}\neq\bot$ and $\acom\not\equiv\freeof{\apvar}$, $\heapcomputsmr{\tau.\anact}=\heapcomputsmr{\tau}[\anup]$.
\end{compactitem}
(Recall that we assume that $\adsraw$ does not perform any $\free$, instead uses $\ansmr$ to perform them safely.)

Now, we can make precise the requirement on $\ads[\ansmr]$ to enable compositional verification.
As stated in \Cref{sec:observers}, we require that $\ansmr$ influences $\adsraw$ only through frees.
Formally, we need a separation of the memory regions they use and that they do not use regions that do not belong to them.
We define this separation on the level of computations.

\begin{definition}[Compositional Computations]
	A computation $\tau\in\allsem[{\ads[\ansmr]}]$ \emph{is compositional}, if for every prefix $\sigma.\anact$ of $\tau$ we have:
	\begin{compactenum}[(i)]
		\item $\heapcomput{\sigma.\anact}=\heapcomputds{\sigma.\anact}\uplus\heapcomputsmr{\sigma.\anact}$,
		\item $\heapcomputof{\sigma.\anact}{e}=\heapcomputdsof{\sigma.\anact}{e}$ for all $e\in\pvarsds\cup\pvarssmr$, and
		\item if $\anact=(\athread,\acom,\anup)$ with $\acom\equiv\mathit{expr}:=\sel{\apvar}{sel}$ and $\heapcomputof{\sigma}{\apvar}=\anadr$, then $\heapcomputof{\sigma}{\sel{\anadr}{sel}}=\heapcomputdsof{\sigma}{\sel{\anadr}{sel}}$.
	\end{compactenum}
\end{definition}

By definition, $\epsilon$ is compositional.
Moreover, if $\tau.\anact$ is compositional, so is $\tau$.
We lift this definition to programs by considering the set of their computations.

\begin{definition}[Compositional Programs]
	We say $\ads[\ansmr]$ \emph{is compositional} if all $\tau\in\allsem[{\ads[\ansmr]}]$ are.
\end{definition}

Recall from \Cref{sec:observers} that $\satisfies{\ansmr}{\smrobs}$ if $\historyof{\allsem[MGC(\ansmr)]}\subseteq\specof{\smrobs}$.
Together with compositionality of $\ads[\ansmr]$ we can show that $\ads[\ansmr]$ and $\ads$ reach the same control states wrt. $\adsraw$.

\begin{lemma}
	\label{thm:compositional-computations}
	Assume $\satisfies{\ansmr}{\smrobs}$.
	Let $\tau\in\allsem$ be compositional.
	There is $\sigma\in\allsem$ such that $\controldsof{\tau}=\controlof{\sigma}$, $\heapcomputds{\tau}=\heapcomput{\sigma}$, $\freshof{\tau}\subseteq\freshof{\sigma}$, and $\freedof{\tau}\subseteq\freedof{\sigma}$.
\end{lemma}

\begin{proof}[Proof of \Cref{thm:compositional-computations}]
	We proceed by induction over the structure of $\tau$.
	In the base case, we have $\epsilon$ and the claim follows by definition.
	So consider some compositional $\tau.\anact$ such that we have already constructed for $\tau$ a $\sigma$ with the desired properties.
	Now, we need to construct a $\sigma'$ for $\tau.\anact$.
	Let $\anact=(\athread,\acom,\anup)$.
	We distinguish three cases.
	\begin{casedistinction}
		\item[${\controlsmrof[\athread]{\tau}}\neq\bot$ and $\acom\not\equiv\freeof{\apvar}$]
			In this case $\anact$ stems from $\athread$ executing SMR code of $\ansmr$.
			By compositionality of $\tau$ we have $\heapcomputds{\tau.\anact}=\heapcomputds{\tau}$.
			So $\heapcomputds{\tau.\anact}=\heapcomput{\sigma}$ by induction.
			Since $\acom$ is not a $\free$, we get $\freshof{\tau.\anact}\subseteq\freshof{\tau}\subseteq\freshof{\sigma}$ and $\freedof{\tau.\anact}\subseteq\freedof{\tau}\subseteq\freedof{\sigma}$.
			Now, recall that we assumed that actions of $\ansmr$ do not change the control of $\adsraw$.
			That is, we have $\controldsof{\tau.\anact}=\controldsof{\tau}=\controlof{\sigma}$.
			Altogether, this means that $\sigma'=\sigma$ is an adequate choice.

		\item[$\acom\equiv\freeof{\apvar}$]
			Recall that we have assumed that only $\ansmr$ performs frees.
			Hence, we get $\historyof{\tau}.\freeof{\anadr}\in\specof{\smrobs}$ from $\satisfies{\ansmr}{\smrobs}$.
			That is, for $\anactp=(\bot,\freeof{\anadr},\anup)$ we have $\sigma.\anactp\in\allsem$.
			By definition we get $\heapcomputds{\tau.\anact}=\heapcomputds{\tau}[\anup]=\heapcomput{\sigma}[\anup]=\heapcomput{\sigma.\anactp}$.
			Moreover, we have
			\begin{align*}
				\freshof{\tau.\anact}=\freshof{\tau}&\subseteq\freshof{\sigma}=\freshof{\sigma.\anactp}
				\\\text{and}\qquad
				\freedof{\tau.\anact}\freedof{\tau}\cup\set{\anadr}&\subseteq\freedof{\sigma}\cup\set{\anadr}=\freedof{\sigma.\anactp}
				\ .
			\end{align*}
			As in the previous case, we get $\controldsof{\tau.\anact}=\controldsof{\tau}=\controlof{\sigma}$ because $\anact$ stems from executing $\ansmr$.
			Altogether, $\sigma'=\sigma.\anactp$ is an adequate choice.

		\item[${\controlsmrof[\athread]{\tau}}=\bot$]
			We choose $\sigma'=\sigma.\anact$.
			We have to show that $\acom$ is enabled and produces the same update $\anup$ after $\sigma$.
			This follows from $\controldsof{\tau}=\controlof{\sigma}$ and compositionality of $\tau.\anact$.
			The former states that $\acom$ can be executed.
			The latter implies that the same update is performed.
			To see this, note that compositionality gives $\heapcomputof{\tau}{\apvar}=\heapcomputdsof{\tau}{\apvar}$.
			So by induction $\heapcomputof{\tau}{\apvar}=\heapcomputof{\sigma}{\apvar}$.
			And similarly for data variables and selectors appearing on the right-hand side of assignments.
			Since the freed and fresh addresses are not changed, $\sigma'=\sigma.\anact$ is an adequate choice.
	\end{casedistinction}
	The above case distinction is complete.
	This concludes the claim.
\end{proof}

We use the above \namecref{thm:compositional-computations} to establish \Cref{thm:compositionality}.
Here, we focus on establishing safety properties.
Hence, correctness boils down to control state reachability.
We assume that the correctness of $\ads[\ansmr]$ depends only on $\adsraw$.
Let $\bad$ be those control locations of $\adsraw$ that a correct program must not reach.

\begin{definition}[Correctness]
	A data structure $\ads[\ansmr]$ \emph{is correct} if there is no $\tau\in\allsem[{\ads[\ansmr]}]$ with $\controldsof{\tau}\in\bad$.
	Similarly, $\ads$ is correct if there is no $\sigma\in\allsem$ with $\controlof{\sigma}\in\bad$.
\end{definition}

With this notion of correctness avoiding \emph{bad} control states, we are able to prove \Cref{thm:compositionality}.

\begin{proof}[Proof of \Cref{thm:compositionality}]
	Let $\satisfies{\ansmr}{\smrobs}$.
	To show that the correctness of $\ads$ implies the correctness of $\ads[\ansmr]$ we show the contra positive.
	So let $\ads[\ansmr]$ be incorrect.
	By definition, there is $\tau\in\allsem[{\ads[\ansmr]}]$ with $\controldsof{\tau}\in\bad$.
	From \Cref{thm:compositional-computations} we get $\sigma\in\allsem$ such that $\controlof{\sigma}=\controldsof{\tau}$.
	That is, $\controlof{\sigma}\in\bad$.
	Thus, $\ads$ is incorrect.
	This concludes the claim.
\end{proof}

\paragraph{A Remark on Checking Compositionality}
Compositionality can be checked using the techniques from \Cref{sec:reduction}.
If $\adsraw$ breaks compositionality, then this manifests as a pointer race (unsafe access) in $\allsem$.
One can show that such pointer races are absent if $\onesem$ is pointer race free.
So our verification approach establishes that if $\ads[\ansmr]$ breaks compositionality, then it is due to $\ansmr$.

To check whether or not $\ansmr$ breaks compositionality, one could use the same pointer race freedom approach as for $\adsraw$.
However, since the precise shape of $\ansmr$ is a parameter to our development, such a proof is beyond the scope of this paper.
As an alternative approach, one can check whether or not $\ansmr$ modifies the memory of $\adsraw$.
This may allow for a much simpler (data flow) analysis.

%% file: content/appendix/theory/generalization.tex

\subsection{Reduction}
\label{appendix:definitions:reduction}

\begin{definition}[Valid Expressions]
	\label{Definition:Validity}
	The \emph{valid pointer expressions} in a computation $\tau\in\allsem$, denoted by $\validof{\tau}\subseteq\pexp$, are defined by:
	\begin{align*}
		\validof{\epsilon}&:=\pvars
			\\
		\validof{\tau.(\athread, \apvar:=\apvarp, \anup)}
			&:= \validof{\tau}\cup\set{\apvar}
			&&\text{if }\apvarp\in \validof{\tau}
			\\
		\validof{\tau.(\athread, \apvar:=\apvarp, \anup)}
			&:= \validof{\tau}\setminus\set{\apvar}
			&&\text{if }\apvarp\notin \validof{\tau}
			\\
		\validof{\tau.(\athread, \psel{\apvar}:=\apvarp, \anup)}
			&:= \validof{\tau}\cup\set{\psel{\anadr}}
			&&\text{if }\heapcomputof{\tau}{\apvar}=\anadr\in\adr \wedge \apvarp\in \validof{\tau}
			\\
		\validof{\tau.(\athread, \psel{\apvar}:=\apvarp, \anup)}
			&:= \validof{\tau}\setminus\set{\psel{\anadr}}
			&&\text{if }\heapcomputof{\tau}{\apvar}=\anadr\in\adr \wedge \apvarp\notin \validof{\tau}
			\\
		\validof{\tau.(\athread, \apvar:=\psel{\apvarp}, \anup)}
			&:= \validof{\tau}\cup\set{\apvar}
			&&\text{if }\heapcomputof{\tau}{\apvarp}=\anadr\in\adr \wedge \psel{\anadr}\in\validof{\tau}
			\\
		\validof{\tau.(\athread, \apvar:=\psel{\apvarp}, \anup)}
			&:= \validof{\tau}\setminus\set{\apvar}
			&&\text{if }\heapcomputof{\tau}{\apvarp}=\anadr\in\adr \wedge \psel{\anadr}\notin\validof{\tau}
			\\
		\validof{\tau.(\athread, \freeof{\anadr}, \anup)}
			&:=\validof{\tau}\setminus \invalidof{\anadr}
			\\
		\validof{\tau.(\athread, \apvar:=\malloc, \anup)}
			&:=\validof{\tau}\cup\set{\apvar,\psel{\anadr}}
			&&\text{if }[\apvar\mapsto\anadr]\in\anup
			\\
		\validof{\tau.(\athread, \assert\ \apvar=\apvarp, \anup)}
			&:=\validof{\tau}\cup\set{\apvar,\apvarp}
			&&\text{if }\set{\apvar,\apvarp}\cap\validof{\tau}\neq\emptyset
			\\
		\validof{\tau.(\athread, \anact, \anup)}
			&:=\validof{\tau}
			&&\text{otherwise.}
	\end{align*}
	We have $\invalidof{\anadr}:=\setcond{\apvar}{\heapcomputof{\tau}{\apvar}=\anadr} \cup \setcond{\psel{\anadrp}}{\heapcomputof{\tau}{\psel{\anadrp}}=\anadr} \cup \set{\psel{\anadr}}$.
\end{definition}

\begin{definition}[Replacements]
	\label{def:replacements}
	A \emph{replacement of address $\anadr$ to $\anadrp$ in a history $\ahist$}, denoted by $\renamingof{\ahist}{\anadr}{\anadrp}$, replaces in $\ahist$ every occurrence of $\anadr$ with $\anadrp$, and vice versa, as follows:
	\begin{align*}
		\renamingof{\epsilon}{\anadr}{\anadrp} = &\;\epsilon
		\\
		\renamingof{\big(\ahist.\afunc(\vecof{\anadrpp},\vecof{\advalue})\big)}{\anadr}{\anadrp} = &\;\big(\renamingof{\ahist}{\anadr}{\anadrp}\big).\big(\afunc(\vecof{\anadrpp'},\vecof{\advalue})\big)
		\\
		&\;\text{with }
		\anadrpp_i'=\anadr \text{ if } \anadrpp_i=\anadrp,~
		\anadrpp_i'=\anadrp \text{ if } \anadrpp_i=\anadr,~\text{and}~
		\anadrpp_i'=\anadrpp \text{ otherwise}
		\\
		\renamingof{\ahist.\anevent}{\anadr}{\anadrp} = &\;\renamingof{\ahist}{\anadr}{\anadrp} \qquad\text{otherwise.}
	\end{align*}
\end{definition}

\subsection{Generalized Reduction}
\label{appendix:theory-generalization}

In \Cref{sec:reduction} we required the absence of harmful ABAs and elision support as premise for the reduction result.
We present a strictly weaker premise that implies the same reduction result.

\begin{definition}[Reduction Compatibility]
	\label{def:reduction-compatible}
	The set $\onesem$ is \emph{reduction compatible} if:
	\begin{align*}
		\forall\,\tau.\anact&\in\allsem~
		\forall\,\sigma_\anadr.\anact\in\adrsem{\anadr}~
		\forall\,\sigma_\anadrp\in\adrsem{\anadrp}.~~
		\\&
		\tau\computequiv\sigma_\anadr
		\:\wedge\:
		\tau\obsrel\sigma_\anadr
		\:\wedge\:
		\tau\heapequiv[\anadr]\sigma_\anadr
		\:\wedge\:
		\tau.\anact\computequiv\sigma_\anadr.\anact
		\:\wedge\:
		\tau.\anact\obsrel\sigma_\anadr.\anact
		\:\wedge\:
		\\&
		\tau.\anact\heapequiv[\anadr]\sigma_\anadr.\anact
		\:\wedge\:
		\tau\computequiv\sigma_\anadrp
		\:\wedge\:
		\tau\obsrel\sigma_\anadrp
		\:\wedge\:
		\tau\heapequiv[\anadrp]\sigma_\anadrp
		\:\wedge\:
		\sigma_\anadr\computequiv\sigma_\anadrp
		\:\wedge\:
		\anadr\neq\anadrp
		\:\wedge\:
		\\&
		\anact\in\big\{\,(\_,\assertof{\_},\_),\, (\_,\apvar:=\malloc,[\apvar\mapsto\anadr,\_])\,\big\}
		\\\:\implies\:&
		\tau.\anact\computequiv\sigma_\anadrp'
		\:\wedge\:
		\tau.\anact\obsrel\sigma_\anadrp'
		\:\wedge\:
		\tau.\anact\heapequiv[\anadrp]\sigma_\anadrp'
		\intertext{and}
		\forall\,\ahist,\anadr,\anadrp.~~&
		\anadr\neq\anadrp\implies
		\freeableof{\ahist.\freeof{\anadr}}{\anadrp}=\freeableof{\ahist}{\anadrp}
		\ .
	\end{align*}
\end{definition}

The analog of the results from \Cref{sec:reduction} hold for the above generalized definition.

\begin{proposition}
	\label{thm:RPRF-implies-GCplus}
	Let $\onesem$ be reduction compatible and pointer race free.
	Then, for all $\tau\in\allsem$ and all $\anadr\in\adr$ there is $\sigma\in\onesem$ with $\tau\computequiv\sigma$, $\tau\heapequiv[\anadr]\sigma$, and $\tau\obsrel[\anadr]\sigma$.
\end{proposition}

\begin{theorem}
 	\label{thm:rprf-guarantee}
	We have $\allsem\computequiv\onesem$ provided $\onesem$ is reduction compatible and pointer race free.
\end{theorem}

The following \namecref{thm:ABA-awareness-and-elision-support-implies-reduction-compatibility} shows that the above is indeed a generalization of the results presented in the paper.

\begin{lemma}
	\label{thm:ABA-awareness-and-elision-support-implies-reduction-compatibility}
	Pointer race freedom (\Cref{definition:RPR}) of $\onesem$, the absence of harmful ABAs (\Cref{def:harmful-ABA}), and elision support (\Cref{def:elision-support}) implies reduction compatibility (\Cref{def:reduction-compatible}).
\end{lemma}

Altogether, the results from \Cref{sec:reduction} follow from combining \Cref{thm:RPRF-implies-GCplus}, \Cref{thm:rprf-guarantee}, and \Cref{thm:ABA-awareness-and-elision-support-implies-reduction-compatibility}.

%% file: content/appendix/theory/theory.tex

\section{Proofs for the Reduction Results}
\label{appendix:proofs-reduction}

We prove the presented results.
To do so, we establish useful lemmas and develop the elision of addresses formally.
Hereafter, we abbreviate pointer race by PR and pointer race free by PRF.

\subsection{Useful Lemmas}
\label{appendix:useful_lemmas}

\begin{lemma}
	\label{thm:computequiv-is-symmetric}
	If $\tau_1\computequiv\tau_2$, then $\tau_2\computequiv\tau_1$.
\end{lemma}

\begin{lemma}
	\label{thm:computequiv-is-transitiv}
	If $\tau_1\computequiv\tau_2\computequiv\tau_3$, then $\tau_1\computequiv\tau_3$.
\end{lemma}

\begin{lemma}
	\label{thm:heapequiv-is-transitiv-provided-same-valid-heap-range}
	If $\heapcomputof{\tau_1}{\validof{\tau_1}}\subseteq\heapcomputof{\tau_2}{\validof{\tau_2}}$ and $\tau_1\heapequiv[A]\tau_2\heapequiv[B]\tau_3$, then $\tau_1\heapequiv[A\cap B]\tau_3$.
\end{lemma}

\begin{lemma}
	\label{thm:obsrel-is-transitiv-provided-same-valid-adr-NEW}
	If $\adrof{\restrict{\heapcomput{\tau_1}}{\validof{\tau_1}}}\subseteq\adrof{\restrict{\heapcomput{\tau_2}}{\validof{\tau_2}}}$ and $\tau_1\obsrel\tau_2\obsrel\tau_3$, then $\tau_1\obsrel\tau_3$.
\end{lemma}

\begin{lemma}
	\label{thm:adrof-valid-heap-restriction}
	We have $\adrof{\restrict{\heapcomput{\tau}}{\validof{\tau}}}=(\validof{\tau}\cap\adr)\cup\heapcomputof{\tau}{\validof{\tau}}$.
\end{lemma}

\begin{lemma}
	\label{thm:valid-subset-dom}
	If $\tau\in\allsem$, then $\validof{\tau}\subseteq\domof{\heapcomput{\tau}}$.
\end{lemma}

\begin{lemma}
	\label{thm:computequiv-implies-same-valid}
	If $\tau\computequiv\sigma$, then $\validof{\tau}=\validof{\sigma}$.
\end{lemma}

\begin{lemma}
	\label{thm:move-event-from-freeable}
	If $\anevent\prall{\in}\set{\afuncof{\athread,\vecof{\anadr},\vecof{\advalue}},\exitof{\athread},\freeof{\anadr}}$,
	then $\ahist_1\prall{\in}\freeableof{\ahist_2.\anevent}{\anadr}{\iff}\anevent.\ahist_1\prall{\in}\freeableof{\ahist_2}{\anadr}$.
\end{lemma}

\begin{lemma}
	\label{thm:fresh-not-referenced}
	If $\tau\in\allsem$ and $\anadr\in\freshof{\tau}$, then $\anadr\notin\rangeof{\heapcomput{\tau}}$.
\end{lemma}

\begin{lemma}
	\label{thm:fresh-not-valid}
	If $\tau\in\allsem$ and $\anadr\in\freshof{\tau}$, then $\anadr\notin\heapcomputof{\tau}{\validof{\tau}}$ and $\psel{\anadr}\notin\validof{\tau}$.
\end{lemma}

\begin{lemma}
	\label{thm:free-not-valid}
	If $\tau\in\allsem$ PRF and $\anadr\prall{\in}\freedof{\tau}$, then $\anadr\prall{\notin}\heapcomputof{\tau}{\validof{\tau}}$ and $\psel{\anadr}\prall{\notin}\validof{\tau}$.
\end{lemma}

\begin{lemma}
	\label{thm:seg-valid}
	If $\tau\in\allsem$ and $\heapcomputof{\tau}{\apexp}=\segval$, then $\apexp\in\validof{\tau}$.
\end{lemma}

\begin{lemma}
	\label{thm:rprf-vs-bot}
	If $\tau\in\allsem$ PRF and $\heapcomputof{\tau}{\apexp}=\bot$, then $\set{\apexp}\cap\adr\not\subseteq\heapcomputof{\tau}{\validof{\tau}}$.
\end{lemma}

\begin{lemma}
	\label{thm:rprf-vs-bot-pvar}
	If $\tau\in\allsem$ PRF, then $\pvars\subseteq\domof{\heapcomput{\tau}}$.
\end{lemma}

\subsection{Technical Preparations}
\label{appendix:technical_preparations}

The following \namecref{thm:invalidpointers-adr} states that pointer race freedom guarantees a separation of the valid and invalid memory.
It may only overlap for those addresses that are tracked precisely using the memory equivalence.

\begin{lemma}
	\label{thm:invalidpointers-adr}
	If $\tau\in\asem{A}$ PRF, then $\adrof{\restrict{\heapcomput{\tau}}{\validof{\tau}}}\cap\heapcomputof{\tau}{\pexp\setminus\validof{\tau}}\subseteq A$.
\end{lemma}

Towards a proof of the core result, \Cref{thm:RPRF-implies-GCplus}, we show that certain actions of a computation $\tau$ can be mimicked in another computation $\sigma$ by simply executing the same command in both executions.
As discussed in \Cref{sec:reduction}, the induced update may differ.
Basically, the following \namecref{thm:mimic-simple-NEW} proves the technical cases of the proof of \Cref{thm:RPRF-implies-GCplus}.

\begin{lemma}
	\label{thm:mimic-simple-NEW}
	Let $\tau.\anact\in\allsem$ and $\sigma\in\asem{A}$ such that
	\begin{align*}
		&
		\tau\computequiv\sigma \quad\wedge\quad \tau\heapequiv[A]\sigma \quad\wedge\quad \tau\obsrel[A]\sigma \quad\wedge\quad \anact=(\athread,\acom,\anup)
		\\\wedge\quad&\sigma.\anact\notin\asem{A}\implies\acom \text{ is an assignment}
		\\\wedge\quad&\acom\text{ contains }\psel{\apvar}\text{ or }\dsel{\apvar}\implies\apvar\in\validof{\tau}
		\\\wedge\quad&\acom\equiv\afunc(\vecof{\apvar},\vecof{\advar})\implies\heapcomputof{\tau}{\vecof{\apvar}}=\heapcomputof{\sigma}{\vecof{\apvar}}
		\\\wedge\quad&\acom\equiv\freeof{\anadr}\implies\big(\forall\anadrp\prall{\in}\adr\prall{\setminus}\set{\anadr}~\forall\gamma\prall{\in}\set{\tau,\sigma}.~\freeableof{\gamma.\anact}{\anadrp}=\freeableof{\gamma}{\anadrp}\big)
		\\\wedge\quad&\acom\equiv\apvar:=\malloc\wedge\heapcomputof{\tau.\anact}{\apvar}=\anadr\notin A\implies\freeableof{\tau}{\anadr}\subseteq\freeableof{\sigma}{\anadr}
		\ .
		\intertext{%
	Then there is some $\anactp=(\athread,\acom,\anupp)$ with $\sigma.\anactp\in\asem{A}$ and
		}
		&\tau.\anact\computequiv\sigma.\anactp \quad\wedge\quad \tau.\anact\heapequiv[A]\sigma.\anactp \quad\wedge\quad \tau.\anact\obsrel\sigma.\anactp
		\ .
	\end{align*}
\end{lemma}

\subsection{Elision Technique}
\label{appendix:elision_technique}

\begin{definition}
	An address mapping is a bijection $\swapadrraw:\adr\to\adr$.
	For convenience, we extend $\swapadrraw$ such that $\swapadr{\bot}=\bot$, $\swapadr{\segval}=\segval$, and $\swapadr{\advalue}=\advalue$ for any $\advalue\in\dom$.
	The address mapping induces an expression mapping $\swapexpraw$ with
	\begin{align*}
		\swapexp{\apvar}&=\apvar
		&
		\swapexp{\psel{\anadr}}&=\psel{\swapadr{\anadr}}
		\\
		\swapexp{\advar}&=\advar
		&
		\swapexp{\dsel{\anadr}}&=\dsel{\swapadr{\anadr}}
	\end{align*}
	and a history mapping $\swaphistraw$ with
	\begin{align*}
		\swaphist{\epsilon}&=\epsilon
		\\
		\swaphist{\ahist.\freeof{\anadr}}&=\swaphist{\ahist}.\freeof{\swapadr{\anadr}}
		\\
		\swaphist{\ahist.\afunc(\athread,\vecof{\anadr},\vecof{\advalue})}&=\swaphist{\ahist}.\afunc(\athread,\swapadr{\vecof{\anadr}},\vecof{\advalue})
		\\
		\swaphist{\ahist.\exitof{\athread}}&=\swaphist{\ahist}.\exitof{\athread}
	\end{align*}
	for any SMR API function $\afunc$.
	We use $\swapadr{\vecof{\anadr}}=\swapadr{\anadr_1},\dots,\swapadr{\anadr_k}$.
\end{definition}

Note that $\renamingof{\ahist}{\anadr}{\anadrp}=\swaphist{\ahist}$ for $\swaphistraw$ being the history mapping induced by the address mapping $\swapadrraw$ with $\swapadr{\anadr}=\anadrp$, $\swapadr{\anadrp}=\anadr$, and $\swapadr{\anadrpp}=\anadrpp$ in all other cases.

\begin{lemma}
	\label{thm:inverse-of-address-mapping-is-again-address-mapping}
	If $\swapadrraw$ is an address mapping, then $\swapadrinvraw$ is also an address mapping.
\end{lemma}

\begin{definition}
	If $\swapadrraw$ is an address mapping, then we write $\swapexpinvraw$ and $\swaphistinvraw$ for the expression and history mapping induced by the inverse address mapping $\swapadrinvraw$.
\end{definition}

\begin{lemma}
	\label{thm:swaphist-is-unique}
	If $\swaphist{\ahist_1}=\swaphist{\ahist_2}$, then $\ahist_1=\ahist_2$.
\end{lemma}

\begin{lemma}
	\label{thm:swaphist-versus-elementof}
	We have $\swaphist{\ahist}\in\swaphist{H}\iff\ahist\in H$.
\end{lemma}

\begin{lemma}
	\label{thm:swap-vs-set-opereations}
	For every $\otimes\in\set{\setminus,\cup,\cap}$ we have:
	\begin{align*}
		\swapadr{A_1}\otimes\swapadr{A_2}&=\swapadr{A_1\otimes A_2}
		\\
		\swapexp{B_1}\otimes\swapexp{B_2}&=\swapexp{B_1\otimes B_2}
		\\
		\swaphist{C_1}\otimes\swaphist{C_2}&=\swaphist{C_1\otimes C_2}
	\end{align*}
\end{lemma}

\begin{lemma}
	\label{thm:swaphist-of-swaphistinv-is-id}
	For all $\ahist$ we have $\swaphistinv{\swaphist{\ahist}}=\ahist$.
\end{lemma}

\begin{lemma}
	\label{thm:swaphist-versus-accpetance}
	For all $\ahist$ we have $\ahist\in\historyof{\smrobs}\iff\swaphist{\ahist}\in\historyof{\smrobs}$.
\end{lemma}

\begin{lemma}
	\label{thm:swaphist-freeable}
	If $\swaphist{\historyof{\tau}}=\historyof{\sigma}$, then $\swaphist{\freeableof{\tau}{\anadr}}=\freeableof{\sigma}{\swapadr{\anadr}}$ for all $\anadr$.
\end{lemma}

\begin{theorem}
	\label{thm:elision}
	For every computation $\tau\in\asem{A}$ and every address mapping $\swapadrraw$, there is some $\sigma\in\asem{\swapadr{A}}$ such that:
	\begin{align*}
		&\heapcomput{\sigma}\circ\swapexpraw=\swapadrraw\circ\heapcomput{\tau}
		&&\historyof{\sigma}=\swaphist{\historyof{\tau}}
		&&\freedof{\sigma}=\swapadr{\freedof{\tau}}
		\\
		&\validof{\sigma}=\swapexp{\validof{\tau}}
		&&\controlof{\sigma}=\controlof{\tau}
		&&\freshof{\sigma}=\swapadr{\freshof{\tau}}
	\end{align*}
\end{theorem}

\begin{lemma}
	\label{thm:supported-elision-of-invalid-address-more}
	Let $\smrobs$ support elision.
	If $\tau\in\asem{A}$ and $\anadr\notin\adrof{\restrict{\heapcomput{\tau}}{\validof{\tau}}}\cup A$,
	then there is $\sigma\in\asem{A}$ with $\tau\computequiv\sigma$, $\tau\heapequiv[A]\sigma$, $\tau\obsrel\sigma$, and $\anadr\in\freshof{\sigma}$.
	Moreover, $\heapcomputof{\tau}{\anexp}\neq\heapcomputof{\tau}{\anexpp}$ implies $\heapcomputof{\sigma}{\anexp}\neq\heapcomputof{\sigma}{\anexpp}$ for all $\anexp,\anexpp\in\pvars\cup\setcond{\psel{\anadrp}}{\anadrp\in\heapcomputof{\tau}{\validof{\tau}}}$.
\end{lemma}

\begin{lemma}
	\label{thm:supported-elision-of-invalid-address}
	Let $\smrobs$ support elision.
	If $\tau\in\asem{A}$ and $\anadr\notin\adrof{\restrict{\heapcomput{\tau}}{\validof{\tau}}}\cup A$,
	then there is $\sigma\in\asem{A}$ with $\tau\computequiv\sigma$, $\tau\heapequiv[A]\sigma$, $\tau\obsrel\sigma$, and $\anadr\in\freshof{\sigma}$.
\end{lemma}


\input{content/appendix/theory/proofs_useful}
\input{content/appendix/theory/proofs_technical}
\input{content/appendix/theory/proofs_elision}
\input{content/appendix/theory/proofs_results}

%% file: content/appendix/theory/proofs_useful.tex

\subsection{Proofs of Useful Lemmas (\CREF{appendix:useful_lemmas})}

\begin{proof}[Proof of \Cref{thm:computequiv-is-symmetric,thm:computequiv-is-transitiv,thm:heapequiv-is-transitiv-provided-same-valid-heap-range,thm:obsrel-is-transitiv-provided-same-valid-adr-NEW}]
	By definition.
\end{proof}

%
\begin{proof}[Proof of \Cref{thm:adrof-valid-heap-restriction}]
	By definition we have:
	\begin{align*}
		\domof{\restrict{\heapcomput{\tau}}{\validof{\tau}}}\cap\adr
		=~&
		(\validof{\tau}\cup\dvars\cup\setcond{\psel{\anadr}}{\anadr\in\heapcomputof{\tau}{\validof{\tau}}})\cap\adr
		\\=~&
		(\validof{\tau}\cap\adr)\cup(\setcond{\psel{\anadr}}{\anadr\in\heapcomputof{\tau}{\validof{\tau}}}\cap\adr)
		\\=~&
		(\validof{\tau}\cap\adr)\cup(\setcond{\anadr}{\anadr\in\heapcomputof{\tau}{\validof{\tau}}})
		\\=~&
		(\validof{\tau}\cap\adr)\cup\heapcomputof{\tau}{\validof{\tau}}
		\intertext{%
	Moreover, we have%
		}
		\rangeof{\restrict{\heapcomput{\tau}}{\validof{\tau}}}\cap\adr
		=~&
		\heapcomputof{\tau}{\domof{\restrict{\heapcomput{\tau}}{\validof{\tau}}}}\cap\adr
		\\=~&
		\heapcomputof{\tau}{\domof{\restrict{\heapcomput{\tau}}{\validof{\tau}}}\cap\pexp}
		=
		\heapcomputof{\tau}{\validof{\tau}}
		\intertext{%
	With this, we conclude as follows:
		}
		\adrof{\restrict{\heapcomput{\tau}}{\validof{\tau}}}\cap\adr
		=~&
		(\domof{\restrict{\heapcomput{\tau}}{\validof{\tau}}}\cap\rangeof{\restrict{\heapcomput{\tau}}{\validof{\tau}}})\cap\adr
		\\=~&
		(\validof{\tau}\cap\adr)\cup\heapcomputof{\tau}{\validof{\tau}}\cup\heapcomputof{\tau}{\validof{\tau}}
		\\=~&
		(\validof{\tau}\cap\adr)\cup\heapcomputof{\tau}{\validof{\tau}}
	\end{align*}
\end{proof}

%
\begin{proof}[Proof of \Cref{thm:valid-subset-dom}]
	The claim holds for $\epsilon$ since $\validof{\tau}=\pvars$ and $\pvars\subseteq\domof{\heapcomput{\epsilon}}$ by definition.
	Towards a contradiction, assume the claim does not hold.
	Then, there is a shortest computation $\tau.\anact\in\allsem$ with $\validof{\tau.\anact}\not\subseteq\domof{\heapcomput{\tau.\anact}}$.
	That is, there is $\apexp\in\validof{\tau.\anact}$ with $\heapcomputof{\tau.\anact}{\apexp}=\bot$.
	First, consider the case where we have $\apexp\notin\validof{\tau}$.
	To validate $\apexp$, $\anact$ must be
	\begin{inparaenum}[(i)]
		\item\label{proof:thm:valid-subset-dom:malloc} an allocation $\apvar:=\malloc$ with $\heapcomputof{\tau.\anact}{\apvar}=\anadr$ and $\apexp\in\set{\apvar,\psel{\anadr}}$,
		\item\label{proof:thm:valid-subset-dom:assert} an assertion of the form $\assertof{\apexp=\apvarp}$ with $\apvarp\in\validof{\tau}$, or
		\item\label{proof:thm:valid-subset-dom:assign} an assignment $\apexp:=\apexpp$ with $\apexpp\in\validof{\tau}$.
	\end{inparaenum}
	Case~\eqref{proof:thm:valid-subset-dom:malloc} cannot apply as it results in $\apexp\in\domof{\heapcomput{\tau.\anact}}$ by definition.
	Case~\eqref{proof:thm:valid-subset-dom:assert} cannot apply because $\heapcomputof{\tau}{\apexp}=\heapcomputof{\tau}{\apvarp}$ together with $\apvarp\in\domof{\heapcomput{\tau}}=\domof{\heapcomput{\tau.\anact}}$ by minimality gives $\apexp\in\domof{\heapcomput{\tau.\anact}}$.
	So case~\eqref{proof:thm:valid-subset-dom:assign} must apply.
	$\apexpp\in\validof{\tau}$ yields $\heapcomputof{\tau}{\apexpp}\neq\bot$ by minimality.
	We get $\heapcomputof{\tau.\anact}{\apexp}\neq\bot$ what contradicts the choice of $\apexp$.
	Overall, $\apexp\in\validof{\tau}$ must hold.

	Consider now the $\apexp\in\validof{\tau}$.
	By minimality, we have $\heapcomputof{\tau}{\apexp}\neq\bot$.
	So $\anact$ must update $\apexp$ to $\bot$.
	To do so, $\anact$ must be
	\begin{inparaenum}[(i)]
		\item $\freeof{\anadr}$ and $\apexp\equiv\psel{\anadr}$, or
		\item an assignment $\apexp:=\apexpp$ with $\heapcomputof{\tau}{\apexpp}=\bot$.
	\end{inparaenum}
	The former case cannot apply as it results in $\apexp\notin\validof{\tau.\anact}$ by definition.
	So the latter case applies.
	By minimality, we have $\apexpp\notin\validof{\tau}$.
	Hence, $\apexp\notin\validof{\tau.\anact}$.
	This contradicts the choice of $\apexp$.
\end{proof}

%
\begin{proof}[Proof of \Cref{thm:computequiv-implies-same-valid}]
	We conclude as follows using the definition of $\tau\computequiv\sigma$, set theory, and the definition of restrictions:
	\begin{align*}
		\tau\computequiv\sigma
		\implies~&
		\restrict{\heapcomput{\tau}}{\validof{\tau}}=\restrict{\heapcomput{\sigma}}{\validof{\sigma}}
		\\\implies~&
		\domof{\restrict{\heapcomput{\tau}}{\validof{\tau}}}=\domof{\restrict{\heapcomput{\sigma}}{\validof{\sigma}}}
		\\\implies~&
		\domof{\restrict{\heapcomput{\tau}}{\validof{\tau}}}\cap\pexp=\domof{\restrict{\heapcomput{\sigma}}{\validof{\sigma}}}\cap\pexp
		\\\implies~&
		\validof{\tau}=\validof{\sigma}
	\end{align*}
	where the last implication is due to \Cref{thm:valid-subset-dom}.
\end{proof}

%
\begin{proof}[Proof of \Cref{thm:move-event-from-freeable}]
	The direction from right to left follows by definition.
	So consider the direction from left to right.
	Let $\ahist_1\in\freeableof{\ahist_2.\anevent}{\anadr}$.
	Towards a contradiction, assume $\anevent.\ahist_1\notin\freeableof{\ahist_2}{\anadr}$.
	Since $\freesof{\anevent}\subseteq\set{\anadr}$, we must have $\ahist_2.\anevent.\ahist_1\notin\specof{\smrobs}$ by definition.
	Then, $\ahist_1\notin\freeableof{\ahist_2.\anevent}{\anadr}$ by definition.
	This contradicts the assumption.
\end{proof}

%
\begin{proof}[Proof of \Cref{thm:fresh-not-referenced}]
	Towards a contradiction, let $\tau.\anact\in\allsem$ be the shortest computation such that there is some $\anadr\in\freshof{\tau.\anact}$ with $\anadr\in\rangeof{\heapcomput{\tau.\anact}}$.
	That is, there is some $\apexp\in\pexp$ with $\heapcomputof{\tau.\anact}{\apexp}=\anadr$.
	By monotonicity, we have $\anadr\in\freshof{\tau}$.
	By minimality then, we must have $\heapcomputof{\tau}{\apexp}\neq\anadr$.
	Hence, $\anact$ updates $\apexp$ to $\anadr$.
	This means it must be an allocation or a pointer assignment.
	In the former case, we must have $\anact=(\athread,\apexp:=\malloc,[\apvar\mapsto\anadr,\dots])$.
	Then, $\anadr\notin\freshof{\tau.\anact}$ holds by definition.
	So this case cannot apply as it contradicts the assumption.
	That is, $\anact$ is of the form $\anact=(\athread,\apexp=\apexpp,\anup)$ with $\heapcomputof{\tau}{\apexpp}=\anadr$.
	This means $\anadr\in\rangeof{\heapcomput{\tau}}$.
	This, contradicts the minimality of $\tau.\anact$ and thus concludes the claim.
\end{proof}

%
\begin{proof}[Proof of \Cref{thm:fresh-not-valid}]
	Towards a contradiction, assume there is a shortest $\tau.\anact\in\allsem$ such that there is some $\anadr\in\freshof{\tau.\anact}$ with $\anadr\in\heapcomputof{\tau.\anact}{\validof{\tau.\anact}}\vee\psel{\anadr}\in\validof{\tau.\anact}$.
	Note that $\tau.\anact$ is indeed the shortest such computation since the claim holds for $\epsilon$.
	By monotonicity, we have $\anadr\in\freshof{\tau}$.
	By minimality of $\tau.\anact$ we have $\anadr\notin\heapcomputof{\tau}{\validof{\tau}}$ and $\psel{\anadr}\notin\validof{\tau}$.
	If $\psel{\anadr}\in\validof{\tau.\anact}$ holds, then $\anact$ must be an assignment of the form $\psel{\apvar}:=\apvarp$ with $\heapcomputof{\tau}{\apvar}=\anadr$.
	This means $\anadr\in\rangeof{\heapcomput{\tau}}$.
	So \Cref{thm:fresh-not-referenced} gives $\anadr\notin\freshof{\tau}$.
	This contradicts $\anadr\in\freshof{\tau}$ from above.
	Hence, this cannot apply and we must have $\psel{\anadr}\notin\validof{\tau.\anact}$.
	So by assumption we have $\anadr\in\heapcomputof{\tau.\anact}{\validof{\tau.\anact}}$.
	By definition, there is some $\apexp\in\validof{\tau.\anact}$ with $\heapcomputof{\tau.\anact}{\apexp}=\anadr$.
	Again by minimality, we have $\apexp\notin\validof{\tau}$ or $\heapcomputof{\tau}{\apexp}\neq\anadr$.
	Consider $\apexp\notin\validof{\tau}$.
	That is, $\anact$ validates $\apexp$.
	To do so, $\anact$ must be an assignment, an allocation, or an assertion:
	\begin{itemize}
		\item
			If $\anact$ is of the form $\anact=(\athread,\apexp:=\apexpp,\anup)$, then $\apexpp\in\validof{\tau}$ and $\heapcomputof{\tau}{\apexpp}=\anadr$ must hold.
			The latter, leads to $\apexpp\notin\validof{\tau}$ by minimality of $\tau.\anact$.
			Hence, $\anact$ cannot be an assignment.
		\item
			If $\anact$ is of the form $\anact=(\athread,\apvar:=\malloc,[\apvar\mapsto\anadr,\dots])$, then $\apexp\equiv\apvar$ must hold.
			(Note that for $\apexp\equiv\psel{\anadr}$ we get $\heapcomputof{\tau.\anact}{\apexp}=\segval\neq\anadr$.)
			This, results in $\anadr\notin\freshof{\tau.\anact}$ which contradicts the assumption.
			Hence, $\anact$ cannot be an allocation.
		\item
			If $\anact$ is of the form $\anact=(\athread,\assertof{\apvar=\apvarp},\anup)$, then wlog. $\apexp\equiv\apvar$ and $\apvarp\in\validof{\tau}$ and $\anadr=\heapcomputof{\tau.\anact}{\apexp}=\heapcomputof{\tau}{\apexp}=\heapcomputof{\tau}{\apvarp}$ must hold.
			Again by minimality, we get $\apvarp\notin\validof{\tau}$ which contradicts the assumption.
			Hence, $\anact$ cannot be an assertion.
	\end{itemize}
	So $\apexp\in\validof{\tau}$ must hold and thus $\heapcomputof{\tau}{\apexp}\neq\anadr$.
	That is, $\anact$ updates $\apexp$ to $\anadr$.
	To do so, $\anact$ must be an assignment or an allocation.
	We then conclude a contradiction as before.
\end{proof}

%
\begin{proof}[Proof of \Cref{thm:free-not-valid}]
	To the contrary, assume there is a shortest $\tau.\anact\in\allsem$ PRF with $\anadr\in\freedof{\tau.\anact}$ and $\anadr\in\heapcomputof{\tau.\anact}{\validof{\tau.\anact}}\vee\psel{\anadr}\in\validof{\tau.\anact}$.
	Note that $\tau.\anact$ is indeed the shortest such computation since the claim is vacuously true for $\epsilon$.
	First, consider the case where we have $\anadr\in\freedof{\tau}$.
	We get $\anadr\notin\heapcomputof{\tau}{\validof{\tau}}$ and $\psel{\anadr}\notin\validof{\tau}$ by minimality of $\tau.\anact$.
	If $\psel{\anadr}\in\validof{\tau.\anact}$ holds, then $\anact$ must be an assignment of the form $\psel{\apvar}:=\apvarp$ with $\heapcomputof{\tau}{\apvar}=\anadr$.
	Moreover, $\apvar\notin\validof{\tau}$ must hold because $\anadr\notin\heapcomputof{\tau}{\validof{\tau}}$.
	Hence, $\anact$ raises a pointer race.
	This contradicts the assumption.
	So we must have $\anadr\in\heapcomputof{\tau.\anact}{\validof{\tau.\anact}}$.
	This allows us to derive a contradiction along the lines of the proof of \Cref{thm:fresh-not-valid}.
	So consider the case where we have $\anadr\notin\freedof{\tau}$ now.
	Then, $\anact$ must execute the command $\freeof{\anadr}$.
	As a consequence, we get $\psel{\anadr}\notin\validof{\tau.\anact}$ and $\apexp\notin\validof{\tau.\anact}$ for all $\apexp$ with $\heapcomputof{\tau}{\apexp}=\anadr$.
	That is, we have $\anadr\notin\heapcomputof{\tau.\anact}{\validof{\tau.\anact}}$ and $\psel{\anadr}\notin\validof{\tau.\anact}$.
	This contradicts the assumption and concludes the claim.
\end{proof}

%
\begin{proof}[Proof of \Cref{thm:seg-valid}]
	To the contrary, assume there is a shortest $\tau.\anact\in\allsem$ with some $\apexp\notin\validof{\tau.\anact}$ and $\heapcomputof{\tau.\anact}{\apexp}=\segval$.
	By minimality, $\heapcomputof{\tau}{\apexp}\neq\segval$ or $\apexp\in\validof{\tau.\anact}$.
	As a first case, consider $\heapcomputof{\tau}{\apexp}\neq\segval$.
	Then, $\anact$ updates $\apexp$ to $\segval$.
	That is, $\anact$ is an allocation or an assignment.
	If $\anact=(\athread,\apvar:=\malloc,[\apvar\mapsto\anadr,\apexp\mapsto\segval,\dots])$, then $\apexp\equiv\psel{\anadr}$ must hold.
	This gives $\apexp\in\validof{\tau.\anact}$ by definition.
	So this case cannot apply by assumption.
	If $\anact=(\athread,\apexp:=\apexpp,[\apexp\mapsto\anup])$, then we must have $\heapcomputof{\tau}{\apexpp}=\segval$.
	By minimality, this gives $\apexpp\in\validof{\tau}$.
	And by definition, $\apexp\in\validof{\tau.\anact}$ then.
	This contradicts the assumption.
	So we must have $\apexp\in\validof{\tau.\anact}$.
	That is, $\anact$ makes $\apexp$ invalid.
	To do so, $\anact$ must be an assignment or a free.
	If $\anact=(\athread,\apexp:=\apexpp,\anup)$, then we must have $\apexpp\notin\validof{\tau}$.
	By minimality, this gives $\heapcomputof{\tau}{\apexpp}\neq\segval$.
	So we get $\heapcomputof{\tau.\anact}{\apexp}\neq\segval$ what contradicts the assumption.
	If $\anact=(\athread,\freeof{\anadr},[\psel{\anadr}\mapsto\bot,\dots])$, then we must have $\apexp\equiv\psel{\anadr}$ or $\heapcomputof{\tau}{\apexp}=\anadr$.
	In the former case, we get $\heapcomputof{\tau.\anact}{\apexp}=\bot\neq\segval$.
	In the latter case, we get $\heapcomputof{\tau.\anact}{\apexp}=\anadr\neq\segval$.
	Hence, we get $\heapcomputof{\tau.\anact}{\apexp}\neq\segval$ what contradicts the assumption.
	This concludes the claim.
\end{proof}

%
\begin{proof}[Proof of \Cref{thm:rprf-vs-bot}]
	To the contrary, assume a shortest PRF $\tau.\anact\in\allsem$ such that there is some $\apexp$ with $\heapcomputof{\tau.\anact}{\apexp}=\bot$ and $\set{\apexp}\cap\adr\subseteq\heapcomputof{\tau.\anact}{\validof{\tau.\anact}}$.
	By minimality, we have $\heapcomputof{\tau}{\apexp}\neq\bot$ or $\set{\apexp}\cap\adr\not\subseteq\heapcomputof{\tau}{\validof{\tau}}$.
	First, consider $\heapcomputof{\tau}{\apexp}\neq\bot$.
	In order for $\anact$ to update $\apexp$ to $\bot$, it must be a free or an assignment.
	If $\anact=(\athread,\freeof{\anadr},[\psel{\anadr}\mapsto\bot,\dots])$, then we must have $\apexp\equiv\psel{\anadr}$.
	By definition, we get $\anadr\in\freedof{\tau.\anact}$.
	Then, \Cref{thm:free-not-valid} gives $\anadr\notin\heapcomputof{\tau.\anact}{\validof{\tau.\anact}}$.
	This contradicts the choice of $\apexp$ to satisfy $\set{\apexp}\cap\adr\not\subseteq\heapcomputof{\tau}{\validof{\tau}}$.
	If $\anact=(\athread,\apexp:=\apexpp,[\apexp\mapsto\bot])$, then $\heapcomputof{\tau}{\apexpp}=\bot$ must hold.
	If $\apexpp\in\pvars$, then we have $\set{\apexpp}\cap\adr=\emptyset\subseteq\heapcomputof{\tau}{\validof{\tau}}$.
	Since this contradicts minimality, $\apexpp$ must be of the form $\apexpp\equiv\psel{\anadrp}$.
	By minimality, we have $\anadrp\notin\heapcomputof{\tau}{\validof{\tau}}$.
	This means $\anact$ raises a pointer race.
	This contradicts the assumption of $\tau.\anact$ being PRF.
	Altogether, this means $\heapcomputof{\tau}{\apexp}=\bot$ holds.

	So consider $\set{\apexp}\cap\adr\not\subseteq\heapcomputof{\tau}{\validof{\tau}}$.
	By definition, $\apexp$ must be of the form $\psel{\anadr}$.
	We get $\anadr\notin\heapcomputof{\tau}{\validof{\tau}}$.
	To arrive at $\set{\apexp}\cap\adr\subseteq\heapcomputof{\tau.\anact}{\validof{\tau.\anact}}$ we must get $\anadr\in\heapcomputof{\tau.\anact}{\validof{\tau.\anact}}$.
	To get this, $\anact$ must be an allocation: $\anact=(\athread,\apvar:=\malloc,[\apvar\mapsto\anadr,\psel{\anadr}\mapsto\segval,\dots])$.
	This, results in $\heapcomputof{\tau.\anact}{\apexp}=\segval\neq\bot$.
	This contradicts the assumption and concludes the claim.
\end{proof}

%
\begin{proof}[Proof of \Cref{thm:rprf-vs-bot-pvar}]
	Follows from \Cref{thm:rprf-vs-bot} and $\set{\apvar}\cap\adr=\emptyset$ for every $\apvar\in\pvars$ together with the fact that $\emptyset\subseteq\heapcomputof{\tau}{\validof{\tau}}$ is always true.
\end{proof}

%% file: content/appendix/theory/proofs_technical.tex

\subsection{Proofs of Technical Preparations (\CREF{appendix:technical_preparations})}

\begin{proof}[Proof of \Cref{thm:invalidpointers-adr}]
	We proceed by induction over the structure of computations.
	\begin{description}[labelwidth=6mm,leftmargin=8mm,itemindent=0mm]
		\item[IB:]
			For $\tau=\epsilon$ choose $\sigma=\tau$.
			This satisfies the claim.

		\item[IH:]
			For every $\tau\in\asem{A}$ RPRF we have $\adrof{\restrict{\heapcomput{\tau}}{\validof{\tau}}}\cap\heapcomputof{\tau}{\pexp\setminus\validof{\tau}}\subseteq A$.

		\item[IS:]
			Consider now $\tau.\anact\in\asem{A}$ RPRF.
			Let $\anact=(\athread,\acom,\anup)$.
			By induction we have $\adrof{\restrict{\heapcomput{\tau}}{\validof{\tau}}}\cap\heapcomputof{\tau}{\pexp\setminus\validof{\tau}}\subseteq A$.
			By \Cref{thm:adrof-valid-heap-restriction} this means:
			\begin{align*}
			 	\heapcomputof{\tau}{\validof{\tau}}\cap\heapcomputof{\tau}{\pexp\setminus\validof{\tau}}\subseteq&\, A
			 	\\\text{and}\quad
			 	(\validof{\tau}\cap\adr)\cap\heapcomputof{\tau}{\pexp\setminus\validof{\tau}}\subseteq&\, A
			 	\ .
			\end{align*}
			The claim follows immediately if $\anact$ does not update pointer expressions and does not modify the validity of pointer expressions.
			Otherwise, $\acom$ must be one of the following: a pointer assignment, an allocation, a free, or an assertion for pointer equality.
			We do a case distinction over $\acom$.
			Note that it suffices to show:
			\begin{align*}
			 	\heapcomputof{\tau.\anact}{\validof{\tau.\anact}}\cap\heapcomputof{\tau.\anact}{\pexp\setminus\validof{\tau.\anact}}\subseteq&\, A
			 	\\\text{and}\quad
			 	(\validof{\tau.\anact}\cap\adr)\cap\heapcomputof{\tau.\anact}{\pexp\setminus\validof{\tau.\anact}}\subseteq&\, A
			\end{align*}
			because this implies the desired $\adrof{\restrict{\heapcomput{\tau.\anact}}{\validof{\tau.\anact}}}\cap\heapcomputof{\tau.\anact}{\pexp\setminus\validof{\tau.\anact}}\subseteq A$ by \Cref{thm:adrof-valid-heap-restriction}.

			\begin{casedistinction}
				\item[$\acom\equiv\apvar=\apvarp$]
				Let $\anadr=\heapcomputof{\tau}{\apvarp}$.
				Then, the update is $\anup=[\apvar\mapsto\anadr]$.
				So we have $\heapcomputof{\tau.\anact}{\apvar}=\anadr=\heapcomputof{\tau}{\apvarp}$.
				By definition, we have $\validof{\tau.\anact}\cap\adr=\validof{\tau}\cap\adr$.
				\begin{itemize}
					\item
						If $\apvarp\in\validof{\tau}$ we get:
						\begin{align*}
							\qquad&\heapcomputof{\tau.\anact}{\validof{\tau.\anact}}
							=
							\heapcomput{\tau}[\apvar\mapsto\anadr](\validof{\tau}\cup\set{\apvar})
							=
							\heapcomput{\tau}[\apvar\mapsto\anadr]((\validof{\tau}\setminus\set{\apvar})\cup\set{\apvar})
							\\=~&
							\heapcomput{\tau}[\apvar\mapsto\anadr](\validof{\tau}\setminus\set{\apvar})\cup\heapcomput{\tau}[\apvar\mapsto\anadr](\set{\apvar})
							\subseteq
							\heapcomput{\tau}(\validof{\tau})\cup\set{\anadr}
							=
							\heapcomput{\tau}(\validof{\tau})
						\intertext{%
						where $\anadr\in\heapcomputof{\tau}{\validof{\tau}}$ because $\apvarp\in\validof{\tau}$ and $\heapcomputof{\tau}{\apvarp}=\anadr$.
						Moreover, we have:
						}
							\qquad\quad&\heapcomputof{\tau.\anact}{\pexp\setminus\validof{\tau.\anact}}
							=
							\heapcomput{\tau}[\apvar\mapsto\anadr](\pexp\setminus(\validof{\tau}\cup\set{\apvar}))
							\\=~&
							\heapcomput{\tau}[\apvar\mapsto\anadr]((\pexp\setminus\validof{\tau})\setminus\set{\apvar})
							=
							\heapcomput{\tau}((\pexp\setminus\validof{\tau})\setminus\set{\apvar})
							\subseteq
							\heapcomput{\tau}(\pexp\setminus\validof{\tau})
							\ .
						\intertext{%
							Altogether we get:
						}
							&\heapcomputof{\tau.\anact}{\validof{\tau.\anact}}\cap\heapcomputof{\tau.\anact}{\pexp\setminus\validof{\tau.\anact}}
							\subseteq
							\heapcomput{\tau}(\validof{\tau})\cap\heapcomput{\tau}(\pexp\setminus\validof{\tau})
							~~~\text{and}\\
							\qquad&(\validof{\tau.\anact}\cap\adr)\cap\heapcomputof{\tau.\anact}{\pexp\setminus\validof{\tau.\anact}}
							\subseteq
							(\validof{\tau}\cap\adr)\cap\heapcomput{\tau}(\pexp\setminus\validof{\tau})
							\ .
						\end{align*}

					\item
						If $\apvarp\notin\validof{\tau}$ we proceed analogously and get:
						\begin{align*}
							\qquad
							&\heapcomputof{\tau.\anact}{\validof{\tau.\anact}}
							=
							\heapcomput{\tau}[\apvar\mapsto\anadr](\validof{\tau}\setminus\set{\apvar})
							=
							\heapcomput{\tau}(\validof{\tau}\setminus\set{\apvar})
							\subseteq
							\heapcomput{\tau}(\validof{\tau})
						\intertext{and}
							&
							\heapcomputof{\tau.\anact}{\pexp\setminus\validof{\tau.\anact}}
							=
							\heapcomput{\tau}[\apvar\mapsto\anadr](\pexp\setminus(\validof{\tau}\setminus\set{\apvar}))
							\\\subseteq~&
							\heapcomput{\tau}[\apvar\mapsto\anadr]((\pexp\setminus(\validof{\tau}\cup\set{\apvar}))\cup\set{\apvar})
							\\=~&
							\heapcomput{\tau}[\apvar\mapsto\anadr](\pexp\setminus(\validof{\tau}\cup\set{\apvar}))\cup\set{\heapcomput{\tau}[\apvar\mapsto\anadr](\apvar)}
							\\=~&
							\heapcomput{\tau}(\pexp\setminus(\validof{\tau}\cup\set{\apvar}))\cup\set{\anadr}
							\subseteq
							\heapcomput{\tau}(\pexp\setminus\validof{\tau})\cup\set{\anadr}
							=
							\heapcomput{\tau}(\pexp\setminus\validof{\tau})
						\intertext{%
						where the last equality holds because $\apvarp\in\pexp\setminus\validof{\tau}$ and $\heapcomputof{\tau}{\apvarp}=\anadr$.
						Altogether we get
						}
							&\heapcomputof{\tau.\anact}{\validof{\tau.\anact}}\cap\heapcomputof{\tau.\anact}{\pexp\setminus\validof{\tau.\anact}}
							\subseteq
							\heapcomput{\tau}(\validof{\tau})\cap\heapcomput{\tau}(\pexp\setminus\validof{\tau})
							~~~\text{and}\\
							&(\validof{\tau.\anact}\cap\adr)\cap\heapcomputof{\tau.\anact}{\pexp\setminus\validof{\tau.\anact}}
							\subseteq
							(\validof{\tau}\cap\adr)\cap\heapcomput{\tau}(\pexp\setminus\validof{\tau})
							\ .
						\end{align*}
				\end{itemize}
				This concludes the property.

				\item[$\acom\equiv\apvar=\psel{\apvarp}$ with $\heapcomputof{\tau}{\apvarp}\in\adr$]
				Let $\heapcomputof{\tau}{\apvarp}=\anadr$ and $\heapcomputof{\tau}{\psel{\anadr}}=\anadrp$.
				Then the update is $\anup=[\apvar\mapsto\anadrp]$.
				By definition, we have $\validof{\tau.\anact}\cap\adr=\validof{\tau}\cap\adr$.
				\begin{itemize}
					\item
						If $\psel{\anadr}\in\validof{\tau}$ we get:
						\begin{align*}
							\qquad\quad&
							\heapcomputof{\tau.\anact}{\validof{\tau.\anact}}
							=
							\heapcomput{\tau}[\apvar\mapsto\anadrp](\validof{\tau}\cup\set{\apvar})
							=
							\heapcomput{\tau}[\apvar\mapsto\anadrp]((\validof{\tau}\setminus\set{\apvar})\cup\set{\apvar})
							\\=~&
							\heapcomput{\tau}[\apvar\mapsto\anadrp](\validof{\tau}\setminus\set{\apvar})\cup\heapcomput{\tau}[\apvar\mapsto\anadrp](\set{\apvar})
							\subseteq
							\heapcomput{\tau}(\validof{\tau})\cup\set{\anadrp}
							=
							\heapcomput{\tau}(\validof{\tau})
						\intertext{%
						where $\anadrp\in\heapcomputof{\tau}{\validof{\tau}}$ holds because $\psel{\anadr}\in\validof{\tau}$ and $\heapcomputof{\tau}{\psel{\anadr}}=\anadrp$.
						Moreover, the following holds:
						}
							&
							\heapcomputof{\tau.\anact}{\pexp\setminus\validof{\tau.\anact}}
							=
							\heapcomput{\tau}[\apvar\mapsto\anadrp](\pexp\setminus(\validof{\tau}\cup\set{\apvar}))
							\\=~&
							\heapcomput{\tau}[\apvar\mapsto\anadrp]((\pexp\setminus\validof{\tau})\setminus\set{\apvar})
							=
							\heapcomput{\tau}((\pexp\setminus\validof{\tau})\setminus\set{\apvar})
							\subseteq
							\heapcomput{\tau}(\pexp\setminus\validof{\tau})
						\intertext{%
						Altogether we get conclude this case by:
						}
							&\heapcomputof{\tau.\anact}{\validof{\tau.\anact}}\cap\heapcomputof{\tau.\anact}{\pexp\setminus\validof{\tau.\anact}}
							\subseteq
							\heapcomput{\tau}(\validof{\tau})\cap\heapcomput{\tau}(\pexp\setminus\validof{\tau})
							~~~\text{and}\\
							&(\validof{\tau.\anact}\cap\adr)\cap\heapcomputof{\tau.\anact}{\pexp\setminus\validof{\tau.\anact}}
							\subseteq
							(\validof{\tau}\cap\adr)\cap\heapcomput{\tau}(\pexp\setminus\validof{\tau})
							\ .
						\end{align*}

					\item
						If $\psel{\anadr}\notin\validof{\tau}$ we proceed analogously and get:
						\begin{align*}
							\quad\qquad&
							\heapcomputof{\tau.\anact}{\validof{\tau.\anact}}
							=
							\heapcomput{\tau}[\apvar\mapsto\anadrp](\validof{\tau}\setminus\set{\apvar})
							=
							\heapcomput{\tau}(\validof{\tau}\setminus\set{\apvar})
							\subseteq
							\heapcomput{\tau}(\validof{\tau})
						\intertext{and}
							&
							\heapcomputof{\tau.\anact}{\pexp\setminus\validof{\tau.\anact}}
							=
							\heapcomput{\tau}[\apvar\mapsto\anadrp](\pexp\setminus(\validof{\tau}\setminus\set{\apvar}))
							\\\subseteq~&
							\heapcomput{\tau}[\apvar\mapsto\anadrp]((\pexp\setminus(\validof{\tau}\cup\set{\apvar}))\cup\set{\apvar})
							\\=~&
							\heapcomput{\tau}[\apvar\mapsto\anadrp](\pexp\setminus(\validof{\tau}\cup\set{\apvar}))\cup\set{\heapcomput{\tau}[\apvar\mapsto\anadrp](\apvar)}
							\\=~&
							\heapcomput{\tau}(\pexp\setminus(\validof{\tau}\cup\set{\apvar}))\cup\set{\anadrp}
							\subseteq
							\heapcomput{\tau}(\pexp\setminus\validof{\tau})\cup\set{\anadrp}
							=
							\heapcomput{\tau}(\pexp\setminus\validof{\tau})
						\intertext{%
						where the last equality holds by $\psel{\anadr}\in\pexp\setminus\validof{\tau}$ and $\heapcomputof{\tau}{\psel{\anadr}}=\anadrp$.
						Altogether we conclude by:
						}
							&\heapcomputof{\tau.\anact}{\validof{\tau.\anact}}\cap\heapcomputof{\tau.\anact}{\pexp\setminus\validof{\tau.\anact}}
							\subseteq
							\heapcomput{\tau}(\validof{\tau})\cap\heapcomput{\tau}(\pexp\setminus\validof{\tau})
							~~~\text{and}\\
							&(\validof{\tau.\anact}\cap\adr)\cap\heapcomputof{\tau.\anact}{\pexp\setminus\validof{\tau.\anact}}
							\subseteq
							(\validof{\tau}\cap\adr)\cap\heapcomput{\tau}(\pexp\setminus\validof{\tau})
							\ .
						\end{align*}
				\end{itemize}
				This concludes the property.

				\item[$\acom\equiv\psel{\apvar}=\apvarp$ with $\heapcomputof{\tau}{\apvar}\in\adr$]
				Similar to the previous case.

				\item[$\acom\equiv\apvar:=\malloc$]
				Let $\anup=[\apvar\mapsto\anadr,\psel{\anadr}\mapsto\segval,\dsel{\anadr}\mapsto\advalue]$.
				Then, $\anadr\in\freshof{\tau}\cup(\freedof{\tau}\cap A)$ holds due to the semantics.

				We have $\validof{\tau.\anact}\cap\adr=(\validof{\tau}\cap\adr)\cup\set{\anadr}$ and:
				\begin{align*}
					\heapcomputof{\tau.\anact}{\validof{\tau.\anact}}
					=~&
					\heapcomputof{\tau.\anact}{(\validof{\tau}\setminus\set{\apvar,\psel{\anadr}})\cup\set{\apvar,\psel{\anadr}}}
					\\=~&
					\heapcomputof{\tau.\anact}{\validof{\tau}\setminus\set{\apvar,\psel{\anadr}}}\cup\set{\anadr}
					\\=~&
					\heapcomputof{\tau}{\validof{\tau}\setminus\set{\apvar,\psel{\anadr}}}\cup\set{\anadr}
					\subseteq
					\heapcomputof{\tau}{\validof{\tau}}\cup\set{\anadr}
					\ .
				\intertext{%
				Moreover, we have:
				}
					\heapcomputof{\tau.\anact}{\pexp\setminus\validof{\tau.\anact}}
					=~&
					\heapcomputof{\tau.\anact}{\pexp\setminus(\validof{\tau}\cup\set{\apvar,\psel{\anadr}})}
					\\=~&
					\heapcomputof{\tau}{\pexp\setminus(\validof{\tau}\cup\set{\apvar,\psel{\anadr}})}
					\subseteq
					\heapcomputof{\tau}{\pexp\setminus\validof{\tau}}
					\ .
				\end{align*}
				Altogether, we get:
				\begin{align*}
					&
					\Big(\heapcomputof{\tau.\anact}{\validof{\tau.\anact}}
					\cup
					(\validof{\tau.\anact}\cap\adr)
					\Big)\cap
					\heapcomputof{\tau.\anact}{\pexp\setminus\validof{\tau.\anact}}
					\\\subseteq~&
					\Big(\heapcomputof{\tau}{\validof{\tau}}
					\cup
					(\validof{\tau}\cap\adr)\cup\set{\anadr}
					\Big)\cap
					\heapcomputof{\tau}{\pexp\setminus\validof{\tau}}
					\subseteq
					A
					\ .
				\end{align*}
				For the last inclusion, consider two cases.
				If $\anadr\in\freedof{\tau}$, then we must have $\anadr\in A$.
				Otherwise, $\anadr\in\freshof{\tau}$.
				Then, \Cref{thm:fresh-not-referenced} provides $\anadr\notin\heapcomputof{\tau}{\pexp\setminus\validof{\tau}}$.

				\item[$\acom\equiv\freeof{\anadr}$]
				The update is $\anup=[\psel{\anadr}=\bot,\dsel{\anadr}=\bot]$.

				We get $\validof{\tau.\anact}\cap\adr\subseteq(\validof{\tau}\cap\adr)\setminus\set{\anadr}$ because we have $\validof{\tau.\anact}\subseteq\validof{\tau}$ and $\psel{\anadr}\notin\validof{\tau.\anact}$ by definition.

				Next we have \[\heapcomputof{\tau.\anact}{\validof{\tau.\anact}}\subseteq\heapcomputof{\tau}{\validof{\tau}}\setminus\set{\anadr} \ .\]
				To see this, note \[\heapcomputof{\tau.\anact}{\validof{\tau.\anact}}\subseteq\heapcomputof{\tau}{\validof{\tau.\anact}}\subseteq\heapcomputof{\tau}{\validof{\tau}}\] where the first inclusion holds because $\psel{\anadr},\dsel{\anadr}\notin\validof{\tau.\anact}$ and the second inclusion holds by $\validof{\tau.\anact}\subseteq\validof{\tau}$.
				To the contrary, assume $\anadr\in\heapcomputof{\tau.\anact}{\validof{\tau.\anact}}$ was true.
				Then there is a pointer expression $\apexp\in\validof{\tau.\anact}$ with $\heapcomputof{\tau.\anact}{\apexp}=\anadr$.
				Since $\psel{\anadr}\notin\validof{\tau.\anact}$ we must have $\apexp\not\equiv\psel{\anadr}$.
				So $\heapcomputof{\tau}{\apexp}=\heapcomputof{\tau.\anact}{\apexp}=\anadr$.
				Hence, by definition we have $\apexp\notin\validof{\tau.\anact}$.
				This contradicts the assumption.
				Thus, $\anadr\notin\heapcomputof{\tau.\anact}{\validof{\tau.\anact}}$ must indeed hold.

				Finally we have \[\heapcomputof{\tau.\anact}{\pexp\setminus\validof{\tau.\anact}}\subseteq\heapcomputof{\tau}{\pexp\setminus\validof{\tau}}\cup\set{\anadr} \ .\]
				To see this, consider some $\anadrp\in\heapcomputof{\tau.\anact}{\pexp\setminus\validof{\tau.\anact}}$ with $\anadrp\notin\heapcomputof{\tau}{\pexp\setminus\validof{\tau}}$.
				There is some $\apexp\in\pexp$ with $\apexp\notin\validof{\tau.\anact}$ and $\heapcomputof{\tau.\anact}{\apexp}=\anadrp\neq\bot$.
				Due to $\anadrp\neq\bot$ we know that $\apexp\not\equiv\psel{\anadr}$.
				Hence, $\heapcomputof{\tau}{\apexp}=\heapcomputof{\tau.\anact}{\apexp}=\anadrp$.
				So we must have $\apexp\in\validof{\tau}$ as for otherwise we would get $\apexp\in\pexp\setminus\validof{\tau}$ and thus $\anadrp\in\heapcomputof{\tau}{\pexp\setminus\validof{\tau}}$ which contradicts the choice of $\anadrp$.
				To sum this up: we now know that $\apexp$ is not $\psel{\anadr}$ and is invalidated by $\anact$.
				Hence, $\anadrp=\anadr$ must hold as for otherwise $\apexp$ would not be invalidated by $\anact$.
				That is, $\anadrp\in\set{\anadr}$.
				This concludes the claim.

				Altogether, we can now conclude as follows:
				\begin{align*}
					&
					(\validof{\tau.\anact}\cap\adr)\cap\heapcomputof{\tau.\anact}{\pexp\setminus\validof{\tau.\anact}}
					\\\subseteq~&
					\Big((\validof{\tau}\cap\adr)\setminus\set{\anadr}\Big)\cap
					\Big(\heapcomputof{\tau}{\pexp\setminus\validof{\tau}}\cup\set{\anadr}\Big)
					\\\subseteq~&
					(\validof{\tau}\cap\adr)\cap
					\heapcomputof{\tau}{\pexp\setminus\validof{\tau}}
				\intertext{and}
					&
					\heapcomputof{\tau.\anact}{\validof{\tau.\anact}}
					\cap
					\heapcomputof{\tau.\anact}{\pexp\setminus\validof{\tau.\anact}}
					\\\subseteq~&
					\big(\heapcomputof{\tau}{\validof{\tau}}\setminus\set{\anadr}\big)
					\cap
					\big(\heapcomputof{\tau}{\pexp\setminus\validof{\tau}}\cup\set{\anadr}\big)
					\\=~&
					\heapcomputof{\tau}{\validof{\tau}}
					\cap
					\heapcomputof{\tau}{\pexp\setminus\validof{\tau}}
					\ .
				\end{align*}

				\item[$\acom\equiv\assert\ \apvar=\apvarp$]
				We have $\heapcomput{\tau}=\heapcomput{\tau.\anact}$ by the semantics.
				If $\validof{\tau}=\validof{\tau.\anact}$, then the claim follows immediately.
				Otherwise, we have $\validof{\tau.\anact}=\validof{\tau}\cup\set{\apvar,\apvarp}$ and $\validof{\tau}\cup\set{\apvar,\apvarp}\neq\emptyset$ by definition.
				And by the semantics we have $\heapcomputof{\tau}{\apvar}=\heapcomputof{\tau}{\apvarp}$.
				So we get:
				\begin{align*}
					\heapcomputof{\tau.\anact}{\validof{\tau.\anact}}
					=
					\heapcomputof{\tau}{\validof{\tau}\cup\set{\apvar,\apvarp}}
					=
					\heapcomputof{\tau}{\validof{\tau}}\cup\set{\heapcomputof{\tau}{\apvar}}
					=
					\heapcomputof{\tau}{\validof{\tau}}
				\end{align*}
				where the last inclusion holds because $\validof{\tau}\cup\set{\apvar,\apvarp}\neq\emptyset$ yields $\heapcomputof{\tau}{\apvar}\in\heapcomputof{\tau}{\validof{\tau}}$.
				Moreover, we have $\validof{\tau.\anact}\cap\adr=\validof{\tau}\cap\adr$ and: \[\heapcomputof{\tau.\anact}{\pexp\setminus\validof{\tau.\anact}}\subseteq\heapcomputof{\tau}{\pexp\setminus\validof{\tau}}\]
				Then, we conclude by induction:
				\begin{align*}
					\heapcomputof{\tau.\anact}{\validof{\tau.\anact}}\cap\heapcomputof{\tau.\anact}{\pexp\setminus\validof{\tau.\anact}}
					\subseteq~&
					\heapcomput{\tau}(\validof{\tau})\cap\heapcomput{\tau}(\pexp\setminus\validof{\tau})
					\intertext{and}
					(\validof{\tau.\anact}\cap\adr)\cap\heapcomputof{\tau.\anact}{\pexp\setminus\validof{\tau.\anact}}
					\subseteq~&
					(\validof{\tau}\cap\adr)\cap\heapcomput{\tau}(\pexp\setminus\validof{\tau})
					\ .
				\end{align*}
			\end{casedistinction}
	\end{description}
\end{proof}

\begin{proof}[Proof of \Cref{thm:mimic-simple-NEW}]
	Unrolling the \textbf{p}remise gives:
	\begin{enumerate}[label=({P}\arabic*),leftmargin=1.3cm] 
		\item \label[property]{proof:gcplus:premis:control} $\controlof{\tau}=\controlof{\sigma}$ 
		\item \label[property]{proof:gcplus:premis:valid-heap} $\restrict{\heapcomput{\tau}}{\validof{\tau}}=\restrict{\heapcomput{\sigma}}{\validof{\sigma}}$ 
		\item \label[property]{proof:gcplus:premis:address-heap-pvars} $\forall\, \apvar\in\pvars.~ \heapcomputof{\tau}{\apvar}=\anadr \iff \heapcomputof{\sigma}{\apvar}=\anadr$ 
		\item \label[property]{proof:gcplus:premis:address-heap-pexp} $\forall\, \anadrp\in\heapcomputof{\tau}{\validof{\tau}}.~ \heapcomputof{\tau}{\psel{\anadrp}}=\anadr \iff \heapcomputof{\sigma}{\psel{\anadrp}}=\anadr$ 
		\item \label[property]{proof:gcplus:premis:allocated} $\forall\anadr\in A.~\anadr\in\allocatedof{\tau}\iff\anadr\in\allocatedof{\sigma}$
		\item \label[property]{proof:gcplus:premis:obsrel-NEW} $\forall\, \anadrp{\,\in\,}\adrof{\restrict{\heapcomput{\tau}}{\validof{\tau}}}\cup A.~\freeableof{\tau}{\anadrp}\subseteq\freeableof{\sigma}{\anadrp}$ 
		\item \label[property]{proof:gcplus:premis:anact} $\anact=(\athread,\acom,\anup)$
		\item \label[property]{proof:gcplus:premis:not-an-assignment-implies-enabled} $\acom\text{ is not an assignment}\implies\sigma.\anact\in\asem{A}$
		\item \label[property]{proof:gcplus:premis:tau-deref-NEW} $\acom\text{ contains }\psel{\apvar}\text{ or }\dsel{\apvar}\implies\apvar\in\validof{\tau}$
		\item \label[property]{proof:gcplus:premis:call-event-restriction-NEW} $\acom\equiv\afunc(\vecof{\apvar},\vecof{\advar}) \implies \heapcomputof{\tau}{\vecof{\apvar}}=\heapcomputof{\sigma}{\vecof{\apvar}}$
		\item \label[property]{proof:gcplus:premis:freeable-tau} $\acom\equiv\freeof{\anadr}\implies\forall\anadrp\in\adr\setminus\set{\anadr}.~\freeableof{\tau.\anact}{\anadrp}=\freeableof{\tau}{\anadrp}$
		\item \label[property]{proof:gcplus:premis:freeable-sigma} $\acom\equiv\freeof{\anadr}\implies\forall\anadrp\in\adr\setminus\set{\anadr}.~\freeableof{\sigma.\anact}{\anadrp}=\freeableof{\sigma}{\anadrp}$
		\item \label[property]{proof:gcplus:premis:malloc-sigma-fresh-implies-same-freeable-NEW} $\acom\equiv\apvar:=\malloc\wedge\heapcomputof{\tau.\anact}{\apvar}\notin A\implies\freeableof{\tau}{\heapcomputof{\tau.\anact}{\apvar}}\subseteq\freeableof{\sigma}{\heapcomputof{\tau.\anact}{\apvar}}$
	\end{enumerate}

	We have to show that there is some $\anactp=(\athread,\acom,\anupp)$ which satisfies the claim.
	That is, our \textbf{g}oal is to show for all threads $\athread$ and all addresses $\anadr\in A$:
	\begin{enumerate}[label=({G}\arabic*),leftmargin=1.3cm] 
		\item \label[property]{proof:gcplus:goal:control} $\controlof{\tau.\anact}=\controlof{\sigma.\anactp}$ 
		\item \label[property]{proof:gcplus:goal:valid-heap} $\restrict{\heapcomput{\tau.\anact}}{\validof{\tau.\anact}}=\restrict{\heapcomput{\sigma.\anactp}}{\validof{\sigma.\anactp}}$ 
		\item \label[property]{proof:gcplus:goal:address-heap-pvars} $\forall\, \apvar\in\pvars.~ \heapcomputof{\tau.\anact}{\apvar}=\anadr \iff \heapcomputof{\sigma.\anactp}{\apvar}=\anadr$ 
		\item \label[property]{proof:gcplus:goal:address-heap-pexp} $\forall\, \anadrp\in\heapcomputof{\tau.\anact}{\validof{\tau.\anact}}.~ \heapcomputof{\tau.\anact}{\psel{\anadrp}}=\anadr \iff \heapcomputof{\sigma.\anactp}{\psel{\anadrp}}=\anadr$ 
		\item \label[property]{proof:gcplus:goal:allocated} $\anadr\in\allocatedof{\tau.\anact}\iff\anadr\in\allocatedof{\sigma.\anactp}$
		\item \label[property]{proof:gcplus:goal:obsrel-NEW} $\forall\, \anadrp{\,\in\,}\adrof{\restrict{\heapcomput{\tau.\anact}}{\validof{\tau.\anact}}}\cup A.\; \freeableof{\tau.\anact}{\anadrp}\subseteq\freeableof{\sigma.\anactp}{\anadrp}$ 
	\end{enumerate}
	If $A=\emptyset$, then \Cref{proof:gcplus:goal:address-heap-pvars,proof:gcplus:goal:address-heap-pexp,proof:gcplus:goal:allocated} are vacuously true.
	For the remainder of the proof, we assume $A\neq\emptyset$ and fix some arbitrary addresses $\anadr\in A$.
	Because $\anactp$ executes the same command $\acom$ as $\anact$, we get \Cref{proof:gcplus:goal:control} from \Cref{proof:gcplus:premis:control}.
	Hence, we do not comment on \Cref{proof:gcplus:goal:control}  in the following.
	Also note that the following holds due to \Cref{proof:gcplus:premis:valid-heap} together with \Cref{thm:computequiv-implies-same-valid}:
	\begin{auxiliary}
		\label[aux]{proof:gcplus:aux:valid}
		\validof{\tau}=\validof{\sigma}
	\end{auxiliary}

	\begin{casedistinction}

		\item[$\acom\equiv\apvar=\apvarp$]
			The update is $\anup=[\apvar\mapsto\anadrp]$ with $\heapcomputof{\tau}{\apvarp}=\anadrp$.
			We choose $\anupp=[\apvar\mapsto\anadrp']$ such that $\heapcomputof{\sigma}{\apvarp}=\anadrp'$.
			Note that we have $\anadrp=\anadrp'$ if $\apvarp\in\validof{\tau}$.
			Otherwise, this equality may not hold.

			\Cref{proof:gcplus:goal:allocated} holds by definition together with \Cref{proof:gcplus:premis:allocated}.
			It remains to show \Cref{proof:gcplus:goal:valid-heap,proof:gcplus:goal:address-heap-pvars,proof:gcplus:goal:address-heap-pexp,proof:gcplus:goal:obsrel-NEW}.

			\ad{\Cref{proof:gcplus:goal:valid-heap}}
			First consider the case where we have $\apvarp\in\validof{\tau}$.
			By \Cref{proof:gcplus:aux:valid} we also have $\apvarp\in\validof{\sigma}$.
			We get \[
				\restrict{\heapcomput{\tau.\anact}}{\validof{\tau.\anact}}
				=
				\restrict{\heapcomput{\tau}[\apvar\mapsto\anadrp]}{\validof{\tau}\cup\set{\apvar}}
				=
				(\restrict{\heapcomput{\tau}}{\validof{\tau}})[\apvar\mapsto\anadrp]
				\ .
			\]
			The second equality holds because $\heapcomputof{\tau}{\apvarp}=\anadrp$.
			That is, we preserve $\dsel{\anadrp}$ in $\restrict{\heapcomput{\tau}}{\validof{\tau}}$ and only need to update the mapping of $\apvar$.
			Similarly, we get \[
				\restrict{\heapcomput{\sigma.\anactp}}{\validof{\sigma.\anactp}}
				=
				\restrict{\heapcomput{\sigma}[\apvar\mapsto\anadrp']}{\validof{\sigma}\cup\set{\apvar}}
				=
				(\restrict{\heapcomput{\sigma}}{\validof{\sigma}})[\apvar\mapsto\anadrp']
				\ .
			\]
			By \Cref{proof:gcplus:premis:valid-heap} we have $\restrict{\heapcomput{\sigma}}{\validof{\sigma}}=\restrict{\heapcomput{\tau}}{\validof{\tau}}$.
			Moreover, it provides $\heapcomputof{\tau}{\apvarp}=\heapcomputof{\sigma}{\apvarp}$ because $\apvarp$ is valid.
			This means $\anadrp=\anadrp'$.
			So we conclude by \[
				\restrict{\heapcomput{\sigma.\anactp}}{\validof{\sigma.\anactp}}
				=
				(\restrict{\heapcomput{\sigma}}{\validof{\sigma}})[\apvar\mapsto\anadrp']
				=
				(\restrict{\heapcomput{\tau}}{\validof{\tau}})[\apvar\mapsto\anadrp]
				=
				\restrict{\heapcomput{\tau.\anact}}{\validof{\tau.\anact}}
				\ .
			\]

			Now consider the case $\apvarp\notin\validof{\tau}$.
			By \Cref{proof:gcplus:aux:valid} we also have $\apvarp\notin\validof{\sigma}$.
			We get \[
				\restrict{\heapcomput{\tau.\anact}}{\validof{\tau.\anact}}
				=
				\restrict{\heapcomput{\tau}[\apvar\mapsto\anadrp]}{\validof{\tau}\setminus\set{\apvar}}
				=
				\restrict{\heapcomput{\tau}}{\validof{\tau}\setminus\set{\apvar}}
				\ .
			\]
			The last equality holds because the update does not survive the restriction to a set which is guaranteed not to contain $\apvar$.
			Similarly, we get \[
				\restrict{\heapcomput{\sigma.\anactp}}{\validof{\sigma.\anactp}}
				=
				\restrict{\heapcomput{\sigma}[\apvar\mapsto\anadrp']}{\validof{\sigma}\setminus\set{\apvar}}
				=
				\restrict{\heapcomput{\sigma}}{\validof{\sigma}\setminus\set{\apvar}}
				\ .
			\]
			We now get:
			\begin{align*}
				&
				\restrict{\heapcomput{\tau}}{\validof{\tau}\setminus\set{\apvar}}
				=
				\restrict{(\restrict{\heapcomput{\tau}}{\validof{\tau}})}{\validof{\tau}\setminus\set{\apvar}}
				=
				\restrict{(\restrict{\heapcomput{\sigma}}{\validof{\sigma}})}{\validof{\tau}\setminus\set{\apvar}}
				=
				\restrict{(\restrict{\heapcomput{\sigma}}{\validof{\sigma}})}{\validof{\sigma}\setminus\set{\apvar}}
				=
				\restrict{\heapcomput{\sigma}}{\validof{\sigma}\setminus\set{\apvar}}
			\end{align*}
			The first equality is by definition of restrictions, the second equality is due to \Cref{proof:gcplus:premis:valid-heap}, the third one is by \Cref{proof:gcplus:aux:valid}, and the last is again by definition.
			This concludes the property.

			\ad{\Cref{proof:gcplus:goal:address-heap-pvars}}
			Only the valuation of $\apvar$ is changed by both $\anact$ and $\anactp$.
			So by \Cref{proof:gcplus:premis:address-heap-pvars} it suffices to show $\heapcomputof{\tau.\anact}{\apvar}{=}\anadr{\iff}\heapcomputof{\sigma.\anactp}{\apvar}{=}\anadr$.
			This follows easily: \[\heapcomputof{\tau.\anact}{\apvar}=\anadr\iff\heapcomputof{\tau}{\apvarp}=\anadr\iff\heapcomputof{\sigma}{\apvarp}=\anadr\iff\heapcomputof{\sigma.\anactp}{\apvar}=\anadr\] where the first and last equivalence hold due to the updates $\anup$ and $\anupp$ and the second equality holds by \Cref{proof:gcplus:premis:address-heap-pvars}.

			\ad{\Cref{proof:gcplus:goal:address-heap-pexp}}
			Let $\anadrpp\in\heapcomputof{\tau.\anact}{\validof{\tau.\anact}}$.
			We have $\anadrpp\in\heapcomputof{\tau}{\validof{\tau}}$.
			To see this, consider some $\apexp\in\validof{\tau.\anact}$ such that $\heapcomputof{\tau.\anact}{\apexp}=\anadrpp$.
			If $\apexp\not\equiv\apvar$, then $\heapcomputof{\tau.\anact}{\apexp}=\heapcomputof{\tau}{\apexp}$ and $\apexp\in\validof{\tau}$ by definition.
			So $\anadrpp\in\heapcomputof{\tau}{\validof{\tau}}$ follows immediately.
			Otherwise, in the case of $\apexp\equiv\apvar$, we have $\heapcomputof{\tau.\anact}{\apvar}=\heapcomputof{\tau}{\apvarp}$.
			Moreover, $\apvar\in\validof{\tau.\anact}$ implies $\apvarp\in\validof{\tau}$.
			Hence, $\anadrpp\in\heapcomputof{\tau}{\validof{\tau}}$.
			We then conclude as follows:
			\begin{align*}
				\heapcomputof{\tau.\anact}{\psel{\anadrpp}}=\anadr
				\iff
				\heapcomputof{\tau}{\psel{\anadrpp}}=\anadr
				\iff
				\heapcomputof{\sigma}{\psel{\anadrpp}}=\anadr
				\iff
				\heapcomputof{\sigma.\anactp}{\psel{\anadrpp}}=\anadr
			\end{align*} where the first/last equivalence hold because $\anact$/$\anactp$ does not update any pointer selector and the second equivalence holds by $\anadrpp\in\heapcomputof{\tau}{\validof{\tau}}$ from above together with \Cref{proof:gcplus:premis:address-heap-pexp}.

			\ad{\Cref{proof:gcplus:goal:obsrel-NEW}}
			It suffices to show $\adrof{\restrict{\heapcomput{\tau.\anact}}{\validof{\tau.\anact}}}\subseteq\adrof{\restrict{\heapcomput{\tau}}{\validof{\tau}}}$ because of \Cref{proof:gcplus:premis:obsrel-NEW} together with the fact that $\anact$/$\anactp$ do not emit an event.
			We have:
			\begin{align*}
				&
				\adrof{\restrict{\heapcomput{\tau.\anact}}{\validof{\tau.\anact}}}
				\\=~&
				(\validof{\tau.\anact}\cap\adr)\cup\heapcomputof{\tau.\anact}{\validof{\tau.\anact}}
				\\\subseteq~&
				((\validof{\tau}\cup\set{\apvar})\cap\adr)\cup\heapcomputof{\tau.\anact}{\validof{\tau.\anact}}
				\\=~&
				(\validof{\tau}\cap\adr)\cup\heapcomputof{\tau.\anact}{\validof{\tau.\anact}}
			\end{align*}
			where the first equality is due to \Cref{thm:adrof-valid-heap-restriction}.
			It remains to show $\heapcomputof{\tau.\anact}{\validof{\tau.\anact}}\subseteq\heapcomputof{\tau}{\validof{\tau}}$.
			\begin{itemize}
				\item 
					If $\apvarp\notin\validof{\tau}$ we get:
					\begin{align*}
						\heapcomputof{\tau.\anact}{\validof{\tau.\anact}}
						=
						\heapcomputof{\tau.\anact}{\validof{\tau}\setminus\set{\apvar}}
						=
						\heapcomputof{\tau}{\validof{\tau}\setminus\set{\apvar}}
						\subseteq
						\heapcomputof{\tau}{\validof{\tau}}
						\ .
					\end{align*}

				\item
					If $\apvarp\in\validof{\tau}$ we get:
					\begin{align*}
						\heapcomputof{\tau.\anact}{\validof{\tau.\anact}}
						=~&
						\heapcomputof{\tau.\anact}{(\validof{\tau}\setminus\set{\apvar})\cup\set{\apvar}}
						=
						\heapcomputof{\tau}{\validof{\tau}\setminus\set{\apvar}}\cup\set{\anadrp}
						\\\subseteq~&
						\heapcomputof{\tau}{\validof{\tau}}\cup\set{\anadrp}
						=
						\heapcomputof{\tau}{\validof{\tau}}
					\end{align*}
					where the last equality holds because of $\heapcomputof{\tau}{\apvarp}=\anadrp$ together with $\apvarp\in\validof{\tau}$ by assumption.
			\end{itemize}
			This concludes the property.

		\item[$\acom\equiv\apvar=\psel{\apvarp}$]
			By \Cref{proof:gcplus:premis:tau-deref-NEW} we have $\apvarp\in\validof{\tau}$.
			And by the semantics we have $\heapcomputof{\tau}{\apvarp}\in\adr$.
			Then the update is $\anup=[\apvar\mapsto\anadrpp]$ with $\heapcomputof{\tau}{\apvarp}=\anadrp$ and $\heapcomputof{\tau}{\psel{\anadrp}}=\anadrpp$.
			\Cref{proof:gcplus:premis:valid-heap} gives $\heapcomputof{\sigma}{\apvarp}=\heapcomputof{\tau}{\apvarp}=\anadrp$.
			So we choose $\anupp=[\apvar\mapsto\anadrpp']$ with $\heapcomputof{\tau}{\psel{\anadrp}}=\anadrpp'$.

			\Cref{proof:gcplus:goal:allocated} holds by definition together with \Cref{proof:gcplus:premis:allocated}.
			It remains to show \Cref{proof:gcplus:goal:valid-heap,proof:gcplus:goal:address-heap-pvars,proof:gcplus:goal:address-heap-pexp,proof:gcplus:goal:obsrel-NEW}.

			\ad{\Cref{proof:gcplus:goal:valid-heap}}
			First consider the case where we have $\psel{\anadrp}\in\validof{\tau}$.
			By \Cref{proof:gcplus:aux:valid} this implies $\psel{\anadrp}\in\validof{\sigma}$.
			That is, $\apvar$ is validated by the assignment.
			We get \[
				\restrict{\heapcomput{\tau.\anact}}{\validof{\tau.\anact}}
				=
				\restrict{\heapcomput{\tau}[\apvar\mapsto\anadrpp]}{\validof{\tau}\cup\set{\apvar}}
				=
				(\restrict{\heapcomput{\tau}}{\validof{\tau}})[\apvar\mapsto\anadrpp]
				\ .
			\]
			The second equality holds because $\heapcomputof{\tau}{\psel{\anadrp}}=\anadrpp$.
			That is, we preserve $\dsel{\anadrpp}$ in $\restrict{\heapcomput{\tau}}{\validof{\tau}}$ and only need to update the mapping of $\apvar$.
			Similarly, we get \[
				\restrict{\heapcomput{\sigma.\anactp}}{\validof{\sigma.\anactp}}
				=
				\restrict{\heapcomput{\sigma}[\apvar\mapsto\anadrpp']}{\validof{\sigma}\cup\set{\apvar}}
				=
				(\restrict{\heapcomput{\sigma}}{\validof{\sigma}})[\apvar\mapsto\anadrpp']
				\ .
			\]
			By \Cref{proof:gcplus:premis:valid-heap} we have $\restrict{\heapcomput{\sigma}}{\validof{\sigma}}=\restrict{\heapcomput{\tau}}{\validof{\tau}}$.
			Moreover, it provides $\heapcomputof{\tau}{\psel{\anadrp}}=\heapcomputof{\sigma}{\psel{\anadrp}}$ because $\psel{\anadrp}$ is valid.
			This means $\anadrpp=\anadrpp'$.
			So we conclude by \[
				\restrict{\heapcomput{\sigma.\anactp}}{\validof{\sigma.\anactp}}
				=
				(\restrict{\heapcomput{\sigma}}{\validof{\sigma}})[\apvar\mapsto\anadrpp']
				=
				(\restrict{\heapcomput{\tau}}{\validof{\tau}})[\apvar\mapsto\anadrpp]
				=
				\restrict{\heapcomput{\tau.\anact}}{\validof{\tau.\anact}}
				\ .
			\]

			Now consider the case $\psel{\anadrp}\notin\validof{\tau}$.
			By \Cref{proof:gcplus:aux:valid} we also have $\psel{\anadrp}\notin\validof{\sigma}$.
			We get \[
				\restrict{\heapcomput{\tau.\anact}}{\validof{\tau.\anact}}
				=
				\restrict{\heapcomput{\tau}[\apvar\mapsto\anadrpp]}{\validof{\tau}\setminus\set{\apvar}}
				=
				\restrict{\heapcomput{\tau}}{\validof{\tau}\setminus\set{\apvar}}
				\ .
			\]
			The last equality holds because the update does not survive the restriction to a set which is guaranteed not to contain $\apvar$.
			Similarly, we have \[
				\restrict{\heapcomput{\sigma.\anactp}}{\validof{\sigma.\anactp}}
				=
				\restrict{\heapcomput{\sigma}[\apvar\mapsto\anadrpp']}{\validof{\sigma}\setminus\set{\apvar}}
				=
				\restrict{\heapcomput{\sigma}}{\validof{\sigma}\setminus\set{\apvar}}
				\ .
			\]
			We now get:
			\begin{align*}
				&
				\restrict{\heapcomput{\tau}}{\validof{\tau}\setminus\set{\apvar}}
				=
				\restrict{(\restrict{\heapcomput{\tau}}{\validof{\tau}})}{\validof{\tau}\setminus\set{\apvar}}
				=
				\restrict{(\restrict{\heapcomput{\sigma}}{\validof{\sigma}})}{\validof{\tau}\setminus\set{\apvar}}
				=
				\restrict{(\restrict{\heapcomput{\sigma}}{\validof{\sigma}})}{\validof{\sigma}\setminus\set{\apvar}}
				=
				\restrict{\heapcomput{\sigma}}{\validof{\sigma}\setminus\set{\apvar}}
			\end{align*}
			The first equality is by definition of restrictions, the second equality is due to \Cref{proof:gcplus:premis:valid-heap}, the third one is by \Cref{proof:gcplus:aux:valid}, and the last is again by definition.
			This concludes the property.

			\ad{\Cref{proof:gcplus:goal:address-heap-pvars}}
			Only the valuation of $\apvar$ is changed by both $\anact$ and $\anactp$.
			So by \Cref{proof:gcplus:premis:address-heap-pvars} it suffices to show $\heapcomputof{\tau.\anact}{\apvar}=\anadr{\iff}\heapcomputof{\sigma.\anactp}{\apvar}=\anadr$.
			This follows easily: \[\heapcomputof{\tau.\anact}{\apvar}=\anadr\iff\heapcomputof{\tau}{\psel{\anadrp}}=\anadr\iff\heapcomputof{\sigma}{\psel{\anadrp}}=\anadr\iff\heapcomputof{\sigma.\anactp}{\apvar}=\anadr\] where the first/last equivalence holds due to the update $\anup$/$\anupp$.
			The second equality holds by \Cref{proof:gcplus:premis:address-heap-pexp} because $\apvarp\in\validof{\tau}$ from above yields $\anadrp=\heapcomputof{\tau}{\apvarp}\in\heapcomputof{\tau}{\validof{\tau}}$.

			\ad{\Cref{proof:gcplus:goal:address-heap-pexp}}
			Let $\tilde{\anadr}\in\heapcomputof{\tau.\anact}{\validof{\tau.\anact}}$.
			We have $\tilde{\anadr}\in\heapcomputof{\tau}{\validof{\tau}}$.
			To see this, consider some $\apexp\in\validof{\tau.\anact}$ such that $\heapcomputof{\tau.\anact}{\apexp}=\tilde{\anadr}$.
			If $\apexp\not\equiv\apvar$, then $\heapcomputof{\tau.\anact}{\apexp}=\heapcomputof{\tau}{\apexp}$ and $\apexp\in\validof{\tau}$ by definition.
			So $\tilde{\anadr}\in\heapcomputof{\tau}{\validof{\tau}}$ follows immediately.
			Otherwise, in the case of $\apexp\equiv\apvar$, we have $\heapcomputof{\tau.\anact}{\apvar}=\heapcomputof{\tau}{\apvarp}$.
			Moreover, $\apvar\in\validof{\tau.\anact}$ implies $\apvarp\in\validof{\tau}$.
			Hence, $\tilde{\anadr}\in\heapcomputof{\tau}{\validof{\tau}}$.
			We then conclude as follows:
			\begin{align*}
				\heapcomputof{\tau.\anact}{\psel{\tilde{\anadr}}}=\anadr
				\iff
				\heapcomputof{\tau}{\psel{\tilde{\anadr}}}=\anadr
				\iff
				\heapcomputof{\sigma}{\psel{\tilde{\anadr}}}=\anadr
				\iff
				\heapcomputof{\sigma.\anactp}{\psel{\tilde{\anadr}}}=\anadr
			\end{align*} where the first/last equivalence hold because $\anact$/$\anactp$ does not update any pointer selector and the second equivalence holds by $\tilde{\anadr}\in\heapcomputof{\tau}{\validof{\tau}}$ from above together with \Cref{proof:gcplus:premis:address-heap-pexp}.

			\ad{\Cref{proof:gcplus:goal:obsrel-NEW}}
			It suffices to show $\adrof{\restrict{\heapcomput{\tau.\anact}}{\validof{\tau.\anact}}}\subseteq\adrof{\restrict{\heapcomput{\tau}}{\validof{\tau}}}$ because of \Cref{proof:gcplus:premis:obsrel-NEW} together with the fact that $\anact$/$\anactp$ do not emit an event.
			We have:
			\begin{align*}
				&
				\adrof{\restrict{\heapcomput{\tau.\anact}}{\validof{\tau.\anact}}}
				=
				(\validof{\tau.\anact}\cap\adr)\cup\heapcomputof{\tau.\anact}{\validof{\tau.\anact}}
				\\\subseteq~&
				((\validof{\tau}\cup\set{\apvar})\cap\adr)\cup\heapcomputof{\tau.\anact}{\validof{\tau.\anact}}
				=
				(\validof{\tau}\cap\adr)\cup\heapcomputof{\tau.\anact}{\validof{\tau.\anact}}
			\end{align*}
			where the first equality is due to \Cref{thm:adrof-valid-heap-restriction}.
			It remains to show $\heapcomputof{\tau.\anact}{\validof{\tau.\anact}}\subseteq\heapcomputof{\tau}{\validof{\tau}}$.
			\begin{itemize}
				\item
					If $\psel{\anadrp}\notin\validof{\tau}$ we get:
					\begin{align*}
						\heapcomputof{\tau.\anact}{\validof{\tau.\anact}}
						=
						\heapcomputof{\tau.\anact}{\validof{\tau}\setminus\set{\apvar}}
						=
						\heapcomputof{\tau}{\validof{\tau}\setminus\set{\apvar}}
						\subseteq
						\heapcomputof{\tau}{\validof{\tau}}
					\end{align*}

				\item
					If $\psel{\anadrp}\in\validof{\tau}$ we get:
					\begin{align*}
						\heapcomputof{\tau.\anact}{\validof{\tau.\anact}}
						=~&
						\heapcomputof{\tau.\anact}{(\validof{\tau}\setminus\set{\apvar})\cup\set{\apvar}}
						=
						\heapcomputof{\tau}{\validof{\tau}\setminus\set{\apvar}}\cup\set{\anadrpp}
						\\\subseteq~&
						\heapcomputof{\tau}{\validof{\tau}}\cup\set{\anadrpp}
						=
						\heapcomputof{\tau}{\validof{\tau}}
					\end{align*}
					where the last equality holds because of $\heapcomputof{\tau}{\psel{\anadrp}}=\anadrpp$ together with $\psel{\anadrp}\in\validof{\tau}$ by assumption.
			\end{itemize}
			This concludes the property.

		\item[$\acom\equiv\psel{\apvar}=\apvarp$]
			This case is analogous to the previous one.

		\item[$\acom\equiv\advar=\opof{\advar_1,\ldots, \advar_n}$]
			The update is $\anup=[\advar\mapsto\advalue]$ with $\advalue=\opof{\heapcomputof{\tau}{\advar_1},\dots,\heapcomputof{\tau}{\advar_n}}$.
			We have $\advar_1,\dots,\advar_n\in\domof{\restrict{\heapcomput{\tau}}{\validof{\tau}}}$ by definition.
			So \Cref{proof:gcplus:premis:valid-heap} gives $\advalue=\opof{\heapcomputof{\sigma}{\advar_1},\dots,\heapcomputof{\sigma}{\advar_n}}$.
			Hence, we choose $\anupp=[\advar\mapsto\advalue]$.
			That is, the same update is performed, $\anup=\anupp$.

			We immediately get \Cref{proof:gcplus:goal:valid-heap,proof:gcplus:goal:address-heap-pvars,proof:gcplus:goal:address-heap-pexp,proof:gcplus:goal:allocated,proof:gcplus:goal:obsrel-NEW} from \Cref{proof:gcplus:premis:valid-heap,proof:gcplus:premis:address-heap-pvars,proof:gcplus:premis:address-heap-pexp,proof:gcplus:premis:allocated,proof:gcplus:premis:obsrel-NEW}.

		\item[$\acom\equiv\advar=\dsel{\apvarp}$]
			By \Cref{proof:gcplus:premis:tau-deref-NEW} we have $\apvarp\in\validof{\tau}$.
			And by the semantics we have $\heapcomputof{\tau}{\apvarp}\in\adr$.
			Then, the update is $\anup=[\advar\mapsto\advalue]$ with $\heapcomputof{\tau}{\apvarp}=\anadrp$ and $\heapcomputof{\tau}{\dsel{\anadrp}}=\advalue$.
			We get $\dsel{\anadrp}\in\domof{\restrict{\heapcomput{\tau}}{\validof{\tau}}}$ because $\apvarp\in\validof{\tau}$.
			From \Cref{proof:gcplus:aux:valid,proof:gcplus:premis:valid-heap} we get $\apvarp\in\validof{\sigma}$, $\heapcomputof{\sigma}{\apvarp}=\heapcomputof{\tau}{\apvarp}=\anadrp$, and $\heapcomputof{\sigma}{\dsel{\anadrp}}=\heapcomputof{\tau}{\dsel{\anadrp}}=\advalue$.
			So we choose $\anupp=\anup$.

			Then, \Cref{proof:gcplus:goal:valid-heap,proof:gcplus:goal:address-heap-pvars,proof:gcplus:goal:address-heap-pexp,proof:gcplus:goal:allocated,proof:gcplus:goal:obsrel-NEW} follow immediately from \Cref{proof:gcplus:premis:valid-heap,proof:gcplus:premis:address-heap-pvars,proof:gcplus:premis:address-heap-pexp,proof:gcplus:premis:allocated,proof:gcplus:premis:obsrel-NEW}.

		\item[$\acom\equiv\dsel{\apvar}=\advarp$]
			By \Cref{proof:gcplus:premis:tau-deref-NEW} we have $\apvar\in\validof{\tau}$.
			And by the semantics we have $\heapcomputof{\tau}{\apvar}\in\adr$.
			Then, the update is $\anup=[\dsel{\anadrp}\mapsto\advalue]$ with $\heapcomputof{\tau}{\apvar}=\anadrp$ and $\heapcomputof{\tau}{\advarp}=\advalue$.
			From \Cref{proof:gcplus:premis:valid-heap} we get $\heapcomputof{\sigma}{\apvar}=\heapcomputof{\tau}{\apvar}=\anadrp$ and $\heapcomputof{\sigma}{\advarp}=\heapcomputof{\tau}{\advarp}=\advalue$.
			Hence we get $\anupp=\anup$.

			Then, \Cref{proof:gcplus:goal:valid-heap,proof:gcplus:goal:address-heap-pvars,proof:gcplus:goal:address-heap-pexp,proof:gcplus:goal:allocated,proof:gcplus:goal:obsrel-NEW} follow immediately from \Cref{proof:gcplus:premis:valid-heap,proof:gcplus:premis:address-heap-pvars,proof:gcplus:premis:address-heap-pexp,proof:gcplus:premis:allocated,proof:gcplus:premis:obsrel-NEW}.

		\item[$\acom\equiv\apvar:=\malloc$]
			The update is $\anup=[\apvar\mapsto\anadrp,\psel{\anadrp}\mapsto\segval,\dsel{\anadrp}\mapsto\advalue]$ with $\anadrp\in\freshof{\tau}\cup\freedof{\tau}$ and arbitrary $\advalue$.
			By \Cref{proof:gcplus:premis:not-an-assignment-implies-enabled} we know $\anadrp\in\freshof{\sigma}\cup\freedof{\sigma}$.
			We choose \[\anupp=[\apvar\mapsto\anadrp,\psel{\anadrp}\mapsto\segval,\dsel{\anadrp}\mapsto\advalue]\ .\]

			We first show two auxiliary statements:
			\begin{auxiliary}
				\heapcomputof{\tau.\anact}{\validof{\tau.\anact}}&=\heapcomputof{\tau}{\validof{\tau}\setminus\set{\apvar,\psel{\anadrp}}}\cup\set{\anadrp}
				\label[aux]{proof:gcplus:case:malloc:valid-heap-tau-new}
				\\
				\heapcomputof{\sigma.\anact}{\validof{\sigma.\anact}}&=\heapcomputof{\sigma}{\validof{\sigma}\setminus\set{\apvar,\psel{\anadrp}}}\cup\set{\anadrp}
				\label[aux]{proof:gcplus:case:malloc:valid-heap-sigma-new}
			\end{auxiliary}

			\ad{\Cref{proof:gcplus:case:malloc:valid-heap-tau-new}}
			Note that $\anact$ changes only the valuation of $\apvar$ and $\psel{\anadrp}$.
			Moreover, by definition, we have $\validof{\tau.\anact}=\validof{\tau}\cup\set{\apvar,\psel{\anadrp}}$.
			So we conclude as follows:
			\begin{align*}
				\heapcomputof{\tau.\anact}{\validof{\tau.\anact}}
				=~&
				\heapcomputof{\tau.\anact}{\validof{\tau}\cup\set{\apvar,\psel{\anadrp}}}
				\\=~&
				\heapcomputof{\tau.\anact}{\validof{\tau}\setminus\set{\apvar,\psel{\anadrp}}}\cup\heapcomputof{\tau.\anact}{\set{\apvar,\psel{\anadrp}}}
				\\=~&
				\heapcomputof{\tau}{\validof{\tau}\setminus\set{\apvar,\psel{\anadrp}}}\cup\set{\anadrp}
				\ .
			\end{align*}

			\ad{\Cref{proof:gcplus:case:malloc:valid-heap-sigma-new}}
			Similar to \Cref{proof:gcplus:case:malloc:valid-heap-tau-new}.

			\ad{\Cref{proof:gcplus:goal:valid-heap}}
			\Cref{proof:gcplus:case:malloc:valid-heap-tau-new} with $\validof{\tau.\anact}=\validof{\tau}\cup\set{\apvar,\psel{\anadrp}}$ gives the following:
			\begin{align*}
				&
				\domof{\restrict{\heapcomput{\tau.\anact}}{\validof{\tau.\anact}}}
				=
				\validof{\tau.\anact}\cup\dvars\cup\setcond{\dsel{\anadrpp}}{\anadrpp\in\heapcomputof{\tau.\anact}{\validof{\tau.\anact}}}
				\\=~&
				(\validof{\tau}\setminus\set{\apvar,\psel{\anadrp}})\cup\set{\apvar,\psel{\anadrp}}\cup\dvars\\&\cup\setcond{\dsel{\anadrpp}}{\anadrpp\in\heapcomputof{\tau}{\validof{\tau}\setminus\set{\apvar,\psel{\anadrp}}}}\cup\set{\dsel{\anadrp}}
				\\=~&
				\domof{\restrict{\heapcomput{\tau}}{\validof{\tau}\setminus\set{\apvar,\psel{\anadrp}}}}\cup\set{\apvar,\psel{\anadrp},\dsel{\anadrp}}
				\ .
			\end{align*}
			By definition then, we have:
			\begin{align*}
			 	\restrict{\heapcomput{\tau.\anact}}{\validof{\tau.\anact}}
			 	=~&
			 	(\restrict{\heapcomput{\tau}}{\validof{\tau}\setminus\set{\apvar,\psel{\anadrp}}})[\apvar\mapsto\anadrp,\psel{\anadrp}\mapsto\segval,\dsel{\anadrp}\mapsto\advalue]
			 	\\=~&
			 	(\restrict{(\restrict{\heapcomput{\tau}}{\validof{\tau}})}{\validof{\tau}\setminus\set{\apvar,\psel{\anadrp}}})[\apvar\mapsto\anadrp,\psel{\anadrp}\mapsto\segval,\dsel{\anadrp}\mapsto\advalue]
			 	\ .
			 	\intertext{
			Along the same lines, using \Cref{proof:gcplus:case:malloc:valid-heap-sigma-new}, we get: 
			 	}
			 	\restrict{\heapcomput{\sigma.\anactp}}{\validof{\sigma.\anactp}}=~&(\restrict{(\restrict{\heapcomput{\sigma}}{\validof{\sigma}})}{\validof{\sigma}\setminus\set{\apvar,\psel{\anadrp}}})[\apvar\mapsto\anadrp,\psel{\anadrp}\mapsto\segval,\dsel{\anadrp}\mapsto\advalue]
			\end{align*}
			The above equalities now allow us to conclude the claim using \Cref{proof:gcplus:premis:valid-heap,proof:gcplus:aux:valid}:
			\[\restrict{(\restrict{\heapcomput{\tau}}{\validof{\tau}})}{\validof{\tau}\setminus\set{\apvar,\psel{\anadrp}}}=\restrict{(\restrict{\heapcomput{\sigma}}{\validof{\sigma}})}{\validof{\sigma}\setminus\set{\apvar,\psel{\anadrp}}}\]

			\ad{\Cref{proof:gcplus:goal:address-heap-pvars}}
			For $\apvar$ we have $\heapcomputof{\tau.\anact}{\apvar}=\heapcomputof{\sigma.\anact}{\apvar}$.
			For $\apvarp\neq\apvar$ we have $\heapcomputof{\tau.\anact}{\apvarp}=\heapcomputof{\tau}{\apvarp}$ and $\heapcomputof{\sigma.\anact}{\apvarp}=\heapcomputof{\sigma}{\apvarp}$.
			Hence, the claim follows from \Cref{proof:gcplus:premis:address-heap-pvars}.

			\ad{\Cref{proof:gcplus:goal:address-heap-pexp}}
			For $\psel{\anadrp}$ we have $\heapcomputof{\tau.\anact}{\psel{\anadrp}}=\heapcomputof{\sigma.\anact}{\psel{\anadrp}}$.
			For $\psel{\anadrpp}\neq\psel{\anadrp}$ we have $\heapcomputof{\tau.\anact}{\psel{\anadrpp}}=\heapcomputof{\tau}{\psel{\anadrpp}}$ and $\heapcomputof{\sigma.\anactp}{\psel{\anadrpp}}=\heapcomputof{\sigma}{\psel{\anadrpp}}$.
			Hence, the claim follows from \Cref{proof:gcplus:premis:address-heap-pexp}.

			\ad{\Cref{proof:gcplus:goal:allocated}}
			Follows from \Cref{proof:gcplus:premis:allocated} because $\allocatedof{\tau.\anact}=\allocatedof{\tau}\cup\set{\anadrp}$ and $\allocatedof{\sigma.\anactp}=\freedof{\sigma}\setminus\set{\anadrp}$.

			\ad{\Cref{proof:gcplus:goal:obsrel-NEW}}
			From \Cref{proof:gcplus:case:malloc:valid-heap-tau-new} and \Cref{thm:adrof-valid-heap-restriction} we get:
			\begin{align*}
				\adrof{\restrict{\heapcomput{\tau.\anact}}{\validof{\tau.\anact}}}
				=~&
				(\validof{\tau.\anact}\cap\adr)\cup\heapcomputof{\tau.\anact}{\validof{\tau.\anact}}
				\\=~&
				((\validof{\tau}\cup\set{\apvar,\psel{\anadrp}})\cap\adr)\cup
				\heapcomputof{\tau}{\validof{\tau}\setminus\set{\apvar,\psel{\anadrp}}}\cup\set{\anadrp}
				\\=~&
				(\validof{\tau}\cap\adr)\cup
				\heapcomputof{\tau}{\validof{\tau}\setminus\set{\apvar,\psel{\anadrp}}}\cup\set{\anadrp}
				\\\subseteq~&
				(\validof{\tau}\cap\adr)\cup\heapcomputof{\tau}{\validof{\tau}}\cup\set{\anadrp}
				=
				\adrof{\restrict{\heapcomput{\tau}}{\validof{\tau}}}\cup\set{\anadrp}
			\end{align*}
			We have $\freeableof{\tau.\anact}{\anadrpp}=\freeableof{\tau}{\anadrpp}$ and $\freeableof{\sigma.\anactp}{\anadrpp}=\freeableof{\sigma}{\anadrpp}$ for all $\anadrpp\in\adr$ because $\anact$/$\anactp$ does not emit an event.
			Hence, due to \Cref{proof:gcplus:premis:obsrel-NEW}, it remains to show that $\freeableof{\tau}{\anadrp}\subseteq\freeableof{\sigma}{\anadrp}$.
			If $\anadrp\in A$, then we get the desired inclusion from \Cref{proof:gcplus:premis:obsrel-NEW}.
			Otherwise, we get it from \Cref{proof:gcplus:premis:malloc-sigma-fresh-implies-same-freeable-NEW}.

		\item[$\acom\equiv\freeof{\anadrp}$]
			The update is $\anup=[\psel{\anadrp}\mapsto\bot,\dsel{\anadrp}\mapsto\bot]$.
			By \Cref{proof:gcplus:premis:not-an-assignment-implies-enabled} we can choose $\anactp=\anact$.

			\ad{\Cref{proof:gcplus:goal:valid-heap}}
			First, note that $\validof{\tau.\anact}=\validof{\sigma.\anactp}$.
			To see this, consider $\apexp\in\validof{\tau.\anact}$.
			By definition, this means $\apexp\in\validof{\tau}$, $\apexp\not\equiv\psel{\anadrp}$, and $\heapcomputof{\tau}{\apexp}\neq\anadrp$.
			From \Cref{proof:gcplus:aux:valid} we get $\apexp\in\validof{\sigma}$.
			From \Cref{proof:gcplus:premis:valid-heap} we get $\heapcomputof{\sigma}{\apexp}=\heapcomputof{\tau}{\apexp}\neq\anadrp$.
			Hence, $\apexp\in\validof{\sigma.\anactp}$ must hold by definition.
			This establishes $\validof{\tau.\anact}\subseteq\validof{\sigma.\anactp}$.
			The reverse inclusion follows analogously.
			So we have $\validof{\tau.\anact}=\validof{\sigma.\anactp}$ indeed.

			Now, we get
			\begin{align*}
				\domof{\restrict{\heapcomput{\tau.\anact}}{\validof{\tau.\anact}}}
				&=
				\validof{\tau.\anact} \cup \dvars \cup \setcond{\dsel{\anadrpp}}{\anadrpp\in\heapcomputof{\tau.\anact}{\validof{\tau.\anact}}}
				\\&=
				\validof{\tau.\anact} \cup \dvars \cup \setcond{\dsel{\anadrpp}}{\anadrpp\in\heapcomputof{\tau}{\validof{\tau.\anact}}}
				=
				\domof{\restrict{\heapcomput{\tau}}{\validof{\tau.\anact}}}
				\intertext{%
			where the first and last equality hold by the definition of restrictions and the second equality holds because $\psel{\anadrpp},\dsel{\anadrpp}\notin\validof{\tau.\anact}$.
			Then we get
				}
				\restrict{\heapcomput{\tau.\anact}}{\validof{\tau.\anact}}
				&=
				\restrict{\heapcomput{\tau}}{\validof{\tau.\anact}}
				\intertext{%
			because $\psel{\anadrpp},\dsel{\anadrpp}\notin\domof{\heapcomput{\tau.\anact}}$ holds due to the update $\anup$.
			Similarly, we conclude
				}
				\restrict{\heapcomput{\sigma.\anactp}}{\validof{\sigma.\anactp}}
				&=
				\restrict{\heapcomput{\sigma}}{\validof{\sigma.\anactp}}
				\ .
			\end{align*}
			Note that we have
			\begin{align*}
				\restrict{\heapcomput{\tau}}{\validof{\tau.\anact}}
				=
				\restrict{(\restrict{\heapcomput{\tau}}{\validof{\tau}})}{\validof{\tau.\anact}}
				=
				\restrict{(\restrict{\heapcomput{\sigma}}{\validof{\sigma}})}{\validof{\tau.\anact}}
				=
				\restrict{(\restrict{\heapcomput{\sigma}}{\validof{\sigma}})}{\validof{\sigma.\anactp}}
				=
				\restrict{\heapcomput{\sigma}}{\validof{\sigma.\anactp}}
			\end{align*}
			where the first holds by $\validof{\tau.\anact}\subseteq\validof{\tau}$, the second equality holds by \Cref{proof:gcplus:premis:valid-heap}, the third equality is shown above, and the last equality holds by $\validof{\sigma.\anact}\subseteq\validof{\sigma}$.
			Hence, we conclude as follows: \[ \restrict{\heapcomput{\tau.\anact}}{\validof{\tau.\anact}}=\restrict{\heapcomput{\tau}}{\validof{\tau.\anact}}=\restrict{\heapcomput{\sigma}}{\validof{\sigma.\anactp}}=\restrict{\heapcomput{\sigma.\anactp}}{\validof{\sigma.\anactp}} \ .\]

			\ad{\Cref{proof:gcplus:goal:address-heap-pvars}}
			The updates $\anup$/$\anupp$ do not change the valuation of any pointer variable.
			Hence, the claim follows immediately from \Cref{proof:gcplus:premis:address-heap-pvars}.

			\ad{\Cref{proof:gcplus:goal:address-heap-pexp}}
			Consider $\psel{\anadrp}$.
			We have $\heapcomputof{\tau.\anact}{\psel{\anadrp}}=\bot=\heapcomputof{\sigma.\anactp}{\psel{\anadrp}}$ due to the updates.
			This means $\heapcomputof{\tau.\anact}{\psel{\anadrp}}=\anadr\iff\heapcomputof{\sigma.\anactp}{\psel{\anadrp}}=\anadr$.

			Consider now $\anadrpp\in\heapcomputof{\tau.\anact}{\validof{\tau.\anact}}\setminus\set{\anadrp}$.
			There is some $\apexp\in\validof{\tau.\anact}$ with $\heapcomputof{\tau.\anact}{\apexp}=\anadrpp$.
			By definition of validity, we have $\apexp\in\validof{\tau}$.
			We also get $\apexp\not\equiv\psel{\anadrp}$ from $\apexp$ being valid.
			Hence, $\heapcomputof{\tau.\anact}{\apexp}=\heapcomputof{\tau}{\apexp}=\anadrpp$.
			So we have $\anadrpp\in\heapcomputof{\tau}{\validof{\tau}}$.
			Now we conclude as follows:
			\begin{align*}
				\heapcomputof{\tau.\anact}{\psel{\anadrpp}}
				\iff
				\heapcomputof{\tau}{\psel{\anadrpp}}
				\iff
				\heapcomputof{\sigma}{\psel{\anadrpp}}
				\iff
				\heapcomputof{\sigma.\anactp}{\psel{\anadrpp}}
			\end{align*}
			where the first/last equivalence hold because $\anup$/$\anupp$ does not modify $\psel{\anadrpp}$ (since $\anadrpp\neq\anadrp$ by choice) and the second equivalence holds by $\anadrpp\in\heapcomputof{\tau}{\validof{\tau}}$ together with \Cref{proof:gcplus:premis:address-heap-pexp}.

			\ad{\Cref{proof:gcplus:goal:allocated}}
			Follows from \Cref{proof:gcplus:premis:allocated} together with $\allocatedof{\tau.\anact}=\allocatedof{\tau}\setminus\set{\anadrp}$ and $\allocatedof{\sigma.\anactp}=\allocatedof{\sigma}\setminus\set{\anadrp}$.

			\ad{\Cref{proof:gcplus:goal:obsrel-NEW}}
			Let $\anadrpp\in\adrof{\restrict{\heapcomput{\tau.\anact}}{\validof{\tau.\anact}}}\cup A$.
			Let $\ahist\in\freeableof{\tau.\anact}{\anadrpp}$.
			We show $\ahist\in\freeableof{\sigma.\anactp}{\anadrpp}$.

			If $\anadrpp=\anadrp$, we get $\freeof{\anadrpp}.\ahist\in\freeableof{\tau}{\anadrpp}$ from \Cref{{thm:move-event-from-freeable}}.
			\Cref{proof:gcplus:premis:obsrel-NEW} gives $\freeof{\anadrpp}.\ahist\in\freeableof{\sigma}{\anadrpp}$.
			Again by \Cref{{thm:move-event-from-freeable}}, we have $\ahist\in\freeableof{\sigma.\anact}{\anadrpp}$.
			This concludes the property due to the choice of $\anact=\anactp$.

			If $\anadrpp\neq\anadrp$, we have $\ahist\in\freeableof{\tau}{\anadrpp}$ by \Cref{proof:gcplus:premis:freeable-tau}.
			\Cref{proof:gcplus:premis:obsrel-NEW} gives $\freeof{\anadrpp}.\ahist\in\freeableof{\sigma}{\anadrpp}$.
			And \Cref{proof:gcplus:premis:freeable-sigma} yields $\ahist\in\freeableof{\sigma.\anact}{\anadrpp}$.
			This concludes the property.

		\item[$\acom\equiv\assert\ \apvar=\apvarp$]
			The updates are $\anup=\anupp=\emptyset$.
			We have $\heapcomput{\tau.\anact}=\heapcomput{\tau}$ and $\heapcomput{\sigma.\anactp}=\heapcomput{\sigma}$.
			Moreover, we have $\validof{\tau.\anact}\subseteq\validof{\tau}\cup\set{\apvar,\apvarp}$ and $\validof{\sigma.\anactp}\subseteq\validof{\sigma}\cup\set{\apvar,\apvarp}$.
			And due to the semantics we have $\heapcomputof{\tau}{\apvar}=\heapcomputof{\tau}{\apvarp}$.
			By \Cref{proof:gcplus:premis:not-an-assignment-implies-enabled} we must also have $\heapcomputof{\sigma}{\apvar}=\heapcomputof{\sigma}{\apvarp}$.

			\Cref{proof:gcplus:goal:address-heap-pvars,proof:gcplus:goal:allocated} follow immediately from \Cref{proof:gcplus:premis:address-heap-pvars,proof:gcplus:premis:allocated}.
			For the remaining properties we show auxiliary statements first:
			\begin{auxiliary}
				\heapcomputof{\tau.\anact}{\validof{\tau.\anact}}&=\heapcomputof{\tau}{\validof{\tau}}
				\label[aux]{proof:gcplus:case:asserteq:valid-heap-tau}
				\\
				\heapcomputof{\sigma.\anactp}{\validof{\sigma.\anactp}}&=\heapcomputof{\sigma}{\validof{\sigma}}
				\label[aux]{proof:gcplus:case:asserteq:valid-heap-sigma}
				\\
				\validof{\tau.\anact}&=\validof{\sigma.\anactp}
				\label[aux]{proof:gcplus:case:asserteq:valid-act}
			\end{auxiliary}

			\ad{\Cref{proof:gcplus:case:asserteq:valid-heap-tau}}
			Since $\heapcomput{\tau.\anact}=\heapcomput{\tau}$, it suffices to show that $\heapcomputof{\tau}{\validof{\tau.\anact}}=\heapcomputof{\tau}{\validof{\tau}}$ holds.
			If $\validof{\tau.\anact}=\validof{\tau}$, then claim follows immediately.
			So assume $\validof{\tau.\anact}\neq\validof{\tau}$.
			By definition, we get $\validof{\tau.\anact}=\validof{\tau}\cup\set{\apvar,\apvarp}$ and $\set{\apvar,\apvarp}\cap\validof{\tau}\neq\emptyset$.
			Wlog. let $\apvar\in\validof{\tau}$.
			So $\heapcomputof{\tau}{\apvar}\in\heapcomputof{\tau}{\validof{\tau}}$.
			And because $\heapcomputof{\tau}{\apvar}=\heapcomputof{\tau}{\apvarp}$ we also get $\heapcomputof{\tau}{\apvarp}\in\heapcomputof{\tau}{\validof{\tau}}$.
			Hence, we conclude by: \[\heapcomputof{\tau}{\validof{\tau.\anact}}=\heapcomputof{\tau}{\validof{\tau}\cup\set{\apvar,\apvarp}}=\heapcomputof{\tau}{\validof{\tau}} \ .\]

			\ad{\Cref{proof:gcplus:case:asserteq:valid-heap-sigma}}
			Analogous to \Cref{proof:gcplus:case:asserteq:valid-heap-tau}.

			\ad{\Cref{proof:gcplus:case:asserteq:valid-act}}
			There are two cases.
			First, assume $\set{\apvar,\apvarp}\cap\validof{\tau}=\emptyset$.
			Then, \Cref{proof:gcplus:aux:valid} gives $\set{\apvar,\apvarp}\cap\validof{\sigma}=\emptyset$.
			By definition, this gives $\validof{\tau.\anact}=\validof{\tau}$ and $\validof{\sigma.\anactp}=\validof{\sigma}$.
			Hence, the claim follows from \Cref{proof:gcplus:aux:valid}.

			Second, $\set{\apvar,\apvarp}\cap\validof{\tau}\neq\emptyset$.
			Wlog. $\apvar\in\validof{\tau}$.
			By \Cref{proof:gcplus:aux:valid} we have $\apvar\in\validof{\sigma}$.
			And by definition, we get $\validof{\tau.\anact}=\validof{\tau}\cup\set{\apvar,\apvarp}$ and $\validof{\sigma.\anactp}=\validof{\sigma}\cup\set{\apvar,\apvarp}$.
			Again, we conclude by \Cref{proof:gcplus:aux:valid}.

			\ad{\Cref{proof:gcplus:goal:valid-heap}}
			We consider two cases.
			First, assume $\set{\apvar,\apvarp}\cap\validof{\tau}=\emptyset$.
			Then, we have $\validof{\tau.\anact}=\validof{\tau}$.
			So \Cref{proof:gcplus:case:asserteq:valid-heap-sigma} together with \Cref{proof:gcplus:aux:valid} gives $\validof{\sigma.\anactp}=\validof{\sigma}$.
			Hence, we can conclude using $\heapcomput{\tau.\anact}=\heapcomput{\tau}$ and $\heapcomput{\sigma.\anactp}=\heapcomput{\sigma}$: \[\restrict{\heapcomput{\tau.\anact}}{\validof{\tau.\anact}}=\restrict{\heapcomput{\tau}}{\validof{\tau}}=\restrict{\heapcomput{\sigma}}{\validof{\sigma}}=\restrict{\heapcomput{\sigma.\anactp}}{\validof{\sigma.\anactp}}\] where the second equality holds by \Cref{proof:gcplus:premis:valid-heap}.

			Second, assume $\set{\apvar,\apvarp}\cap\validof{\tau}\neq\emptyset$.
			Then, $\validof{\tau.\anact}=\validof{\tau}\cup{\apvar,\apvarp}$.
			So \Cref{proof:gcplus:case:asserteq:valid-heap-sigma} together with \Cref{proof:gcplus:aux:valid} gives $\validof{\sigma.\anactp}=\validof{\sigma}\cup\set{\apvar,\apvarp}$.
			Using \Cref{proof:gcplus:case:asserteq:valid-heap-tau,proof:gcplus:case:asserteq:valid-heap-sigma} and \Cref{proof:gcplus:premis:valid-heap} we now get:
			\begin{align*}
				&\domof{\restrict{\heapcomput{\tau.\anact}}{\validof{\tau.\anact}}}
				\\=~&
				\validof{\tau.\anact}\cup\dvars\cup\setcond{\dsel{\anadrpp}}{\anadrpp\in\heapcomputof{\tau.\anact}{\validof{\tau.\anact}}}
				\\=~&
				\validof{\tau}\cup\set{\apvar,\apvarp}\cup\dvars\cup\setcond{\dsel{\anadrpp}}{\anadrpp\in\heapcomputof{\tau}{\validof{\tau}}}
				\\=~&
				\domof{\restrict{\heapcomput{\tau}}{\validof{\tau}}}\cup\set{\apvar,\apvarp}
				\\=~&
				\domof{\restrict{\heapcomput{\sigma}}{\validof{\sigma}}}\cup\set{\apvar,\apvarp}
				\\=~&
				\validof{\sigma}\cup\set{\apvar,\apvarp}\cup\dvars\cup\setcond{\dsel{\anadrpp}}{\anadrpp\in\heapcomputof{\sigma}{\validof{\sigma}}}
				\\=~&
				\validof{\sigma.\anactp}\cup\dvars\cup\setcond{\dsel{\anadrpp}}{\anadrpp\in\heapcomputof{\sigma.\anactp}{\validof{\sigma.\anactp}}}
				\\=~&
				\domof{\restrict{\heapcomput{\sigma.\anactp}}{\validof{\sigma.\anactp}}}
			\end{align*}
			By \Cref{proof:gcplus:premis:valid-heap} together with $\heapcomput{\tau.\anact}=\heapcomput{\tau}$ and $\heapcomput{\sigma.\anactp}=\heapcomput{\sigma}$ it suffices to show \[\heapcomputof{\tau.\anact}{\apvar}=\heapcomputof{\sigma.\anactp}{\apvar} \qquad\text{and}\qquad \heapcomputof{\tau.\anact}{\apvarp}=\heapcomputof{\sigma.\anactp}{\apvarp}\] to conclude the claim.
			Wlog. $\apvar\in\validof{\tau}$.
			Then, $\heapcomputof{\tau}{\apvar}=\heapcomputof{\sigma}{\apvar}$.
			Hence, $\heapcomputof{\tau.\anact}{\apvar}=\heapcomputof{\sigma.\anactp}{\apvar}$.
			And the assertion in $\acom$ requires $\apvar$ and $\apvarp$ to have the same valuation, ${\heapcomputof{\tau.\anact}{\apvarp}=\heapcomputof{\sigma.\anactp}{\apvarp}}$.
			This concludes the claim.

			\ad{\Cref{proof:gcplus:goal:address-heap-pexp}}
			Follows from $\heapcomput{\tau.\anact}=\heapcomput{\tau}$ and $\heapcomput{\sigma.\anactp}=\heapcomput{\sigma}$ together with \Cref{proof:gcplus:premis:address-heap-pexp} and \Cref{proof:gcplus:case:asserteq:valid-heap-tau}.

			\ad{\Cref{proof:gcplus:goal:obsrel-NEW}}
			We have $\freeableof{\tau.\anact}{\anadrpp}=\freeableof{\tau}{\anadrpp}$ and $\freeableof{\sigma.\anactp}{\anadrpp}=\freeableof{\sigma}{\anadrpp}$ for all $\anadrpp\in\adr$ since $\anact$/$\anactp$ does not emit an event.
			Now, using \Cref{thm:adrof-valid-heap-restriction} and \Cref{proof:gcplus:case:asserteq:valid-heap-tau} we get:
			\begin{align*}
				&
				\adrof{\restrict{\heapcomput{\tau.\anact}}{\validof{\tau.\anact}}}
				=
				(\validof{\tau.\anact}\cap\adr)\cup\heapcomputof{\tau.\anact}{\validof{\tau.\anact}}
				\\\subseteq~&
				((\validof{\tau}\cup{\apvar,\apvarp})\cap\adr)\cup\heapcomputof{\tau}{\validof{\tau}}
				\\=~&
				(\validof{\tau}\cap\adr)\cup\heapcomputof{\tau}{\validof{\tau}}
				=
				\adrof{\restrict{\heapcomput{\tau}}{\validof{\tau}}}
				\ .
			\end{align*}
			This concludes the claim by \Cref{proof:gcplus:premis:obsrel-NEW}.

		\item[$\acom\equiv\assert\ \apvar\neq\apvarp$]
			We have $\anup=\anupp=\emptyset$.
			By \Cref{proof:gcplus:premis:not-an-assignment-implies-enabled} $\anact$ is enabled after $\sigma$.
			\Cref{proof:gcplus:goal:valid-heap,proof:gcplus:goal:address-heap-pvars,proof:gcplus:goal:address-heap-pexp,proof:gcplus:goal:allocated,proof:gcplus:goal:obsrel-NEW} follow immediately from \Cref{proof:gcplus:premis:valid-heap,proof:gcplus:premis:address-heap-pvars,proof:gcplus:premis:address-heap-pexp,proof:gcplus:premis:allocated,proof:gcplus:premis:obsrel-NEW}.

		\item[$\acom\equiv\assert\ \advar\prec\advarp$ with $\prec\in\set{=,\neq,<}$]
			We have $\anup=\anupp=\emptyset$ by definition.
			By \Cref{proof:gcplus:premis:not-an-assignment-implies-enabled} $\anact$ is enabled after $\sigma$.
			\Cref{proof:gcplus:goal:valid-heap,proof:gcplus:goal:address-heap-pvars,proof:gcplus:goal:address-heap-pexp,proof:gcplus:goal:allocated,proof:gcplus:goal:obsrel-NEW} follow immediately from \Cref{proof:gcplus:premis:valid-heap,proof:gcplus:premis:address-heap-pvars,proof:gcplus:premis:address-heap-pexp,proof:gcplus:premis:allocated,proof:gcplus:premis:obsrel-NEW}.

		\item[$\acom\equiv\enterof{\afuncof{\vecof{\apvar},\vecof{\advar}}}$]
			We have $\anup=\anupp=\emptyset$.
			As in the previous case, $\anact$/$\anactp$ does not have an effect on the memory, validity, freeness, and freshness.
			Hence, \Cref{proof:gcplus:goal:valid-heap,proof:gcplus:goal:address-heap-pvars,proof:gcplus:goal:address-heap-pexp,proof:gcplus:goal:allocated} follow from \Cref{proof:gcplus:premis:valid-heap,proof:gcplus:premis:address-heap-pvars,proof:gcplus:premis:address-heap-pexp,proof:gcplus:premis:allocated}.

			For the remaining property note that \Cref{proof:gcplus:premis:call-event-restriction-NEW} provides $\heapcomputof{\tau}{\vecof{\apvar}}=\heapcomputof{\sigma}{\vecof{\apvar}}$.
			And \Cref{proof:gcplus:premis:valid-heap} gives $\heapcomputof{\tau}{\vecof{\advar}}=\heapcomputof{\sigma}{\vecof{\advar}}$.
			Hence, $\anact$ and $\anactp$ emit the same event.
			That is, we have $\historyof{\tau.\anact}=\ahist_1.\anevent$ and $\historyof{\sigma.\anactp}=\ahist_2.\anevent$.
			Thus, \Cref{proof:gcplus:goal:obsrel-NEW} follows from \Cref{proof:gcplus:premis:obsrel-NEW} together with \Cref{thm:move-event-from-freeable}.

		\item[$\acom\equiv\exit$]
			Analogous to the previous case since both $\tau$ and $\sigma$ emit the same event for $\exit$ commands by definition.

	\end{casedistinction}
\end{proof}

%% file: content/appendix/theory/proofs_elision.tex

\subsection{Proofs of Elision Technique (\CREF{appendix:elision_technique})}

%
\begin{proof}[Proof of \Cref{thm:inverse-of-address-mapping-is-again-address-mapping}]
	Follows immediately from the fact that $\swapadrraw$ is a bijection.
\end{proof}

%
\begin{proof}[Proof of \Cref{thm:swaphist-is-unique}]
	Towards a contradiction, assume there is a shortest $\ahist.\anevent$ such that there are histories $\ahist_1.\anevent_1\neq\ahist_2.\anevent_2$ with $\swapadr{\ahist_1.\anevent_1}=\ahist.\anevent=\swapadr{\ahist_2.\anevent_2}$.
	Note that $\ahist_1.\anevent_1$, $\ahist_2.\anevent_2$, and $\ahist'.\anevent'$ all must have the same length.
	Also note that $\ahist.\anevent$ is indeed a shortest such computation because the claim holds for the empty history $\epsilon$. 
	Let $\anevent$ be of the form $\anevent\equiv\afunc(\athread,\vecof{\anadr},\vecof{\advalue})$.
	Let $\anevent_i\equiv\afunc^i(\athread^i,\vecof{\anadr^i},\vecof{\advalue^i})$.
	Then, $\swaphist{\anevent_i}=\afunc^i(\athread^i,\swapadr{\vecof{\anadr^i}},\vecof{\advalue^i})$.
	By choice of $\anevent_i$ we have $\swapadr{\anevent_i}=\afunc(\athread,\vecof{\anadr},\vecof{\advalue})$.
	To arrive there, we must have $\afunc^i=\afunc$, $\swapadr{\vecof{\anadr^i}}=\vecof{\anadr}$, and $\vecof{\advalue^i}=\vecof{\advalue}$.
	Hence, $\vecof{\anadr^i}=\swapadrinv{\vecof{\anadr}}$.
	Consequently, we have $\vecof{\anadr^1}=\vecof{\anadr^2}$ because $\swapadrraw$ is a bijection.
	So we have $\anevent_1=\anevent_2$.
	And by minimality of $\ahist'.\anevent'$ we get $\ahist_1=\ahist_2$.
	Altogether, we have $\ahist_1.\anevent_1=\ahist_2.\anevent_2$.
	This contradicts the assumption.
	The remaining cases for $\exit$ and $\free$ events follow analogously.
	This concludes the claim.
\end{proof}

%
\begin{proof}[Proof of \Cref{thm:swaphist-versus-elementof}]
	If $\ahist\in H$ holds, then we get $\ahist'\in\swaphist{H}$ immediately due to $\swaphist{\ahist}=\ahist'$.
	For the reverse direction, we know that there is some $\tilde\ahist\in H$ with $\swaphist{\tilde\ahist}=\ahist'$.
	\Cref{thm:swaphist-is-unique} yields $\tilde\ahist=\ahist$.
	So $\ahist\in H$.
	This concludes the claim.
\end{proof}

%
\begin{proof}[Proof of \Cref{thm:swap-vs-set-opereations}]
	Let $\anadr'\in\adr$.
	Since $\swapadrraw$ is a bijection, there is $\anadr\in\adr$ such that $\anadr'=\swapadr{\anadr}$ and $\anadr'\notin\swapadr{\adr\setminus\set{\anadr}}$.
	Hence we have $\anadr\notin A_1\implies\anadr'\notin\swapadr{A_1}$.
	And $\anadr\in A_1\implies\anadr'\in\swapadr{A_1}$ by choice of $\anadr$.
	So altogether we have: $\anadr'\in\swapadr{A_1}\iff\anadr\in A_1$.
	Similarly, $\anadr'\prall{\in}\swapadr{A_2}\iff\anadr\prall{\in} A_2$.
	Thus: $\anadr'\prall{\in}\swapadr{A_1}\otimes\swapadr{A_2}\iff\anadr\prall{\in} A_1\otimes A_2$.
	And with the same arguments we derive $\anadr\in A_1\otimes A_2\iff\anadr'\in\swapadr{A_1\otimes A_2}$.
	This concludes the first equality.

	Similarly, we have \[\psel{\anadr'}\in\swapexp{B_1}\iff\psel{\anadr}\in B_1 \quad\text{and}\quad \psel{\anadr'}\in\swapexp{B_2}\iff\psel{\anadr}\in B_2 \ .\]
	As before, we derive $\psel{\anadr'}\in\swapexp{B_1}\otimes\swapexp{B_2}\iff\psel{\anadr}\in B_1\otimes B_2$.
	And with the same arguments we get $\psel{\anadr}\in B_1\otimes B_2\iff\psel{\anadr'}\in\swapexp{B_1\otimes B_2}$.
	This concludes the second equality.

	Consider now some history $\ahist'$.
	Since $\swapadrraw$ is a bijection, there is $\ahist$ with $\swaphist{\ahist}=\ahist'$.
	Then, \Cref{thm:swaphist-versus-elementof} yields $\ahist'\in\swaphist{C_1}\iff\ahist\in C_1$.
	Similarly, $\ahist'\in\swaphist{C_2}\iff\ahist\in C_2$ follows.
	So $\ahist'\in\swaphist{C_1}\otimes\swaphist{C_2}\iff\ahist\in C_1\otimes C_2$.
	Again by \Cref{thm:swaphist-versus-elementof} we get $\ahist\in C_1\otimes C_2\iff\ahist'\in\swaphist{C_1\otimes C_2}$.
	This concludes the third equality.
\end{proof}

%
\begin{proof}[Proof of \Cref{thm:swaphist-of-swaphistinv-is-id}]
	Consider some event $\anevent\equiv\afunc(\athread,\vecof{\anadr},\vecof{\advalue})$.
	Then we have:
	\begin{align*}
		&
		\swaphistinv{\swaphist{\anevent}}
		=
		\swaphistinv{\swaphist{\afunc(\athread,\vecof{\anadr},\vecof{\advalue})}}
		=
		\swaphistinv{\afunc(\athread,\swapadr{\vecof{\anadr}},\vecof{\advalue})}
		\\=~&
		\afunc(\athread,\swapadrinv{\swapadr{\vecof{\anadr}}},\vecof{\advalue})
		=
		\afunc(\athread,\vecof{\anadr},\vecof{\advalue})
	\end{align*}
	Analogously, $\swaphistinv{\swaphist{\freeof{\anadr}}}=\freeof{\anadr}$ and $\swaphistinv{\swaphist{\exitof{\athread}}}=\exitof{\athread}$.

	The overall claim follows then from inductively applying the above to $\ahist$.
	In the base case one has $\swaphistinv{\swaphist{\epsilon}}=\epsilon$ by definition.
	For $\ahist.\anevent$ one has
	\begin{align*}
		\swaphistinv{\swaphist{\ahist.\anevent}}
		=~&
		\swaphistinv{\swaphist{\ahist}.\swaphist{\anevent}}
		\\=~&
		\swaphistinv{\swaphist{\ahist}}.\swaphistinv{\swaphist{\anevent}}
		=
		\ahist.\anevent
	\end{align*}
	where the last equality is due to induction (for $\ahist$) and due to the above reasoning (for $\anevent$).
\end{proof}

%
\begin{proof}[Proof of \Cref{thm:swaphist-versus-accpetance}]
	We first show the following auxiliary: \[ (\alocation,\varphi)\trans{\ahist}(\alocation',\varphi)\iff(\alocation,\swapadrraw\circ\varphi)\trans{\swaphist{\ahist}}(\alocation',\swapadrraw\circ\varphi) \ . \]
	To that end, let $(\alocation,\varphi)\trans{\anevent}(\alocation',\varphi)$ be some observer step and let $\vecof{\anovar}=\domof{\varphi}$ be the observer variables.
	We show that also \[ (\alocation,\varphi')\trans{\swaphist{\anevent}}(\alocation',\varphi') \quad\text{with}\quad \varphi'=\swapadrraw\circ\varphi \] is an observer step.
	The event $\anevent$ is of the form $\anevent\equiv\evt{\afunc}{\vecof{\avalue}}$.
	By the definition, there is a transition
	\begin{align*}
	 	\alocation\trans{\translab{\evt{\afunc}{\vecof{r}}}{\aguard}}\alocationp
	 	\quad\text{such that}\quad
	 	\aguard[\vecof{r}\mapsto\vecof{\avalue},\vecof{\anovar}\mapsto\varphi(\vecof{\anovar})]
	 	\logicequiv\mathit{true}\ .
	\end{align*}
	Now, turn to $\swaphist{\anevent}$.
	It is of the form $\swaphist{\anevent}\equiv\evt{\afunc}{\vecof{\avaluep}}$ with $\vecof{\avaluep}=\swapadr{\vecof{\avalue}}$.
	Hence, the above transition matches.
	It remains to show that it is enabled.
	To that end, we need show that $\aguard[\vecof{r}\mapsto\vecof{\avaluep},\vecof{\anovar}\mapsto\varphi'(\vecof{\anovar})]\logicequiv\mathit{true}$.
	Intuitively, this holds because guards are composed of (in)equalities which are \emph{stable} under the bijection $\swapadrraw$.
	Formally, we know that $\aguard$ is equivalent to \[\aguard\logicequiv\bigvee_i\bigwedge_j\avar_{i,j,1}\triangleq\avar_{i,j,2} \quad\text{with}\quad \avar_{i,j,k}\in\set{p_1,p_2,\anovar,\anovarp} ~\text{and}~ \triangleq\in\set{=,\neq} \ .\]
	Hence, we have to show
	\begin{align*}
	 	&(\avar\triangleq\avarp)[\vecof{r}\mapsto\vecof{\avalue},\vecof{\anovar}\mapsto\varphi(\vecof{\anovar})]\logicequiv\mathit{true}
	 	\\\iff~&
	 	(\avar\triangleq\avarp)[\vecof{r}\mapsto\vecof{\avaluep},\vecof{\anovar}\mapsto\varphi'(\vecof{\anovar})]\logicequiv\mathit{true}
	\end{align*}
	for every $i,j$ and $\avar=\avar_{i,j,1}$ and $\avarp=\avar_{i,j,2}$.
	We conclude this as follows:
	\begin{align*}
		&
		(\avar\triangleq\avarp)[\vecof{r}\mapsto\vecof{\avalue},\vecof{\anovar}\mapsto\varphi(\vecof{\anovar})]\logicequiv\mathit{true}
		\\\iff~&
		\avar[\vecof{r}\mapsto\vecof{\avalue},\vecof{\anovar}\mapsto\varphi(\vecof{\anovar})]
		\triangleq
		\avarp[\vecof{r}\mapsto\vecof{\avalue},\vecof{\anovar}\mapsto\varphi(\vecof{\anovar})]
		\\\iff~&
		\swapadr{\avar[\vecof{r}\mapsto\vecof{\avalue},\vecof{\anovar}\mapsto\varphi(\vecof{\anovar})]}
		\triangleq
		\swapadr{\avarp[\vecof{r}\mapsto\vecof{\avalue},\vecof{\anovar}\mapsto\varphi(\vecof{\anovar})]}
		\\\iff~&
		\avar[\vecof{r}\mapsto\swapadr{\vecof{\avalue}},\vecof{\anovar}\mapsto\swapadr{\varphi(\vecof{\anovar})}]
		\\&\triangleq
		\avarp[\vecof{r}\mapsto\swapadr{\vecof{\avalue}},\vecof{\anovar}\mapsto\swapadr{\varphi(\vecof{\anovar})}]
		\\\iff~&
		\avar[\vecof{r}\mapsto\vecof{\avaluep},\vecof{\anovar}\mapsto\varphi'(\vecof{\anovar})]
		\triangleq
		\avarp[\vecof{r}\mapsto\vecof{\avaluep},\vecof{\anovar}\mapsto\varphi'(\vecof{\anovar})]
		\\\iff~&
		(\avar\triangleq\avarp)[\vecof{r}\mapsto\vecof{\avaluep},\vecof{\anovar}\mapsto\varphi'(\vecof{\anovar})]\logicequiv\mathit{true}
	\end{align*}
	where the second equivalence holds because $\swapadrraw$ is a bijections, and the third equivalence holds because $\avar$/$\avarp$ are either contained in $\vecof{r}$ or $\vecof{\anovar}$ by the definition of observers.
	This concludes the auxiliary claim.

	Altogether, the overall implication "${\implies}$" \[(\alocation,\varphi)\trans{\ahist}(\alocation',\varphi)\implies(\alocation,\swapadrraw\circ\varphi)\trans{\swaphist{\ahist}}(\alocation',\swapadrraw\circ\varphi)\] follows by applying the above argument inductively to every event in the history $\ahist$.

	For the reverse direction "${\impliedby}$" we use \Cref{thm:inverse-of-address-mapping-is-again-address-mapping}.
	That is, we apply the above result to $(\alocation,\swapadrraw\circ\varphi)\trans{\swaphist{\ahist}}(\alocation',\swapadrraw\circ\varphi)$ using the address mapping $\swapadrinvraw$.
	This yields:
	\begin{align*}
		&
		(\alocation,\swapadrraw\circ\varphi)
		\trans{\swaphist{\ahist}}
		(\alocation',\swapadrraw\circ\varphi)
		\\\implies\;&
		(\alocation,\swapadrinvraw\circ\swapadrraw\circ\varphi)
		\trans{\swaphistinv{\swaphist{\ahist}}}
		(\alocation',\swapadrinvraw\circ\swapadrraw\circ\varphi)
		\ .
	\end{align*}
	Then, \Cref{thm:swaphist-of-swaphistinv-is-id} together with $\swapadrinvraw\circ\swapadrraw=\mathit{id}$ gives:
	\begin{align*}
		&
		(\alocation,\swapadrinvraw\circ\swapadrraw\circ\varphi)
		\trans{\swaphistinv{\swaphist{\ahist}}}
		(\alocation',\swapadrinvraw\circ\swapadrraw\circ\varphi)
		\\\implies\;&
		(\alocation,\swapadrraw\circ\varphi)
		\trans{\swaphist{\ahist}}
		(\alocation',\swapadrraw\circ\varphi)
	\end{align*}
	This concludes the proof.
\end{proof}

%
\begin{proof}[Proof of \Cref{thm:swaphist-freeable}]
	Let $\anadr\in\adr$.
	We conclude as follows:
	\begin{align*}
		&
		\ahist\in\freeableof{\tau}{\anadr}
		\\\iff&~
		\historyof{\tau}.\ahist\in\historyof{\smrobs}\wedge\mathit{frees}(\ahist)\subseteq\set{\anadr}
		\\\iff&~
		\swaphist{\historyof{\tau}}.\swaphist{\ahist}\in\historyof{\smrobs}\wedge\mathit{frees}(\swaphist{\ahist})\subseteq\set{\swapadr{\anadr}}
		\\\iff&~
		\swaphist{\ahist}\in\freeableof{\sigma}{\swapadr{\anadr}}
	\end{align*}
	where the second equivalence holds because of \Cref{thm:swaphist-versus-accpetance}.
	This concludes the claim.
\end{proof}

%
\begin{proof}[Proof of \Cref{thm:elision}]
	We proceed by induction over the structure of computations.
	\begin{description}[labelwidth=6mm,leftmargin=8mm,itemindent=0mm]
		\item[IB:]
			For $\tau=\epsilon$ or $\tau=\anact$ choose $\sigma=\tau$.
			This immediately satisfies the claim.

		\item[IH:]
			For every $\tau\in\asem{A}$ there is some $\sigma$ that satisfies the properties of the \namecref{thm:elision}.

		\item[IS:]
			Consider now $\tau.\anact\in\asem{A}$.
			By induction, there is some $\sigma$ with the following properties:
			\begin{enumerate}[label=({P}\arabic*),leftmargin=1.3cm]
				\item \label[property]{proof:renaming:premis:sigmaact-enabled} $\sigma\in\asem{\swapadr{A}}$
				\item \label[property]{proof:renaming:premis:heap-pexp} $\forall\apexp\in\pexp.~\heapcomputof{\sigma}{\swapexp{\apexp}}=\swapadr{\heapcomputof{\tau}{\apexp}}$
				\item \label[property]{proof:renaming:premis:heap-dexp} $\forall\adexp\in\dexp.~\heapcomputof{\sigma}{\swapexp{\adexp}}=\heapcomputof{\tau}{\adexp}$
				\item \label[property]{proof:renaming:premis:valid} $\validof{\sigma}=\swapexp{\validof{\tau}}$
				\item \label[property]{proof:renaming:premis:free} $\freedof{\sigma}=\swapadr{\freedof{\tau}}$
				\item \label[property]{proof:renaming:premis:fresh} $\freshof{\sigma}=\swapadr{\freshof{\tau}}$
				\item \label[property]{proof:renaming:premis:history} $\historyof{\sigma}=\swaphist{\historyof{\tau}}$
				\item \label[property]{proof:renaming:premis:control} $\controlof{\sigma}=\controlof{\tau}$
			\end{enumerate}
			Let $\anact=(\athread,\acom,\anup)$.
			We show that there is some $\anactp=(\athread,\acom',\anupp)$ such that $\sigma.\anactp$ satisfies the claim.
			That is, our \textbf{g}oal is to show the following:
			\begin{enumerate}[label=({G}\arabic*),leftmargin=1.3cm]
				\item \label[property]{proof:renaming:goal:sigmaact-enabled} $\sigma.\anactp\in\asem{\swapadr{A}}$
				\item \label[property]{proof:renaming:goal:heap-pexp} $\forall\apexp\in\pexp.~\heapcomputof{\sigma.\anactp}{\swapexp{\apexp}}=\swapadr{\heapcomputof{\tau.\anact}{\apexp}}$
				\item \label[property]{proof:renaming:goal:heap-dexp} $\forall\adexp\in\dexp.~\heapcomputof{\sigma.\anactp}{\swapexp{\adexp}}=\heapcomputof{\tau.\anact}{\adexp}$
				\item \label[property]{proof:renaming:goal:valid} $\validof{\sigma.\anactp}=\swapexp{\validof{\tau.\anact}}$
				\item \label[property]{proof:renaming:goal:free} $\freedof{\sigma.\anactp}=\swapadr{\freedof{\tau.\anact}}$
				\item \label[property]{proof:renaming:goal:fresh} $\freshof{\sigma.\anactp}=\swapadr{\freshof{\tau.\anact}}$
				\item \label[property]{proof:renaming:goal:history} $\historyof{\sigma.\anactp}=\swaphist{\historyof{\tau.\anact}}$
				\item \label[property]{proof:renaming:goal:control} $\controlof{\sigma.\anactp}=\controlof{\tau.\anact}$
			\end{enumerate}
			We choose $\acom'=\acom$ for all cases except $\acom=\freeof{\anadr}$.
			For deletions, we replace the address: $\acom=\freeof{\swapadr{\anadr}}$.
			Hence, \Cref{proof:renaming:goal:control} will follow from \Cref{proof:renaming:premis:control} together with \Cref{assumption:frees-do-not-affect-control}; we will not comment on this property hereafter.
			For \Cref{proof:renaming:goal:sigmaact-enabled} we will only argue that $\anactp$ is enabled after $\sigma$.
			This, together with \Cref{proof:renaming:premis:sigmaact-enabled} yields the desired property.

			We do a case distinction on the executed command $\acom$.
			\begin{casedistinction}
				\item[$\acom\equiv\apvar:=\apvarp$]
				We have $\anup=[\apvar\mapsto\anadr]$ with $\anadr=\heapcomputof{\tau}{\apvarp}$.
				Choose $\anupp=[\apvar\mapsto\swapadr{\anadr}]$.

				\Cref{proof:renaming:goal:free,proof:renaming:goal:fresh,proof:renaming:goal:history} follow immediately from \Cref{proof:renaming:premis:free,proof:renaming:premis:fresh,proof:renaming:premis:history} because the fresh and freed addresses are not changed and no event is emitted, formally: $\freshof{\tau.\anact}=\freshof{\tau}$, $\freshof{\sigma.\anactp}=\freshof{\sigma}$, $\freedof{\tau.\anact}=\freedof{\tau}$, $\freedof{\sigma.\anactp}=\freedof{\sigma}$, $\historyof{\tau.\anact}=\historyof{\tau}$, and $\historyof{\sigma.\anactp}=\historyof{\sigma}$.
				So it remains to show \Cref{proof:renaming:goal:sigmaact-enabled,proof:renaming:goal:heap-pexp,proof:renaming:goal:heap-dexp,proof:renaming:goal:valid}.

				\ad{\Cref{proof:renaming:goal:sigmaact-enabled}}
				We have to show that $\swapadr{\anadr}=\heapcomputof{\sigma}{\apvarp}$ holds.
				By choice of $\anadr$, we have to show $\heapcomputof{\sigma}{\apvarp}=\swapadr{\heapcomputof{\tau}{\apvarp}}$.
				This follows from \Cref{proof:renaming:premis:heap-pexp} with $\swapexp{\apvarp}=\apvarp$.

				\ad{\Cref{proof:renaming:goal:heap-pexp}}
				For $\apvar\in\pexp$ the claim follows by choice of $\anupp$.
				So consider $\apexp\in\pexp$ with $\apexp\not\equiv\apvar$.
				Then, we conclude as follows:
				\begin{align*}
					\heapcomputof{\sigma.\anactp}{\swapexp{\apexp}}
					=~&
					\heapcomputof{\sigma}{\swapexp{\apexp}}
					=
					\swapadr{\heapcomputof{\tau}{\apexp}}
					\\=~&
					\swapadr{\heapcomputof{\tau.\anact}{\apexp}}
				\end{align*}
				where the first equality holds because only $\apvar$ is updated by $\anupp$ and $\swapexp{\apexp}\neq\apvar$, the second equality holds by \Cref{proof:renaming:premis:heap-pexp}, and the third equality holds because only $\apvar$ is updated by $\anup$.

				\ad{\Cref{proof:renaming:goal:heap-dexp}}
				Neither $\anup$ nor $\anupp$ update data expressions.
				So $\heapcomputof{\tau.\anact}{\adexp}=\heapcomputof{\tau}{\adexp}$ and $\heapcomputof{\sigma.\anactp}{\adexp}=\heapcomputof{\sigma}{\adexp}$ for every $\adexp\in\dexp$.
				Hence, the claim follows from \Cref{proof:renaming:premis:heap-dexp}.

				\ad{\Cref{proof:renaming:goal:valid}}
				By $\swapexp{\apvarp}=\apvarp$ we have $\apvarp\in\validof{\tau}\iff\apvarp\in\validof{\sigma}$.
				Hence, we conclude using \Cref{proof:renaming:premis:valid}, \Cref{thm:swap-vs-set-opereations}, and the definition of validity and $\swapexpraw$:
				\begin{align*}
					\validof{\sigma.\anactp}
					=
					\validof{\sigma}\otimes\set{\apvar}
					=~&
					\swapexp{\validof{\tau}}\otimes\swapexp{\set{\apvar}}
					\\=~&
					\swapexp{\validof{\tau}\otimes\set{\apvar}}
					=
					\swapexp{\validof{\tau.\anact}}
				\end{align*}
				with $\otimes:=\cup$ if $\apvarp\in\validof{\tau}$ and $\otimes:=\setminus$ otherwise.

				\item[$\acom\equiv\apvar=\psel{\apvarp}$ with $\heapcomputof{\tau}{\apvarp}\in\adr$]
				We have $\anup=[\apvar\mapsto\anadrp]$ with $\anadr=\heapcomputof{\tau}{\apvarp}$ and $\anadrp=\heapcomputof{\tau}{\psel{\anadr}}$.
				Choose $\anupp=[\apvar\mapsto\swapadr{\anadrp}]$.

				\Cref{proof:renaming:goal:free,proof:renaming:goal:fresh,proof:renaming:goal:history} follow immediately from \Cref{proof:renaming:premis:free,proof:renaming:premis:fresh,proof:renaming:premis:history} because the fresh and freed addresses are not changed and no event is emitted.

				\ad{\Cref{proof:renaming:goal:sigmaact-enabled}}
				Let $\anadr'=\heapcomputof{\sigma}{\apvarp}$ and $\anadrp'=\heapcomputof{\sigma}{\psel{\anadr'}}$.
				For the claim to hold we have to show: $\anadrp'=\swapadr{\anadrp}$.
				First, note that we have: \[\anadr'=\heapcomputof{\sigma}{\apvarp}=\heapcomputof{\sigma}{\swapexp{\apvarp}}=\swapadr{\heapcomputof{\tau}{\apvarp}}=\swapadr{\anadr}\] by \Cref{proof:renaming:premis:heap-pexp}.
				Then, we conclude as follows:
				\begin{align*}
					\anadrp'
					=~&
					\heapcomputof{\sigma}{\psel{\anadr'}}
					=
					\heapcomputof{\sigma}{\psel{\swapadr{\anadr}}}
					=
					\heapcomputof{\sigma}{\swapexp{\psel{\anadr}}}
					\\=~&
					\swapadr{\heapcomputof{\tau}{\psel{\anadr}}}
					=
					\swapadr{\anadrp}
				\end{align*}

				\ad{\Cref{proof:renaming:goal:heap-pexp}}
				For $\apvar\in\pexp$ the claim follows by choice of $\anupp$.
				For $\apexp\in\pexp$ with $\apexp\not\equiv\apvar$ the claim follows from \Cref{proof:renaming:premis:heap-pexp} together with the fact that $\anup$ and $\anupp$ do not modify $\apexp$.

				\ad{\Cref{proof:renaming:goal:heap-dexp}}
				Neither $\anup$ nor $\anupp$ update data expressions.
				Hence, the claim follows from \Cref{proof:renaming:premis:heap-dexp}.

				\ad{\Cref{proof:renaming:goal:valid}}
				By $\swapexp{\apvarp}=\apvarp$ we have $\apvarp\in\validof{\tau}\iff\apvarp\in\validof{\sigma}$.
				Sicne \Cref{proof:renaming:goal:sigmaact-enabled} from above gives $\anadr'=\swapadr{\anadr}$, we have $\swapexp{\psel{\anadr}}=\psel{\anadr'}$.
				Hence, $\psel{\anadr}\in\validof{\tau}\iff\psel{\anadr'}\in\validof{\sigma}$ holds by \Cref{proof:renaming:premis:valid}.
				This means we have $\apvar\in\validof{\tau.\anact}\iff\apvar\in\validof{\sigma.\anactp}$ because $\heapcomputof{\tau}{\apvarp}\in\adr$.
				Together with \Cref{proof:renaming:premis:valid} this concludes the claim because only the validity of $\apvar$ is altered by $\anact$/$\anactp$.

				\item[$\acom\equiv\psel{\apvar}=\apvarp$ with $\heapcomputof{\tau}{\apvar}\in\adr$]
				The update is of the form $\anup=[\psel{\anadr}\mapsto\anadrp]$ with $\anadr=\heapcomputof{\tau}{\apvar}$ and $\anadrp=\heapcomputof{\tau}{\apvarp}$.
				For $\anupp$ we choose $\anupp=[\psel{\swapadr{\anadr}}\mapsto\swapadr{\anadrp}]$.

				\Cref{proof:renaming:goal:free,proof:renaming:goal:fresh,proof:renaming:goal:history} follow immediately from \Cref{proof:renaming:premis:free,proof:renaming:premis:fresh,proof:renaming:premis:history} because the fresh and freed addresses are not changed and no event is emitted.

				\ad{\Cref{proof:renaming:goal:sigmaact-enabled}}
				Let $\anadr'=\heapcomputof{\sigma}{\apvar}$ and $\anadrp'=\heapcomputof{\sigma}{\apvarp}$.
				For the claim to hold we have to show: $\anadr'=\swapadr{\anadr}$ and $\anadrp'=\swapadr{\anadrp}$.
				This follows immediately from \Cref{proof:renaming:premis:heap-pexp}.

				\ad{\Cref{proof:renaming:goal:heap-pexp}}
				For $\psel{\anadr}\in\pexp$ we have:
				\begin{align*}
					\heapcomputof{\sigma.\anactp}{\swapexp{\psel{\anadr}}}
					=~&
					\heapcomputof{\sigma.\anactp}{\psel{\swapadr{\anadr}}}
					=
					\swapadr{\anadrp}
					\\=~&
					\swapadr{\heapcomputof{\tau.\anact}{\psel{\anadr}}}
				\end{align*}
				where the first equality is by definition, the second equality is due to $\anupp$, and the last equality is due to $\anup$.
				For $\apexp\in\pexp$ with $\apexp\not\equiv\psel{\anadr}$ the claim follows from \Cref{proof:renaming:premis:heap-pexp} together with the fact that $\anup$ and $\anupp$ do not modify $\apexp$ and $\swapexp{\apexp}$, respectively.

				\ad{\Cref{proof:renaming:goal:heap-dexp}}
				Neither $\anup$ nor $\anupp$ update data expressions.
				Hence, the claim follows from \Cref{proof:renaming:premis:heap-dexp}.

				\ad{\Cref{proof:renaming:goal:valid}}
				We have $\apvarp\in\validof{\tau}\iff\apvarp\in\validof{\sigma}$ due to $\swapexp{\apvarp}=\apvarp$ and \Cref{proof:renaming:premis:valid}.
				Hence, $\psel{\anadr}\in\validof{\tau.\anact}\iff\psel{\swapadr{\anadr}}\in\validof{\sigma.\anactp}$ because $\heapcomputof{\tau}{\apvar}\in\adr$.
				Together with \Cref{proof:renaming:premis:valid} this concludes the property because only the validity of $\psel{\anadr}$/$\psel{\swapadr{\anadr}}$ is altered by $\anact$/$\anactp$.

				\item[$\acom\equiv\apvar:=\psel{\apvarp}$ or $\acom\equiv\psel{\apvarp}:=\apvar$ with $\heapcomputof{\tau}{\apvarp}\notin\adr$]
				The updates are $\anup=\emptyset=\anupp$ because the command segfaults.
				By definition of validity we get $\validof{\tau.\anact}=\validof{\tau}$.
				We have $\heapcomputof{\sigma}{\apvarp}=\swapadr{\heapcomputof{\tau}{\apvarp}}=\heapcomputof{\tau}{\apvarp}\notin\adr$ by \Cref{proof:renaming:premis:heap-pexp}.
				Hence, $\anact$ also segfaults in $\sigma$ and we conclude \Cref{proof:renaming:goal:sigmaact-enabled}.
				The remaining properties follow because the heap, the valid expressions, the freed addresses, the fresh addresses, and the history remain unchanged.

				\item[$\acom\equiv\advar=\opof{\advar_1,\ldots, \advar_n}$]
				The update is $\anup=[\advar\mapsto\advalue]$ with $\advalue=\opof{\heapcomputof{\tau}{\advar_1},\dots,\heapcomputof{\tau}{\advar_1}}$.
				Choose $\anupp=\anup$.

				Since the pointer expression valuations, the validity, the freed address, and the fresh address are not altered and no event is emitted by $\anact$/$\anactp$, \Cref{proof:renaming:goal:heap-pexp,proof:renaming:goal:valid,proof:renaming:goal:free,proof:renaming:goal:fresh,proof:renaming:goal:history} follow immediately from \Cref{proof:renaming:premis:heap-pexp,proof:renaming:premis:valid,proof:renaming:premis:free,proof:renaming:premis:fresh,proof:renaming:premis:history}

				\ad{\Cref{proof:renaming:goal:sigmaact-enabled}}
				By \Cref{proof:renaming:premis:heap-dexp} we have $\heapcomputof{\sigma}{\advar_i}=\heapcomputof{\tau}{\advar_i}$.
				Hence, $\anactp$ is enabled after $\sigma$.

				\ad{\Cref{proof:renaming:goal:heap-dexp}}
				We have $\heapcomputof{\sigma.\anactp}{\swapexp{\advar}}=\heapcomputof{\sigma.\anactp}{\advar}=\advalue=\heapcomputof{\tau.\anact}{\advar}$.
				And because no other data expressions are updated by $\anup$/$\anupp$, the claim follows from \Cref{proof:renaming:premis:heap-dexp}.

				\item[$\acom\equiv\advar=\dsel{\apvarp}$ with $\heapcomputof{\tau}{\apvarp}\in\adr$]
				The update is $\anup=[\advar\mapsto\advalue]$ with $\anadr=\heapcomputof{\tau}{\apvarp}$ and $\advalue=\heapcomputof{\tau}{\dsel{\anadr}}$.
				Choose $\anupp=\anup$.

				Since the pointer expression valuations, the validity, the freed address, and the fresh address are not altered and no event is emitted by $\anact$/$\anactp$, \Cref{proof:renaming:goal:heap-pexp,proof:renaming:goal:valid,proof:renaming:goal:free,proof:renaming:goal:fresh,proof:renaming:goal:history} follow immediately from \Cref{proof:renaming:premis:heap-pexp,proof:renaming:premis:valid,proof:renaming:premis:free,proof:renaming:premis:fresh,proof:renaming:premis:history}

				\ad{\Cref{proof:renaming:goal:sigmaact-enabled}}
				Let $\anadr'=\heapcomputof{\sigma}{\apvarp}$ and $\advalue'=\heapcomputof{\sigma}{\dsel{\anadr'}}$.
				For enabledness of $\anactp$, we have to show that $\advalue'=\advalue$.
				To see this, first note that \[\anadr'=\heapcomputof{\sigma}{\apvarp}=\heapcomputof{\sigma}{\swapexp{\apvarp}}=\swapadr{\heapcomputof{\tau}{\apvarp}}=\swapadr{\anadr}\] \Cref{proof:renaming:premis:heap-pexp}.
				Together \Cref{proof:renaming:premis:heap-dexp} we conclude as follows:
				\begin{align*}
					\advalue'
					=\heapcomputof{\sigma}{\dsel{\anadr'}}
					=\heapcomputof{\sigma}{\swapexp{\dsel{\anadr}}}
					=\heapcomputof{\tau}{\dsel{\anadr}}
					=\advalue
					\ .
				\end{align*}

				\ad{\Cref{proof:renaming:goal:heap-dexp}}
				We have $\heapcomputof{\sigma.\anactp}{\advar}=\advalue=\heapcomputof{\tau.\anact}{\advar}=\heapcomputof{\tau.\anact}{\swapexp{\advar}}$.
				And because no other data expressions are updated by $\anup$/$\anupp$, the claim follows from \Cref{proof:renaming:premis:heap-dexp}.

				\item[$\acom\equiv\dsel{\apvar}=\advarp$ with $\heapcomputof{\tau}{\apvar}\in\adr$]
				By the semantics the update is $\anup=[\dsel{\anadr}\mapsto\advalue]$ with $\anadr=\heapcomputof{\tau}{\apvar}$ and $\advalue=\heapcomputof{\tau}{\apvar}$.
				We choose $\anupp=[\dsel{\swapadr{\anadr}}\mapsto\advalue]$.

				Since the pointer expression valuations, the validity, the freed address, and the fresh address are not altered and no event is emitted by $\anact$/$\anactp$, \Cref{proof:renaming:goal:heap-pexp,proof:renaming:goal:valid,proof:renaming:goal:free,proof:renaming:goal:fresh,proof:renaming:goal:history} follow immediately from \Cref{proof:renaming:premis:heap-pexp,proof:renaming:premis:valid,proof:renaming:premis:free,proof:renaming:premis:fresh,proof:renaming:premis:history}

				\ad{\Cref{proof:renaming:goal:sigmaact-enabled}}
				For enabledness of $\anactp$ we have to show that $\heapcomputof{\sigma}{\apvar}=\swapadr{\anadr}$ and $\heapcomputof{\sigma}{\advar}=\advalue$ hold.
				This follows immediately from \Cref{proof:renaming:premis:heap-pexp,proof:renaming:premis:heap-dexp}.

				\ad{\Cref{proof:renaming:goal:heap-dexp}}
				We have: \[\heapcomputof{\sigma.\anactp}{\swapexp{\dsel{\anadr}}}=\heapcomputof{\sigma.\anactp}{\dsel{\swapadr{\anadr}}}=\advalue=\heapcomputof{\tau.\anact}{\dsel{\anadr}}\ .\]
				And because no other data expressions are updated by $\anup$/$\anupp$, the claim follows from \Cref{proof:renaming:premis:heap-dexp}.

				\item[$\acom\equiv\dsel{\apvar}=\advarp$ or $\acom\equiv\advar=\dsel{\apvar}$ with $\heapcomputof{\tau}{\apvar}\notin\adr$]
				The updates are $\anup=\emptyset=\anupp$.
				We have $\heapcomputof{\sigma}{\apvar}=\swapadr{\heapcomputof{\tau}{\apvar}}=\heapcomputof{\tau}{\apvar}\notin\adr$ by \Cref{proof:renaming:premis:heap-pexp}.
				So we conclude \Cref{proof:renaming:goal:sigmaact-enabled}.
				Then, the claim follows because the heap, the valid expressions, the freed addresses, the fresh addresses, and the history remain unchanged.

				\item[$\acom\equiv\apvar:=\malloc$]
				The update is
				\begin{align*}
					\anup&=[\apvar\mapsto\anadr,\psel{\anadr}\mapsto\segval,\dsel{\anadr}\mapsto\advalue]
					\intertext{
				with $\anadr\in\freshof{\tau}\cup(\freedof{\tau}\cap A)$ and some $\advalue\in\dom$.
				Choose
					}
					\anupp&=[\apvar\mapsto\swapadr{\anadr},\psel{\swapadr{\anadr}}\mapsto\segval,\dsel{\swapadr{\anadr}}\mapsto\advalue]\ .
				\end{align*}
				Then, \Cref{proof:renaming:goal:history} follows immediately from \Cref{proof:renaming:premis:history} because no event is emitted.

				\ad{\Cref{proof:renaming:goal:sigmaact-enabled}}
				We have to show that $\swapadr{\anadr}\in\freshof{\sigma}\cup(\freedof{\sigma}\cap\swapadr{A})$.
				By \Cref{proof:renaming:premis:fresh,proof:renaming:premis:free} and \Cref{thm:swap-vs-set-opereations} we have:
				\begin{align*}
					&
					\freshof{\sigma}\cup(\freedof{\sigma}\cap\swapadr{A})
					\\=~&
				 	\swapadr{\freshof{\tau}}\cup(\swapadr{\freedof{\tau}}\cap\swapadr{A})
				 	\\=~&
				 	\swapadr{\freshof{\tau}\cup(\freedof{\tau}\cap A)}
				\end{align*}
				So since $\anadr\in\freshof{\tau}\cup(\freedof{\tau}\cap A)$ holds by enabledness of $\anact$, the above yields the desired $\swapadr{\anadr}\in\freshof{\sigma}\cup(\freedof{\sigma}\cap\swapadr{A})$.
				Hence, $\anactp$ is enabled after $\sigma$.

				\ad{\Cref{proof:renaming:goal:heap-pexp}}
				We have
				\begin{align*}
					\heapcomputof{\sigma.\anactp}{\swapexp{\apvar}}
					=\heapcomputof{\sigma.\anactp}{\apvar}
					=\swapadr{\anadr}
					=\swapadr{\heapcomputof{\tau.\anact}{\apvar}}
				\end{align*}
				and
				\begin{align*}
					\heapcomputof{\sigma.\anactp}{\swapexp{\psel{\anadr}}}
					=~&
					\heapcomputof{\sigma.\anactp}{\psel{\swapadr{\anadr}}}
					=
					\segval
					\\=~&
					\swapadr{\segval}
					=
					\swapadr{\heapcomputof{\tau.\anact}{\psel{\anadr}}}
				\end{align*}
				Since no other pointer expressions are updated, the claim follows from \Cref{proof:renaming:premis:heap-pexp}.

				\ad{\Cref{proof:renaming:goal:heap-dexp}}
				We have:
				\begin{align*}
					\heapcomputof{\sigma.\anactp}{\swapexp{\dsel{\anadr}}}
					=
					\heapcomputof{\sigma.\anactp}{\dsel{\swapadr{\anadr}}}
					=
					\advalue
					=
					\heapcomputof{\tau.\anact}{\dsel{\anadr}}
				\end{align*}
				Since no other data expression is updated, the follows from \Cref{proof:renaming:premis:heap-dexp}.

				\ad{\Cref{proof:renaming:goal:valid}}
				We conclude using \Cref{proof:renaming:premis:valid} and \Cref{thm:swap-vs-set-opereations} as follows:
				\begin{align*}
					\validof{\sigma.\anactp}
					=~&
					\validof{\sigma}\cup\set{\apvar,\psel{\swapadr{\anadr}}}
					=
					\swapexp{\validof{\tau}}\cup\swapexp{\set{\apvar,\psel{\anadr}}}
					\\=~&
					\swapexp{\validof{\tau}\cup\set{\apvar,\psel{\anadr}}}
					=
					\swapexp{\validof{\tau.\anact}}
				\end{align*}

				\ad{\Cref{proof:renaming:goal:free}}
				We conclude using \Cref{proof:renaming:premis:free} and \Cref{thm:swap-vs-set-opereations} as follows:
				\begin{align*}
					\freedof{\sigma.\anactp}
					=~&
					\freedof{\sigma}\setminus{\swapadr{\anadr}}
					=
					\swapadr{\freedof{\tau}}\setminus{\swapadr{\anadr}}
					\\=~&
					\swapadr{\freedof{\tau}\setminus\anadr}
					=
					\swapadr{\freedof{\tau.\anact}}
				\end{align*}

				\ad{\Cref{proof:renaming:goal:fresh}}
				We conclude using \Cref{proof:renaming:premis:fresh} and \Cref{thm:swap-vs-set-opereations} as follows:
				\begin{align*}
					\freshof{\sigma.\anactp}
					=~&
					\freshof{\sigma}\setminus{\swapadr{\anadr}}
					=
					\swapadr{\freshof{\tau}}\setminus{\swapadr{\anadr}}
					\\=~&
					\swapadr{\freshof{\tau}\setminus\anadr}
					=
					\swapadr{\freshof{\tau.\anact}}
				\end{align*}

				\item[$\acom\equiv\freeof{\anadr}$]
				The update is $\anup=[\psel{\anadr}\mapsto\bot,\dsel{\anadr}\mapsto\bot]$.
				Choose $\acom'=\freeof{\swapadr{\anadr}}$ and $\anupp=[\psel{\swapadr{\anadr}}\mapsto\bot,\dsel{\swapadr{\anadr}}\mapsto\bot]$.

				\ad{\Cref{proof:renaming:goal:sigmaact-enabled}}
				By the semantics we have $\freeof{\anadr}\in\freeableof{\tau}{\anadr}$.
				By \Cref{proof:renaming:premis:history} together with \Cref{thm:swaphist-freeable} we get $\freeof{\swapadr{\anadr}}\in\freeableof{\sigma}{\swapadr{\anadr}}$.
				Hence, $\anactp$ is enabled after $\sigma$.

				\ad{\Cref{proof:renaming:goal:heap-pexp}}
				We have \[\heapcomputof{\sigma.\anactp}{\swapexp{\apvar}}=\heapcomputof{\sigma.\anactp}{\apvar}=\swapadr{\anadr}=\swapadr{\heapcomputof{\tau.\anact}{\apvar}}\] and \[\heapcomputof{\sigma.\anactp}{\swapexp{\psel{\anadr}}}=\heapcomputof{\sigma.\anactp}{\psel{\swapadr{\anadr}}}=\bot=\heapcomputof{\tau.\anact}{\psel{\anadr}}\ .\]
				Since no other pointer expressions are affected by $\anup$/$\anupp$ the claim follows from \Cref{proof:renaming:goal:heap-pexp}.

				\ad{\Cref{proof:renaming:goal:heap-dexp}}
				We have \[\heapcomputof{\sigma.\anactp}{\swapexp{\dsel{\anadr}}}=\heapcomputof{\sigma.\anactp}{\dsel{\swapadr{\anadr}}}=\bot=\heapcomputof{\tau.\anact}{\dsel{\anadr}}\ .\]
				Since no other data expressions are modified by $\anup$/$\anupp$ the claim follows from \Cref{proof:renaming:goal:heap-dexp}.

				\ad{\Cref{proof:renaming:goal:valid}}
				Consider some $\apexp'\in\pexp$.
				Since $\swapadrraw$ is a bijection, there is some $\apexp\in\pexp$ such that $\swapexp{\apexp}=\apexp'$.
				First, we get:
				\begin{align*}
					\apexp'\in\validof{\sigma}
					\iff~&
					\swapexp{\apexp}\in\validof{\sigma}
					\\\iff~&
					\swapexp{\apexp}\in\swapexp{\validof{\tau}}
					\iff
					\apexp\in\validof{\tau}
				\end{align*}
				where the second equivalence is due to \Cref{proof:renaming:premis:valid} and the third equivalence holds since $\swapadrraw$ is a bijection.
				Second, we have:
				\begin{align*}
					\heapcomputof{\sigma}{\apexp'}\neq\swapadr{\anadr}
					\iff~&
					\heapcomputof{\sigma}{\swapexp{\apexp}}\neq\swapadr{\anadr}
					\\\iff~&
					\swapadr{\heapcomputof{\tau}{\apexp}}\neq\swapadr{\anadr}
					\iff
					\heapcomputof{\tau}{\apexp}\neq\anadr
				\end{align*}
				where the second equivalence is due to \Cref{proof:renaming:premis:heap-pexp} and the third equivalence holds because $\swapadrraw$ is a bijection.
				Last, we have:
				\begin{align*}
					\apexp'\cap\adr\neq\set{\swapadr{\anadr}}
					\iff~&
					\swapexp{\apexp}\cap\adr\neq\set{\swapadr{\anadr}}
					\\\iff~&
					\apexp\cap\adr\neq\set{\anadr}
				\end{align*}
				Altogether, this gives:
				\begin{align*}
					&
					\apexp'\in\validof{\sigma.\anactp}
					\\\iff~&
					\apexp'\in\validof{\sigma}
					\wedge
					\heapcomputof{\sigma}{\apexp'}\neq\swapadr{\anadr}
					\wedge
					\apexp'\cap\adr\neq\set{\swapadr{\anadr}}
					\\\iff~&
					\apexp\in\validof{\tau}
					\wedge
					\heapcomputof{\tau}{\apexp}\neq\anadr
					\wedge
					\apexp\cap\adr\neq\set{\anadr}
					\\\iff~&
					\apexp\in\validof{\tau.\anact}
					\\\iff~&
					\swapexp{\apexp}\in\swapexp{\validof{\tau.\anact}}
					\\\iff~&
					\apexp'\in\swapexp{\validof{\tau.\anact}}
				\end{align*}
				where the last but first equivalence holds because $\swapadrraw$ is a bijection.

				\ad{\Cref{proof:renaming:goal:free}}
				We conclude using \Cref{proof:renaming:premis:free} and \Cref{thm:swap-vs-set-opereations}:
				\begin{align*}
					\freedof{\sigma.\anactp}
					=~&
					\freedof{\sigma}\cup\set{\swapadr{\anadr}}
					=
					\swapadr{\freedof{\tau}}\cup\set{\swapadr{\anadr}}
					\\=~&
					\swapadr{\freedof{\tau}\cup\set{\anadr}}
					=
					\swapadr{\freedof{\tau.\anact}}
				\end{align*}

				\ad{\Cref{proof:renaming:goal:fresh}}
				We conclude using \Cref{proof:renaming:premis:fresh} and \Cref{thm:swap-vs-set-opereations}:
				\begin{align*}
					\freshof{\sigma.\anactp}
					=~&
					\freshof{\sigma}\setminus\set{\swapadr{\anadr}}
					=
					\swapadr{\freshof{\tau}}\setminus\set{\swapadr{\anadr}}
					\\=~&
					\swapadr{\freshof{\tau}\setminus\set{\anadr}}
					=
					\swapadr{\freshof{\tau.\anact}}
				\end{align*}

				\ad{\Cref{proof:renaming:goal:history}}
				We conclude using \Cref{proof:renaming:premis:history}
				\begin{align*}
					\historyof{\sigma.\anactp}
					=~&
					\historyof{\sigma}.\freeof{\swapadr{\anadr}}
					=
					\swaphist{\historyof{\tau}}.\swaphist{\freeof{\anadr}}
					\\=~&
					\swaphist{\historyof{\tau}.\freeof{\anadr}}
					=
					\swaphist{\historyof{\tau.\anact}}
				\end{align*}

				\item[$\acom\equiv\assert\ \mathit{cond}$]
				The update is $\anup=\emptyset$.
				Choose $\anupp=\emptyset$.

				Since the memory, the freed address, and the fresh address are not altered and no event is emitted by $\anact$/$\anactp$, \Cref{proof:renaming:goal:heap-pexp,proof:renaming:goal:heap-dexp,proof:renaming:goal:free,proof:renaming:goal:fresh,proof:renaming:goal:history} follow immediately from \Cref{proof:renaming:premis:heap-pexp,proof:renaming:premis:heap-dexp,proof:renaming:premis:free,proof:renaming:premis:fresh,proof:renaming:premis:history}.

				\ad{\Cref{proof:renaming:goal:sigmaact-enabled}}
				First, consider the case where $\mathit{cond}$ is a condition over data variables $\advar,\advarp$.
				By \Cref{proof:renaming:premis:heap-dexp} we have $\heapcomputof{\sigma}{\advar}=\heapcomputof{\tau}{\advar}$ and $\heapcomputof{\sigma}{\advarp}=\heapcomputof{\tau}{\advarp}$.
				Hence, $\anactp$ is enabled after $\sigma$.
				Now, consider the case where $\mathit{cond}$ is a condition over pointer variables $\apvar,\apvarp$.
				We have $\heapcomputof{\sigma}{\apvar}=\swapadr{\heapcomputof{\tau}{\apvar}}$ and $\heapcomputof{\sigma}{\apvarp}=\swapadr{\heapcomputof{\tau}{\apvarp}}$.
				So we get $\heapcomputof{\sigma}{\apvar}=\heapcomputof{\sigma}{\apvarp}$ iff $\heapcomputof{\tau}{\apvar}=\heapcomputof{\tau}{\apvarp}$.
				Hence, $\anactp$ is enabled after $\sigma$.
				This concludes the claim.

				\ad{\Cref{proof:renaming:goal:valid}}
				Note that we have:
				\begin{align*}
					&
					\validof{\sigma.\anactp}\neq\validof{\sigma}
					\\\iff~&
					\mathit{cond}\equiv\apvar=\apvarp \wedge \set{\apvar,\apvarp}\cap\validof{\sigma}\neq\emptyset \wedge \set{\apvar,\apvarp}\not\subseteq\validof{\sigma}
					\\\iff~&
					\mathit{cond}\equiv\apvar=\apvarp \wedge \set{\apvar,\apvarp}\cap\swapexp{\validof{\tau}}\neq\emptyset \wedge \set{\apvar,\apvarp}\not\subseteq\swapexp{\validof{\tau}}
					\\\iff~&
					\mathit{cond}\equiv\apvar=\apvarp \wedge \set{\apvar,\apvarp}\cap\validof{\tau}\neq\emptyset \wedge \set{\apvar,\apvarp}\not\subseteq\validof{\tau}
					\\\iff~&
					\validof{\tau.\anact}\neq\validof{\tau}
				\end{align*}
				where the second equivalence is by \Cref{proof:renaming:premis:valid}.
				Consider the case $\validof{\sigma.\anactp}\neq\validof{\sigma}$.
				We conclude by \Cref{proof:renaming:premis:valid} and \Cref{thm:swap-vs-set-opereations}:
				\begin{align*}
					\validof{\sigma.\anactp}
					=~&\validof{\sigma}\cup\set{\apvar,\apvarp}
					=\swapexp{\validof{\tau}}\cup\swapexp{\set{\apvar,\apvarp}}
					\\=~&\swapexp{\validof{\tau}\cup\set{\apvar,\apvarp}}
					=\swapexp{\validof{\tau.\anact}}
				\end{align*}
				In all other cases, we have $\validof{\tau.\anact}=\validof{\tau}$ and $\validof{\sigma.\anactp}=\validof{\sigma}$.
				Then, the claim follows immediately from \Cref{proof:renaming:premis:valid}.

				\item[$\acom\equiv\enterof{\afunc(\vecof{\apvar},\vecof{\advar})}$]
				The update is $\anup=\emptyset$ and we choose $\anupp=\emptyset$.

				Since the memory, the validity, the freed address, and the fresh address are not altered, \Cref{proof:renaming:goal:heap-pexp,proof:renaming:goal:heap-dexp,proof:renaming:goal:valid,proof:renaming:goal:free,proof:renaming:goal:fresh} follow immediately from \Cref{proof:renaming:premis:heap-pexp,proof:renaming:premis:heap-dexp,proof:renaming:premis:valid,proof:renaming:premis:free,proof:renaming:premis:fresh}.

				\ad{\Cref{proof:renaming:goal:sigmaact-enabled}}
				Consider some $i$.
				By \Cref{proof:renaming:premis:sigmaact-enabled} we have $\heapcomputof{\tau}{\apvar_i}\in\adr$.
				Hence, we have $\swapadr{\heapcomputof{\tau}{\apvar_i}}\in\adr$ by definition.
				So \Cref{proof:renaming:premis:heap-pexp} yields $\heapcomputof{\sigma}{\apvar_i}\in\adr$.
				Hence, $\anact$ is enabled.

				\ad{\Cref{proof:renaming:goal:history}}
				If $\heapcomputof{\tau}{\apvar_i}\notin\adr$, then \Cref{thm:rprf-vs-bot-pvar,thm:seg-valid} give $\heapcomputof{\tau}{\apvar_i}=\segval$ and $\apvar_i\in\validof{\tau}$.
				So we have $\heapcomputof{\sigma}{\apvar_i}=\swapadr{\heapcomputof{\tau}{\apvar_i}}=\heapcomputof{\tau}{\apvar_i}\notin\adr$.
				Hence, neither $\anact$/$\anactp$ emit an event and we conclude $\historyof{\sigma.\anactp}=\historyof{\sigma}=\swaphist{\historyof{\tau}}=\swaphist{\historyof{\tau.\anact}}$ by \Cref{proof:renaming:premis:history}.
				Otherwise, $\anact$ emits the event $\evt{\afunc}{\athread,\vecof{\anadr},\vecof{\advalue}}$ with $\vecof{\anadr}=\heapcomputof{\tau}{\vecof{\apvar}}$ and $\vecof{\advalue}=\heapcomputof{\tau}{\vecof{\advar}}$.
				By \Cref{proof:renaming:premis:heap-pexp} we have $\heapcomputof{\sigma}{\vecof{\apvar}}=\swapadr{\vecof{\anadr}}$.
				Hence, $\anactp$ emits the event $\evt{\afunc}{\athread,\swapadr{\anadr}}$.
				So we conclude using \Cref{proof:renaming:premis:history}:
				\begin{align*}
					\historyof{\sigma.\anactp}
					=~&\historyof{\sigma}.\evt{\afunc}{\athread,\swapadr{\vecof{\anadr}},\vecof{\advalue}}
					=\swaphist{\historyof{\tau}}.\swaphist{\evt{\afunc}{\athread,\vecof{\anadr},\vecof{\advalue}}}
					\\=~&\swaphist{\historyof{\tau}.\evt{\afunc}{\athread,\vecof{\anadr},\vecof{\advalue}}}
					=\swaphist{\historyof{\tau.\anact}}
				\end{align*}

				\item[$\acom\equiv\exit$]
				Follows analogously to the previous case.

			\end{casedistinction}
	\end{description}
	The case distinction is complete and thus concludes the claim.
\end{proof}

%
\begin{proof}[Proof of \Cref{thm:supported-elision-of-invalid-address-more}]
	Let $\smrobs$ support elision (\Cref{def:elision-support}).
	Let $\anadrp\in\freshof{\tau}\setminus A$.
	Let $\swapadrraw:\adr\to\adr$ be the address mapping defined by $\swapadr{\anadr}=\anadrp$, $\swapadr{\anadrp}=\anadr$, and $\swapadr{\anadrpp}=\anadrpp$ for $\anadr\neq\anadrpp\neq\anadrp$.
	\Cref{thm:elision} yields $\sigma\in\asem{\swapadr{A}}$ with
	\begin{compactitem}
		\item $\heapcomput{\sigma}\circ\swapexpraw=\swapadrraw\circ\heapcomput{\tau}$,
		\item $\historyof{\sigma}=\swaphist{\historyof{\tau}}$,
		\item $\validof{\sigma}=\swapexp{\validof{\tau}}$,
		\item $\freshof{\sigma}=\swapadr{\freshof{\tau}}$,
		\item $\freedof{\sigma}=\swapadr{\freedof{\tau}}$, and
		\item $\controlof{\sigma}=\controlof{\tau}$.
	\end{compactitem}
	We show that $\sigma$ satisfies the claim.
	First, note that $\swapadr{A}=A$ by $\anadr,\anadrp\notin A$.
	So we have $\sigma\in\asem{A}$.

	We first show the following auxiliaries:
	\begin{auxiliary}
		\psel{\anadr},\psel{\anadrp}&\notin\validof{\tau}
		\label[aux]{proof:supported-elision-of-invalid-address:anext-bnext}
		\\
		f=f^{-1} ~~&\text{for }f\in\set{\swapadrraw,\swapexpraw,\swaphistraw}
		\label[aux]{proof:supported-elision-of-invalid-address:inverse}
		\\
		\forall\apexp\in\validof{\tau}.~&\apexp\equiv\swapexp{\apexp}
		\label[aux]{proof:supported-elision-of-invalid-address:valid-pexp}
		\\
		\forall\apexp\in\validof{\tau}.~&\heapcomputof{\tau}{\apexp}=\swapadr{\heapcomputof{\tau}{\apexp}}
		\label[aux]{proof:supported-elision-of-invalid-address:valid-pointers}
		\\
		\forall\anadrpp\in\heapcomputof{\tau}{\validof{\tau}}.~&\dsel{\anadrpp}\equiv\swapexp{\dsel{\anadrpp}}
		\label[aux]{proof:supported-elision-of-invalid-address:valid-dexp}
		\\
		\forall\anadrpp\in\adr\setminus\set{\anadr,\anadrp}&.~\freeableof{\tau}{\anadrpp}\subseteq\freeableof{\sigma}{\anadrpp}
		\label[aux]{proof:supported-elision-of-invalid-address:freeable}
	\end{auxiliary}

	\ad{\Cref{proof:supported-elision-of-invalid-address:anext-bnext}}
	Holds by $\anadr\notin\vadrof{\tau}$, and by $\anadrp\in\freshof{\tau}$ together with \Cref{thm:free-not-valid}.

	\ad{\Cref{proof:supported-elision-of-invalid-address:inverse}} 
	Holds by choice of $\swapadrraw$.

	\ad{\Cref{proof:supported-elision-of-invalid-address:valid-pexp}}
	Holds by \Cref{proof:supported-elision-of-invalid-address:anext-bnext} and the definition of $\swapadrraw$ and its induced $\swapexpraw$.

	\ad{\Cref{proof:supported-elision-of-invalid-address:valid-pointers}}
	Holds due to $\anadr\notin\vadrof{\tau}$ giving $\anadr\notin\heapcomputof{\tau}{\validof{\tau}}$ by \Cref{thm:adrof-valid-heap-restriction}, and $\anadrp\in\freshof{\tau}$ giving $\anadrp\notin\rangeof{\heapcomput{\tau}}$ by \Cref{thm:fresh-not-referenced} and thus $\anadrp\notin\heapcomputof{\tau}{\validof{\tau}}$.

	\ad{\Cref{proof:supported-elision-of-invalid-address:valid-dexp}}
	From $\anadrpp\in\heapcomputof{\tau}{\validof{\tau}}$ we get some $\apexp\in\validof{\tau}$ with $\heapcomputof{\tau}{\apexp}=\anadrpp$.
	Thus, $\anadrpp=\heapcomputof{\tau}{\apexp}=\swapadr{\heapcomputof{\tau}{\apexp}}=\swapadr{\anadrpp}$ by \Cref{proof:supported-elision-of-invalid-address:valid-pointers}.

	\ad{\Cref{proof:supported-elision-of-invalid-address:freeable}}
	Let $\anadrpp\in\adr\setminus\set{\anadr,\anadrpp}$.
	Then, \[\freeableof{\tau}{\anadrpp}=\freeableof{\historyof{\tau}}{\anadrpp}=\freeableof{\swaphist{\historyof{\tau}}}{\anadrpp}=\freeableof{\historyof{\sigma}}{\anadrpp}=\freeableof{\sigma}{\anadrpp}\] where the second equality holds because $\smrobs$ supports elision (\Cref{def:elision-support}\ref{def:elision-support:replace}) by assumption and the third equality holds by the properties given by \Cref{thm:elision} listed above.
	So we have the desired $\freeableof{\tau}{\anadrpp}\subseteq\freeableof{\sigma}{\anadrpp}$.

	\ad{$\tau\computequiv\sigma$}
	We already have $\controlof{\sigma}=\controlof{\tau}$.
	We get:
	\begin{align*}
		\domof{\restrict{\heapcomput{\sigma}}{\validof{\sigma}}}
		=~&
		\validof{\sigma}\cup\dvars\cup\setcond{\dsel{\anadrpp}}{\anadrpp\in\heapcomputof{\sigma}{\validof{\sigma}}}
		\\=~&
		\swapexp{\validof{\tau}}\cup\dvars\cup\setcond{\dsel{\anadrpp}}{\anadrpp\in\swapadr{\heapcomputof{\tau}{\swapexpinv{\validof{\tau}}}}}
		\\=~&
		\swapexp{\validof{\tau}}\cup\dvars\cup\setcond{\dsel{\anadrpp}}{\anadrpp\in\swapadr{\heapcomputof{\tau}{\swapexp{\validof{\tau}}}}}
		\\=~&
		\validof{\tau}\cup\dvars\cup\setcond{\dsel{\anadrpp}}{\anadrpp\in\swapadr{\heapcomputof{\tau}{\validof{\tau}}}}
		\\=~&
		\validof{\tau}\cup\dvars\cup\setcond{\dsel{\anadrpp}}{\anadrpp\in\heapcomputof{\tau}{\validof{\tau}}}
		=
		\domof{\restrict{\heapcomput{\tau}}{\validof{\tau}}}
	\end{align*}
	where the first equality is by definition, the second by the properties of $\sigma$, the third by \Cref{proof:supported-elision-of-invalid-address:inverse}, the fourth by \Cref{proof:supported-elision-of-invalid-address:valid-pexp}, the fifth by \Cref{proof:supported-elision-of-invalid-address:valid-pointers}, and the last again by definition.
	Then, we get for $\apexp\in\domof{\restrict{\heapcomput{\sigma}}{\validof{\sigma}}}\cap\pexp$:
	\begin{align*}
		\heapcomputof{\sigma}{\apexp}
		=
		\swapadr{\heapcomputof{\tau}{\swapexpinv{\apexp}}}
		=
		\swapadr{\heapcomputof{\tau}{\apexp}}
		=
		\heapcomputof{\tau}{\apexp}
	\end{align*}
	using the properties of $\sigma$, the fact that $\apexp\in\validof{\tau}$ must holds, and \Cref{proof:supported-elision-of-invalid-address:inverse,proof:supported-elision-of-invalid-address:valid-pexp,proof:supported-elision-of-invalid-address:valid-pointers}.
	Moreover, we get for $\adexp\in\domof{\restrict{\heapcomput{\sigma}}{\validof{\sigma}}}\cap\dexp$:
	\begin{align*}
		\heapcomputof{\sigma}{\adexp}
		=
		\heapcomputof{\tau}{\swapexpinv{\adexp}}
		=
		\heapcomputof{\tau}{\adexp}
		\ .
	\end{align*}
	by \Cref{proof:supported-elision-of-invalid-address:valid-dexp} together with the fact that $\dsel{\anadrpp}\in\domof{\restrict{\heapcomput{\sigma}}{\validof{\sigma}}}=\domof{\restrict{\heapcomput{\tau}}{\validof{\tau}}}$ implies $\anadrpp\in\heapcomputof{\tau}{\validof{\tau}}$.
	Altogether, this yields $\restrict{\heapcomput{\tau}}{\validof{\tau}}=\restrict{\heapcomput{\sigma}}{\validof{\sigma}}$.
	So we have $\tau\computequiv\sigma$.

	\ad{$\tau\heapequiv[A]\sigma$}
	Let $\anadrpp\in A$.
	We show $\tau\heapequiv[\anadrpp]\sigma$.
	We have $\anadr\neq\anadrpp\neq\anadrp$ and thus $\swapadr{\anadrpp}=\anadrpp$.
	So using \Cref{proof:supported-elision-of-invalid-address:inverse} we get:
	\begin{align*}
		\anadrpp\in\allocatedof{\sigma}
		\iff~&
		\anadrpp\notin\freshof{\sigma}\cup\freedof{\sigma}
		\iff
		\anadrpp\notin\swapadr{\freshof{\tau}}\cup\swapadr{\freedof{\tau}}
		\\\iff~&
		\anadrpp\notin\swapadr{\freshof{\tau}\cup\freedof{\tau}}
		\iff
		\anadrpp\in\swapadr{\allocatedof{\tau}}
		\\\iff~&
		\swapadr{\anadrpp}\in\swapadr{\swapadr{\allocatedof{\tau}}}
		\iff
		\anadrpp\in\allocatedof{\tau}
	\end{align*}
	From \Cref{proof:supported-elision-of-invalid-address:freeable} we get $\freeableof{\tau}{\anadrpp}\subseteq\freeableof{\sigma}{\anadrpp}$.
	For $\apvar\in\pvars$ we have:
	\begin{align*}
		\anadrpp=\heapcomputof{\sigma}{\apvar}
		\iff~&
		\anadrpp=\swapadr{\heapcomputof{\tau}{\swapexp{\apvar}}}
		\iff
		\swapadr{\anadrpp}=\swapadr{\swapadr{\heapcomputof{\tau}{\apvar}}}
		\\\iff~&
		\anadrpp=\heapcomputof{\tau}{\apvar}
	\end{align*}
	Now, consider $\anadrpp'\in\heapcomputof{\tau}{\validof{\tau}}$.
	We have:
	\begin{align*}
		\anadrpp=\heapcomputof{\sigma}{\psel{\anadrpp'}}
		\iff~&
		\anadrpp=\swapadr{\heapcomputof{\tau}{\swapexp{\psel{\anadrpp'}}}}
		\\\iff~&
		\swapadr{\anadrpp}=\swapadr{\swapadr{\heapcomputof{\tau}{\swapexp{\psel{\anadrpp'}}}}}
		\\\iff~&
		\anadrpp=\heapcomputof{\tau}{\swapexp{\psel{\anadrpp'}}}
	\end{align*}
	In order to conclude, it remains to show that $\swapexp{\psel{\anadrpp'}}=\psel{\anadrpp}$.
	To that end, it suffices to show that $\anadrpp'\neq\anadr$ and $\anadrpp'\neq\anadrp$.
	The former follows from $\anadr\notin\vadrof{\tau}$ together with \Cref{thm:adrof-valid-heap-restriction}.
	The latter follows form $\anadrp\in\freshof{\tau}$ together with \Cref{thm:fresh-not-valid}.

	\ad{$\tau\obsrel\sigma$}
	Follows from \Cref{proof:supported-elision-of-invalid-address:freeable} with $\anadr,\anadrp\notin\vadrof{\tau}$ by \Cref{thm:adrof-valid-heap-restriction,thm:fresh-not-valid}.

	\ad{$\anadr\in\freshof{\sigma}$}
	We have $\anadrp\in\freshof{\tau}$.
	So $\anadr=\swapadr{\anadrp}\in\swapadr{\freshof{\tau}}=\freshof{\sigma}$.

	\ad{Implication}
	Consider some $\anexp,\anexpp\in\pvars\cup\setcond{\psel{\anadr}}{\anadr\in\heapcomputof{\tau}{\validof{\tau}}}$ with $\heapcomputof{\tau}{\anexp}\neq\heapcomputof{\tau}{\anexpp}$.
	If $\anexp\in\pvars$, then $\swapexp{\anexp}=\anexp$.
	Otherwise, $\anexp\equiv\psel{\anadrpp}$ with $\anadrpp\in\heapcomputof{\tau}{\validof{\tau}}$.
	By assumption together with \Cref{thm:adrof-valid-heap-restriction}, $\anadrpp\neq\anadr$.
	And by \Cref{thm:fresh-not-valid}, $\anadrpp\neq\anadrp$.
	So $\swapexp{\anexp}=\anexp$ must hold.
	Similarly, we get $\swapexp{\anexpp}=\anexpp$.
	Then, we have:
	\begin{align*}
		&\heapcomputof{\sigma}{\anexp}=\swapadr{\heapcomputof{\tau}{\swapexp{\anexp}}}=\swapadr{\heapcomputof{\tau}{\anexp}}
		\\\text{and}\quad
		&\heapcomputof{\sigma}{\anexpp}=\swapadr{\heapcomputof{\tau}{\swapexp{\anexpp}}}=\swapadr{\heapcomputof{\tau}{\anexpp}}.
	\end{align*}
	Since $\swapadrraw$ is a bijection, $\heapcomputof{\tau}{\anexp}\neq\heapcomputof{\tau}{\anexpp}$ implies the desired $\heapcomputof{\sigma}{\anexp}\neq\heapcomputof{\sigma}{\anexpp}$.
\end{proof}

\begin{proof}[Proof of \Cref{thm:supported-elision-of-invalid-address}]
	Follows from \Cref{thm:supported-elision-of-invalid-address-more}.
\end{proof}

%% file: content/appendix/theory/proofs_results.tex

\subsection{Proofs of the Generalized Reduction (\CREF{appendix:theory-generalization})}

\begin{proof}[Proof of \Cref{thm:RPRF-implies-GCplus}]
	Let $\gcplussem[k]{\ads}$ be reduction compatible according to \Cref{def:reduction-compatible} and poitner race free (PRF) according to \Cref{definition:RPR}.
	We proceed by induction on the structure of computations.

	\begin{description}[labelwidth=6mm,leftmargin=8mm,itemindent=0mm]
		\item[IB:]
			Consider $\tau=\epsilon$.
			Then, $\sigma=\epsilon$ satisfies the claim.

		\item[IH:]
			For every $\tau\in\allsem$ and every $\anadr\in\adr$ there exists some $\sigma\in\adrsem{\anadr}$ with $\tau\computequiv\sigma$, $\tau\heapequiv[\anadr]\sigma$, and $\tau\obsrel\sigma$.

		\item[IS:]
			Consider now $\tau.\anact\in\allsem$.
			Let $\anact$ be of the form $\anact=(\athread,\acom,\anup)$.
			Let $\anadr\in\adr$ be some addresses.
			Invoke the induction hypothesis for $\tau$ and $\anadr$ to get some $\sigma\in\adrsem{\anadr}$ with $\tau\computequiv\sigma$, $\tau\heapequiv[\anadr]\sigma$, and $\tau\obsrel\sigma$.
			We show that there is some $\sigmagoal\in\adrsem{\anadr}$ with $\tau.\anact\computequiv\sigmagoal$, $\tau.\anact\heapequiv[\anadr]\sigmagoal$, and $\tau.\anact\obsrel\sigmagoal$.
			To do so, we do a case distinction on $\acom$.

			\begin{casedistinction}
				\item[$\acom$ is an assignment]
				We want to use \Cref{thm:mimic-simple-NEW} to find some $\anactp$ which mimics $\anact$.
				We have to show that \Cref{thm:mimic-simple-NEW} is enabled.
				To that end, consider $\acom$ contains $\psel{\apvar}$ or $\dsel{\apvar}$.
				By the semantics, we know $\heapcomputof{\tau}{\apvar}\in\adr$.
				By \Cref{thm:rprf-vs-bot-pvar} we have $\heapcomputof{\sigma}{\apvar}\neq\bot$.
				Towards a contradiction, assume $\heapcomputof{\sigma}{\apvar}=\segval$.
				By \Cref{thm:seg-valid} this means $\apvar\in\validof{\sigma}$.
				So we have $\heapcomputof{\tau}{\apvar}=\segval$ due to $\tau\computequiv\sigma$.
				This contradicts enabledness of $\anact$ after $\tau$.
				Altogether, this means we must have $\heapcomputof{\sigma}{\apvar}\in\adr$.
				So $\acom$ is enabled after $\sigma$.
				That is, there is some update $\anupp$ such that $\sigma.(\athread,\acom,\anupp)\in\adrsem{\anadr}$.
				Due to the premise, we know that $\sigma.(\athread,\acom,\anupp)$ is PRF.
				That is, we have $\apvar\in\validof{\sigma}$.
				Thus, $\apvar\in\validof{\tau}$ by $\tau\computequiv\sigma$.
				Hence, \Cref{thm:mimic-simple-NEW} is enabled.

				Now, \Cref{thm:mimic-simple-NEW} yields some $\anactp$ such that $\sigma.\anactp\in\asem{\anadr}$, $\tau.\anact\computequiv\sigma.\anactp$, $\tau.\anact\heapequiv[A]\sigma.\anactp$, and $\tau.\anact\obsrel\sigma.\anactp$.
				That is, $\sigmagoal:=\sigma.\anactp$ satisfies the claim.

				\item[$\sigma.\anact\in\adrsem{\anadr}$ and $\acom\not\equiv\apvar:=\malloc$ and $\acom\not\equiv\afunc(\vecof{\apvar},\vecof{\advar})$]
				As for assignments, we use \Cref{thm:mimic-simple-NEW} to get some $\anactp$ and choose $\sigmagoal:=\sigma.\anactp$.
				(Note that assertions over data variables are always enabled. Similarly, this case covers all $\exit$ actions.)

				\item[$\sigma.\anact\in\adrsem{\anadr}$ and $\acom\equiv\afunc(\vecof{\apvar},\vecof{\advar})$ and $\heapcomputof{\tau}{\vecof{\apvar}}=\heapcomputof{\sigma}{\vecof{\apvar}}$]
				As for assignments, we use \Cref{thm:mimic-simple-NEW} to get some $\anactp$ and choose $\sigmagoal:=\sigma.\anactp$.

				\item[$\sigma.\anact\in\adrsem{\anadr}$ and $\acom\equiv\afunc(\vecof{\apvar},\vecof{\advar})$ and $\heapcomputof{\tau}{\vecof{\apvar}}\neq\heapcomputof{\sigma}{\vecof{\apvar}}$]

				We show that $\sigmagoal:=\sigma.\anact$ is an adequate choice.
				Let $\heapcomputof{\tau}{\vecof{\apvar}}=\vecof{\anadrp}_1$, $\heapcomputof{\sigma}{\vecof{\apvar}}=\vecof{\anadrp}_2$, and $\heapcomputof{\tau}{\vecof{\advar}}=\vecof{\advalue}$.
				By $\tau\computequiv\sigma$ we have $\heapcomputof{\sigma}{\vecof{\advar}}=\vecof{\advalue}$.
				By $\heapcomputof{\tau}{\vecof{\apvar}}\neq\heapcomputof{\sigma}{\vecof{\apvar}}$ we have $\vecof{\anadrp}_1\neq\vecof{\anadrp}_2$.
				The task is to show: \[\forall\,\anadrpp\in\vadrof{\tau.\anact}\cup\set{\anadr}.~~\freeableof{\tau.\anact}{\anadrpp}\subseteq\freeableof{\sigma.\anact}{\anadrpp}\ .\]
				To that end, consider some arbitrary $\anadrpp\in\vadrof{\tau.\anact}\cup\set{\anadr}$.
				Because $\anact$ does not alter the heap or validity of expressions, we have $\anadrpp\in\vadrof{\tau}\cup\set{\anadr}$ by definition.
				From $\tau\heapequiv[\anadr]\sigma$ and $\tau\obsrel\sigma$ we get $\freeableof{\tau}{\anadrpp}\subseteq\freeableof{\sigma}{\anadrpp}$.
				This implies
				\begin{align}
					\label{thm:rprf-fresh-and-free-not-freeable:freeable-same-event}
				 	\freeableof{\historyof{\tau}.\afunc(\vecof{\anadrp}_1,\vecof{\advalue})}{\anadrpp}\subseteq\freeableof{\historyof{\sigma}.\afunc(\vecof{\anadrp}_1,\vecof{\advalue})}{\anadrpp}\ .
				\end{align}
				To see this, let $\ahist\in\freeableof{\historyof{\tau}.\afunc(\vecof{\anadrp}_1,\vecof{\advalue})}{\anadrpp}$.
				So we have $\afunc(\vecof{\anadrp}_1,\vecof{\advalue}).\ahist\in\freeableof{\historyof{\tau}}{\anadrpp}$ by \Cref{thm:move-event-from-freeable}.
				Using $\tau\heapequiv[\anadr]\sigma$ and $\tau\obsrel\sigma$ then yields $\afunc(\vecof{\anadrp}_1,\vecof{\advalue}).\ahist\in\freeableof{\historyof{\sigma}}{\anadrpp}$.
				Again by \Cref{thm:move-event-from-freeable}, we conclude $\ahist\in\freeableof{\historyof{\sigma}.\afunc(\vecof{\anadrp}_1,\vecof{\advalue})}{\anadrpp}$.

				Now, recall that $\sigma.\anact\in\adrsem{\anadr}$ is PRF by assumption.
				That is, the SMR call executed by $\anact$ is not racy.
				From this we get:
				\begin{align}
					\label{thm:rprf-fresh-and-free-not-freeable:freeable-same-history}
				 	\freeableof{\historyof{\sigma}.\afunc(\vecof{\anadrp}_1,\vecof{\advalue})}{\anadrpp}\subseteq\freeableof{\historyof{\sigma}.\afunc(\vecof{\anadrp}_2,\vecof{\advalue})}{\anadrpp}\ .
				\end{align}
				To see this, we have to show that $\anadrp_{2,i}=\anadrpp\vee\apvar_i\in\validof{\sigma}\implies\anadrp_{1,i}=\anadrp_{2,i}$ holds for every $i$.
				If $\apvar_i\in\validof{\sigma}$, then we have $\anadrp_{2,i}=\heapcomputof{\sigma}{\apvar_i}=\heapcomputof{\tau}{\apvar_i}=\anadrp_{1,i}$ due to $\tau\computequiv\sigma$.
				So consider now the case where we have $\anadrp_{2,i}=\anadrpp$ and $\apvar_i\notin\validof{\sigma}$.
				Towards a contradiction, assume $\anadrpp\neq\anadr$.
				So $\anadrpp\in\vadrof{\sigma}$ must hold by choice.
				Moreover, we get $\anadrpp\in\heapcomputof{\sigma}{\set{\apvar_i}}\subseteq\heapcomputof{\sigma}{\pexp\setminus\validof{\sigma}}$.
				Then, \Cref{thm:invalidpointers-adr} yields $\anadrpp\in{\anadr}$ which contradicts the assumption.
				This means we must have $\anadrpp=\anadr$.
				Hence, $\heapcomputof{\sigma}{\apvar_i}=\anadrpp$ implies $\heapcomputof{\tau}{\apvar_i}=\anadrpp$ due to $\tau\heapequiv[\anadr]\sigma$.
				That is, $\anadrp_{1,i}=\anadrp_{2,i}$ as desired.
				Altogether, this proves the desired implication required to apply race freedom of the invocation.

				Combining \eqref{thm:rprf-fresh-and-free-not-freeable:freeable-same-event} and \eqref{thm:rprf-fresh-and-free-not-freeable:freeable-same-history} we conclude as follows:
				\begin{align*}
					\freeableof{\tau.\anact}{\anadrpp}\:
					=~&
					\freeableof{\historyof{\tau}.\afunc(\vecof{\anadr},\vecof{\advalue})}{\anadrpp}
					\subseteq
					\freeableof{\historyof{\sigma}.\afunc(\vecof{\anadr},\vecof{\advalue})}{\anadrpp}
					\\\subseteq~&
					\freeableof{\historyof{\sigma}.\afunc(\vecof{\anadrp},\vecof{\advalue})}{\anadrpp}
					=
					\freeableof{\sigma.\anact}{\anadrpp}
				\end{align*}
				The first equality holds by the semantics, i.e., $\heapcomputof{\tau}{\apvar_i}\in\adr$.
				To see the last equality, assume to the contrary that we had $\heapcomputof{\sigma}{\apvar_i}\notin\adr$.
				By \Cref{thm:rprf-vs-bot-pvar} this means $\heapcomputof{\sigma}{\apvar_i}=\segval$.
				By \Cref{thm:seg-valid}, we must have $\apvar_i\in\validof{\sigma}$ then.
				This, however, yields $\heapcomputof{\tau}{\apvar_i}=\segval$ what contradicts $\heapcomputof{\tau}{\apvar_i}\in\adr$.

				Since $\anact$ does not modify the memory nor the validity of expressions nor the freed and fresh addresses, the above implies the desired $\tau.anact\computequiv\sigma.\anact$, $\tau.\anact\heapequiv[\anadr]\sigma.\anact$, and $\tau.\anact\obsrel\sigma.\anact$.

				\item[$\acom\equiv\apvar:=\malloc$]
				Let the update be $\anup=[\apvar\mapsto\anadrp,\psel{\anadrp}\mapsto\segval,\dsel{\anadrp}\mapsto\advalue]$.
				We have $\anadrp\in\freshof{\tau}\cup\freedof{\tau}$.

				If $\anadrp=\anadr$, then we use, similar to the case of assignments, \Cref{thm:mimic-simple-NEW} to get some $\anactp$ and choose $\sigmagoal:=\sigma.\anactp$ with the desired $\tau.\anact\computequiv\sigmagoal$, $\tau.\anact\heapequiv[\anadr]\sigmagoal$, and $\tau.\anact\obsrel\sigmagoal$.
				
				So consider $\anadrp\neq\anadr$ hereafter.
				By the induction hypothesis there is $\gamma\in\adrsem{\anadrp}$ with $\tau\computequiv\gamma$, $\tau\obsrel\gamma$, and $\tau\heapequiv[\anadrp]\gamma$.
				\Cref{thm:computequiv-is-transitiv} gives $\gamma\computequiv\sigma$.
				Moreover, $\anadrp\in\freshof{\gamma}\cup\freedof{\gamma}$ follows from $\tau\heapequiv[\anadrp]\gamma$.
				So $\anact$ is enabled after $\gamma$, that is, $\gamma.\anact\in\adrsem{\anadrp}$.
				Then, \Cref{thm:mimic-simple-NEW} gives $\anactp=(\athread,\acom,\anupp)$ such that $\gamma.\anactp\in\adrsem{\anadrp}$, $\tau.\anact\computequiv\gamma.\anactp$, $\tau.\anact\obsrel\sigma.\anactp$, $\tau.\anact\heapequiv[\anadrp]\sigma.\anactp$.
				Because $\apvar,\psel{\anadrp},\dsel{\anadrp}\in\domof{\restrict{\heapcomput{\tau.\anact}}{\validof{\tau.\anact}}}$ holds by definition we must have $\anup=\anupp$.
				To sum this up, we have:
				\begin{itemize}[nosep,leftmargin=.75cm]
					\item $\tau.\anact\in\allsem$
					\item $\gamma.\anact\in\adrsem{\anadrp}$
					\item $\sigma\in\adrsem{\anadr}$
					\item $\tau\computequiv\gamma$, $\tau\obsrel\gamma$, $\tau\heapequiv[\anadrp]\gamma$
					\item $\tau.\anact\computequiv\gamma.\anact$, $\tau.\anact\obsrel\gamma.\anact$, $\tau.\anact\heapequiv[\anadrp]\gamma.\anact$
					\item $\tau\computequiv\sigma$, $\tau\obsrel\sigma$, $\tau\heapequiv[\anadr]\sigma$
					\item $\anadr\neq\anadrp$
					\item $\gamma\computequiv\sigma$
					\item $\anact=(\_,\apvar:=\malloc,[\apvar\mapsto\anadrp,\_])$
				\end{itemize}
				Now we can use reduction compatibility (\Cref{def:reduction-compatible}).
				This yields some $\sigmagoal\in\adrsem{\anadr}$ with the desired $\tau.\anact\computequiv\sigmagoal$, $\tau.\anact\obsrel\sigmagoal$, and $\tau.\anact\heapequiv[A]\sigmagoal$.
				This concludes the claim.

				\item[$\sigma.\anact\notin\adrsem{\anadr}$ and $\acom\equiv\freeof{\anadrp}$]
				The update is $\anup=[\psel{\anadrp}\mapsto\bot,\dsel{\anadrp}\mapsto\bot]$.
				By definition, we have $\freeof{\anadrp}\in\freeableof{\tau}{\anadrp}$.
				And the assumption that $\freeof{\anadrp}$ is not enabled after $\sigma$ implies $\freeof{\anadrp}\notin\freeableof{\sigma}{\anadrp}$.
				Note that $\anadrp\neq\anadr$ must hold as for otherwise $\tau\heapequiv[\anadr]\sigma$ implies $\freeableof{\tau}{\anadrp}=\freeableof{\sigma}{\anadrp}$ what would contradict $\freeof{\anadrp}\notin\freeableof{\sigma}{\anadrp}$.

				Because of $\tau\obsrel\sigma$ we conclude that $\anadrp\notin\adrof{\restrict{\heapcomput{\tau}}{\validof{\tau}}}$ must hold as for otherwise we had $\freeof{\anadrp}\in\freeableof{\sigma}{\anadrp}$ what contradicts the assumption of $\anact$ not being enabled after $\sigma$.
				By $\tau\computequiv\sigma$ this means $\anadrp\notin\heapcomputof{\sigma}{\validof{\sigma}}$.
				We now show that $\sigmagoal=\sigma$ is an adequate choice, that is, that we do not need to mimic $\anact$ at all.

				\ad{$\tau.\anact\computequiv\sigma$}
				We have $\controlof{\tau}=\controlof{\sigma}$ due to $\tau\computequiv\sigma$.
				Moreover, \Cref{assumption:frees-do-not-affect-control} gives $\controlof{\tau}=\controlof{\tau.\anact}$.
				That is, we have $\controlof{\tau.\anact}=\controlof{\sigma}$.

				To show $\restrict{\heapcomput{\tau.\anact}}{\validof{\tau.\anact}}=\restrict{\heapcomput{\sigma}}{\validof{\sigma}}$ it suffices to show that $\restrict{\heapcomput{\tau}}{\validof{\tau}}=\restrict{\heapcomput{\tau.\anact}}{\validof{\tau.\anact}}$ holds because of $\tau\computequiv\sigma$.
				To arrive at this equality, we first show $\validof{\tau}=\validof{\tau.\anact}$.
				The inclusion $\validof{\tau.\anact}\subseteq\validof{\tau}$ holds by definition.
				To see $\validof{\tau}\subseteq\validof{\tau.\anact}$, consider $\apexp\in\validof{\tau}$.
				Then, $\apexp\not\equiv\psel{\anadrp}$ must hold as for otherwise we had $\anadrp\in\adrof{\restrict{\heapcomput{\tau}}{\validof{\tau}}}$ which does not hold as shown before.
				Moreover, $\heapcomputof{\tau}{\apexp}\neq\anadrp$ must hold as for otherwise we would again get $\anadrp\in\adrof{\restrict{\heapcomput{\tau}}{\validof{\tau}}}$.
				Hence, by the definition of validity, we have $\apexp\in\validof{\tau.\anact}$.

				With this we get:
				\begin{align*}
					\domof{\restrict{\heapcomput{\tau.\anact}}{\validof{\tau.\anact}}}
					=~&
					\validof{\tau.\anact}\cup\dvars\cup\setcond{\dsel{\anadrpp}}{\anadrpp\in\heapcomputof{\tau.\anact}{\validof{\tau.\anact}}}
					\\=~&
					\validof{\tau.\anact}\cup\dvars\cup\setcond{\dsel{\anadrpp}}{\anadrpp\in\heapcomputof{\tau}{\validof{\tau.\anact}}}
					\\=~&
					\validof{\tau}\cup\dvars\cup\setcond{\dsel{\anadrpp}}{\anadrpp\in\heapcomputof{\tau}{\validof{\tau}}}
					=
					\domof{\restrict{\heapcomput{\tau}}{\validof{\tau}}}
				\end{align*}
				where the first equality is by definition, the second equality holds because $\psel{\anadrp}\notin\validof{\tau.\anact}$ together with the update $\anup$, the third equality holds by $\validof{\tau}=\validof{\tau.\anact}$ from above, and the last equality holds by definition again.
				The same reasoning as above also allows us to conclude that $\psel{\anadrp},\dsel{\anadrp}\notin\domof{\restrict{\heapcomput{\tau.\anact}}{\validof{\tau.\anact}}}$ must hold because $\anadrp\notin\heapcomputof{\tau}{\validof{\tau}}$ due to $\anadrp\notin\adrof{\restrict{\heapcomput{\tau}}{\validof{\tau}}}$.
				Hence, $\restrict{\heapcomput{\tau}}{\validof{\tau}}=\restrict{\heapcomput{\tau.\anact}}{\validof{\tau.\anact}}$ follows immediately.
				Altogether, this gives the desired $\tau.\anact\computequiv\sigma$.

				\ad{$\tau.\anact\obsrel\sigma$}
				Let $\anadrpp\in\adrof{\restrict{\heapcomput{\tau.\anact}}{\validof{\tau.\anact}}}$.
				We get $\anadrpp\in\adrof{\restrict{\heapcomput{\tau}}{\validof{\tau}}}$ by definition.
				This means $\anadrp\neq\anadrpp$ due to the above.
				Moreover, by $\tau\obsrel\sigma$ we have $\freeableof{\tau}{\anadrpp}\subseteq\freeableof{\sigma}{\anadrpp}$.
				So it suffices to show that $\freeableof{\tau}{\anadrpp}=\freeableof{\tau.\anact}{\anadrpp}$ holds.
				This holds by reduction compatibility (\Cref{def:reduction-compatible}).

				\ad{$\tau.\anact\heapequiv[\anadr]\sigma$}
				First, recall $\anadrp\neq\anadr$.
				So $\tau\heapequiv[\anadr]\sigma$ gives \[\anadr\in\allocatedof{\tau.\anact}\iff\anadr\in\allocatedof{\tau}\iff\anadr\in\allocatedof{\sigma}\ .\]
				And similarly to $\tau.\anact\obsrel\sigma$, we use reduction compatibility (\Cref{def:reduction-compatible}) together with $\anadrp\neq\anadr$ and $\tau\heapequiv[\anadr]\sigma$ to get $\freeableof{\tau.\anact}{\anadr}=\freeableof{\tau}{\anadr}\subseteq\freeableof{\sigma}{\anadr}$.
				Since $\anact$ does not modify the valuation of pointer variables, we get $\heapcomputof{\tau.\anact}{\apvar}=\heapcomputof{\tau}{\apvar}$ for all $\apvar\in\pvars$.
				Hence, we have:
				\begin{align*}
					\forall\apvar\in\pvars.~\heapcomputof{\tau.\anact}{\apvar}=\anadrp\iff\heapcomputof{\tau}{\apvar}=\anadrp\iff\heapcomputof{\sigma}{\apvar}=\anadrp
					\ .
				\end{align*}
				It remains to show
				\[\forall\anadrpp\in\heapcomputof{\tau.\anact}{\validof{\tau.\anact}}.~\heapcomputof{\tau.\anact}{\psel{\anadrpp}}=\anadr\iff\heapcomputof{\sigma}{\psel{\anadrpp}}=\anadr\ .\]
				So consider $\anadrpp\in\heapcomputof{\tau.\anact}{\validof{\tau.\anact}}$.
				Then there is some $\apexp\in\validof{\tau.\anact}$ with $\heapcomputof{\tau.\anact}{\apexp}=\anadrpp$.
				By definition, $\apexp\in\validof{\tau}$.
				And because $\psel{\anadrp}\notin\validof{\tau.\anact}$ by definition, we must have $\apexp\not\equiv\psel{\anadrp}$.
				So we get $\heapcomputof{\tau}{\apexp}=\heapcomputof{\tau.\anact}{\apexp}=\anadrpp$.
				Hence, $\anadrpp\in\heapcomputof{\tau}{\validof{\tau}}$.
				We already showed $\anadrp\notin\adrof{\restrict{\heapcomput{\tau}}{\validof{\tau}}}$.
				That is, $\anadrpp\neq\anadrp$ must hold.
				So we conclude as follows:
				\[\heapcomputof{\tau.\anact}{\psel{\anadrpp}}=\anadr\iff\heapcomputof{\tau}{\psel{\anadrpp}}=\anadr\iff\heapcomputof{\sigma}{\psel{\anadrpp}}=\anadr\]
				where the first equivalence holds because $\anup$ does not update $\psel{\anadrpp}$ (due to $\anadrpp\neq\anadrp$) and the second equivalence holds by $\tau\heapequiv[\anadr]\sigma$.
				This concludes the claim.

				\item[$\sigma.\anact\notin\adrsem{\anadr}$ and $\acom\equiv\assertof{\apvar\triangleq\apvarp}$ with $\triangleq\,\in\set{=,\neq}$]
				Let $\heapcomputof{\tau}{\apvar}=\anadrp$, $\heapcomputof{\tau}{\apvarp}=\anadrpp$, $\heapcomputof{\sigma}{\apvar}=\anadrp'$, and $\heapcomputof{\sigma}{\apvarp}=\anadrpp'$.
				By the semantics we have $\anadrp\triangleq\anadrpp$.
				Since $\anact$ is not enable after $\sigma$, we have $\neg(\anadrp'\triangleq\anadrpp')$.
				That is, we have $\anadrp=\anadrpp\iff\anadrp'\neq\anadrpp'$.
				Note that $\set{\anadrp,\anadrpp,\anadrp',\anadrpp'}\cap\set{\anadr}=\emptyset$ must hold as for otherwise $\tau\heapequiv[\anadr]\sigma$ would yield $\anadrp'\triangleq\anadrpp'$.
				In order to find an appropriate $\sigmagoal$ we use reduction compatibility (\Cref{def:reduction-compatible}).

				By the induction hypothesis there is $\gamma\in\asem{\anadrp}$ with $\tau\computequiv\gamma$, $\tau\obsrel\gamma$, and $\tau\heapequiv[\anadrp]\gamma$.
				By \Cref{thm:computequiv-is-transitiv} we have $\gamma\computequiv\sigma$.
				The latter yields $\heapcomputof{\gamma}{\apvar}=\anadrp$ and $\heapcomputof{\gamma}{\apvarp}=\anadrp\iff\heapcomputof{\tau}{\apvarp}=\anadrp$.
				This means $\heapcomputof{\gamma}{\apvar}\triangleq\heapcomputof{\gamma}{\apvarp}$.
				So $\anact$ is enabled after $\gamma$: $\gamma.\anact\in\adrsem{\anadrp}$.
				Then, \Cref{thm:mimic-simple-NEW} gives $\anactp=(\athread,\acom,\anupp)$ such that $\gamma.\anactp\in\adrsem{\anadrp}$, $\tau.\anact\computequiv\gamma.\anactp$, $\tau.\anact\obsrel\sigma.\anactp$, $\tau.\anact\heapequiv[\anadrp]\sigma.\anactp$.
				Because $\anup=\emptyset=\anupp$ must hold by the semantics we have $\anact=\anactp$.
				To sum this up, we have:
				\begin{itemize}[nosep,leftmargin=.75cm]
					\item $\tau.\anact\in\allsem$
					\item $\gamma.\anact\in\adrsem{\anadrp}$
					\item $\sigma\in\adrsem{\anadr}$
					\item $\tau\computequiv\gamma$, $\tau\obsrel\gamma$, $\tau\heapequiv[\anadrp]\gamma$
					\item $\tau.\anact\computequiv\gamma.\anact$, $\tau.\anact\obsrel\gamma.\anact$, $\tau.\anact\heapequiv[\anadrp]\gamma.\anact$
					\item $\tau\computequiv\sigma$, $\tau\obsrel\sigma$, $\tau\heapequiv[\anadr]\sigma$
					\item $\anadr\neq\anadrp$
					\item $\gamma\computequiv\sigma$
					\item $\anact=(\_,\assertof{\_},\_)$
				\end{itemize}
				Now we can use reduction compatibility (\Cref{def:reduction-compatible}).
				This yields some $\sigmagoal\in\adrsem{\anadr}$ with the desired $\tau.\anact\computequiv\sigmagoal$, $\tau.\anact\obsrel\sigmagoal$, and $\tau.\anact\heapequiv[A]\sigmagoal$.
				This concludes the claim.

			\end{casedistinction}
	\end{description}
	The above case distinction is complete.
	This concludes the induction.
\end{proof}


\begin{proof}[Proof of \Cref{thm:rprf-guarantee}]
	Follows immediately from \Cref{thm:RPRF-implies-GCplus}.
\end{proof}

\begin{proof}[Proof of \Cref{thm:ABA-awareness-and-elision-support-implies-reduction-compatibility}]
	Consider some $\tau.\anact$, $\sigma_\anadr.\anact$, and $\sigma_\anadrp$ satisfying the premise of the implication from \Cref{def:reduction-compatible}.
	We show that there is some $\gamma$ with $\tau.\anact\computequiv\gamma$, $\tau.\anact\obsrel\gamma$, and $\tau.\anact\heapequiv[\anadrp]\gamma$ indeed.
	We distinguish two cases.

	\begin{casedistinction}
		\item[$\anact=(\textnormal{\_},\assertof{\apvar\triangleq\apvarp},\textnormal{\_})$]
			\Cref{def:reduction-compatible} gives $\tau\computequiv\sigma_\anadr$, $\tau.\anact\computequiv\sigma_\anadr.\anact$, $\tau\computequiv\sigma_\anadrp$, $\tau\obsrel\sigma_\anadrp$, and $\tau\heapequiv[\anadrp]\sigma_\anadrp$.
			From the absence of harmful ABAs (\Cref{def:harmful-ABA}) we get $\gamma$ with $\sigma_\anadr.\anact\computequiv\gamma$, $\sigma_\anadrp\heapequiv[\anadrp]\gamma$, and $\sigma_\anadrp\obsrel\gamma$.
			Then, we get $\tau.\anact\computequiv\gamma$ by \Cref{thm:computequiv-is-transitiv}.
			We get $\tau\obsrel\gamma$ by \Cref{thm:obsrel-is-transitiv-provided-same-valid-adr-NEW}.
			And we get $\tau\heapequiv[\anadrp]\gamma$ by \Cref{thm:heapequiv-is-transitiv-provided-same-valid-heap-range}.
			We show that $\sigma_\anadrp':=\gamma$ is an adequate choice.
			To do so, we show two auxiliary properties first.
			\begin{auxiliary}
				\heapcomputof{\tau.\anact}{\validof{\tau.\anact}} &= \heapcomputof{\tau}{\validof{\tau}}
				\label[aux]{proof:abaawareness:range-valid-heap}
				\\
				\adrof{\restrict{\heapcomput{\tau.\anact}}{\validof{\tau.\anact}}} &= \adrof{\restrict{\heapcomput{\tau}}{\validof{\tau}}}
				\label[aux]{proof:abaawareness:adr-valid-heap}
			\end{auxiliary}

			\ad{\Cref{proof:abaawareness:range-valid-heap}}
			We have $\heapcomput{\tau.\anact}=\heapcomput{\tau}$ because $\anact$ is an assertion.
			So it remains to show that $\heapcomputof{\tau}{\validof{\tau.\anact}}=\heapcomputof{\tau}{\validof{\tau}}$ holds.
			If $\validof{\tau.\anact}=\validof{\tau}$, then the claim follows immediately.

			So consider the case where we have $\validof{\tau.\anact}\neq\validof{\tau}$.
			Wlog. we have $\apvar\in\validof{\tau}$.
			By definition, $\validof{\tau.\anact}=\validof{\tau}\cup\set{\apvar,\apvarp}=\validof{\tau}\cup\set{\apvarp}$ and $\apvar\triangleq\apvarp\equiv\apvar=\apvarp$ follows.
			Hence, $\heapcomputof{\tau}{\apvar}=\heapcomputof{\tau}{\apvarp}$.
			So \[\heapcomputof{\tau}{\validof{\tau.\anact}}=\heapcomputof{\tau}{\validof{\tau}}\cup\set{\heapcomputof{\tau}{\apvar}}=\heapcomputof{\tau}{\validof{\tau}}\] where the last equality holds because $\apvar\in\validof{\tau}$ and thus $\heapcomputof{\tau}{\apvar}\in\heapcomputof{\tau}{\validof{\tau}}$.

			\ad{\Cref{proof:abaawareness:adr-valid-heap}}
			Note that we have:
			\begin{align*}
				\validof{\tau.\anact}\cap\adr&\subseteq(\validof{\tau}\cup\set{\apvar,\apvarp})\cap\adr=\validof{\tau}\cap\adr
				\\\text{and}\qquad
				\validof{\tau}\cap\adr&\subseteq\validof{\tau.\anact}\cap\adr
			\end{align*}
			So, $\validof{\tau.\anact}\cap\adr=\validof{\tau}\cap\adr$ holds.
			We conclude as follows:
			\begin{align*}
				\adrof{\restrict{\heapcomput{\tau.\anact}}{\validof{\tau.\anact}}}
				\:=~&
				(\validof{\tau.\anact}\cap\adr)\cup\heapcomputof{\tau.\anact}{\validof{\tau.\anact}}
				=
				(\validof{\tau}\cap\adr)\cup\heapcomputof{\tau}{\validof{\tau}}
				\\=~&
				\adrof{\restrict{\heapcomput{\tau}}{\validof{\tau}}}
			\end{align*}
			where the first and last equality are due to \Cref{thm:adrof-valid-heap-restriction} and the second equality is due to the note above and \Cref{proof:abaawareness:range-valid-heap}.

			\ad{$\tau.\anact\obsrel\gamma$}
			We already have $\tau\obsrel\gamma$.
			Assume for the moment that we have $\tau.\anact\obsrel\tau$.
			Then, \Cref{thm:obsrel-is-transitiv-provided-same-valid-adr-NEW} together with \Cref{proof:abaawareness:adr-valid-heap} yields the desired $\tau.\anact\obsrel\gamma$.
			So it remains to show that $\tau.\anact\obsrel\tau$ holds indeed.
			This follows immediately from the fact that $\anact$ does not emit an event and \Cref{proof:abaawareness:adr-valid-heap} so that we have $\freeableof{\tau.\anact}{\anadrpp}=\freeableof{\tau}{\anadrpp}$ for every $\anadrpp\in\adr$.
			This concludes the property.

			\ad{$\tau.\anact\heapequiv[\anadrp]\gamma$}
			We already have $\tau\heapequiv[\anadrp]\gamma$.
			Assume for the moment we have $\tau.\anact\heapequiv[\anadrp]\tau$.
			Then, \Cref{proof:abaawareness:adr-valid-heap} together with \Cref{thm:heapequiv-is-transitiv-provided-same-valid-heap-range} yields the desired $\tau.\anact\heapequiv[\anadrp]\gamma$.
			So it remains to show that $\tau.\anact\heapequiv[\anadrp]\tau$ holds indeed.
			This follows from the definition due to that we have: $\heapcomput{\tau}=\heapcomput{\tau.\anact}$, $\freedof{\tau.\anact}=\freedof{\tau}$, $\freshof{\tau.\anact}=\freshof{\tau}$, $\historyof{\tau.\anact}=\historyof{\tau}$, and \Cref{proof:abaawareness:range-valid-heap}.

		\item[$\anact=(\_,\apvar:=\malloc,{[}\apvar\mapsto\anadrp,\_{]})$]
			\Cref{def:reduction-compatible} gives $\sigma_\anadr.\anact\in\adrsem{\anadr}$ and $\sigma_\anadrp\in\adrsem{\anadrp}$ with $\sigma_\anadr\computequiv\sigma_\anadrp$, $\sigma_\anadr\computequiv\sigma_\anadrp$, $\tau\computequiv\sigma_\anadrp$, $\tau\heapequiv[\anadrp]\sigma_\anadrp$, $\tau\obsrel\sigma_\anadrp$, and $\anadr\neq\anadrp$.
			Since $\anact$ is enabled after $\sigma_\anadr$ we have $\anadr\in\freshof{\sigma_\anadr}\cup\freedof{\sigma_\anadr}$.
			From $\onesem$ being PRF together with \Cref{thm:adrof-valid-heap-restriction,thm:fresh-not-valid,thm:free-not-valid} we get $\anadr\notin\vadrof{\sigma_\anadr}$.
			So $\sigma_\anadr\computequiv\sigma_\anadrp$ gives $\anadr\notin\vadrof{\sigma_\anadrp}$.
			Then, \Cref{thm:supported-elision-of-invalid-address} gives $\gamma\in\adrsem{\anadrp}$ with $\sigma_\anadrp\computequiv\gamma$, $\sigma_\anadrp\heapequiv[\anadrp]\gamma$, $\sigma_\anadrp\obsrel\gamma$, and $\anadr\in\freshof{\gamma}$.
			Combining this with $\tau\computequiv\sigma_\anadrp$, $\tau\heapequiv[\anadrp]\sigma_\anadrp$, and $\tau\obsrel\sigma_\anadrp$ using \Cref{thm:computequiv-is-transitiv,thm:obsrel-is-transitiv-provided-same-valid-adr-NEW,thm:heapequiv-is-transitiv-provided-same-valid-heap-range} gives $\tau\computequiv\gamma$, $\tau\obsrel\gamma$, and $\tau\heapequiv[\anadrp]\gamma$.
			The latter provides $\freeableof{\tau}{\anadrp}\subseteq\freeableof{\gamma}{\anadrp}$.
			This together with $\anadr\in\freshof{\gamma}$ yields $\freeableof{\tau}{\anadr}\subseteq\freeableof{\gamma}{\anadr}$ by elision support according to \Cref{def:elision-support}\ref{def:elision-support:fresh}.
			Moreover, since $\anadr\in\freshof{\gamma}$ we have $\gamma.\anact\in\asem{\anadrp}$.
			Then, \Cref{thm:mimic-simple-NEW} gives $\anactp=(\athread,\acom,\anupp)$ with $\tau.\anact\computequiv\gamma.\anactp$, $\tau.\anact\heapequiv[\anadrp]\gamma.\anactp$, and $\tau.\anact\obsrel\sigma.\anactp$.
			Then $\sigma_\anadrp':=\sigma.\anactp$ is an adequate choice.

	\end{casedistinction}
\end{proof}

%% file: main.bbl

\begin{thebibliography}{90}


\ifx \showCODEN    \undefined \def \showCODEN     #1{\unskip}     \fi
\ifx \showDOI      \undefined \def \showDOI       #1{#1}\fi
\ifx \showISBNx    \undefined \def \showISBNx     #1{\unskip}     \fi
\ifx \showISBNxiii \undefined \def \showISBNxiii  #1{\unskip}     \fi
\ifx \showISSN     \undefined \def \showISSN      #1{\unskip}     \fi
\ifx \showLCCN     \undefined \def \showLCCN      #1{\unskip}     \fi
\ifx \shownote     \undefined \def \shownote      #1{#1}          \fi
\ifx \showarticletitle \undefined \def \showarticletitle #1{#1}   \fi
\ifx \showURL      \undefined \def \showURL       {\relax}        \fi
\providecommand\bibfield[2]{#2}
\providecommand\bibinfo[2]{#2}
\providecommand\natexlab[1]{#1}
\providecommand\showeprint[2][]{arXiv:#2}

\bibitem[\protect\citeauthoryear{Abdulla, Haziza, Hol{\'{\i}}k, Jonsson, and
  Rezine}{Abdulla et~al\mbox{.}}{2013}]%
        {DBLP:conf/tacas/AbdullaHHJR13}
\bibfield{author}{\bibinfo{person}{Parosh~Aziz Abdulla},
  \bibinfo{person}{Fr{\'{e}}d{\'{e}}ric Haziza}, \bibinfo{person}{Luk{\'{a}}s
  Hol{\'{\i}}k}, \bibinfo{person}{Bengt Jonsson}, {and} \bibinfo{person}{Ahmed
  Rezine}.} \bibinfo{year}{2013}\natexlab{}.
\newblock \showarticletitle{An Integrated Specification and Verification
  Technique for Highly Concurrent Data Structures}. In
  \bibinfo{booktitle}{\emph{{TACAS}}} \emph{(\bibinfo{series}{{LNCS}})},
  Vol.~\bibinfo{volume}{7795}. \bibinfo{publisher}{Springer},
  \bibinfo{pages}{324--338}.
\newblock
\urldef\tempurl%
\url{https://doi.org/10.1007/978-3-642-36742-7\_23}
\showDOI{\tempurl}


\bibitem[\protect\citeauthoryear{Abdulla, Haziza, Hol{\'{\i}}k, Jonsson, and
  Rezine}{Abdulla et~al\mbox{.}}{2017}]%
        {DBLP:journals/sttt/AbdullaHHJR17}
\bibfield{author}{\bibinfo{person}{Parosh~Aziz Abdulla},
  \bibinfo{person}{Fr{\'{e}}d{\'{e}}ric Haziza}, \bibinfo{person}{Luk{\'{a}}s
  Hol{\'{\i}}k}, \bibinfo{person}{Bengt Jonsson}, {and} \bibinfo{person}{Ahmed
  Rezine}.} \bibinfo{year}{2017}\natexlab{}.
\newblock \showarticletitle{An Integrated Specification and Verification
  Technique for Highly Concurrent Data Structures}.
\newblock \bibinfo{journal}{\emph{{STTT}}} \bibinfo{volume}{19},
  \bibinfo{number}{5} (\bibinfo{year}{2017}), \bibinfo{pages}{549--563}.
\newblock
\urldef\tempurl%
\url{https://doi.org/10.1007/s10009-016-0415-4}
\showDOI{\tempurl}


\bibitem[\protect\citeauthoryear{Abdulla, Jonsson, and Trinh}{Abdulla
  et~al\mbox{.}}{2016}]%
        {DBLP:conf/sas/AbdullaJT16}
\bibfield{author}{\bibinfo{person}{Parosh~Aziz Abdulla}, \bibinfo{person}{Bengt
  Jonsson}, {and} \bibinfo{person}{Cong~Quy Trinh}.}
  \bibinfo{year}{2016}\natexlab{}.
\newblock \showarticletitle{Automated Verification of Linearization Policies}.
  In \bibinfo{booktitle}{\emph{{SAS}}} \emph{(\bibinfo{series}{{LNCS}})},
  Vol.~\bibinfo{volume}{9837}. \bibinfo{publisher}{Springer},
  \bibinfo{pages}{61--83}.
\newblock
\urldef\tempurl%
\url{https://doi.org/10.1007/978-3-662-53413-7\_4}
\showDOI{\tempurl}


\bibitem[\protect\citeauthoryear{Aghazadeh, Golab, and Woelfel}{Aghazadeh
  et~al\mbox{.}}{2014}]%
        {DBLP:conf/podc/AghazadehGW13a}
\bibfield{author}{\bibinfo{person}{Zahra Aghazadeh},
  \bibinfo{person}{Wojciech~M. Golab}, {and} \bibinfo{person}{Philipp
  Woelfel}.} \bibinfo{year}{2014}\natexlab{}.
\newblock \showarticletitle{Making objects writable}. In
  \bibinfo{booktitle}{\emph{{PODC}}}. \bibinfo{publisher}{{ACM}},
  \bibinfo{pages}{385--395}.
\newblock
\urldef\tempurl%
\url{https://doi.org/10.1145/2611462.2611483}
\showDOI{\tempurl}


\bibitem[\protect\citeauthoryear{Alglave, Kroening, and Tautschnig}{Alglave
  et~al\mbox{.}}{2013}]%
        {DBLP:conf/cav/AlglaveKT13}
\bibfield{author}{\bibinfo{person}{Jade Alglave}, \bibinfo{person}{Daniel
  Kroening}, {and} \bibinfo{person}{Michael Tautschnig}.}
  \bibinfo{year}{2013}\natexlab{}.
\newblock \showarticletitle{Partial Orders for Efficient Bounded Model Checking
  of Concurrent Software}. In \bibinfo{booktitle}{\emph{{CAV}}}
  \emph{(\bibinfo{series}{{LNCS}})}, Vol.~\bibinfo{volume}{8044}.
  \bibinfo{publisher}{Springer}, \bibinfo{pages}{141--157}.
\newblock
\urldef\tempurl%
\url{https://doi.org/10.1007/978-3-642-39799-8\_9}
\showDOI{\tempurl}


\bibitem[\protect\citeauthoryear{Alistarh, Eugster, Herlihy, Matveev, and
  Shavit}{Alistarh et~al\mbox{.}}{2014}]%
        {DBLP:conf/eurosys/AlistarhEHMS14}
\bibfield{author}{\bibinfo{person}{Dan Alistarh}, \bibinfo{person}{Patrick
  Eugster}, \bibinfo{person}{Maurice Herlihy}, \bibinfo{person}{Alexander
  Matveev}, {and} \bibinfo{person}{Nir Shavit}.}
  \bibinfo{year}{2014}\natexlab{}.
\newblock \showarticletitle{StackTrack: an automated transactional approach to
  concurrent memory reclamation}. In \bibinfo{booktitle}{\emph{EuroSys}}.
  \bibinfo{publisher}{{ACM}}, \bibinfo{pages}{25:1--25:14}.
\newblock
\urldef\tempurl%
\url{https://doi.org/10.1145/2592798.2592808}
\showDOI{\tempurl}


\bibitem[\protect\citeauthoryear{Alistarh, Leiserson, Matveev, and
  Shavit}{Alistarh et~al\mbox{.}}{2015}]%
        {DBLP:conf/spaa/AlistarhLMS15}
\bibfield{author}{\bibinfo{person}{Dan Alistarh}, \bibinfo{person}{William~M.
  Leiserson}, \bibinfo{person}{Alexander Matveev}, {and} \bibinfo{person}{Nir
  Shavit}.} \bibinfo{year}{2015}\natexlab{}.
\newblock \showarticletitle{ThreadScan: Automatic and Scalable Memory
  Reclamation}. In \bibinfo{booktitle}{\emph{{SPAA}}}.
  \bibinfo{publisher}{{ACM}}, \bibinfo{pages}{123--132}.
\newblock
\urldef\tempurl%
\url{https://doi.org/10.1145/2755573.2755600}
\showDOI{\tempurl}


\bibitem[\protect\citeauthoryear{Amit, Rinetzky, Reps, Sagiv, and Yahav}{Amit
  et~al\mbox{.}}{2007}]%
        {DBLP:conf/cav/AmitRRSY07}
\bibfield{author}{\bibinfo{person}{Daphna Amit}, \bibinfo{person}{Noam
  Rinetzky}, \bibinfo{person}{Thomas~W. Reps}, \bibinfo{person}{Mooly Sagiv},
  {and} \bibinfo{person}{Eran Yahav}.} \bibinfo{year}{2007}\natexlab{}.
\newblock \showarticletitle{Comparison Under Abstraction for Verifying
  Linearizability}. In \bibinfo{booktitle}{\emph{{CAV}}}
  \emph{(\bibinfo{series}{{LNCS}})}, Vol.~\bibinfo{volume}{4590}.
  \bibinfo{publisher}{Springer}, \bibinfo{pages}{477--490}.
\newblock
\urldef\tempurl%
\url{https://doi.org/10.1007/978-3-540-73368-3\_49}
\showDOI{\tempurl}


\bibitem[\protect\citeauthoryear{Balmau, Guerraoui, Herlihy, and
  Zablotchi}{Balmau et~al\mbox{.}}{2016}]%
        {DBLP:conf/spaa/BalmauGHZ16}
\bibfield{author}{\bibinfo{person}{Oana Balmau}, \bibinfo{person}{Rachid
  Guerraoui}, \bibinfo{person}{Maurice Herlihy}, {and} \bibinfo{person}{Igor
  Zablotchi}.} \bibinfo{year}{2016}\natexlab{}.
\newblock \showarticletitle{Fast and Robust Memory Reclamation for Concurrent
  Data Structures}. In \bibinfo{booktitle}{\emph{{SPAA}}}.
  \bibinfo{publisher}{{ACM}}, \bibinfo{pages}{349--359}.
\newblock
\urldef\tempurl%
\url{https://doi.org/10.1145/2935764.2935790}
\showDOI{\tempurl}


\bibitem[\protect\citeauthoryear{B{\"{a}}umler, Schellhorn, Tofan, and
  Reif}{B{\"{a}}umler et~al\mbox{.}}{2011}]%
        {DBLP:journals/fac/BaumlerSTR11}
\bibfield{author}{\bibinfo{person}{Simon B{\"{a}}umler},
  \bibinfo{person}{Gerhard Schellhorn}, \bibinfo{person}{Bogdan Tofan}, {and}
  \bibinfo{person}{Wolfgang Reif}.} \bibinfo{year}{2011}\natexlab{}.
\newblock \showarticletitle{Proving linearizability with temporal logic}.
\newblock \bibinfo{journal}{\emph{Formal Asp. Comput.}} \bibinfo{volume}{23},
  \bibinfo{number}{1} (\bibinfo{year}{2011}), \bibinfo{pages}{91--112}.
\newblock
\urldef\tempurl%
\url{https://doi.org/10.1007/s00165-009-0130-y}
\showDOI{\tempurl}


\bibitem[\protect\citeauthoryear{Berdine, Lev{-}Ami, Manevich, Ramalingam, and
  Sagiv}{Berdine et~al\mbox{.}}{2008}]%
        {DBLP:conf/cav/BerdineLMRS08}
\bibfield{author}{\bibinfo{person}{Josh Berdine}, \bibinfo{person}{Tal
  Lev{-}Ami}, \bibinfo{person}{Roman Manevich}, \bibinfo{person}{G.
  Ramalingam}, {and} \bibinfo{person}{Shmuel Sagiv}.}
  \bibinfo{year}{2008}\natexlab{}.
\newblock \showarticletitle{Thread Quantification for Concurrent Shape
  Analysis}. In \bibinfo{booktitle}{\emph{{CAV}}}
  \emph{(\bibinfo{series}{{LNCS}})}, Vol.~\bibinfo{volume}{5123}.
  \bibinfo{publisher}{Springer}, \bibinfo{pages}{399--413}.
\newblock
\urldef\tempurl%
\url{https://doi.org/10.1007/978-3-540-70545-1\_37}
\showDOI{\tempurl}


\bibitem[\protect\citeauthoryear{Bouajjani, Emmi, Enea, and
  Mutluergil}{Bouajjani et~al\mbox{.}}{2017}]%
        {DBLP:conf/cav/BouajjaniEEM17}
\bibfield{author}{\bibinfo{person}{Ahmed Bouajjani}, \bibinfo{person}{Michael
  Emmi}, \bibinfo{person}{Constantin Enea}, {and} \bibinfo{person}{Suha~Orhun
  Mutluergil}.} \bibinfo{year}{2017}\natexlab{}.
\newblock \showarticletitle{Proving Linearizability Using Forward Simulations}.
  In \bibinfo{booktitle}{\emph{{CAV} {(2)}}} \emph{(\bibinfo{series}{{LNCS}})},
  Vol.~\bibinfo{volume}{10427}. \bibinfo{publisher}{Springer},
  \bibinfo{pages}{542--563}.
\newblock
\urldef\tempurl%
\url{https://doi.org/10.1007/978-3-319-63390-9\_28}
\showDOI{\tempurl}


\bibitem[\protect\citeauthoryear{Braginsky, Kogan, and Petrank}{Braginsky
  et~al\mbox{.}}{2013}]%
        {DBLP:conf/spaa/BraginskyKP13}
\bibfield{author}{\bibinfo{person}{Anastasia Braginsky}, \bibinfo{person}{Alex
  Kogan}, {and} \bibinfo{person}{Erez Petrank}.}
  \bibinfo{year}{2013}\natexlab{}.
\newblock \showarticletitle{Drop the anchor: lightweight memory management for
  non-blocking data structures}. In \bibinfo{booktitle}{\emph{{SPAA}}}.
  \bibinfo{publisher}{{ACM}}, \bibinfo{pages}{33--42}.
\newblock
\urldef\tempurl%
\url{https://doi.org/10.1145/2486159.2486184}
\showDOI{\tempurl}


\bibitem[\protect\citeauthoryear{Brown}{Brown}{2015}]%
        {DBLP:conf/podc/Brown15}
\bibfield{author}{\bibinfo{person}{Trevor~Alexander Brown}.}
  \bibinfo{year}{2015}\natexlab{}.
\newblock \showarticletitle{Reclaiming Memory for Lock-Free Data Structures:
  There has to be a Better Way}. In \bibinfo{booktitle}{\emph{{PODC}}}.
  \bibinfo{publisher}{{ACM}}, \bibinfo{pages}{261--270}.
\newblock
\urldef\tempurl%
\url{https://doi.org/10.1145/2767386.2767436}
\showDOI{\tempurl}


\bibitem[\protect\citeauthoryear{Burckhardt, Dern, Musuvathi, and
  Tan}{Burckhardt et~al\mbox{.}}{2010}]%
        {DBLP:conf/pldi/BurckhardtDMT10}
\bibfield{author}{\bibinfo{person}{Sebastian Burckhardt},
  \bibinfo{person}{Chris Dern}, \bibinfo{person}{Madanlal Musuvathi}, {and}
  \bibinfo{person}{Roy Tan}.} \bibinfo{year}{2010}\natexlab{}.
\newblock \showarticletitle{Line-up: a complete and automatic linearizability
  checker}. In \bibinfo{booktitle}{\emph{{PLDI}}}. \bibinfo{publisher}{{ACM}},
  \bibinfo{pages}{330--340}.
\newblock
\urldef\tempurl%
\url{https://doi.org/10.1145/1806596.1806634}
\showDOI{\tempurl}


\bibitem[\protect\citeauthoryear{Cern{\'{y}}, Radhakrishna, Zufferey,
  Chaudhuri, and Alur}{Cern{\'{y}} et~al\mbox{.}}{2010}]%
        {DBLP:conf/cav/CernyRZCA10}
\bibfield{author}{\bibinfo{person}{Pavol Cern{\'{y}}}, \bibinfo{person}{Arjun
  Radhakrishna}, \bibinfo{person}{Damien Zufferey}, \bibinfo{person}{Swarat
  Chaudhuri}, {and} \bibinfo{person}{Rajeev Alur}.}
  \bibinfo{year}{2010}\natexlab{}.
\newblock \showarticletitle{Model Checking of Linearizability of Concurrent
  List Implementations}. In \bibinfo{booktitle}{\emph{{CAV}}}
  \emph{(\bibinfo{series}{{LNCS}})}, Vol.~\bibinfo{volume}{6174}.
  \bibinfo{publisher}{Springer}, \bibinfo{pages}{465--479}.
\newblock
\urldef\tempurl%
\url{https://doi.org/10.1007/978-3-642-14295-6\_41}
\showDOI{\tempurl}


\bibitem[\protect\citeauthoryear{Cohen and Petrank}{Cohen and Petrank}{2015a}]%
        {DBLP:conf/oopsla/CohenP15}
\bibfield{author}{\bibinfo{person}{Nachshon Cohen} {and} \bibinfo{person}{Erez
  Petrank}.} \bibinfo{year}{2015}\natexlab{a}.
\newblock \showarticletitle{Automatic memory reclamation for lock-free data
  structures}. In \bibinfo{booktitle}{\emph{{OOPSLA}}}.
  \bibinfo{publisher}{{ACM}}, \bibinfo{pages}{260--279}.
\newblock
\urldef\tempurl%
\url{https://doi.org/10.1145/2814270.2814298}
\showDOI{\tempurl}


\bibitem[\protect\citeauthoryear{Cohen and Petrank}{Cohen and Petrank}{2015b}]%
        {DBLP:conf/spaa/CohenP15}
\bibfield{author}{\bibinfo{person}{Nachshon Cohen} {and} \bibinfo{person}{Erez
  Petrank}.} \bibinfo{year}{2015}\natexlab{b}.
\newblock \showarticletitle{Efficient Memory Management for Lock-Free Data
  Structures with Optimistic Access}. In \bibinfo{booktitle}{\emph{{SPAA}}}.
  \bibinfo{publisher}{{ACM}}, \bibinfo{pages}{254--263}.
\newblock
\urldef\tempurl%
\url{https://doi.org/10.1145/2755573.2755579}
\showDOI{\tempurl}


\bibitem[\protect\citeauthoryear{Colvin, Doherty, and Groves}{Colvin
  et~al\mbox{.}}{2005}]%
        {DBLP:journals/entcs/ColvinDG05}
\bibfield{author}{\bibinfo{person}{Robert Colvin}, \bibinfo{person}{Simon
  Doherty}, {and} \bibinfo{person}{Lindsay Groves}.}
  \bibinfo{year}{2005}\natexlab{}.
\newblock \showarticletitle{Verifying Concurrent Data Structures by
  Simulation}.
\newblock \bibinfo{journal}{\emph{Electr. Notes Theor. Comput. Sci.}}
  \bibinfo{volume}{137}, \bibinfo{number}{2} (\bibinfo{year}{2005}),
  \bibinfo{pages}{93--110}.
\newblock
\urldef\tempurl%
\url{https://doi.org/10.1016/j.entcs.2005.04.026}
\showDOI{\tempurl}


\bibitem[\protect\citeauthoryear{Colvin, Groves, Luchangco, and Moir}{Colvin
  et~al\mbox{.}}{2006}]%
        {DBLP:conf/cav/ColvinGLM06}
\bibfield{author}{\bibinfo{person}{Robert Colvin}, \bibinfo{person}{Lindsay
  Groves}, \bibinfo{person}{Victor Luchangco}, {and} \bibinfo{person}{Mark
  Moir}.} \bibinfo{year}{2006}\natexlab{}.
\newblock \showarticletitle{Formal Verification of a Lazy Concurrent List-Based
  Set Algorithm}. In \bibinfo{booktitle}{\emph{{CAV}}}
  \emph{(\bibinfo{series}{{LNCS}})}, Vol.~\bibinfo{volume}{4144}.
  \bibinfo{publisher}{Springer}, \bibinfo{pages}{475--488}.
\newblock
\urldef\tempurl%
\url{https://doi.org/10.1007/11817963\_44}
\showDOI{\tempurl}


\bibitem[\protect\citeauthoryear{Delbianco, Sergey, Nanevski, and
  Banerjee}{Delbianco et~al\mbox{.}}{2017}]%
        {DBLP:conf/ecoop/DelbiancoSNB17}
\bibfield{author}{\bibinfo{person}{Germ{\'{a}}n~Andr{\'{e}}s Delbianco},
  \bibinfo{person}{Ilya Sergey}, \bibinfo{person}{Aleksandar Nanevski}, {and}
  \bibinfo{person}{Anindya Banerjee}.} \bibinfo{year}{2017}\natexlab{}.
\newblock \showarticletitle{Concurrent Data Structures Linked in Time}. In
  \bibinfo{booktitle}{\emph{{ECOOP}}} \emph{(\bibinfo{series}{LIPIcs})},
  Vol.~\bibinfo{volume}{74}. \bibinfo{publisher}{Schloss Dagstuhl -
  Leibniz-Zentrum fuer Informatik}, \bibinfo{pages}{8:1--8:30}.
\newblock
\urldef\tempurl%
\url{https://doi.org/10.4230/LIPIcs.ECOOP.2017.8}
\showDOI{\tempurl}


\bibitem[\protect\citeauthoryear{Derrick, Schellhorn, and Wehrheim}{Derrick
  et~al\mbox{.}}{2011}]%
        {DBLP:journals/toplas/DerrickSW11}
\bibfield{author}{\bibinfo{person}{John Derrick}, \bibinfo{person}{Gerhard
  Schellhorn}, {and} \bibinfo{person}{Heike Wehrheim}.}
  \bibinfo{year}{2011}\natexlab{}.
\newblock \showarticletitle{Mechanically verified proof obligations for
  linearizability}.
\newblock \bibinfo{journal}{\emph{{ACM} Trans. Program. Lang. Syst.}}
  \bibinfo{volume}{33}, \bibinfo{number}{1} (\bibinfo{year}{2011}),
  \bibinfo{pages}{4:1--4:43}.
\newblock
\urldef\tempurl%
\url{https://doi.org/10.1145/1889997.1890001}
\showDOI{\tempurl}


\bibitem[\protect\citeauthoryear{Desnoyers, McKenney, and Dagenais}{Desnoyers
  et~al\mbox{.}}{2013}]%
        {DBLP:journals/sigops/DesnoyersMD13}
\bibfield{author}{\bibinfo{person}{Mathieu Desnoyers}, \bibinfo{person}{Paul~E.
  McKenney}, {and} \bibinfo{person}{Michel~R. Dagenais}.}
  \bibinfo{year}{2013}\natexlab{}.
\newblock \showarticletitle{Multi-core systems modeling for formal verification
  of parallel algorithms}.
\newblock \bibinfo{journal}{\emph{Operating Systems Review}}
  \bibinfo{volume}{47}, \bibinfo{number}{2} (\bibinfo{year}{2013}),
  \bibinfo{pages}{51--65}.
\newblock
\urldef\tempurl%
\url{https://doi.org/10.1145/2506164.2506174}
\showDOI{\tempurl}


\bibitem[\protect\citeauthoryear{Detlefs, Martin, Moir, and Jr.}{Detlefs
  et~al\mbox{.}}{2001}]%
        {DBLP:conf/podc/DetlefsMMS01}
\bibfield{author}{\bibinfo{person}{David Detlefs}, \bibinfo{person}{Paul~Alan
  Martin}, \bibinfo{person}{Mark Moir}, {and} \bibinfo{person}{Guy L.~Steele
  Jr.}} \bibinfo{year}{2001}\natexlab{}.
\newblock \showarticletitle{Lock-free reference counting}. In
  \bibinfo{booktitle}{\emph{{PODC}}}. \bibinfo{publisher}{{ACM}},
  \bibinfo{pages}{190--199}.
\newblock
\urldef\tempurl%
\url{https://doi.org/10.1145/383962.384016}
\showDOI{\tempurl}


\bibitem[\protect\citeauthoryear{Dice, Herlihy, and Kogan}{Dice
  et~al\mbox{.}}{2016}]%
        {DBLP:conf/iwmm/DiceHK16}
\bibfield{author}{\bibinfo{person}{Dave Dice}, \bibinfo{person}{Maurice
  Herlihy}, {and} \bibinfo{person}{Alex Kogan}.}
  \bibinfo{year}{2016}\natexlab{}.
\newblock \showarticletitle{Fast non-intrusive memory reclamation for
  highly-concurrent data structures}. In \bibinfo{booktitle}{\emph{{ISMM}}}.
  \bibinfo{publisher}{{ACM}}, \bibinfo{pages}{36--45}.
\newblock
\urldef\tempurl%
\url{https://doi.org/10.1145/2926697.2926699}
\showDOI{\tempurl}


\bibitem[\protect\citeauthoryear{Dodds, Haas, and Kirsch}{Dodds
  et~al\mbox{.}}{2015}]%
        {DBLP:conf/popl/DoddsHK15}
\bibfield{author}{\bibinfo{person}{Mike Dodds}, \bibinfo{person}{Andreas Haas},
  {and} \bibinfo{person}{Christoph~M. Kirsch}.}
  \bibinfo{year}{2015}\natexlab{}.
\newblock \showarticletitle{A Scalable, Correct Time-Stamped Stack}. In
  \bibinfo{booktitle}{\emph{{POPL}}}. \bibinfo{publisher}{{ACM}},
  \bibinfo{pages}{233--246}.
\newblock
\urldef\tempurl%
\url{https://doi.org/10.1145/2676726.2676963}
\showDOI{\tempurl}


\bibitem[\protect\citeauthoryear{Doherty, Detlefs, Groves, Flood, Luchangco,
  Martin, Moir, Shavit, and Jr.}{Doherty et~al\mbox{.}}{2004a}]%
        {DBLP:conf/spaa/DohertyDGFLMMSS04}
\bibfield{author}{\bibinfo{person}{Simon Doherty}, \bibinfo{person}{David
  Detlefs}, \bibinfo{person}{Lindsay Groves}, \bibinfo{person}{Christine~H.
  Flood}, \bibinfo{person}{Victor Luchangco}, \bibinfo{person}{Paul~Alan
  Martin}, \bibinfo{person}{Mark Moir}, \bibinfo{person}{Nir Shavit}, {and}
  \bibinfo{person}{Guy L.~Steele Jr.}} \bibinfo{year}{2004}\natexlab{a}.
\newblock \showarticletitle{{DCAS} is not a silver bullet for nonblocking
  algorithm design}. In \bibinfo{booktitle}{\emph{{SPAA}}}.
  \bibinfo{publisher}{{ACM}}, \bibinfo{pages}{216--224}.
\newblock
\urldef\tempurl%
\url{https://doi.org/10.1145/1007912.1007945}
\showDOI{\tempurl}


\bibitem[\protect\citeauthoryear{Doherty, Groves, Luchangco, and Moir}{Doherty
  et~al\mbox{.}}{2004b}]%
        {DBLP:conf/forte/DohertyGLM04}
\bibfield{author}{\bibinfo{person}{Simon Doherty}, \bibinfo{person}{Lindsay
  Groves}, \bibinfo{person}{Victor Luchangco}, {and} \bibinfo{person}{Mark
  Moir}.} \bibinfo{year}{2004}\natexlab{b}.
\newblock \showarticletitle{Formal Verification of a Practical Lock-Free Queue
  Algorithm}. In \bibinfo{booktitle}{\emph{{FORTE}}}
  \emph{(\bibinfo{series}{{LNCS}})}, Vol.~\bibinfo{volume}{3235}.
  \bibinfo{publisher}{Springer}, \bibinfo{pages}{97--114}.
\newblock
\urldef\tempurl%
\url{https://doi.org/10.1007/978-3-540-30232-2\_7}
\showDOI{\tempurl}


\bibitem[\protect\citeauthoryear{Doherty and Moir}{Doherty and Moir}{2009}]%
        {DBLP:conf/wdag/DohertyM09}
\bibfield{author}{\bibinfo{person}{Simon Doherty} {and} \bibinfo{person}{Mark
  Moir}.} \bibinfo{year}{2009}\natexlab{}.
\newblock \showarticletitle{Nonblocking Algorithms and Backward Simulation}. In
  \bibinfo{booktitle}{\emph{{DISC}}} \emph{(\bibinfo{series}{{LNCS}})},
  Vol.~\bibinfo{volume}{5805}. \bibinfo{publisher}{Springer},
  \bibinfo{pages}{274--288}.
\newblock
\urldef\tempurl%
\url{https://doi.org/10.1007/978-3-642-04355-0\_28}
\showDOI{\tempurl}


\bibitem[\protect\citeauthoryear{Dongol and Derrick}{Dongol and
  Derrick}{2014}]%
        {DBLP:journals/corr/DongolD14}
\bibfield{author}{\bibinfo{person}{Brijesh Dongol} {and} \bibinfo{person}{John
  Derrick}.} \bibinfo{year}{2014}\natexlab{}.
\newblock \showarticletitle{Verifying linearizability: {A} comparative survey}.
\newblock \bibinfo{journal}{\emph{CoRR}}  \bibinfo{volume}{abs/1410.6268}
  (\bibinfo{year}{2014}).
\newblock
\urldef\tempurl%
\url{http://arxiv.org/abs/1410.6268}
\showURL{%
\tempurl}


\bibitem[\protect\citeauthoryear{Dragojevic, Herlihy, Lev, and Moir}{Dragojevic
  et~al\mbox{.}}{2011}]%
        {DBLP:conf/podc/DragojevicHLM11}
\bibfield{author}{\bibinfo{person}{Aleksandar Dragojevic},
  \bibinfo{person}{Maurice Herlihy}, \bibinfo{person}{Yossi Lev}, {and}
  \bibinfo{person}{Mark Moir}.} \bibinfo{year}{2011}\natexlab{}.
\newblock \showarticletitle{On the power of hardware transactional memory to
  simplify memory management}. In \bibinfo{booktitle}{\emph{{PODC}}}.
  \bibinfo{publisher}{{ACM}}, \bibinfo{pages}{99--108}.
\newblock
\urldef\tempurl%
\url{https://doi.org/10.1145/1993806.1993821}
\showDOI{\tempurl}


\bibitem[\protect\citeauthoryear{Elmas, Qadeer, Sezgin, Subasi, and
  Tasiran}{Elmas et~al\mbox{.}}{2010}]%
        {DBLP:conf/tacas/ElmasQSST10}
\bibfield{author}{\bibinfo{person}{Tayfun Elmas}, \bibinfo{person}{Shaz
  Qadeer}, \bibinfo{person}{Ali Sezgin}, \bibinfo{person}{Omer Subasi}, {and}
  \bibinfo{person}{Serdar Tasiran}.} \bibinfo{year}{2010}\natexlab{}.
\newblock \showarticletitle{Simplifying Linearizability Proofs with Reduction
  and Abstraction}. In \bibinfo{booktitle}{\emph{{TACAS}}}
  \emph{(\bibinfo{series}{{LNCS}})}, Vol.~\bibinfo{volume}{6015}.
  \bibinfo{publisher}{Springer}, \bibinfo{pages}{296--311}.
\newblock
\urldef\tempurl%
\url{https://doi.org/10.1007/978-3-642-12002-2\_25}
\showDOI{\tempurl}


\bibitem[\protect\citeauthoryear{Emmi and Enea}{Emmi and Enea}{2018}]%
        {DBLP:journals/pacmpl/EmmiE18}
\bibfield{author}{\bibinfo{person}{Michael Emmi} {and}
  \bibinfo{person}{Constantin Enea}.} \bibinfo{year}{2018}\natexlab{}.
\newblock \showarticletitle{Sound, complete, and tractable linearizability
  monitoring for concurrent collections}.
\newblock \bibinfo{journal}{\emph{{PACMPL}}} \bibinfo{volume}{2},
  \bibinfo{number}{{POPL}} (\bibinfo{year}{2018}),
  \bibinfo{pages}{25:1--25:27}.
\newblock
\urldef\tempurl%
\url{https://doi.org/10.1145/3158113}
\showDOI{\tempurl}


\bibitem[\protect\citeauthoryear{Emmi, Enea, and Hamza}{Emmi
  et~al\mbox{.}}{2015}]%
        {DBLP:conf/pldi/EmmiEH15}
\bibfield{author}{\bibinfo{person}{Michael Emmi}, \bibinfo{person}{Constantin
  Enea}, {and} \bibinfo{person}{Jad Hamza}.} \bibinfo{year}{2015}\natexlab{}.
\newblock \showarticletitle{Monitoring refinement via symbolic reasoning}. In
  \bibinfo{booktitle}{\emph{{PLDI}}}. \bibinfo{publisher}{{ACM}},
  \bibinfo{pages}{260--269}.
\newblock
\urldef\tempurl%
\url{https://doi.org/10.1145/2737924.2737983}
\showDOI{\tempurl}


\bibitem[\protect\citeauthoryear{Fraser}{Fraser}{2004}]%
        {DBLP:phd/ethos/Fraser04}
\bibfield{author}{\bibinfo{person}{Keir Fraser}.}
  \bibinfo{year}{2004}\natexlab{}.
\newblock \emph{\bibinfo{title}{Practical lock-freedom}}.
\newblock \bibinfo{thesistype}{Ph.D. Dissertation}. \bibinfo{school}{University
  of Cambridge, {UK}}.
\newblock
\urldef\tempurl%
\url{http://ethos.bl.uk/OrderDetails.do?uin=uk.bl.ethos.599193}
\showURL{%
\tempurl}


\bibitem[\protect\citeauthoryear{Fu, Li, Feng, Shao, and Zhang}{Fu
  et~al\mbox{.}}{2010}]%
        {DBLP:conf/concur/FuLFSZ10}
\bibfield{author}{\bibinfo{person}{Ming Fu}, \bibinfo{person}{Yong Li},
  \bibinfo{person}{Xinyu Feng}, \bibinfo{person}{Zhong Shao}, {and}
  \bibinfo{person}{Yu Zhang}.} \bibinfo{year}{2010}\natexlab{}.
\newblock \showarticletitle{Reasoning about Optimistic Concurrency Using a
  Program Logic for History}. In \bibinfo{booktitle}{\emph{{CONCUR}}}
  \emph{(\bibinfo{series}{{LNCS}})}, Vol.~\bibinfo{volume}{6269}.
  \bibinfo{publisher}{Springer}, \bibinfo{pages}{388--402}.
\newblock
\urldef\tempurl%
\url{https://doi.org/10.1007/978-3-642-15375-4\_27}
\showDOI{\tempurl}


\bibitem[\protect\citeauthoryear{Gidenstam, Papatriantafilou, Sundell, and
  Tsigas}{Gidenstam et~al\mbox{.}}{2005}]%
        {DBLP:conf/ispan/GidenstamPST05}
\bibfield{author}{\bibinfo{person}{Anders Gidenstam}, \bibinfo{person}{Marina
  Papatriantafilou}, \bibinfo{person}{H{\aa}kan Sundell}, {and}
  \bibinfo{person}{Philippas Tsigas}.} \bibinfo{year}{2005}\natexlab{}.
\newblock \showarticletitle{Efficient and Reliable Lock-Free Memory Reclamation
  Based on Reference Counting}. In \bibinfo{booktitle}{\emph{{ISPAN}}}.
  \bibinfo{publisher}{{IEEE} Computer Society}, \bibinfo{pages}{202--207}.
\newblock
\urldef\tempurl%
\url{https://doi.org/10.1109/ISPAN.2005.42}
\showDOI{\tempurl}


\bibitem[\protect\citeauthoryear{Gotsman, Rinetzky, and Yang}{Gotsman
  et~al\mbox{.}}{2013}]%
        {DBLP:conf/esop/GotsmanRY13}
\bibfield{author}{\bibinfo{person}{Alexey Gotsman}, \bibinfo{person}{Noam
  Rinetzky}, {and} \bibinfo{person}{Hongseok Yang}.}
  \bibinfo{year}{2013}\natexlab{}.
\newblock \showarticletitle{Verifying Concurrent Memory Reclamation Algorithms
  with Grace}. In \bibinfo{booktitle}{\emph{{ESOP}}}
  \emph{(\bibinfo{series}{{LNCS}})}, Vol.~\bibinfo{volume}{7792}.
  \bibinfo{publisher}{Springer}, \bibinfo{pages}{249--269}.
\newblock
\urldef\tempurl%
\url{https://doi.org/10.1007/978-3-642-37036-6\_15}
\showDOI{\tempurl}


\bibitem[\protect\citeauthoryear{Groves}{Groves}{2007}]%
        {DBLP:conf/iceccs/Groves07}
\bibfield{author}{\bibinfo{person}{Lindsay Groves}.}
  \bibinfo{year}{2007}\natexlab{}.
\newblock \showarticletitle{Reasoning about Nonblocking Concurrency using
  Reduction}. In \bibinfo{booktitle}{\emph{{ICECCS}}}.
  \bibinfo{publisher}{{IEEE} Computer Society}, \bibinfo{pages}{107--116}.
\newblock
\urldef\tempurl%
\url{https://doi.org/10.1109/ICECCS.2007.39}
\showDOI{\tempurl}


\bibitem[\protect\citeauthoryear{Groves}{Groves}{2008}]%
        {DBLP:conf/cats/Groves08}
\bibfield{author}{\bibinfo{person}{Lindsay Groves}.}
  \bibinfo{year}{2008}\natexlab{}.
\newblock \showarticletitle{Verifying Michael and Scott's Lock-Free Queue
  Algorithm using Trace Reduction}. In \bibinfo{booktitle}{\emph{{CATS}}}
  \emph{(\bibinfo{series}{{CRPIT}})}, Vol.~\bibinfo{volume}{77}.
  \bibinfo{publisher}{Australian Computer Society}, \bibinfo{pages}{133--142}.
\newblock
\urldef\tempurl%
\url{http://crpit.com/abstracts/CRPITV77Groves.html}
\showURL{%
\tempurl}


\bibitem[\protect\citeauthoryear{Harris}{Harris}{2001}]%
        {DBLP:conf/wdag/Harris01}
\bibfield{author}{\bibinfo{person}{Timothy~L. Harris}.}
  \bibinfo{year}{2001}\natexlab{}.
\newblock \showarticletitle{A Pragmatic Implementation of Non-blocking
  Linked-Lists}. In \bibinfo{booktitle}{\emph{{DISC}}}
  \emph{(\bibinfo{series}{{LNCS}})}, Vol.~\bibinfo{volume}{2180}.
  \bibinfo{publisher}{Springer}, \bibinfo{pages}{300--314}.
\newblock
\urldef\tempurl%
\url{https://doi.org/10.1007/3-540-45414-4\_21}
\showDOI{\tempurl}


\bibitem[\protect\citeauthoryear{Hart, McKenney, Brown, and Walpole}{Hart
  et~al\mbox{.}}{2007}]%
        {DBLP:journals/jpdc/HartMBW07}
\bibfield{author}{\bibinfo{person}{Thomas~E. Hart}, \bibinfo{person}{Paul~E.
  McKenney}, \bibinfo{person}{Angela~Demke Brown}, {and}
  \bibinfo{person}{Jonathan Walpole}.} \bibinfo{year}{2007}\natexlab{}.
\newblock \showarticletitle{Performance of memory reclamation for lockless
  synchronization}.
\newblock \bibinfo{journal}{\emph{J. Parallel Distrib. Comput.}}
  \bibinfo{volume}{67}, \bibinfo{number}{12} (\bibinfo{year}{2007}),
  \bibinfo{pages}{1270--1285}.
\newblock
\urldef\tempurl%
\url{https://doi.org/10.1016/j.jpdc.2007.04.010}
\showDOI{\tempurl}


\bibitem[\protect\citeauthoryear{Haziza, Hol{\'{\i}}k, Meyer, and Wolff}{Haziza
  et~al\mbox{.}}{2016}]%
        {DBLP:conf/vmcai/HazizaHMW16}
\bibfield{author}{\bibinfo{person}{Fr{\'{e}}d{\'{e}}ric Haziza},
  \bibinfo{person}{Luk{\'{a}}s Hol{\'{\i}}k}, \bibinfo{person}{Roland Meyer},
  {and} \bibinfo{person}{Sebastian Wolff}.} \bibinfo{year}{2016}\natexlab{}.
\newblock \showarticletitle{Pointer Race Freedom}. In
  \bibinfo{booktitle}{\emph{{VMCAI}}} \emph{(\bibinfo{series}{{LNCS}})},
  Vol.~\bibinfo{volume}{9583}. \bibinfo{publisher}{Springer},
  \bibinfo{pages}{393--412}.
\newblock
\urldef\tempurl%
\url{https://doi.org/10.1007/978-3-662-49122-5\_19}
\showDOI{\tempurl}


\bibitem[\protect\citeauthoryear{Hemed, Rinetzky, and Vafeiadis}{Hemed
  et~al\mbox{.}}{2015}]%
        {DBLP:conf/wdag/HemedRV15}
\bibfield{author}{\bibinfo{person}{Nir Hemed}, \bibinfo{person}{Noam Rinetzky},
  {and} \bibinfo{person}{Viktor Vafeiadis}.} \bibinfo{year}{2015}\natexlab{}.
\newblock \showarticletitle{Modular Verification of Concurrency-Aware
  Linearizability}. In \bibinfo{booktitle}{\emph{{DISC}}}
  \emph{(\bibinfo{series}{{LNCS}})}, Vol.~\bibinfo{volume}{9363}.
  \bibinfo{publisher}{Springer}, \bibinfo{pages}{371--387}.
\newblock
\urldef\tempurl%
\url{https://doi.org/10.1007/978-3-662-48653-5\_25}
\showDOI{\tempurl}


\bibitem[\protect\citeauthoryear{Henzinger, Sezgin, and Vafeiadis}{Henzinger
  et~al\mbox{.}}{2013}]%
        {DBLP:conf/concur/HenzingerSV13}
\bibfield{author}{\bibinfo{person}{Thomas~A. Henzinger}, \bibinfo{person}{Ali
  Sezgin}, {and} \bibinfo{person}{Viktor Vafeiadis}.}
  \bibinfo{year}{2013}\natexlab{}.
\newblock \showarticletitle{Aspect-Oriented Linearizability Proofs}. In
  \bibinfo{booktitle}{\emph{{CONCUR}}} \emph{(\bibinfo{series}{{LNCS}})},
  Vol.~\bibinfo{volume}{8052}. \bibinfo{publisher}{Springer},
  \bibinfo{pages}{242--256}.
\newblock
\urldef\tempurl%
\url{https://doi.org/10.1007/978-3-642-40184-8\_18}
\showDOI{\tempurl}


\bibitem[\protect\citeauthoryear{Herlihy, Luchangco, Martin, and Moir}{Herlihy
  et~al\mbox{.}}{2005}]%
        {DBLP:journals/tocs/HerlihyLMM05}
\bibfield{author}{\bibinfo{person}{Maurice Herlihy}, \bibinfo{person}{Victor
  Luchangco}, \bibinfo{person}{Paul~A. Martin}, {and} \bibinfo{person}{Mark
  Moir}.} \bibinfo{year}{2005}\natexlab{}.
\newblock \showarticletitle{Nonblocking memory management support for
  dynamic-sized data structures}.
\newblock \bibinfo{journal}{\emph{{ACM} Trans. Comput. Syst.}}
  \bibinfo{volume}{23}, \bibinfo{number}{2} (\bibinfo{year}{2005}),
  \bibinfo{pages}{146--196}.
\newblock
\urldef\tempurl%
\url{https://doi.org/10.1145/1062247.1062249}
\showDOI{\tempurl}


\bibitem[\protect\citeauthoryear{Herlihy and Wing}{Herlihy and Wing}{1990}]%
        {DBLP:journals/toplas/HerlihyW90}
\bibfield{author}{\bibinfo{person}{Maurice Herlihy} {and}
  \bibinfo{person}{Jeannette~M. Wing}.} \bibinfo{year}{1990}\natexlab{}.
\newblock \showarticletitle{Linearizability: {A} Correctness Condition for
  Concurrent Objects}.
\newblock \bibinfo{journal}{\emph{{ACM} Trans. Program. Lang. Syst.}}
  \bibinfo{volume}{12}, \bibinfo{number}{3} (\bibinfo{year}{1990}),
  \bibinfo{pages}{463--492}.
\newblock
\urldef\tempurl%
\url{https://doi.org/10.1145/78969.78972}
\showDOI{\tempurl}


\bibitem[\protect\citeauthoryear{Hol{\'{\i}}k, Meyer, Vojnar, and
  Wolff}{Hol{\'{\i}}k et~al\mbox{.}}{2017}]%
        {DBLP:conf/sas/HolikMVW17}
\bibfield{author}{\bibinfo{person}{Luk{\'{a}}s Hol{\'{\i}}k},
  \bibinfo{person}{Roland Meyer}, \bibinfo{person}{Tom{\'{a}}s Vojnar}, {and}
  \bibinfo{person}{Sebastian Wolff}.} \bibinfo{year}{2017}\natexlab{}.
\newblock \showarticletitle{Effect Summaries for Thread-Modular Analysis -
  Sound Analysis Despite an Unsound Heuristic}. In
  \bibinfo{booktitle}{\emph{{SAS}}} \emph{(\bibinfo{series}{{LNCS}})},
  Vol.~\bibinfo{volume}{10422}. \bibinfo{publisher}{Springer},
  \bibinfo{pages}{169--191}.
\newblock
\urldef\tempurl%
\url{https://doi.org/10.1007/978-3-319-66706-5\_9}
\showDOI{\tempurl}


\bibitem[\protect\citeauthoryear{Horn and Kroening}{Horn and Kroening}{2015}]%
        {DBLP:conf/forte/HornK15a}
\bibfield{author}{\bibinfo{person}{Alex Horn} {and} \bibinfo{person}{Daniel
  Kroening}.} \bibinfo{year}{2015}\natexlab{}.
\newblock \showarticletitle{Faster Linearizability Checking via
  P-Compositionality}. In \bibinfo{booktitle}{\emph{{FORTE}}}
  \emph{(\bibinfo{series}{{LNCS}})}, Vol.~\bibinfo{volume}{9039}.
  \bibinfo{publisher}{Springer}, \bibinfo{pages}{50--65}.
\newblock
\urldef\tempurl%
\url{https://doi.org/10.1007/978-3-319-19195-9\_4}
\showDOI{\tempurl}


\bibitem[\protect\citeauthoryear{Jones}{Jones}{1983}]%
        {DBLP:journals/toplas/Jones83}
\bibfield{author}{\bibinfo{person}{Cliff~B. Jones}.}
  \bibinfo{year}{1983}\natexlab{}.
\newblock \showarticletitle{Tentative Steps Toward a Development Method for
  Interfering Programs}.
\newblock \bibinfo{journal}{\emph{{ACM} Trans. Program. Lang. Syst.}}
  \bibinfo{volume}{5}, \bibinfo{number}{4} (\bibinfo{year}{1983}),
  \bibinfo{pages}{596--619}.
\newblock
\urldef\tempurl%
\url{https://doi.org/10.1145/69575.69577}
\showDOI{\tempurl}


\bibitem[\protect\citeauthoryear{Jonsson}{Jonsson}{2012}]%
        {DBLP:journals/fac/Jonsson12}
\bibfield{author}{\bibinfo{person}{Bengt Jonsson}.}
  \bibinfo{year}{2012}\natexlab{}.
\newblock \showarticletitle{Using refinement calculus techniques to prove
  linearizability}.
\newblock \bibinfo{journal}{\emph{Formal Asp. Comput.}} \bibinfo{volume}{24},
  \bibinfo{number}{4-6} (\bibinfo{year}{2012}), \bibinfo{pages}{537--554}.
\newblock
\urldef\tempurl%
\url{https://doi.org/10.1007/s00165-012-0250-7}
\showDOI{\tempurl}


\bibitem[\protect\citeauthoryear{Khyzha, Dodds, Gotsman, and Parkinson}{Khyzha
  et~al\mbox{.}}{2017}]%
        {DBLP:conf/esop/KhyzhaDGP17}
\bibfield{author}{\bibinfo{person}{Artem Khyzha}, \bibinfo{person}{Mike Dodds},
  \bibinfo{person}{Alexey Gotsman}, {and} \bibinfo{person}{Matthew~J.
  Parkinson}.} \bibinfo{year}{2017}\natexlab{}.
\newblock \showarticletitle{Proving Linearizability Using Partial Orders}. In
  \bibinfo{booktitle}{\emph{{ESOP}}} \emph{(\bibinfo{series}{{LNCS}})},
  Vol.~\bibinfo{volume}{10201}. \bibinfo{publisher}{Springer},
  \bibinfo{pages}{639--667}.
\newblock
\urldef\tempurl%
\url{https://doi.org/10.1007/978-3-662-54434-1\_24}
\showDOI{\tempurl}


\bibitem[\protect\citeauthoryear{Kokologiannakis and Sagonas}{Kokologiannakis
  and Sagonas}{2017}]%
        {DBLP:conf/spin/Kokologiannakis17}
\bibfield{author}{\bibinfo{person}{Michalis Kokologiannakis} {and}
  \bibinfo{person}{Konstantinos Sagonas}.} \bibinfo{year}{2017}\natexlab{}.
\newblock \showarticletitle{Stateless model checking of the Linux kernel's
  hierarchical read-copy-update (tree {RCU)}}. In
  \bibinfo{booktitle}{\emph{{SPIN}}}. \bibinfo{publisher}{{ACM}},
  \bibinfo{pages}{172--181}.
\newblock
\urldef\tempurl%
\url{https://doi.org/10.1145/3092282.3092287}
\showDOI{\tempurl}


\bibitem[\protect\citeauthoryear{Krishna, Shasha, and Wies}{Krishna
  et~al\mbox{.}}{2018}]%
        {DBLP:journals/pacmpl/KrishnaSW18}
\bibfield{author}{\bibinfo{person}{Siddharth Krishna},
  \bibinfo{person}{Dennis~E. Shasha}, {and} \bibinfo{person}{Thomas Wies}.}
  \bibinfo{year}{2018}\natexlab{}.
\newblock \showarticletitle{Go with the flow: compositional abstractions for
  concurrent data structures}.
\newblock \bibinfo{journal}{\emph{{PACMPL}}} \bibinfo{volume}{2},
  \bibinfo{number}{{POPL}} (\bibinfo{year}{2018}),
  \bibinfo{pages}{37:1--37:31}.
\newblock
\urldef\tempurl%
\url{https://doi.org/10.1145/3158125}
\showDOI{\tempurl}


\bibitem[\protect\citeauthoryear{Liang and Feng}{Liang and Feng}{2013}]%
        {DBLP:conf/pldi/LiangF13}
\bibfield{author}{\bibinfo{person}{Hongjin Liang} {and} \bibinfo{person}{Xinyu
  Feng}.} \bibinfo{year}{2013}\natexlab{}.
\newblock \showarticletitle{Modular verification of linearizability with
  non-fixed linearization points}. In \bibinfo{booktitle}{\emph{{PLDI}}}.
  \bibinfo{publisher}{{ACM}}, \bibinfo{pages}{459--470}.
\newblock
\urldef\tempurl%
\url{https://doi.org/10.1145/2462156.2462189}
\showDOI{\tempurl}


\bibitem[\protect\citeauthoryear{Liang, Feng, and Fu}{Liang
  et~al\mbox{.}}{2012}]%
        {DBLP:conf/popl/LiangFF12}
\bibfield{author}{\bibinfo{person}{Hongjin Liang}, \bibinfo{person}{Xinyu
  Feng}, {and} \bibinfo{person}{Ming Fu}.} \bibinfo{year}{2012}\natexlab{}.
\newblock \showarticletitle{A rely-guarantee-based simulation for verifying
  concurrent program transformations}. In \bibinfo{booktitle}{\emph{{POPL}}}.
  \bibinfo{publisher}{{ACM}}, \bibinfo{pages}{455--468}.
\newblock
\urldef\tempurl%
\url{https://doi.org/10.1145/2103656.2103711}
\showDOI{\tempurl}


\bibitem[\protect\citeauthoryear{Liang, Feng, and Fu}{Liang
  et~al\mbox{.}}{2014}]%
        {DBLP:journals/toplas/LiangFF14}
\bibfield{author}{\bibinfo{person}{Hongjin Liang}, \bibinfo{person}{Xinyu
  Feng}, {and} \bibinfo{person}{Ming Fu}.} \bibinfo{year}{2014}\natexlab{}.
\newblock \showarticletitle{Rely-Guarantee-Based Simulation for Compositional
  Verification of Concurrent Program Transformations}.
\newblock \bibinfo{journal}{\emph{{ACM} Trans. Program. Lang. Syst.}}
  \bibinfo{volume}{36}, \bibinfo{number}{1} (\bibinfo{year}{2014}),
  \bibinfo{pages}{3:1--3:55}.
\newblock
\urldef\tempurl%
\url{https://doi.org/10.1145/2576235}
\showDOI{\tempurl}


\bibitem[\protect\citeauthoryear{Liang, McKenney, Kroening, and Melham}{Liang
  et~al\mbox{.}}{2018}]%
        {DBLP:conf/date/LiangMKM18}
\bibfield{author}{\bibinfo{person}{Lihao Liang}, \bibinfo{person}{Paul~E.
  McKenney}, \bibinfo{person}{Daniel Kroening}, {and} \bibinfo{person}{Tom
  Melham}.} \bibinfo{year}{2018}\natexlab{}.
\newblock \showarticletitle{Verification of tree-based hierarchical read-copy
  update in the Linux kernel}. In \bibinfo{booktitle}{\emph{{DATE}}}.
  \bibinfo{publisher}{{IEEE}}, \bibinfo{pages}{61--66}.
\newblock
\urldef\tempurl%
\url{https://doi.org/10.23919/DATE.2018.8341980}
\showDOI{\tempurl}


\bibitem[\protect\citeauthoryear{Liu, Chen, Liu, and Sun}{Liu
  et~al\mbox{.}}{2009}]%
        {DBLP:conf/fm/LiuCLS09}
\bibfield{author}{\bibinfo{person}{Yang Liu}, \bibinfo{person}{Wei Chen},
  \bibinfo{person}{Yanhong~A. Liu}, {and} \bibinfo{person}{Jun Sun}.}
  \bibinfo{year}{2009}\natexlab{}.
\newblock \showarticletitle{Model Checking Linearizability via Refinement}. In
  \bibinfo{booktitle}{\emph{{FM}}} \emph{(\bibinfo{series}{{LNCS}})},
  Vol.~\bibinfo{volume}{5850}. \bibinfo{publisher}{Springer},
  \bibinfo{pages}{321--337}.
\newblock
\urldef\tempurl%
\url{https://doi.org/10.1007/978-3-642-05089-3\_21}
\showDOI{\tempurl}


\bibitem[\protect\citeauthoryear{Liu, Chen, Liu, Sun, Zhang, and Dong}{Liu
  et~al\mbox{.}}{2013}]%
        {DBLP:journals/tse/Liu0L0ZD13}
\bibfield{author}{\bibinfo{person}{Yang Liu}, \bibinfo{person}{Wei Chen},
  \bibinfo{person}{Yanhong~A. Liu}, \bibinfo{person}{Jun Sun},
  \bibinfo{person}{Shao~Jie Zhang}, {and} \bibinfo{person}{Jin~Song Dong}.}
  \bibinfo{year}{2013}\natexlab{}.
\newblock \showarticletitle{Verifying Linearizability via Optimized Refinement
  Checking}.
\newblock \bibinfo{journal}{\emph{{IEEE} Trans. Software Eng.}}
  \bibinfo{volume}{39}, \bibinfo{number}{7} (\bibinfo{year}{2013}),
  \bibinfo{pages}{1018--1039}.
\newblock
\urldef\tempurl%
\url{https://doi.org/10.1109/TSE.2012.82}
\showDOI{\tempurl}


\bibitem[\protect\citeauthoryear{Lowe}{Lowe}{2017}]%
        {DBLP:journals/concurrency/Lowe17}
\bibfield{author}{\bibinfo{person}{Gavin Lowe}.}
  \bibinfo{year}{2017}\natexlab{}.
\newblock \showarticletitle{Testing for linearizability}.
\newblock \bibinfo{journal}{\emph{Concurrency and Computation: Practice and
  Experience}} \bibinfo{volume}{29}, \bibinfo{number}{4}
  (\bibinfo{year}{2017}).
\newblock
\urldef\tempurl%
\url{https://doi.org/10.1002/cpe.3928}
\showDOI{\tempurl}


\bibitem[\protect\citeauthoryear{McKenney}{McKenney}{2004}]%
        {phd/McKenney04}
\bibfield{author}{\bibinfo{person}{Paul~E. McKenney}.}
  \bibinfo{year}{2004}\natexlab{}.
\newblock \emph{\bibinfo{title}{Exploiting Deferred Destruction: an Analysis of
  Read-Copy-Update Techniques in Operating System Kernels}}.
\newblock \bibinfo{thesistype}{Ph.D. Dissertation}. \bibinfo{school}{Oregon
  Health \& Science University}.
\newblock


\bibitem[\protect\citeauthoryear{McKenney and Slingwine}{McKenney and
  Slingwine}{1998}]%
        {McKenney1998ReadcopyUU}
\bibfield{author}{\bibinfo{person}{Paul~E. McKenney} {and}
  \bibinfo{person}{John~D. Slingwine}.} \bibinfo{year}{1998}\natexlab{}.
\newblock \showarticletitle{Read-copy Update: Using Execution History to Solve
  Concurrency Problems}.
\newblock


\bibitem[\protect\citeauthoryear{Meyer and Wolff}{Meyer and Wolff}{2018}]%
        {popl}
\bibfield{author}{\bibinfo{person}{Roland Meyer} {and}
  \bibinfo{person}{Sebastian Wolff}.} \bibinfo{year}{2018}\natexlab{}.
\newblock \showarticletitle{Decoupling Lock-Free Data Structures from Memory
  Reclamation for Static Analysis}.
\newblock \bibinfo{journal}{\emph{{PACMPL}}} \bibinfo{volume}{3},
  \bibinfo{number}{{POPL}} (\bibinfo{year}{2018}),
  \bibinfo{pages}{58:1--58:31}.
\newblock
\urldef\tempurl%
\url{https://doi.org/10.1145/3290371}
\showDOI{\tempurl}


\bibitem[\protect\citeauthoryear{Michael}{Michael}{2002}]%
        {DBLP:conf/podc/Michael02}
\bibfield{author}{\bibinfo{person}{Maged~M. Michael}.}
  \bibinfo{year}{2002}\natexlab{}.
\newblock \showarticletitle{Safe memory reclamation for dynamic lock-free
  objects using atomic reads and writes}. In
  \bibinfo{booktitle}{\emph{{PODC}}}. \bibinfo{publisher}{{ACM}},
  \bibinfo{pages}{21--30}.
\newblock
\urldef\tempurl%
\url{https://doi.org/10.1145/571825.571829}
\showDOI{\tempurl}


\bibitem[\protect\citeauthoryear{Michael}{Michael}{2004}]%
        {DBLP:journals/tpds/Michael04}
\bibfield{author}{\bibinfo{person}{Maged~M. Michael}.}
  \bibinfo{year}{2004}\natexlab{}.
\newblock \showarticletitle{Hazard Pointers: Safe Memory Reclamation for
  Lock-Free Objects}.
\newblock \bibinfo{journal}{\emph{{IEEE} Trans. Parallel Distrib. Syst.}}
  \bibinfo{volume}{15}, \bibinfo{number}{6} (\bibinfo{year}{2004}),
  \bibinfo{pages}{491--504}.
\newblock
\urldef\tempurl%
\url{https://doi.org/10.1109/TPDS.2004.8}
\showDOI{\tempurl}


\bibitem[\protect\citeauthoryear{Michael and Scott}{Michael and Scott}{1995}]%
        {Michael:1995:CMM:898203}
\bibfield{author}{\bibinfo{person}{Maged~M. Michael} {and}
  \bibinfo{person}{Michael~L. Scott}.} \bibinfo{year}{1995}\natexlab{}.
\newblock \bibinfo{booktitle}{\emph{Correction of a Memory Management Method
  for Lock-Free Data Structures}}.
\newblock \bibinfo{type}{{T}echnical {R}eport}. \bibinfo{address}{Rochester,
  NY, USA}.
\newblock


\bibitem[\protect\citeauthoryear{Michael and Scott}{Michael and Scott}{1996}]%
        {DBLP:conf/podc/MichaelS96}
\bibfield{author}{\bibinfo{person}{Maged~M. Michael} {and}
  \bibinfo{person}{Michael~L. Scott}.} \bibinfo{year}{1996}\natexlab{}.
\newblock \showarticletitle{Simple, Fast, and Practical Non-Blocking and
  Blocking Concurrent Queue Algorithms}. In \bibinfo{booktitle}{\emph{{PODC}}}.
  \bibinfo{publisher}{{ACM}}, \bibinfo{pages}{267--275}.
\newblock
\urldef\tempurl%
\url{https://doi.org/10.1145/248052.248106}
\showDOI{\tempurl}


\bibitem[\protect\citeauthoryear{O'Hearn, Rinetzky, Vechev, Yahav, and
  Yorsh}{O'Hearn et~al\mbox{.}}{2010}]%
        {DBLP:conf/podc/OHearnRVYY10}
\bibfield{author}{\bibinfo{person}{Peter~W. O'Hearn}, \bibinfo{person}{Noam
  Rinetzky}, \bibinfo{person}{Martin~T. Vechev}, \bibinfo{person}{Eran Yahav},
  {and} \bibinfo{person}{Greta Yorsh}.} \bibinfo{year}{2010}\natexlab{}.
\newblock \showarticletitle{Verifying linearizability with hindsight}. In
  \bibinfo{booktitle}{\emph{{PODC}}}. \bibinfo{publisher}{{ACM}},
  \bibinfo{pages}{85--94}.
\newblock
\urldef\tempurl%
\url{https://doi.org/10.1145/1835698.1835722}
\showDOI{\tempurl}


\bibitem[\protect\citeauthoryear{Parkinson, Bornat, and O'Hearn}{Parkinson
  et~al\mbox{.}}{2007}]%
        {DBLP:conf/popl/ParkinsonBO07}
\bibfield{author}{\bibinfo{person}{Matthew~J. Parkinson},
  \bibinfo{person}{Richard Bornat}, {and} \bibinfo{person}{Peter~W. O'Hearn}.}
  \bibinfo{year}{2007}\natexlab{}.
\newblock \showarticletitle{Modular verification of a non-blocking stack}. In
  \bibinfo{booktitle}{\emph{{POPL}}}. \bibinfo{publisher}{{ACM}},
  \bibinfo{pages}{297--302}.
\newblock
\urldef\tempurl%
\url{https://doi.org/10.1145/1190216.1190261}
\showDOI{\tempurl}


\bibitem[\protect\citeauthoryear{Ramalhete and Correia}{Ramalhete and
  Correia}{2017}]%
        {DBLP:conf/spaa/RamalheteC17}
\bibfield{author}{\bibinfo{person}{Pedro Ramalhete} {and}
  \bibinfo{person}{Andreia Correia}.} \bibinfo{year}{2017}\natexlab{}.
\newblock \showarticletitle{Brief Announcement: Hazard Eras - Non-Blocking
  Memory Reclamation}. In \bibinfo{booktitle}{\emph{{SPAA}}}.
  \bibinfo{publisher}{{ACM}}, \bibinfo{pages}{367--369}.
\newblock
\urldef\tempurl%
\url{https://doi.org/10.1145/3087556.3087588}
\showDOI{\tempurl}


\bibitem[\protect\citeauthoryear{Schellhorn, Wehrheim, and Derrick}{Schellhorn
  et~al\mbox{.}}{2012}]%
        {DBLP:conf/cav/SchellhornWD12}
\bibfield{author}{\bibinfo{person}{Gerhard Schellhorn}, \bibinfo{person}{Heike
  Wehrheim}, {and} \bibinfo{person}{John Derrick}.}
  \bibinfo{year}{2012}\natexlab{}.
\newblock \showarticletitle{How to Prove Algorithms Linearisable}. In
  \bibinfo{booktitle}{\emph{{CAV}}} \emph{(\bibinfo{series}{{LNCS}})},
  Vol.~\bibinfo{volume}{7358}. \bibinfo{publisher}{Springer},
  \bibinfo{pages}{243--259}.
\newblock
\urldef\tempurl%
\url{https://doi.org/10.1007/978-3-642-31424-7\_21}
\showDOI{\tempurl}


\bibitem[\protect\citeauthoryear{Segalov, Lev{-}Ami, Manevich, Ramalingam, and
  Sagiv}{Segalov et~al\mbox{.}}{2009}]%
        {DBLP:conf/aplas/SegalovLMGS09}
\bibfield{author}{\bibinfo{person}{Michal Segalov}, \bibinfo{person}{Tal
  Lev{-}Ami}, \bibinfo{person}{Roman Manevich}, \bibinfo{person}{Ganesan
  Ramalingam}, {and} \bibinfo{person}{Mooly Sagiv}.}
  \bibinfo{year}{2009}\natexlab{}.
\newblock \showarticletitle{Abstract Transformers for Thread Correlation
  Analysis}. In \bibinfo{booktitle}{\emph{{APLAS}}}
  \emph{(\bibinfo{series}{{LNCS}})}, Vol.~\bibinfo{volume}{5904}.
  \bibinfo{publisher}{Springer}, \bibinfo{pages}{30--46}.
\newblock
\urldef\tempurl%
\url{https://doi.org/10.1007/978-3-642-10672-9\_5}
\showDOI{\tempurl}


\bibitem[\protect\citeauthoryear{Sergey, Nanevski, and Banerjee}{Sergey
  et~al\mbox{.}}{2015a}]%
        {DBLP:conf/pldi/SergeyNB15}
\bibfield{author}{\bibinfo{person}{Ilya Sergey}, \bibinfo{person}{Aleksandar
  Nanevski}, {and} \bibinfo{person}{Anindya Banerjee}.}
  \bibinfo{year}{2015}\natexlab{a}.
\newblock \showarticletitle{Mechanized verification of fine-grained concurrent
  programs}. In \bibinfo{booktitle}{\emph{{PLDI}}}. \bibinfo{publisher}{{ACM}},
  \bibinfo{pages}{77--87}.
\newblock
\urldef\tempurl%
\url{https://doi.org/10.1145/2737924.2737964}
\showDOI{\tempurl}


\bibitem[\protect\citeauthoryear{Sergey, Nanevski, and Banerjee}{Sergey
  et~al\mbox{.}}{2015b}]%
        {DBLP:conf/esop/SergeyNB15}
\bibfield{author}{\bibinfo{person}{Ilya Sergey}, \bibinfo{person}{Aleksandar
  Nanevski}, {and} \bibinfo{person}{Anindya Banerjee}.}
  \bibinfo{year}{2015}\natexlab{b}.
\newblock \showarticletitle{Specifying and Verifying Concurrent Algorithms with
  Histories and Subjectivity}. In \bibinfo{booktitle}{\emph{{ESOP}}}
  \emph{(\bibinfo{series}{{LNCS}})}, Vol.~\bibinfo{volume}{9032}.
  \bibinfo{publisher}{Springer}, \bibinfo{pages}{333--358}.
\newblock
\urldef\tempurl%
\url{https://doi.org/10.1007/978-3-662-46669-8\_14}
\showDOI{\tempurl}


\bibitem[\protect\citeauthoryear{Sethi, Talupur, and Malik}{Sethi
  et~al\mbox{.}}{2013}]%
        {DBLP:conf/spin/SethiTM13}
\bibfield{author}{\bibinfo{person}{Divjyot Sethi}, \bibinfo{person}{Muralidhar
  Talupur}, {and} \bibinfo{person}{Sharad Malik}.}
  \bibinfo{year}{2013}\natexlab{}.
\newblock \showarticletitle{Model Checking Unbounded Concurrent Lists}. In
  \bibinfo{booktitle}{\emph{{SPIN}}} \emph{(\bibinfo{series}{{LNCS}})},
  Vol.~\bibinfo{volume}{7976}. \bibinfo{publisher}{Springer},
  \bibinfo{pages}{320--340}.
\newblock
\urldef\tempurl%
\url{https://doi.org/10.1007/978-3-642-39176-7\_20}
\showDOI{\tempurl}


\bibitem[\protect\citeauthoryear{Tassarotti, Dreyer, and Vafeiadis}{Tassarotti
  et~al\mbox{.}}{2015}]%
        {DBLP:conf/pldi/TassarottiDV15}
\bibfield{author}{\bibinfo{person}{Joseph Tassarotti}, \bibinfo{person}{Derek
  Dreyer}, {and} \bibinfo{person}{Viktor Vafeiadis}.}
  \bibinfo{year}{2015}\natexlab{}.
\newblock \showarticletitle{Verifying read-copy-update in a logic for weak
  memory}. In \bibinfo{booktitle}{\emph{{PLDI}}}. \bibinfo{publisher}{{ACM}},
  \bibinfo{pages}{110--120}.
\newblock
\urldef\tempurl%
\url{https://doi.org/10.1145/2737924.2737992}
\showDOI{\tempurl}


\bibitem[\protect\citeauthoryear{Tofan, Schellhorn, and Reif}{Tofan
  et~al\mbox{.}}{2011}]%
        {DBLP:conf/ictac/TofanSR11}
\bibfield{author}{\bibinfo{person}{Bogdan Tofan}, \bibinfo{person}{Gerhard
  Schellhorn}, {and} \bibinfo{person}{Wolfgang Reif}.}
  \bibinfo{year}{2011}\natexlab{}.
\newblock \showarticletitle{Formal Verification of a Lock-Free Stack with
  Hazard Pointers}. In \bibinfo{booktitle}{\emph{{ICTAC}}}
  \emph{(\bibinfo{series}{{LNCS}})}, Vol.~\bibinfo{volume}{6916}.
  \bibinfo{publisher}{Springer}, \bibinfo{pages}{239--255}.
\newblock
\urldef\tempurl%
\url{https://doi.org/10.1007/978-3-642-23283-1\_16}
\showDOI{\tempurl}


\bibitem[\protect\citeauthoryear{Travkin, M{\"{u}}tze, and Wehrheim}{Travkin
  et~al\mbox{.}}{2013}]%
        {DBLP:conf/hvc/TravkinMW13}
\bibfield{author}{\bibinfo{person}{Oleg Travkin}, \bibinfo{person}{Annika
  M{\"{u}}tze}, {and} \bibinfo{person}{Heike Wehrheim}.}
  \bibinfo{year}{2013}\natexlab{}.
\newblock \showarticletitle{{SPIN} as a Linearizability Checker under Weak
  Memory Models}. In \bibinfo{booktitle}{\emph{Haifa Verification Conference}}
  \emph{(\bibinfo{series}{{LNCS}})}, Vol.~\bibinfo{volume}{8244}.
  \bibinfo{publisher}{Springer}, \bibinfo{pages}{311--326}.
\newblock
\urldef\tempurl%
\url{https://doi.org/10.1007/978-3-319-03077-7\_21}
\showDOI{\tempurl}


\bibitem[\protect\citeauthoryear{Treiber}{Treiber}{1986}]%
        {opac-b1015261}
\bibfield{author}{\bibinfo{person}{R.Kent Treiber}.}
  \bibinfo{year}{1986}\natexlab{}.
\newblock \bibinfo{booktitle}{\emph{Systems programming: coping with
  parallelism}}.
\newblock \bibinfo{type}{{T}echnical {R}eport} RJ 5118.
  \bibinfo{institution}{IBM}.
\newblock


\bibitem[\protect\citeauthoryear{Vafeiadis}{Vafeiadis}{2010a}]%
        {DBLP:conf/cav/Vafeiadis10}
\bibfield{author}{\bibinfo{person}{Viktor Vafeiadis}.}
  \bibinfo{year}{2010}\natexlab{a}.
\newblock \showarticletitle{Automatically Proving Linearizability}. In
  \bibinfo{booktitle}{\emph{{CAV}}} \emph{(\bibinfo{series}{{LNCS}})},
  Vol.~\bibinfo{volume}{6174}. \bibinfo{publisher}{Springer},
  \bibinfo{pages}{450--464}.
\newblock
\urldef\tempurl%
\url{https://doi.org/10.1007/978-3-642-14295-6\_40}
\showDOI{\tempurl}


\bibitem[\protect\citeauthoryear{Vafeiadis}{Vafeiadis}{2010b}]%
        {DBLP:conf/vmcai/Vafeiadis10}
\bibfield{author}{\bibinfo{person}{Viktor Vafeiadis}.}
  \bibinfo{year}{2010}\natexlab{b}.
\newblock \showarticletitle{RGSep Action Inference}. In
  \bibinfo{booktitle}{\emph{{VMCAI}}} \emph{(\bibinfo{series}{{LNCS}})},
  Vol.~\bibinfo{volume}{5944}. \bibinfo{publisher}{Springer},
  \bibinfo{pages}{345--361}.
\newblock
\urldef\tempurl%
\url{https://doi.org/10.1007/978-3-642-11319-2\_25}
\showDOI{\tempurl}


\bibitem[\protect\citeauthoryear{Vardi}{Vardi}{1987}]%
        {DBLP:conf/lics/Vardi87}
\bibfield{author}{\bibinfo{person}{Moshe~Y. Vardi}.}
  \bibinfo{year}{1987}\natexlab{}.
\newblock \showarticletitle{Verification of Concurrent Programs: The
  Automata-Theoretic Framework}. In \bibinfo{booktitle}{\emph{{LICS}}}.
  \bibinfo{publisher}{{IEEE} Computer Society}, \bibinfo{pages}{167--176}.
\newblock


\bibitem[\protect\citeauthoryear{Vechev and Yahav}{Vechev and Yahav}{2008}]%
        {DBLP:conf/pldi/VechevY08}
\bibfield{author}{\bibinfo{person}{Martin~T. Vechev} {and}
  \bibinfo{person}{Eran Yahav}.} \bibinfo{year}{2008}\natexlab{}.
\newblock \showarticletitle{Deriving linearizable fine-grained concurrent
  objects}. In \bibinfo{booktitle}{\emph{{PLDI}}}. \bibinfo{publisher}{{ACM}},
  \bibinfo{pages}{125--135}.
\newblock
\urldef\tempurl%
\url{https://doi.org/10.1145/1375581.1375598}
\showDOI{\tempurl}


\bibitem[\protect\citeauthoryear{Vechev, Yahav, and Yorsh}{Vechev
  et~al\mbox{.}}{2009}]%
        {DBLP:conf/spin/VechevYY09}
\bibfield{author}{\bibinfo{person}{Martin~T. Vechev}, \bibinfo{person}{Eran
  Yahav}, {and} \bibinfo{person}{Greta Yorsh}.}
  \bibinfo{year}{2009}\natexlab{}.
\newblock \showarticletitle{Experience with Model Checking Linearizability}. In
  \bibinfo{booktitle}{\emph{{SPIN}}} \emph{(\bibinfo{series}{{LNCS}})},
  Vol.~\bibinfo{volume}{5578}. \bibinfo{publisher}{Springer},
  \bibinfo{pages}{261--278}.
\newblock
\urldef\tempurl%
\url{https://doi.org/10.1007/978-3-642-02652-2\_21}
\showDOI{\tempurl}


\bibitem[\protect\citeauthoryear{Wen, Izraelevitz, Cai, Beadle, and Scott}{Wen
  et~al\mbox{.}}{2018}]%
        {DBLP:conf/ppopp/WenICBS18}
\bibfield{author}{\bibinfo{person}{Haosen Wen}, \bibinfo{person}{Joseph
  Izraelevitz}, \bibinfo{person}{Wentao Cai}, \bibinfo{person}{H.~Alan Beadle},
  {and} \bibinfo{person}{Michael~L. Scott}.} \bibinfo{year}{2018}\natexlab{}.
\newblock \showarticletitle{Interval-based memory reclamation}. In
  \bibinfo{booktitle}{\emph{{PPOPP}}}. \bibinfo{publisher}{{ACM}},
  \bibinfo{pages}{1--13}.
\newblock
\urldef\tempurl%
\url{https://doi.org/10.1145/3178487.3178488}
\showDOI{\tempurl}


\bibitem[\protect\citeauthoryear{Yang and Wrigstad}{Yang and Wrigstad}{2017}]%
        {DBLP:conf/iwmm/YangW17}
\bibfield{author}{\bibinfo{person}{Albert~Mingkun Yang} {and}
  \bibinfo{person}{Tobias Wrigstad}.} \bibinfo{year}{2017}\natexlab{}.
\newblock \showarticletitle{Type-assisted automatic garbage collection for
  lock-free data structures}. In \bibinfo{booktitle}{\emph{{ISMM}}}.
  \bibinfo{publisher}{{ACM}}, \bibinfo{pages}{14--24}.
\newblock
\urldef\tempurl%
\url{https://doi.org/10.1145/3092255.3092274}
\showDOI{\tempurl}


\bibitem[\protect\citeauthoryear{Yang, Katoen, Lin, and Wu}{Yang
  et~al\mbox{.}}{2017}]%
        {DBLP:journals/corr/YangKLW17}
\bibfield{author}{\bibinfo{person}{Xiaoxiao Yang},
  \bibinfo{person}{Joost{-}Pieter Katoen}, \bibinfo{person}{Huimin Lin}, {and}
  \bibinfo{person}{Hao Wu}.} \bibinfo{year}{2017}\natexlab{}.
\newblock \showarticletitle{Verifying Concurrent Stacks by Divergence-Sensitive
  Bisimulation}.
\newblock \bibinfo{journal}{\emph{CoRR}}  \bibinfo{volume}{abs/1701.06104}
  (\bibinfo{year}{2017}).
\newblock
\urldef\tempurl%
\url{http://arxiv.org/abs/1701.06104}
\showURL{%
\tempurl}


\bibitem[\protect\citeauthoryear{Zhang}{Zhang}{2011}]%
        {DBLP:conf/icse/Zhang11a}
\bibfield{author}{\bibinfo{person}{Shao~Jie Zhang}.}
  \bibinfo{year}{2011}\natexlab{}.
\newblock \showarticletitle{Scalable automatic linearizability checking}. In
  \bibinfo{booktitle}{\emph{{ICSE}}}. \bibinfo{publisher}{{ACM}},
  \bibinfo{pages}{1185--1187}.
\newblock
\urldef\tempurl%
\url{https://doi.org/10.1145/1985793.1986037}
\showDOI{\tempurl}


\bibitem[\protect\citeauthoryear{Zhu, Petri, and Jagannathan}{Zhu
  et~al\mbox{.}}{2015}]%
        {DBLP:conf/cav/ZhuPJ15}
\bibfield{author}{\bibinfo{person}{He Zhu}, \bibinfo{person}{Gustavo Petri},
  {and} \bibinfo{person}{Suresh Jagannathan}.} \bibinfo{year}{2015}\natexlab{}.
\newblock \showarticletitle{Poling: {SMT} Aided Linearizability Proofs}. In
  \bibinfo{booktitle}{\emph{{CAV} {(2)}}} \emph{(\bibinfo{series}{{LNCS}})},
  Vol.~\bibinfo{volume}{9207}. \bibinfo{publisher}{Springer},
  \bibinfo{pages}{3--19}.
\newblock
\urldef\tempurl%
\url{https://doi.org/10.1007/978-3-319-21668-3\_1}
\showDOI{\tempurl}


\end{thebibliography}
